\documentclass[reqno,11pt,letterpaper]{amsart}

\usepackage{lipsum}
\usepackage{amsmath}
\usepackage{amssymb}
\usepackage{amsthm}
\usepackage{mathrsfs}
\usepackage{accents}
\usepackage{calc}
\usepackage{arydshln}
\usepackage{upgreek}
\usepackage{slashed}
\usepackage{xifthen}
\usepackage{graphicx}
\usepackage{longtable}
\usepackage[inline]{enumitem}

\usepackage{xcolor}
\definecolor{winered}{rgb}{0.6,0,0}
\definecolor{lessblue}{rgb}{0,0,0.7}

\usepackage[pdftex,colorlinks=true,linkcolor=winered,citecolor=lessblue,urlcolor=lessblue,breaklinks=true,bookmarksopen=true]{hyperref}

\makeatletter
\newcommand{\myitem}[2]{\item[\rm(#2)]\def\@currentlabel{#2}\label{#1}}
\makeatother

\usepackage{titletoc}

\makeatletter

\def\@tocline#1#2#3#4#5#6#7{
\begingroup
  \par
    \parindent\z@ \leftskip#3 \relax \advance\leftskip\@tempdima\relax
                  \rightskip\@pnumwidth plus 4em \parfillskip-\@pnumwidth
    \ifcase #1 
       \vskip 0.6em \hskip 0em 
       \or
       \or \hskip 0em 
       \or \hskip 1em 
    \fi%
    #6
    \nobreak\relax{\leavevmode\leaders\hbox{\,.}\hfill}
    \hbox to\@pnumwidth {\@tocpagenum{#7}}
  \par
\endgroup
}

 \def\l@section{\@tocline{0}{0pt}{0pc}{}{}}

\renewcommand{\tocsection}[3]{%
  \indentlabel{\@ifnotempty{#2}{ 
    \ignorespaces\bfseries{#2. #3}}}
  \indentlabel{\@ifempty{#2}{\ignorespaces\bfseries{#3}}{}} 
    \vspace{1.5pt}}

\renewcommand{\tocsubsection}[3]{%
  \indentlabel{\@ifnotempty{#2}{
    \ignorespaces#2. #3}}
  \indentlabel{\@ifempty{#2}{\ignorespaces #3}{}}
    \vspace{1.5pt}}

\renewcommand{\tocsubsubsection}[3]{%
  \indentlabel{\@ifnotempty{#2}{
    \ignorespaces#2. #3}}
  \indentlabel{\@ifempty{#2}{\ignorespaces #3}{}}
    \vspace{1.5pt}}

\makeatother

\makeatletter
\def\@nomenstarted{0}
\newlength{\@nomenoldtabcolsep}

\newcommand{\nomenstart}
  {%
    \def\@nomenstarted{1}%
    \setlength{\@nomenoldtabcolsep}{\tabcolsep}%
    \setlength{\tabcolsep}{3.5pt}%
    \begin{longtable}{p{0.11\textwidth} p{0.86\textwidth}}
  }

\newcommand{\nomenitem}[2]{%
    \ifcase\@nomenstarted%
      \or 
      \or \\ 
    \fi%
    #1\,{\leavevmode\leaders\hbox{\,.}\hfill} & #2%
    \def\@nomenstarted{2}%
  }%
\newcommand{\nomenend}
  {\\%
      \end{longtable}%
      \setlength{\tabcolsep}{\@nomenoldtabcolsep}%
      \def\@nomenstarted{0}%
  }
\makeatother

\makeatletter
\newcommand{\BIG}{\bBigg@{3.5}}
\newcommand{\vast}{\bBigg@{4}}
\newcommand{\Vast}{\bBigg@{5}}
\newcommand{\VAST}[1]{\bBigg@{#1}}
\makeatother

\allowdisplaybreaks

\numberwithin{equation}{section}
\numberwithin{figure}{section}
\newtheorem{thm}{Theorem}[section]

\newtheorem{prop}[thm]{Proposition}
\newtheorem{lemma}[thm]{Lemma}
\newtheorem{cor}[thm]{Corollary}

\newtheorem*{thm*}{Theorem}
\newtheorem*{prop*}{Proposition}
\newtheorem*{cor*}{Corollary}
\newtheorem*{conj*}{Conjecture}

\theoremstyle{definition}
\newtheorem{definition}[thm]{Definition}

\theoremstyle{remark}
\newtheorem{rmk}[thm]{Remark}

\makeatletter
\newcommand{\fakephantomsection}{%
  \Hy@MakeCurrentHref{\@currenvir.\the\Hy@linkcounter}
  \Hy@raisedlink{\hyper@anchorstart{\@currentHref}\hyper@anchorend}%
  \Hy@GlobalStepCount\Hy@linkcounter%
}
\makeatother

\newcommand{\mc}{\mathcal}
\newcommand{\cA}{\mc A}

\newcommand{\cC}{\mc C}
\newcommand{\cD}{\mc D}
\newcommand{\cE}{\mc E}
\newcommand{\cF}{\mc F}
\newcommand{\cG}{\mc G}
\newcommand{\cH}{\mc H}
\newcommand{\cI}{\mc I}

\newcommand{\cK}{\mc K}
\newcommand{\cL}{\mc L}
\newcommand{\cM}{\mc M}

\newcommand{\cO}{\mc O}

\newcommand{\cQ}{\mc Q}
\newcommand{\cR}{\mc R}
\newcommand{\cS}{\mc S}
\newcommand{\cT}{\mc T}
\newcommand{\cU}{\mc U}
\newcommand{\cV}{\mc V}

\newcommand{\cX}{\mc X}

\newcommand{\ms}{\mathscr}

\newcommand{\sC}{\ms C}
\newcommand{\sD}{\ms D}

\newcommand{\scri}{\ms I}

\newcommand{\sS}{\ms S}

\newcommand{\TT}{\mathbb{T}}

\newcommand{\C}{\mathbb{C}}
\newcommand{\N}{\mathbb{N}}
\newcommand{\R}{\mathbb{R}}
\newcommand{\Z}{\mathbb{Z}}

\newcommand{\Sph}{\mathbb{S}}

\newcommand{\sfe}{\mathsf{e}}

\newcommand{\sfm}{\mathsf{m}}

\newcommand{\sfs}{\mathsf{s}}

\newcommand{\sfH}{\mathsf{H}}

\newcommand{\bfB}{\mathbf{B}}

\newcommand{\fa}{\mathfrak{a}}

\newcommand{\fC}{\mathfrak{C}}

\newcommand{\fm}{\mathfrak{m}}
\newcommand{\fp}{\mathfrak{p}}

\newcommand{\ft}{\mathfrak{t}}

\newcommand{\slg}{\slashed{g}{}}

\newcommand{\slstar}{\slashed{\star}}

\newcommand{\Err}{{\mathrm{Err}}{}}

\newcommand{\codim}{\operatorname{codim}}

\renewcommand{\Re}{\operatorname{Re}}
\renewcommand{\Im}{\operatorname{Im}}

\newcommand{\mathspan}{\operatorname{span}}
\newcommand{\supp}{\operatorname{supp}}

\newcommand{\rank}{\operatorname{rank}}

\newcommand{\eps}{\epsilon}

\newcommand{\hra}{\hookrightarrow}
\newcommand{\la}{\langle}

\newcommand{\ol}{\overline}
\newcommand{\pa}{\partial}
\newcommand{\dd}{{\mathrm d}}
\newcommand{\ra}{\rangle}
\newcommand{\spec}{\operatorname{spec}}
\newcommand{\specb}{\operatorname{spec}_\bop}

\newcommand{\wh}{\widehat}
\newcommand{\wt}{\widetilde}
\newcommand{\xra}{\xrightarrow}

\newcommand{\ubar}[1]{\underaccent{\bar}#1}
\newcommand{\pfstep}[1]{$\bullet$\ \underline{\textit{#1}}}

\newcommand{\bop}{{\mathrm{b}}}

\newcommand{\sop}{{\mathrm{s}}}
\newcommand{\seop}{{\mathrm{se}}}

\newcommand{\scop}{{\mathrm{sc}}}

\newcommand{\scl}{{\mathrm{sc}}}

\newcommand{\eop}{{\mathrm{e}}}
\newcommand{\tbop}{{3\mathrm{b}}}
\newcommand{\tscop}{{3\mathrm{sc}}}

\newcommand{\scbtop}{{\mathrm{sc}\text{-}\mathrm{b}}}
\newcommand{\schop}{{\mathrm{sc},\semi}}

\newcommand{\semi}{\hbar}

\newcommand{\ff}{\mathrm{ff}}

\newcommand{\scface}{{\mathrm{scf}}}

\newcommand{\tface}{{\mathrm{tf}}}

\newcommand{\zface}{{\mathrm{zf}}}

\newcommand{\cp}{{\mathrm{c}}}

\newcommand{\Diff}{\mathrm{Diff}}

\DeclareMathOperator{\Op}{Op}

\newcommand{\Vb}{\cV_\bop}
\newcommand{\Ve}{\cV_\eop}
\newcommand{\Vscbt}{\cV_\scbtop}

\newcommand{\Vs}{\cV_\sop}
\newcommand{\Vse}{\cV_\seop}

\newcommand{\Diffb}{\Diff_\bop}
\newcommand{\Diffe}{\Diff_\eop}
\newcommand{\Diffscbt}{\Diff_\scbtop}

\newcommand{\Diffs}{\Diff_\sop}
\newcommand{\Diffse}{\Diff_\seop}

\newcommand{\Psie}{\Psi_\eop}

\newcommand{\Vtb}{\cV_\tbop}
\newcommand{\Vtsc}{\cV_{3\scl}}

\newcommand{\Difftb}{\Diff_\tbop}

\newcommand{\Vsc}{\cV_\scop}
\newcommand{\Diffsc}{\Diff_\scop}
\newcommand{\Vsch}{\cV_\schop}
\newcommand{\Diffsch}{\Diff_\schop}

\newcommand{\WF}{\mathrm{WF}}
\newcommand{\Ell}{\mathrm{Ell}}
\newcommand{\Char}{\mathrm{Char}}

\newcommand{\Omegab}{{}^{\bop}\Omega}

\newcommand{\Tb}{{}^{\bop}T}

\newcommand{\Tse}{{}^\seop T}
\newcommand{\Tsc}{{}^{\scop}T}

\newcommand{\Te}{{}^{\eop}T}

\newcommand{\Ttb}{{}^{\tbop}T}
\newcommand{\Ttsc}{{}^{\tscop}T}

\newcommand{\Sb}{{}^{\bop}S}
\newcommand{\Se}{{}^{\eop}S}

\newcommand{\Ssc}{{}^{\scop}S}
\newcommand{\Sse}{{}^\seop S}

\newcommand{\Stb}{{}^{\tbop}S}

\newcommand{\sigmase}{{}^\seop\upsigma}

\newcommand{\sigmatb}{{}^\tbop\upsigma}

\newcommand{\loc}{{\mathrm{loc}}}
\newcommand{\CI}{\cC^\infty}
\newcommand{\CIdot}{\dot\cC^\infty}

\newcommand{\CIc}{\cC^\infty_\cp}

\newcommand{\Hb}{H_{\bop}}

\newcommand{\He}{H_{\eop}}

\newcommand{\Htb}{H_\tbop}

\newcommand{\Hsc}{H_{\scop}}

\newcommand{\phg}{{\mathrm{phg}}}

\newcommand{\bhm}{\fm}
\newcommand{\bha}{\fa}

\newcommand{\openbigpmatrix}[1]
  {%
    \def\@bigpmatrixsize{#1}%
    \addtolength{\arraycolsep}{-#1}%
    \begin{pmatrix}%
  }
\newcommand{\closebigpmatrix}
  {%
    \end{pmatrix}%
    \addtolength{\arraycolsep}{\@bigpmatrixsize}%
  }

\newlength{\enummargin}

\newcommand{\usref}[1]{{\upshape\ref{#1}}}

\DeclareGraphicsExtensions{.mps}

\makeatletter
\newcommand*{\fwbw}[1]{\expandafter\@fwbw\csname c@#1\endcsname}
\newcommand*{\@fwbw}[1]{\ifcase #1 \or {\rm fw}\or {\rm bw}\fi}
\AddEnumerateCounter{\fwbw}{\@fwbw}
\makeatother

\begin{document}

\title[Gluing small black holes along timelike geodesics II: uniform analysis]{Gluing small black holes along timelike geodesics II: uniform analysis on glued spacetimes}

\date{\today}

\begin{abstract}
  Given a smooth globally hyperbolic $(3+1)$-dimensional spacetime $(M,g)$ satisfying the Einstein vacuum equations (possibly with cosmological constant) and an inextendible timelike geodesic $\cC$, we constructed in Part I \cite{HintzGlueLocI} a family of metrics $g_\eps$ on the complement $M_\eps\subset M$ of an $\eps$-neighborhood of $\cC$ with the following behavior: away from $\cC$ one has $g_\eps\to g$ as $\eps\to 0$, while the $\eps^{-1}$-rescaling of $g_\eps$ around every point of $\cC$ tends to a fixed subextremal Kerr metric. Furthermore, $g_\eps$ solves the Einstein vacuum equation modulo $\cO(\eps^\infty)$ errors. The ultimate goal, which we achieve in Part III \cite{HintzGlueLocIII}, is to correct $g_\eps$ to a true solution on any fixed precompact subset of $M$ by addition of a size $\cO(\eps^\infty)$ metric perturbation which needs to satisfy a quasilinear wave equation (namely, the Einstein vacuum equations in a suitable gauge).

  The present paper lays the necessary analytical foundations. We develop a framework for proving estimates (including tame estimates) for solutions of (tensorial) wave equations on $(M_\eps,g_\eps)$ which, on a suitable scale of Sobolev spaces, are uniform on $\eps$-independent precompact subsets of the original spacetime $M$. These estimates are proved by combining two ingredients: the spectral theory for the corresponding wave equation on Kerr; and uniform microlocal estimates governing the propagation of regularity through the small black hole, including radial point estimates reminiscent of diffraction by conic singularities and long-time estimates near perturbations of normally hyperbolic trapped sets.

  As an illustration of this framework, we construct solutions of a toy nonlinear scalar wave equation on $(M_\eps,g_\eps)$ for uniform timescales and with full control in all asymptotic regimes as $\eps\to 0$.
\end{abstract}


\author{Peter Hintz}
\address{Department of Mathematics, ETH Z\"urich, R\"amistrasse 101, 8092 Z\"urich, Switzerland}
\email{peter.hintz@math.ethz.ch}

\maketitle

\setlength{\parskip}{0.00pt}
\tableofcontents
\setlength{\parskip}{0.05in}

\section{Introduction}
\label{SI}

In this work, we develop analytical tools for proving uniform estimates for solutions of linear wave equations on families $(M_\eps,g_\eps)$ of $(3+1)$-dimensional spacetimes in which a mass $\eps$ Kerr black hole, in the limit $\eps\searrow 0$, travels along a timelike curve $\cC$ in a spacetime $(M,g)$. These tools are used in the companion paper \cite{HintzGlueLocIII} to complete the construction, begun in \cite{HintzGlueLocI}, of metrics of this type which are moreover solutions of the Einstein vacuum equations. (In the present paper, the Einstein vacuum equations play no role.)

\subsection{A simple linear example}
\label{SsILin}

We denote the metric of a Kerr black hole with mass $\bhm>0$ and subextremal specific angular momentum $\bha$, $|\bha|<\bhm$, by $\hat g_{\bhm,\bha}$; see \cite{KerrKerr} and Definitions~\ref{DefGlKerrBL}, \ref{DefGlKerr}. We consider this as a Lorentzian metric on a subset $\{|\hat x|\geq\bhm\}\subset\R_{\hat t}\times\R^3_{\hat x}$ for which $\dd\hat t$ is timelike, and which is stationary ($\cL_{\pa_{\hat t}}\hat g_{\bhm,\bha}=0$), axisymmetric, and asymptotically flat in that the metric coefficients $(\hat g_{\bhm,\bha})_{\hat\mu\hat\nu}$ in the coordinates $\hat t,\hat x$ tend to those of the Minkowski metric $-\dd\hat t^2+\dd\hat x^2$ as $|\hat x|\to\infty$. (The boundary $|\hat x|=\bhm$ is a spacelike hypersurface lying behind the black hole event horizon $\{|\hat x|=\hat r_{\bhm,\bha}=\bhm-\sqrt{\bhm^2-|\bha|^2}\}$.) We moreover have the scaling behavior
\begin{equation}
\label{EqIResc}
  (\hat g_{\bhm,\bha})_{\hat\mu\hat\nu}(x/\eps) = (\hat g_{\eps\bhm,\eps\bha})_{\hat\mu\hat\nu}(x).
\end{equation}

Fix local coordinates $t\in\R$, $x\in\R^3$ on $M:=\R_t\times\R^3_x$ in which $g$ is equal to the Minkowski metric $-\dd t^2+\dd x^2$ at $\cC=\R_t\times\{0\}$. For $\eps\in(0,1)$, we then consider Lorentzian metrics $g_\eps$ on $M_\eps=M\cap\{|x|\geq\eps\bhm\}$ so that on $M\setminus\cC$ we have $g_\eps\to g$ in $\CI$ (i.e.\ locally uniformly with all derivatives), while $(g_\eps)_{\mu\nu}(t,\eps\hat x)\to(\hat g_{\bhm,\bha})_{\hat\mu\hat\nu}(\hat x)$ in $\CI$ for $t\in\R$, $\hat x\in\R^3$, $|\hat x|\geq\bhm$. (Identifying $\hat x=\frac{x}{\eps}$, the metric $g_\eps$ thus describes a mass $\eps\bhm$ Kerr black hole near $x=0$ in view of~\eqref{EqIResc}.) We also require uniform control in intermediate regions $\eps\lesssim|x|\lesssim 1$, in that the metric coefficients $(g_\eps)_{\mu\nu}$ are smooth functions of $t,|x|,\frac{\eps}{|x|},\frac{x}{|x|}$.\footnote{They thus restrict to $\frac{\eps}{|x|}=0$, resp.\ $|x|=0$ to $g_{\mu\nu}(t,|x|\frac{x}{|x|})=g_{\mu\nu}(t,x)$, resp.\ $(\hat g_{\bhm,\bha})_{\hat\mu\hat\nu}(\frac{x/|x|}{\eps/|x|})=(\hat g_{\bhm,\bha})_{\hat\mu\hat\nu}(\frac{x}{\eps})$.} See Figure~\ref{FigI}.

\begin{figure}[!ht]
\centering
\includegraphics{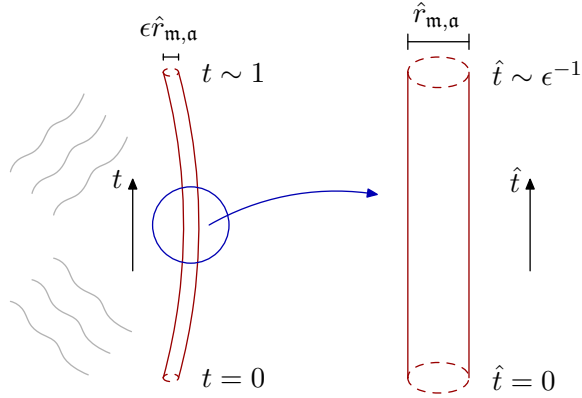}
\caption{Illustration of the metric $g_\eps$ for $\eps\ll 1$, and (on the right) its rescaling by $\eps^{-1}$ around a point $(t,0)\in\cC$. The red tube indicates the event horizon of the small black hole. The gray lines on the left are meant to illustrate waves incident to or reflected off the small black hole.}
\label{FigI}
\end{figure}

In such geometries, it is natural to measure regularity with respect to the vector fields $\pa_t,\pa_x$ in $|x|\gtrsim 1$ (i.e.\ standard regularity on $M$ away from $\cC$) and $\eps\pa_t$, $\eps\pa_x$ in $|x|\lesssim\eps$ (i.e.\ regularity in the time and space coordinates adapted to the scale of the small Kerr black hole). Interpolating suitably in the intermediate region, we thus define $H_{\seop,\eps}^s(M):=H^s(M)$ with the $\eps$-dependent norm
\begin{equation}
\label{EqIHse}
  \|u\|_{H_{\seop,\eps}^s(M)} := \sum_{i+|\beta|\leq s} \| (\hat\rho\pa_t)^i (\hat\rho\pa_x)^\beta u \|_{L^2(M)}^2,\qquad \hat\rho:=(\eps^2+|x|^2)^{\frac12}.
\end{equation}
We refer to $H_{\seop,\eps}^s(M)$ as an \emph{se-Sobolev space}, and regularity with respect to $\hat\rho\pa_t,\hat\rho\pa_x$ as \emph{se-regularity}.

\begin{thm}[Uniform estimates for linear scalar waves]
\label{ThmI}
  Let $T,r_0>0$, and set
  \begin{equation}
  \label{EqIOmega}
    \Omega_\eps=\{0\leq t\leq T,\ \eps\bhm\leq|x|\leq r_0+2(T-t)\}.
  \end{equation}
  Assume that the boundary hypersurfaces of $\Omega_\eps$ are spacelike with respect to $g_\eps$ for $\eps\in(0,\eps_0]$. (This always holds when $r_0,T,\eps_0>0$ are sufficiently small.) Write $L_\eps=\Box_{g_\eps}$ for the linear scalar wave operator on $(M_\eps,g_\eps)$. Let $s=0$ and fix $\hat\alpha\in(\frac12,\frac32)$. Then there exists a constant $C$ so that the following holds. For $f\in H^{s+6}(M_\eps)$ vanishing in $t\leq 0$, denote by $u\in H^{s+7}(M_\eps)$ the forward solution of $L_\eps u=f$, i.e.\ the unique solution of
  \[
    L_\eps u=f\ \text{on}\ \Omega_\eps,\qquad u=0\ \text{in}\ t\leq 0.
  \]
  Then $u$ satisfies the \emph{uniform} estimate
  \begin{equation}
  \label{EqIEst}
    \|\hat\rho^{-\hat\alpha}u\|_{H_{\seop,\eps}^s(\Omega_\eps)} \leq C\|\hat\rho^{-\hat\alpha+2}f\|_{H_{\seop,\eps}^{s+6}(\Omega_\eps)},\qquad \eps\in(0,\eps_0].
  \end{equation}
  More generally, for all $k\in\N_0$, there exists a constant $C_k$ so that for all $\eps\in(0,\eps_0]$,
  \begin{equation}
  \label{EqIEstHi}
    \sum_{i+|\beta|\leq k} \|\hat\rho^{-\hat\alpha} \pa_t^i(\hat\rho\pa_x)^\beta u\|_{H_{\seop,\eps}^s(\Omega_\eps)} \leq C_k\sum_{i+|\beta|\leq k} \|\hat\rho^{-\hat\alpha+2} \pa_t^i(\hat\rho\pa_x)^\beta f\|_{H_{\seop,\eps}^{s+6}(\Omega_\eps)}.
  \end{equation}
\end{thm}

This is a consequence of Theorem~\ref{ThmScUnif} (with $\alpha_\circ=0$) for $k=0$. See Remark~\ref{RmkScSStd} (and Theorem~\ref{ThmScS}) for the case of general $k\in\N$. We comment on a few aspects of Theorem~\ref{ThmI}.

\begin{enumerate}
\item\label{ItILoss}{\rm (se-regularity losses.)} In Theorems~\ref{ThmScUnif} and \ref{ThmScS}, we utilize \emph{variable order} se-Sobolev spaces $H_{\seop,\eps}^\sfs$, i.e.\ the order $\sfs$ depends on the point in phase space and needs to be above, resp.\ below an explicit threshold value at certain incoming, resp.\ outgoing radial sets; see Definition~\ref{DefEstAdm} and Remark~\ref{RmkIFMicroThr}. For suitable such orders $\sfs$, one can use the $H_{\seop,\eps}^\sfs(\Omega_\eps)$-norm \emph{on both sides} of the estimates~\eqref{EqIEst}--\eqref{EqIEstHi}. The resulting uniform estimates thus only lose one se-derivative relative to standard hyperbolic estimates; we will argue that this loss (or at least some small se-regularity loss) arises from trapping near the small black hole. The large se-regularity loss in~\eqref{EqIEst}--\eqref{EqIEstHi} arises from the lossy translation of such a rather precise variable order estimate to constant integer orders.
\item{\rm (Why se-regularity; weights.)} In $|x|\gtrsim 1$, we have $L_\eps\xra{\eps\to 0} L_\circ:=\Box_g$. On the other hand, under the identification $\hat t=\frac{t}{\eps}$ and $\hat x=\frac{x}{\eps}$, we have $\eps^2 L_\eps\xra{\eps\to 0}\Box_{\hat g_{\bhm,\bha}}$ for bounded $|\hat x|$. This motivates the relative weight $\hat\rho^2$ in~\eqref{EqIEst} as well as the usage of se-regularity ($\pa_t,\pa_x$ in $|x|\gtrsim 1$ and $\pa_{\hat t},\pa_{\hat x}$ in $|x|\lesssim\eps$).
\item{\rm (s-regularity.)} We shall refer to regularity under $\pa_t,\hat\rho\pa_x$ as \emph{s-regularity}. In rescaled coordinates in $|\hat x|=\frac{|x|}{\eps}\lesssim 1$, this is regularity with respect to $\eps^{-1}\pa_{\hat t},\pa_{\hat x}$. Thus, the estimate~\eqref{EqIEstHi} means, for $k>0$, that $u$ varies only on unit time scales from the perspective of $M$ (i.e.\ measured with $t$), or equivalently that $u$ behaves \emph{adiabatically} near the small black hole, i.e.\ varies very slowly on unit time scales as measured with $\hat t$.
\end{enumerate}

\begin{rmk}[Semiglobal results]
\label{RmkISG}
  On domains $\Omega\cap\{|x|\geq\eps\bhm\}$ where $\Omega\subset M$ is a compact manifold with corners and spacelike boundary hypersurfaces which near their intersection with $\cC$ equal to level sets of $t$, one can prove uniform estimates of the form~\eqref{EqIEst} (possibly with further regularity losses) by concatenating finitely many such estimates. Note that on domains which do not intersect the curve $\cC$, the estimate~\eqref{EqIEst} follows (with only the standard hyperbolic loss, and for all $s\in\N_0$) from standard hyperbolic theory, applied to a continuous (in $\eps$) family of wave equations.
\end{rmk}

\begin{rmk}[Classical solutions]
\label{RmkIClassical}
  If the norm on the right hand side of~\eqref{EqIEstHi} is bounded by $C_{N,k}\eps^N$ for all $N,k$, then so is the left hand side. A variant of Sobolev embedding thus implies that $u\in\CI(\Omega_\eps)$, and all its derivatives along $\pa_t$, $\hat\rho\pa_x$ have $L^\infty$-norm bounded by $\eps^N$ for all $N$. In this way, Theorem~\ref{ThmI} and its generalizations (discussed below) can be used to control the non-perturbative part of solutions of linear and nonlinear wave equations whose perturbative part (generalized Taylor expansion in $\eps$) is already known. See Theorem~\ref{ThmNToy} for a concrete example.
\end{rmk}

Let us consider the statement of Theorem~\ref{ThmI} from the perspective of the small black hole by using $\hat t=\frac{t}{\eps}$ and $\hat x=\frac{x}{\eps}$; we are thus working on the domain
\begin{equation}
\label{EqIDomKerr}
  \Omega_\eps=\bigl\{(\hat t,\hat x)\colon 0\leq\hat t\leq\eps^{-1}T,\ \bhm\leq|\hat x|\leq\eps^{-1}(r_0+2(T-t))\bigr\}.
\end{equation}
In terms of $z=(t,x)$ and $\hat z=(\hat z^0,\ldots,\hat z^3)=(\hat t,\hat x)$, the metric coefficients of $g_\eps$ are thus
\[
  (g_\eps)_{\mu\nu}(t,\hat x) = (\hat g_{\bhm,\bha})_{\hat\mu\hat\nu}(\hat x) + \eps|\hat x| h_{\hat\mu\hat\nu}\Bigl(t,\eps|\hat x|,\frac{\hat x}{|\hat x|}\Bigr) + \cO((\eps|\hat x|)^2),
\]
where $h_{\hat\mu\hat\nu}$ is smooth in its arguments. (The error term, which captures the deviation of $g_\eps$ from $g$ away from $\cC$, is small when $|x|=\eps|\hat x|$ is small.) Therefore, $(g_\eps)_{\mu\nu}$ is a size $\eps$ perturbation of the Kerr metric $\hat g_{\bhm,\bha}$ (with mass now fixed) in $|\hat x|\lesssim 1$ for times $0\leq\hat t\lesssim\eps^{-1}$, while $(g_\eps)_{\mu\nu}$ differs by a unit amount from the Kerr values for $|\hat x|\sim\eps^{-1}$. Waves which originate in the far-field region $|\hat x|\sim\eps^{-1}$ can thus arrive in the near-field region $|\hat x|\sim 1$ in time $\hat t\sim\eps^{-1}$ and fall into the black hole, linger for some time near the trapped set of Kerr, or scatter back to the far-field region. The estimate~\eqref{EqIEst} must account for all these possibilities.

\begin{rmk}[Scattering by the small black hole]
\label{RmkIScatter}
  One may expect scattering by the small black hole to produce waves traveling into its radiation zone with unit frequency, i.e.\ waves oscillating on unit scales in $\hat t-|\hat x|=\frac{t-|x|}{\eps}$; from the perspective of $M$, these are thus waves with frequency $\eps^{-1}$. The choice $s=0$ in Theorem~\ref{ThmI} (also $s<-\frac12+\hat\alpha$ would be admissible) ensures the finiteness of norms of such outgoing waves once they have reached the far field region $|x|\gtrsim 1$. On the other hand, for sources which vary only on unit scales in $t$ (and in particular are slowly varying near the small black hole), the solution also varies only on such scales---this is the content of~\eqref{EqIEstHi}.
\end{rmk}

\subsubsection{Generalizations}
\label{SssIGen}

In order to be able to solve linear, semilinear, or quasilinear tensorial wave equations on glued spacetimes $(M_\eps,g_\eps)$ for uniform finite time intervals (measured using $t$), we shall generalize the setting and sharpen the conclusions of Theorem~\ref{ThmI} in the following ways.
\begin{enumerate}
\item{\rm (Tensorial operators.)} $L_\eps$ may be a \emph{tensorial} wave operator on $(M_\eps,g_\eps)$, i.e.\ it may act on sections of a vector bundle. Assuming for notational simplicity that the vector bundle is trivial, we thus have
  \begin{equation}
  \label{EqIGenLeps}
    L_\eps = \sum_{i+|\beta|\leq 2} \ell_{\eps,i\beta}(t,x)\pa_t^i\pa_x^\beta = \Box_{g_\eps} + {\rm l.o.t.},
  \end{equation}
  with $\ell_{\eps,i\beta}(t,x)$ being a smooth matrix-valued function of $t,|x|,\frac{\eps}{|x|},\frac{x}{|x|}$, and with the principal symbol of $L_\eps$ being equal to that of $\Box_{g_\eps}$. To $L_\eps$ we can then associate a tensorial wave operator on Kerr,
  \[
    L = \sum_{i+|\beta|\leq 2} \hat\ell_{i\beta}(\hat x)\pa_{\hat t}^i\pa_{\hat x}^\beta = \Box_{\hat g_{\bhm,\bha}} + {\rm l.o.t.},\qquad \hat\ell_{i\beta}(\hat x):=\lim_{\eps\to 0}\ell_{\eps,i\beta}(t,\eps\hat x),
  \]
  which we require here to be independent of $t$. Theorem~\ref{ThmI} has a direct analogue for every such family $L_\eps$, provided the operator $L$ satisfies a mode stability assumption in $\Im\sigma\geq 0$ (discussed in~\S\ref{SssIFModel} below).
\item{\rm (Regularity of coefficients; tame estimates.)} The smoothness of $\ell_{\eps,i\beta}$ can be relaxed considerably: it suffices to assume that $\ell_{\eps,i\beta}=\ell^{(0)}_{\eps,i\beta}+\ell^{(1)}_{\eps,i\beta}$ where $\ell^{(0)}_{\eps,i\beta}$ is smooth as above, while $\ell^{(1)}_{\eps,i\beta}$ only satisfies uniform se-estimates in that
  \[
    |(\hat\rho\pa_t,\hat\rho\pa_x)^\gamma\ell^{(1)}_{\eps,i\beta}| \leq C_\gamma\eps,\qquad \eps\in(0,\eps_0],
  \]
  for all $\gamma\in\N_0^4$ with $|\gamma|\leq d_0$ where $d_0$ is large and finite (depending only on $L$). Assuming the same for $(\pa_t,\hat\rho\pa_x)$-derivatives of $\ell_{\eps,i\beta}^{(1)}$, we moreover prove \emph{tame estimates} for forward solutions of $L_\eps u=f$; omitting weights (which are as in~\eqref{EqIEstHi}), these take the form
  \[
    \|u\|_{H_{\sop,\eps}^k} \leq C_k\bigl(\|\ell_\eps^{(1)}\|_{\cC_{\sop,\eps}^d}\bigr) \bigl( \|f\|_{H_{\sop,\eps}^{k+d}} + \|\ell^{(1)}_\eps\|_{\cC_{\sop,\eps}^{k+d}}\|f\|_{H_{\sop,\eps}^d}\bigr).
  \]
  Here, the $H_{\sop,\eps}^k$-norm is defined like~\eqref{EqIHse} except one uses $\pa_t,\hat\rho\pa_x$ instead of $\hat\rho\pa_t,\hat\rho\pa_x$, and the $\cC_{\sop,\eps}^k$-norm is the supremum norm of all up to $k$-fold derivatives along $\pa_t,\hat\rho\pa_x$. See Theorem~\ref{ThmNTame} for the precise version in the case $L=\Box_{\hat g_{\bhm,\bha}}$. Such tame estimates are the key ingredient for the applicability of a Nash--Moser iteration scheme for solving nonlinear equations; see Theorem~\ref{ThmNToy} for an example.
\item{\rm (Estimates without mode stability.)} Mode stability for $L$ fails in important cases, e.g.\ when $L$ is the linearized gauge-fixed Einstein operator on Kerr in generalized harmonic gauge, as studied in \cite{HaefnerHintzVasyKerr,AnderssonHaefnerWhitingMode}; the estimate~\eqref{EqIEst} is then \emph{false}. Nonetheless, we prove estimates which provide uniform control on $u$ in spaces with high se-regularity in terms of bounds on $u$ in lower se-regularity spaces: omitting weights, this reads
  \begin{equation}
  \label{EqIGenRegEst}
    \|u\|_{H_{\seop,\eps}^\sfs(\Omega_\eps)} \leq C\bigl(\|L_\eps u\|_{H_{\seop,\eps}^\sfs(\Omega_\eps)} + \|u\|_{H_{\seop,\eps}^{\sfs_0}(\Omega_\eps)}\bigr)
  \end{equation}
  where $\sfs_0<\sfs$; here $u=0$ for $t\leq 0$. See Theorem~\ref{ThmEstStd} for a precise version. In fact, it is via the combination of such an estimate with estimates for $L^{-1}$ how we prove Theorem~\ref{ThmI}. In the absence of mode stability at zero energy for $L$, one can still combine~\eqref{EqIGenRegEst} with certain weaker estimates for $L^{-1}$ to conclude, as we demonstrate in the linearized Einstein setting in \cite{HintzGlueLocIII}.
\end{enumerate}

We will describe the main ideas of how to do uniform (tame) analysis for wave-type operators in this generality in~\S\ref{SsIF}.

\begin{rmk}[Other geometries]
\label{RmkIGenOther}
  Using (a subset of) the methods developed in this paper, one can prove uniform estimates when the model metric $\hat g_{\bhm,\bha}$ of $(g_\eps)_{\eps\in(0,1)}$ is replaced by suitable other types of metrics $\hat g$, such as stationary, asymptotically flat, nontrapping metrics. The simplest example is the Minkowski metric, which is the model arising from the choice $g_\eps:=g$, and an operator family to which our methods apply is
  \[
    L_\eps=\Box_g+\eps^{-2}V\Bigl(\frac{x}{\eps}\Bigr)
  \]
  for $V\in\CIc(\R_{\hat x}^3)$ (or for $V$ having at least inverse quadratic decay at infinity), assuming mode stability for $\Box_{\hat g}+V$. In particular, we can correct the formal solution for the model hyperbolic singular perturbation problem discussed in \cite[\S{1.2.1}]{HintzGlueLocI} to a true solution via a correction which is trivial in perturbation theory, i.e.\ which vanishes to all orders (and with all derivatives) as $\eps\searrow 0$; see Remark~\ref{RmkScSGV}. We shall not discuss such geometries in any detail in this paper, and instead refer the interested reader to Remark~\ref{RmkGlOther} and the further remarks referenced there.
\end{rmk}

\subsubsection{Related work}
\label{SssIRel}

Yang \cite{YangGeodesicHypo} considers the Einstein field equations coupled to a nonlinear Klein--Gordon equation which admits (exponentially decaying) soliton solutions. More precisely, \cite[Theorem~1.1]{YangGeodesicHypo} establishes the existence, for uniform finite time intervals, of solutions to the initial value problem when the scalar field data are close to soliton shrunk to spatial scale $\eps$, and the scalar field potential is scaled in an $\eps$-dependent manner so that also the amplitude of the soliton tends to $0$ as $\eps\to 0$. In fact, Yang shows that the soliton approximately moves along a timelike geodesic in the unperturbed spacetime ($M,g)$. This improves on an earlier result by Stuart \cite{StuartGeodesicsEinstein} (see also \cite{StuartSolitonsGeodesics} for related work on semilinear wave equations admitting solitons) which established short-time existence and required a stronger order of vanishing of the scalar field potential in the limit $\eps\to 0$. The spacetime metric remains uniformly close in $\cC^1$ to $g$. The $\|\cdot\|_{H_\eps^s}$-norms introduced in \cite[Equation~(9)]{YangGeodesicHypo} are equivalent (up to an overall power of $\eps$ related to Sobolev embedding) to the $H_{\seop,\eps}^s$-norm~\eqref{EqIHse} over $|x|\lesssim\eps$, but are weaker when $\eps\ll|x|\lesssim 1$ (where the norm~\eqref{EqIHse} does \emph{not} lose a power of $\eps$ with each differentiation); this is related to Remark~\ref{RmkIScatter}.

Hyperbolic singular gluing problems, with precise control of the evolving solutions for a uniform time interval, were also considered in the context of the Euler equations by Davila--del Pi\~{n}o--Musso--Wei \cite{DaviladelPinoMussoWeiVortex,DaviladelPinoMussoWeiHelices}; again spatial derivatives are scaled by $\eps$, and the relevant uniform estimates are closed at a low level of regularity. Literature concerned with the construction of (multi-)soliton solutions to nonlinear Schr\"o\-ding\-er, KdV, water waves, and wave or Klein--Gordon equations includes \cite{MerleCritNLSBlowup,MartelGenKdVMultiSoliton,MartelMerleMultiSolitonNLS,CoteMartelMerlegKdVNLSMultiSoliton,CoteMunozNLKGMultiSoliton,BellazziniGhimentiLeCozNLKGMultiSoliton,MingRoussetTzvetkovWaterWavesMulti,MartelMerleNLW5MultiSoliton,JendrejTwoBubble,CoteMartelNLKGTravelling,JendrejMartelNLWMultiBubble,KadarDynamical3DSoliton}. See also the recent construction \cite{KriegerSchmidZakharovFormal,KriegerSchmidZakharovTrue} of finite-time blow-up solutions for the energy-critical Zakharov system.

An extensive review of the literature concerning gravitational self-force and matched asymptotic expansions in general relativity as well as gluing methods for the constraint equations is given in \cite[\S{1.1}]{HintzGlueLocI}.

\subsection{Framework for uniform analysis}
\label{SsIF}

Our approach towards proving uniform high frequency estimates of the form~\eqref{EqIGenRegEst} is fully microlocal except for the usage of simple energy estimates near the boundaries of the domain $\Omega_\eps$; see~\S\ref{SssIFMicro}. The upgrade of~\eqref{EqIGenRegEst} to uniform estimates~\eqref{EqIEst} requires bounds for the Kerr model operator $L$; see~\S\ref{SssIFModel}. The overall structure of the argument is thus essentially standard (even if the execution is quite delicate):
\begin{itemize}
\item use microlocal techniques to control $u$ in terms of $L_\eps u$ at high frequencies (i.e.\ in the sense of se-differentiability) (\S\ref{SssIFMicro});
\item use the invertibility of certain normal operators of $L_\eps$ as $\eps\searrow 0$ to control $u$ also to leading order in the sense of $\eps$-decay (though with a loss of regularity) (\S\ref{SssIFModel});
\item combine both types of estimates to control $u$ by $L_\eps u$ modulo error terms involving weaker norms of $u$ multiplied by a positive power of $\eps$. For small enough $\eps$, the error terms can be absorbed, giving~\eqref{EqIEst} (also \S\ref{SssIFModel}).
\end{itemize}

Estimates with higher s-regularity such as~\eqref{EqIEstHi} are subsequently proved by commuting $L_\eps$ with appropriate s-vector fields. Technical challenges arise in this step from the need to avoid regularity losses which would preclude the applicability to nonlinear problems via tame estimates; see~\S\ref{SssIFs}.

\emph{Throughout this section, we work in local coordinates $t\in\R$, $x\in\R^3$ on $M$ in which the curve $\cC\subset M$ is given by $\R_t\times\{0\}$.}

\begin{rmk}[Constructions in perturbation theory vs.\ nonpertubative solutions]
\label{RmkIFPertVsNonpert}
  The techniques developed in the present paper are entirely different from those in \cite{HintzGlueLocI} (where formal solutions $g_\eps$ of the Einstein vacuum equations are constructed which are of the type considered in~\S\ref{SsILin}). This is common for singular or asymptotic gluing problems. Indeed, the construction of the perturbative part of a glued solution (i.e.\ in generalized Taylor series in the asymptotic parameter, here the mass $\eps$ of the small black hole) is typically accomplished via repeated inversion of linear model problems, often with infinite regularity throughout the construction; moreover, only certain aspects of the linear problems need to be understood, e.g.\ in \cite{HintzGlueLocI} only linearized gravity on Kerr \emph{at zero frequency}, or in \cite{KadarDynamical3DSoliton} the linearization of the nonlinear wave equation around the soliton \emph{at zero frequency}. Correcting the perturbative solution by adding an appropriate non-perturbative part requires facing the full difficulties caused by the \emph{nonlinear}, hyperbolic (i.e.\ \emph{evolution}) character of the partial differential equation under consideration. This in particular requires precise tracking of regularity, and depends on full spectral information. For example, if mode stability fails for the Kerr model operator $L$ (concrete such settings being studied in \cite{MoschidisSuperradiant}), i.e.\ typical solutions of $L u=0$ have $e^{c\hat t}$ growth, then uniform estimates of the form~\eqref{EqIEst} cannot hold since $u$ is then expected to grow like $e^{c t/\eps}$, i.e.\ \emph{exponentially} in $\eps^{-1}$ on unit size intervals in $t$.
\end{rmk}

\begin{rmk}[Energy methods]
\label{RmkIFEnergy}
  It is conceivable that the toy problem considered in Theorem~\ref{ThmI} can be treated using energy methods. We do not explore this possibility here. The reason is that the framework developed here, based on robust microlocal analysis and exploiting only spectral information for the Kerr model, is highly versatile; this is essential for the application to the (gauge-fixed) Einstein vacuum equation in \cite{HintzGlueLocIII}.
\end{rmk}

\subsubsection{Total space, null-geodesics, phase space}
\label{SssIFPhase}

As in \cite{HintzGlueLocI}, we find it convenient to encode uniformity as $\eps\searrow 0$ by working on a suitable \emph{total space}, i.e.\ a compactification $\wt M$ of $\bigcup_{\eps\in(0,1)}\{\eps\}\times M_\eps$ to a manifold with corners. Concretely, we define $\wt M$ as the real blow-up $\wt M:=[[0,1)_\eps\times M;\{0\}\times\cC]$. Thus, $\wt M$ can be written as the union of three overlapping regions:
\begin{itemize}
\item the \emph{near-field region} $[0,1)_\eps\times\R_t\times\R^3_{\hat x}$ where $\hat x=\frac{x}{\eps}$;
\item the \emph{far-field region} $[0,1)_\eps\times\R_t\times(\R^3_x\setminus\{0\})$;
\item the \emph{transition region} $[0,1)_{\rho_\circ}\times[0,1)_{\hat\rho}\times\R_t\times\Sph_\omega^2$ where $\rho_\circ=\frac{\eps}{|x|}$, $\hat\rho=|x|$, $\omega=\frac{x}{|x|}$.
\end{itemize}
Furthermore, the identity map in $\eps,t,x$ coordinates extends to a smooth \emph{blow-down map}
\[
  \wt\upbeta \colon \wt M \to [0,1)_\eps\times M
\]
which is a diffeomorphism on $\eps>0$. The space $\wt M$ fibers over $[0,1)$ in a singular manner: the fiber over $\eps>0$ is $M_\eps:=\{\eps\}\times M\cong M$, whereas the fiber over $\eps=0$ is the union of two boundary hypersurfaces:
\begin{itemize}
\item $M_\circ$, given in the far-field region by $\eps=0$ and in the transition region by $\rho_\circ=0$;
\item $\hat M$, given in the near-field region by $\eps=0$ and in the transition region by $\hat\rho=0$.
\end{itemize}
Thus, $M_\circ=[M;\cC]$ is a compactified version of the far-field region, arising from $M$ by blowing up $\cC$ (i.e.\ regarding polar coordinates around $\cC$ as valid and nondegenerate down to the center), while $\hat M=\R_t\times\ol{\R^3_{\hat x}}$ is a compactification of the near-field region. See Figure~\ref{FigIwtM}. Below, we shall describe the sense in which each fiber
\[
  \hat M_{t_0} := \hat M \cap \{t=t_0\} \cong \ol{\R^3_{\hat x}},\qquad t_0\in\R,
\]
should be thought of as the Kerr spacetime manifold quotiented out by time translations.

\begin{figure}[!ht]
\centering
\includegraphics{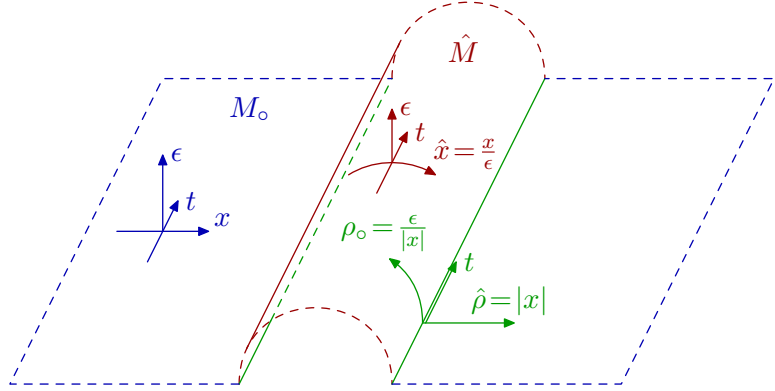}
\caption{The total space $\wt M$ on which our uniform analysis will be phrased, and in the phase space over which we will analyze $\wt L=(L_\eps)_{\eps\in(0,1)}$ microlocally. We also indicate local coordinates (near-field region: red, far-field region: blue, intermediate region: green).}
\label{FigIwtM}
\end{figure}

Equipping $M_\eps\subset\wt M$ with the metric $g_\eps$ as in~\S\ref{SsILin} for all $\eps\in(0,1)$ produces a total metric $\wt g$, whose coefficients (with respect to $t,x$) are smooth functions on $\wt M$. We can then regard null-geodesics on $(M_\eps,g_\eps)$ as curves on $\wt M$ which are tangent to the level sets of $\eps$.

We proceed to qualitatively describe the behavior of null-geodesics inside the domain $\Omega_\eps$ from~\eqref{EqIOmega} in the limit $\eps\to 0$. For every fixed $\eps>0$, null-geodesics simply go from an initial to a final spacelike boundary hypersurface of $\Omega_\eps$. (The hypersurface $|\hat x|=\bhm$ deep inside the small Kerr black hole is one of these final hypersurfaces.) But as $\eps\to 0$, the null-geodesic dynamics of the small Kerr black hole become visible: for example, a null-geodesic can start in the far-field region, approach the small black hole and orbit $\sim\eps^{-1}$ many times before exiting $\Omega_\eps$ and escaping again into the far-field region (or crossing the event horizon, or remaining trapped). Such a null-geodesic is thus \emph{almost} trapped from the perspective of the small black hole. To resolve the dynamics in the near-field region, we rescale the affine parameter by $\eps$ when the null-geodesic is in a region of bounded $|\hat x|$ so that one orbit takes parameter time $\sim 1$. However, this implies that the $\eps\to 0$ limit of such null-geodesics remains stuck at the value of $t$ which it had when entering the near-field region, and indeed one gets a forward trapped null-geodesic on Kerr (or more precisely the spatial manifold of the Kerr black hole at the time $t$). Taking the limit while appropriately shifting the affine parameters may instead produce a backward trapped null-geodesic which escapes the near-field region. Furthermore, when transitioning from the near- to the far-field region, limits of null-geodesics pass through the corner $M_\circ\cap\hat M$ of $\wt M$. Other limits of null-geodesics include those which enter the near-field region but remain far from trapping or horizons and instead exit into the far-field region again; and of course those which always remain in the far-field region, which are thus null-geodesics of $(M\setminus\cC,g)$ simply.

\begin{rmk}[Scattering of null-geodesics]
\label{RmkIScatterNull}
  In the context of the works \cite{MelroseWunschConic,MelroseVasyWunschEdge,MelroseVasyWunschDiffraction,HintzConicWave} and from the perspective of $M\setminus\cC=M_\circ\setminus\pa M_\circ$, one may be tempted to regard the curve $\cC$ (or its resolved version $\pa M_\circ=\wt\upbeta^{-1}(\cC)$) as a timelike curve of cone points. However, $\pa M_\circ=\pa\hat M$ is now itself the boundary at infinity of $\R_t$ times the spatial manifold of an asymptotically flat space. Thus, the `conic singularity' at $\pa M_\circ$ has nontrivial interior dynamics (namely those of the Kerr metric). One can therefore no longer characterize geometric and diffractive null-geodesics via distance $\pi$ propagation along $\pa M_\circ$ (see also \cite{MelroseZworskiFIO}); rather, the scattering map for null-geodesics on Kerr enters. This is very complicated \cite[\S{20}]{ChandrasekharBlackHoles} and only partially defined---but it is also irrelevant for the purposes of the present paper since we do not need to prove diffractive improvements in the spirit of \cite{MelroseWunschConic,MelroseVasyWunschEdge,MelroseVasyWunschDiffraction}.
\end{rmk}

We shall encode the rescaling of parameters, the trapped set of Kerr, etc.\ by working in a suitable phase space over $\wt M$, the \emph{se-cotangent bundle}
\[
  \Tse^*\wt M \to \wt M,
\]
which is a rank $4$ vector bundle over the $5$-dimensional manifold $\wt M$. Since we are working in local coordinates on $M$ here, this bundle is trivial and given by $\wt M\times\R^4$; the point is that over $M_\eps$, we identify
\begin{equation}
\label{EqITseIdent}
\begin{split}
  &M_\eps \times \R^4 \ni (\eps,t,x;\zeta_\seop) \mapsto -\sigma_\seop\frac{\dd t}{\hat\rho} + \xi_\seop\cdot\frac{\dd x}{\hat\rho} \in T^*_{(t,x)}M, \\
  &\hspace{8em} \zeta_\seop=(\sigma_\seop,\xi_\seop)\in\R\times\R^3,\quad \hat\rho=(\eps^2+|x|^2)^{\frac12}.
\end{split}
\end{equation}
This gives a natural isomorphism $\wt T^*_{M_\eps}\wt M\cong T^*M$ which, however, does not extend to $\eps=0$. In other words, the standard momentum variables associated with $t,x$ are rescalings by $\hat\rho^{-1}$ of the se-momentum variables $\sigma_\seop,\xi_\seop$.\footnote{Correspondingly, what constitutes large momenta---which is of central importance for microlocal analysis---differs dramatically between $T^*M$ and $\Tse^*\wt M$ as one approaches points in $\hat M$ in the base.} Note that uniformly equivalent se-momentum variables, i.e.\ smooth and nondegenerate fiber-linear coordinates on $\Tse^*\wt M$, are defined in the near-field region by writing covectors as $-\sigma\,\dd\hat t+\xi\,\dd\hat x$ (where $\dd\hat t=\frac{\dd t}{\eps}$ and $\dd\hat x=\frac{\dd x}{\eps}$), and in the far-field region via $-\sigma\,\dd t+\xi\,\dd x$; thus, $\Tse^*\wt M$ is indeed well-adapted to resolve null-geodesic dynamics (lifted to phase space) in both regimes.

Null-bicharacteristics---lifts of null-geodesics to phase space---are integral curves of the Hamiltonian vector field of the function $G_\eps|_{(t,x)}(\zeta):=|\zeta|_{g_\eps^{-1}|_{(t,x)}}^2$. In the far-field region $|x|\gtrsim 1$, we simply have convergence $H_{G_\eps}\to H_G$ where $G(\zeta)=|\zeta|_{g^{-1}}^2$. In the near-field region $|\hat x|\lesssim 1$ on the other hand, we have
\[
  \eps^2 G_\eps|_{(t,\hat x)}(-\sigma\,\dd\hat t+\xi\,\dd\hat x) = |{-}\sigma\,\dd t+\xi\,\dd x|_{g_\eps|_{(t,\eps\hat x)}}^2 \to |{-}\sigma\,\dd\hat t+\xi\,\dd\hat x|_{\hat g_{\bhm,\bha}^{-1}|_{\hat x}}^2,
\]
which explains why $H_{\eps^2 G_\eps}\to H_{\hat G_{\bhm,\bha}}$ where $\hat G_{\bhm,\bha}(\zeta)=|\zeta|_{\hat g_{\bhm,\bha}^{-1}}^2$. Carefully note that the $\pa_{\hat t}$-component of $H_{\hat G_{\bhm,\bha}}$ remains bounded for bounded se-momenta; but $\pa_{\hat t}=\eps\pa_t$, and thus $t$ is constant along $H_{\eps^2 G_\eps}$-integral curves over $\hat M$ (where $\eps=0$), as already observed prior to Remark~\ref{RmkIScatterNull}.

In order to work uniformly near $\eps=0$, one considers $H_{\hat\rho^2 G_\eps}$. Its integral curves over $M_\circ$ either miss $\pa M_\circ$ altogether, or they tend to $\pa M_\circ$ in the forward or past direction, with limiting momentum a multiple of $(\pa_t\mp\pa_r)^\flat$, $r:=|x|$, over $\pa M_\circ$. Thus, there are \emph{radial sets}
\[
  \cR_{\rm in},\ \cR_{\rm out} \subset \Tse^*_{\pa M_\circ}\wt M \setminus o
\]
over the corner $\pa M_\circ=M_\circ\cap\hat M$; these are the sets through which entrance into and exit from the near-field region take place. From the perspective of the Kerr black hole, these are the sets from which, resp.\ towards which null-bicharacteristics scatter. Their intersections with a fixed $t$-level set are thus closely related to the radial sets of \cite{MelroseEuclideanSpectralTheory,VasyZworskiScl} in stationary scattering theory, and in some sense also to past and null infinity \cite{PenroseAsymptotics} (although in the present setting we do not keep track of the fast time variable $\hat t$ over $\hat M$). The trapped set of Kerr is likewise a smooth conic submanifold
\[
  \Gamma\subset\Tse^*_{\hat M}\wt M\setminus o,
\]
which was characterized in \cite{WunschZworskiNormHypResolvent,DyatlovWaveAsymptotics}. There is a further (generalized) radial set over the event horizon which played an important role already in \cite{VasyMicroKerrdS} (and, non-microlocally, in \cite{DafermosRodnianskiRedShift}).

We point out that all of these sets lie over $\eps=0$, consistent with the fact that scattering and trapping only take place in the limit $\eps\searrow 0$. (For the analysis near the trapped set, we do need to construct an extension of the backwards trapped set over $\hat M$ to small $\eps>0$. This is done in Proposition~\ref{PropTrap} using a general result on normally hyperbolic trapping, Theorem~\ref{ThmTrap}, which follows the constructions given in \cite{HirschPughShubInvariantManifolds,HintzPolyTrap}.)

The detailed description of the null-bicharacteristic dynamics is the subject of~\S\ref{SsGlDyn}; see already Figure~\ref{FigGlDynFlow}.

\subsubsection{Uniform microlocal analysis on se-Sobolev spaces}
\label{SssIFMicro}

A natural class of vector fields adapted to the geometry on $\wt M$ is given by the space $\Vse(\wt M)$ of \emph{se-vector fields}: these are spanned over $\CI(\wt M)$ by $\hat\rho\pa_t$, $\hat\rho\pa_x$ (cf.\ \eqref{EqIHse}). (By~\eqref{EqITseIdent}, their principal symbols are thus smooth fiber-linear functions on $\Tse^*\wt M$.) This is a Lie algebra, with the Lie bracket given by the commutator of vector fields. (An element of $\Vse(\wt M)$ is thus a family, parameterized by $\eps\in(0,1)$, of vector fields on $M\cong M_\eps$ which degenerate in a particular fashion at $\hat M$.) Therefore, we can define corresponding spaces $\Diffse^m(\wt M)$ of \emph{se-differential operators}, and more generally also spaces of weighted operators, with weights $\rho_\circ^{-\ell_\circ}\hat\rho^{-\hat\ell}$ where $\ell_\circ,\hat\ell\in\R$. The key example in the present paper is the total operator
\[
  \wt L = (L_\eps)_{\eps\in(0,1)} \in \hat\rho^{-2}\Diffse^2,
\]
which acts on $\CIc(\wt M^\circ)$ via $L_\eps$ on the $\eps$-level set $M_\eps\subset\wt M$. As further justification for the weight $\hat\rho^{-2}$, we note that
\[
  -\pa_t^2+\pa_x^2 = \hat\rho^{-2}\bigl( -(\hat\rho\pa_t)^2 + (\hat\rho\pa_x)^2 - [\hat\rho\pa_x,\hat\rho]\hat\rho^{-1}\,\hat\rho\pa_x \bigr) \in \hat\rho^{-2}\Diffse^2
\]
indeed, since $[\hat\rho\pa_x,\hat\rho]\hat\rho^{-1}\in\CI(\wt M)$ by direct computation.

We shall use \emph{se-pseudodifferential operators} to effect microlocalization in (conic subsets of) $\Tse^*\wt M\setminus o$ for the purpose of analyzing the operator $\wt L$. We thus quantize symbols $a=a(\eps,t,x;\sigma_\seop,\xi_\seop)\in S^m(\Tse^*\wt M)$ (omitting cutoff functions to suitable neighborhoods of $\{(t,x)=(t',x')\}$) as
\[
  \Op_{\seop,\eps}(a)u(t,x) = (2\pi)^{-4}\iiiint e^{i(-(t-t')\sigma+(x-x')\cdot\xi)} a(\eps,t,x;\hat\rho\sigma,\hat\rho\xi)u(t',x')\,\dd t'\,\dd x'\,\dd\sigma\,\dd\xi.
\]
For now, we only discuss the class $\tilde\Psi_\seop^m$ of families $(\Op_{\seop,\eps}(a))_{\eps\in(0,1)}$ of operators obtained by inserting cutoff functions localizing to $|t-t'|,|x-x'|\lesssim\hat\rho$ in this oscillatory integral, and by allowing symbols $a$ which only feature se-regularity in the base variables in that
\begin{equation}
\label{EqIFMicroSymb}
  \bigl|(\hat\rho\pa_t)^i(\hat\rho\pa_x)^\beta \pa_{(\sigma_\seop,\xi_\seop)}^\gamma a\bigr|\leq C_{i\beta\gamma}(1+|\sigma_\seop|+|\xi_\seop|)^{m-|\gamma|}\qquad\forall\,i,\beta,\gamma.
\end{equation}
Operators of this class are families of bounded geometry ps.d.o.s (pseudodifferential operators) \cite{ShubinBounded}, with the bounded geometry structure depending on $\eps$; see \S\ref{SssFVarse}. The principal symbol of an element of $\tilde\Psi_\seop^m$ captures the symbol $a$ modulo symbols with an extra power of $(1+|\sigma_\seop|+|\xi_\seop|)^{-1}$.

The key reason why se-ps.d.o.s are useful for uniform analysis is that for every $\wt A=(A_\eps)_{\eps\in(0,1)}\in\tilde\Psi_\seop^m$ and $s\in\R$, there exists a constant $C$ so that, \emph{for all $\eps\in(0,1)$},
\[
  \|A_\eps u\|_{H_{\seop,\eps}^{s-m}} \leq C\|u\|_{H_{\seop,\eps}^s}.
\]

\begin{rmk}[Fine localization; variable orders]
\label{RmkIFMicroFine}
  Using $\tilde\Psi_\seop$ we can localize rather finely also in the base---namely, to sets of size $\sim\hat\rho$ in $t,x$, which for bounded $|\hat x|$ means: to sets of size $\sim 1$ in $\hat t,\hat x$. For example, multiplication by $\chi\in\CIc(\R^4_{\hat t,\hat x})$ defines an element of $\tilde\Psi_\seop^0$. Furthermore, following a general principle in microlocal analysis going back to \cite{UnterbergerVariable} and used e.g.\ in \cite{BaskinVasyWunschRadMink,VasyMinicourse,HintzNonstat}, we can define se-ps.d.o.s whose differential order is \emph{variable}, i.e.\ a function on the cosphere bundle $\Sse^*\wt M=(\Tse^*\wt M\setminus o)/\R_+$; via testing with these, we can define se-Sobolev spaces with variable regularity order.
\end{rmk}

To analyze $\wt L$ se-microlocally, observe that its principal symbol is given by the function $\wt G$ given over $M_\eps$ by $G_\eps$---the Hamiltonian flow of which (in the characteristic set) we discussed in~\S\ref{SssIFPhase}. Therefore, we can utilize the full toolkit of principal symbol-based microlocal analysis to prove uniform microlocal estimates for $\wt L$ on se-Sobolev spaces. For example, microlocal elliptic regularity takes the form of a uniform estimate (omitting weights)
\begin{equation}
\label{EqIEll}
  \| B_\eps u \|_{H_{\seop,\eps}^s} \leq C\Bigl( \| G_\eps L_\eps u \|_{H_{\seop,\eps}^{s-2}} + \| \chi u \|_{H_{\seop,\eps}^{-N}} \Bigr),\qquad \eps\in(0,1),
\end{equation}
where $\wt B=(B_\eps)$, $\wt G=(G_\eps)\in\tilde\Psi_\seop^0$, with $\wt G$ and $\wt L$ elliptic on the se-operator wave front set of $\wt B$ (defined in the usual fashion in terms of the full symbol of $\wt B$), and with\footnote{The notation means that $(\hat\rho\pa_t)^i(\hat\rho\pa_x)^\beta\chi\in L^\infty(\wt M)$ for all $i,\beta$.} $\chi\in\cC_\seop^\infty(\wt M)$ a cutoff function which equals $1$ near the projections of the supports of the Schwartz kernels of $\wt B,\wt G$ to $\wt M$ along both projection maps $\wt M^2\to\wt M$.

Similarly, we have uniform estimates which are analogues of real principal type propagation estimates, originating in \cite{DuistermaatHormanderFIO2,HormanderEnseignement}, in the form given in \cite[Theorem~26.1.6]{HormanderAnalysisPDE4}, \cite[Theorem~5.4]{VasyMinicourse}, \cite[\S{E.4}]{DyatlovZworskiBook}, or \cite[Theorem~8.7]{HintzMicro}.

Near the (generalized) radial sets over $\pa M_\circ$ and over the event horizon, we can exploit the saddle point or source/sink structure of the flow to prove uniform estimates quantifying the propagation of se-regularity from the unstable manifold (which is absent for the radial set over the event horizon) into the radial set itself. See~\S\S\ref{SsEstRad} and \ref{SsEstHor}.

\begin{rmk}[Threshold conditions]
\label{RmkIFMicroThr}
  A subtlety which is well-known from asymptotically Euclidean scattering (see e.g.\ \cite[Proposition~5.28]{VasyMinicourse}) is that propagation estimates through the radial set $\cR_{\rm in}$---which connects null-bicharacteristics from the far-field region incident on (spatial infinity of) the black hole and null-bicharacteristics in the near-field region coming in from infinity---yield uniform control in weighted spaces $H_{\seop,\eps}^{s,\alpha_\circ,\hat\alpha}=\rho_\circ^{\alpha_\circ}\hat\rho^{\hat\alpha}H_{\seop,\eps}$ only under the above-threshold condition $s>-\frac12+\alpha_\circ-\hat\alpha$ (plus further constant shifts in case the Kerr model operator $L$ is not symmetric), roughly corresponding to the absence of (high frequency) incoming radiation. On the other hand, when propagating se-estimates out of the near-field region through $\cR_{\rm out}$, there is a below-threshold condition $s<-\frac12+\alpha_\circ-\hat\alpha$ (again with shifts when $L\neq L^*$), roughly corresponding to having to allow for outgoing radiation from the small black hole. There is a further threshold condition of the form $s>s_{\rm thr}$ arising from the radial point estimate at the event horizon. Accommodating all threshold conditions typically requires the differential order $s$ to be variable.
\end{rmk}

Lastly, near the trapped set, we are able to adapt the proof from \cite{HintzPolyTrap}---which in turn is based on \cite{DyatlovSpectralGaps}---to propagate uniform control on solutions of $L_\eps u=f$ from the forward trapped set into the trapped set $\Gamma$ itself, with a loss of two se-derivatives on $u$ compared to elliptic estimates. This is the most delicate microlocal estimate in this paper; see~\S\ref{SsEstTrap}.

In all of these estimates, we need to localize in time, as we are ultimately interested in proving uniform estimates on the bounded (in $M$) domain $M_\eps$. Such localizations are, however, harmless due to the monotonicity of $t$ and $\hat t$ along (future) null-bicharacteristics. Upon concatenating the se-microlocal estimates, we obtain a uniform (schematic) estimate
\begin{equation}
\label{EqIFMicroPre}
  \|u\|_{H_{\seop,\eps}^\sfs(\Omega_\eps^\flat)} \leq C\Bigl( \|L_\eps u\|_{H_{\seop,\eps}^\sfs(\Omega_\eps)} + \|u\|_{H_{\seop,\eps}^{\sfs_0}(\Omega_\eps)}\Bigr)
\end{equation}
where $\Omega_\eps^\flat=\{(t,x)\colon 0\leq t\leq T+\delta\hat\rho,\ \eps(\bhm-\delta)\leq|x|\leq r_0-\delta+2(T-t) \}$ is slightly smaller than $\Omega_\eps$ in~\eqref{EqIOmega} (here $0<\delta\ll 1$ is fixed). The cause for the discrepancy of domains on the left and the right hand sides is the ever-present error term in microlocal estimates, e.g.\ $\chi u$ in~\eqref{EqIEll}.

It is crucial at this point that using $\tilde\Psi_\seop$ we can localize \emph{to unit size intervals in the fast time $\hat t$}, cf.\ Remark~\ref{RmkIFMicroFine}.\footnote{Otherwise, we would have to shrink $\Omega_\eps$ by a fixed amount in $t$, e.g.\ to $t\leq T-\delta$. This would be disastrous: controlling $u$ uniformly on the time interval $[T-\delta,T]$ which would not yet be covered is just as difficult as controlling it on $[0,T]$, which we have not yet done at this point.} The point is that we can bridge the gap between $\Omega_\eps^\flat$ and $\Omega_\eps$ via straightforward energy estimates. Indeed, from the perspective of the Kerr black hole at time $t=T$, i.e.\ in terms of $\hat t=\frac{t-T}{\eps}$ and $\hat x=\frac{x}{\eps}$, the remaining region $T-\delta\hat\rho\leq t\leq T$ is equal to $-\delta\la\hat x\ra\leq\hat t\leq 0$; but for $\delta<1$, this is a region of Kerr which is disjoint from past and future null infinity, and indeed simple energy estimates allow one to control solutions of $\Box_{\hat g_{\bhm,\bha}}v=h$ (with $v,h$ vanishing for $\hat t<-\delta\la\hat x\ra$, say) in this region. The gap near the inner boundary at $|x|=\eps\bhm$ is similarly bridged via energy estimates (with $|\hat x|$ being a time variable there), and the gap near the outer boundary at $|x|=r_0+2(T-t)$ is likewise bridged via standard energy estimates. The details are given in~\S\ref{SsEstEn}.

Altogether, these considerations prove the uniform se-regularity estimate~\eqref{EqIGenRegEst}; see~\S\ref{SsEstStd}. We state this now with the correct weights: writing $\|u\|_{H_{\seop,\eps}^{\sfs,\alpha_\circ,\hat\alpha}(\Omega_\eps)} := \|\rho_\circ^{-\alpha_\circ}\hat\rho^{-\hat\alpha}u\|_{H_{\seop,\eps}^\sfs(\Omega_\eps)}$, this is
\begin{equation}
\label{EqIFMicroReg}
  \|u\|_{H_{\seop,\eps}^{\sfs,\alpha_\circ,\hat\alpha}(\Omega_\eps)} \leq C\bigl(\|L_\eps u\|_{H_{\seop,\eps}^{\sfs,\alpha_\circ,\hat\alpha-2}(\Omega_\eps)} + \|u\|_{H_{\seop,\eps}^{\sfs_0,\alpha_\circ,\hat\alpha}(\Omega_\eps)}\bigr),
\end{equation}
with $\sfs_0<\sfs$ (and with $\sfs_0,\sfs$ satisfying the threshold conditions of Remark~\ref{RmkIFMicroThr} and being non-strictly monotonically decreasing along the future-directed null-bicharacteristic flow).

We stress that up to this point, we have only used symbolic properties of $\wt L$; mode stability for the Kerr model $L$ does not play any role yet. In particular, our arguments go through without any changes when $\wt L$ acts on sections of a vector bundle; no symmetry conditions are required, and the only nontrivial requirement is a bound on the subprincipal symbol of the trapped set (see Proposition~\ref{PropEstTrap}).

\subsubsection{Model operators}
\label{SssIFModel}

The estimate~\eqref{EqIFMicroReg} does not yield unconditional uniform control of $u$ as $\eps\searrow 0$. For example, when mode stability for $L$ fails at some $\sigma\in\C$ with $\Im\sigma>0$, then both norms on $u$ on either side are typically exponentially growing in $\eps$.

In order to proceed, we thus need to impose a mode stability assumption on $L$ at frequencies $\sigma\in\C$ with $\Im\sigma\geq 0$. One says that mode stability holds at a frequency $\sigma\neq 0$ with $\Im\sigma\geq 0$ if there do not exist any \emph{outgoing} solutions $u=u(\hat x)$ of the equation $L(e^{-i\sigma\hat t}u(\hat x))=0$. In the present paper, the precise meaning of the outgoing condition is the triviality of the nullspace of $\hat L(\sigma)=e^{i\sigma\hat t}L e^{-i\sigma\hat t}$ on appropriate scattering Sobolev spaces \cite{MelroseEuclideanSpectralTheory} with variable decay order (as used also in \cite[Proposition~5.28]{VasyMinicourse}); see Proposition~\ref{PropEstFTbdd}. In the case that $L=\Box_{\hat g_{\bhm,\bha}}$ is the scalar wave operator, this can be shown to be equivalent to the requirement that $e^{-i\sigma\hat t}u$ be smooth across the future event horizon, and $u\sim \hat r^{-1}e^{i\sigma\hat r_*}$ as $\hat r=|\hat x|\to\infty$, where $\hat r_*\sim\hat r+2\bhm\log\hat r$ (see \cite[Definition~1.1]{ShlapentokhRothmanModeStability} for the precise conditions for separated mode solutions $u$); we recall in Lemma~\ref{LemmaSc3bSpec} the relevant results from \cite[\S{3}]{HintzKdSMS}.\footnote{We remark that mode stability for \emph{large} $|\sigma|$ is a consequence of symbolic \emph{high energy} (or semiclassical, with $|\sigma|^{-1}$ being the semiclassical parameter) estimates; that is, this is \emph{automatic}. See Proposition~\ref{PropEstFThi}.}

For $\sigma=0$, the operator $\hat L(0)$ is, near spatial infinity, an elliptic operator with good mapping properties on \emph{b-Sobolev spaces} \cite{MelroseMendozaB,MelroseAPS}, i.e.\ Sobolev spaces which measure regularity with respect to $\hat r\pa_{\hat x}$ (recall that we are working in $\hat r\geq\bhm>0$). Mode stability at zero frequency is then the statement that on $L^2$-based b-Sobolev spaces with some range of weights $\hat r^{-\alpha_\cD}$, the operator $\hat L(0)$ is invertible. For $L=\Box_{\hat g_b}$, this is true for $\alpha_\cD\in(-\frac32,-\frac12)$. See Proposition~\ref{PropEstFT0} for the precise statement.

The qualitative mode stability assumption on $L$, coupled with quantitative Fredholm and high energy estimates, implies uniform bounds on $\hat L(\sigma)^{-1}$ as an operator on (semiclassical, when $|\sigma|\gg 1$) scattering Sobolev spaces when $\sigma\neq 0$ and b-Sobolev spaces when $\sigma=0$, and scattering-b-transition Sobolev spaces for uniform estimates near $\sigma=0$. (The latter spaces were defined in \cite[Appendices~A.3--A.4]{HintzKdSMS}, based on earlier work by Guillarmou--Hassell \cite{GuillarmouHassellResI}.) Using the Plancherel theorem (and Paley--Wiener), these quantitative estimates immediately imply the boundedness of the forward solution operator for $L$ on certain spacetime $L^2$-based Sobolev spaces $\Htb^{s,\alpha}=\hat r^{-\alpha}\Htb^s$---namely, weighted 3b-Sobolev spaces as introduced in \cite{Hintz3b}. These are function spaces on $\R_{\hat t}\times\R^3_{\hat x}$ which measure regularity with respect to
\begin{equation}
\label{EqIFModel3b}
  \la\hat x\ra\pa_{\hat t},\quad \la\hat x\ra\pa_{\hat x};
\end{equation}
on the Fourier transform side, these are $L^2$-spaces in $\sigma\in\R$ with values in precisely the aforementioned function spaces for resolvent estimates. We recall this relationship in~\eqref{EqF3FT} following \cite[Proposition~4.29(1)]{Hintz3b}; this was used previously already in the proofs of \cite[Theorem~7.2]{Hintz3b} and \cite[Proposition~5.19]{HintzNonstat}, and we thus refer the reader to these works for further discussion of the relationship of 3b- and resolvent estimates.

In the case $L=\Box_{\hat g_b}$, we obtain
\begin{equation}
\label{EqIFModelKerr}
  L^{-1} \colon H_\tbop^{\sfs_0,\alpha_\cD+2} = \hat r^{-\alpha_\cD-2}H_\tbop^{\sfs_0} \to H_\tbop^{\sfs_0,\alpha_\cD};
\end{equation}
we refer the reader to Theorem~\ref{ThmSc3b} for the precise statement and the conditions on $\sfs_0$ and $\alpha_\cD\in(-\frac32,-\frac12)$. The estimate~\eqref{EqIFModelKerr} in particular implies that solutions of initial value problems with suitably decaying data are in $L^2$ in time, which is significantly weaker than the quantitative decay estimates proved in \cite{DafermosRodnianskiShlapentokhRothmanDecay} following earlier work \cite{AnderssonBlueHiddenKerr,TataruTohaneanuKerrLocalEnergy}, or the precise decay rates established in \cite{HintzPrice,AngelopoulosAretakisGajicKerr} (see also \cite{LukOhTwoTails,LooiXiongSemilinearAsymp,LooiSussmanSchrodingerAsymp} for further recent work in similar directions). The utility of~\eqref{EqIFModelKerr} for present purposes lies in the fact that the 3b-spaces \emph{precisely} match the se-Sobolev spaces on which our microlocal analysis for $\wt L$, leading to the estimate~\eqref{EqIFMicroReg}, is based.

To explain this, note that the basic se-vector fields $\hat\rho\pa_t,\hat\rho\pa_x$, with $\hat\rho=(\eps^2+|x|^2)^{\frac12}$, are, in terms of the coordinates $\hat t=\frac{t}{\eps}$, $\hat x=\frac{x}{\eps}$ (thus $\hat\rho=\eps\la\hat x\ra$), \emph{precisely} equal to the basic 3b-vector fields~\eqref{EqIFModel3b}. Therefore, we have a uniform equivalence of norms (up to powers of $\eps$ required to accommodate changes in volume densities)
\[
  \|u\|_{H_{\seop,\eps}^{\sfs_0,\alpha_\circ,\hat\alpha}(\Omega_\eps)} \sim \eps^{-\hat\alpha}\|u\|_{H_\tbop^{\sfs_0,\alpha_\cD}(\Omega_\eps)},\qquad \alpha_\cD=\alpha_\circ-\hat\alpha,
\]
where on the right we regard $u$ as a function of $(\hat t,\hat x)$ on the Kerr domain~\eqref{EqIDomKerr}. See Lemma~\ref{LemmaFseHse3b} and Proposition~\ref{PropFVarseSobRel} for detailed statements. From~\eqref{EqIFModelKerr} we get $\|u\|_{H_\tbop^{\sfs_0,\alpha_\cD}}\leq C\|L u\|_{H_\tbop^{\sfs_0,\alpha_\cD+2}}$. Since $L$ is equal to $\eps^2 L_\eps$ up to error terms which vanish at $\hat M$ (as se-differential operators of order 2), we can improve~\eqref{EqIFMicroReg} to
\begin{equation}
\label{EqIFMicroMhat}
  \|u\|_{H_{\seop,\eps}^{\sfs,\alpha_\circ,\hat\alpha}(\Omega_\eps)} \leq C\Bigl(\|L_\eps u\|_{H_{\seop,\eps}^{\sfs,\alpha_\circ,\hat\alpha-2}(\Omega_\eps)} + \|u\|_{H_{\seop,\eps}^{\sfs_0+2,\alpha_\circ,\hat\alpha-\delta}(\Omega_\eps)}\Bigr)
\end{equation}
for any fixed $\delta\leq 1$, provided $\alpha_\circ-\hat\alpha=\alpha_\cD\in(-\frac32,-\frac12)$. In other words, we can control $u$ not only in the sense of se-regularity, but also to leading order at $\hat M$, by $L_\eps u$.

\begin{rmk}[se-regularity order]
\label{RmkIFMicroReg}
  We need $2$ more orders of se-regularity in the error term here compared to~\eqref{EqIFMicroReg} due to the loss of 2 derivatives in~\eqref{EqIFModelKerr} and the second order nature of $\eps^2 L_\eps-L$. Therefore, the error term in the estimate~\eqref{EqIFMicroMhat} is only weaker than the left hand side if $\sfs_0<\sfs-2$. In other words, the microlocal se-regularity analysis of~\S\ref{SssIFMicro} is indeed necessary. This is of course a typical feature of non-elliptic problems: the ability to solve $L u=f$ does not suffice to solve perturbed problems $L_\eps u=f$ via $L u=f-(L_\eps-L)u$ via the application of a fixed point argument, since $L_\eps-L$ uses more derivatives than $L$ gains.
\end{rmk}

\begin{rmk}[More precise spaces for resolvent analysis]
\label{RmkIFMicro2nd}
  We cannot use Vasy's resolvent estimates on very precise second microlocal spaces \cite{VasyLAPLag,VasyLowEnergyLag}, as this would force us to work with similarly precise spaces for the entire gluing problem in lieu of the softer se-spaces.
\end{rmk}

To further improve the estimate~\eqref{EqIFMicroMhat}, we have two options.
\begin{enumerate}
\item{\rm (Option I: small domains.)} For $\sfs_0<\sfs-2$, we can estimate the error term of~\eqref{EqIFMicroMhat} using $\|u\|_{H_{\seop,\eps}^{\sfs_0+2,\alpha_\circ,\hat\alpha-\delta}}\leq C\|\hat\rho^\delta u\|_{H_{\seop,\eps}^{\sfs,\alpha_\circ,\hat\alpha}}$. But when the domain $\Omega_\eps$ is \emph{small} (i.e.\ $T,r_0>0$ are small in~\eqref{EqIOmega}), $\hat\rho^\delta$ is \emph{small} on $\Omega_\eps$, and thus this error term can be absorbed into the left hand side of~\eqref{EqIFMicroMhat}. We implement this (by working with a rescaled operator on a fixed domain) in~\S\ref{SssScUnifSm}.
\item{\rm (Option II: inversion of the model operator at $M_\circ$.)} A more systematic option is to improve the $M_\circ$-decay order $\alpha_\circ$ of the error term in~\eqref{EqIFMicroMhat} by using suitable estimates for the model operator of $\wt L$ at $M_\circ$, given by restricting the coefficients of $\wt L$ over $M\setminus\cC$ to $\eps=0$. (In the notation used in~\eqref{EqIGenLeps}, this model operator is thus $L_\circ=\sum_{i+|\beta|\leq 2}\ell_{0,i\beta}(t,x)\pa_t^i\pa_x^\beta$ where $\ell_{0,i\beta}(t,x)=\lim_{\eps\to 0}\ell_{\eps,i\beta}(t,x)$ for $x\neq 0$.) On the level of function spaces, the basic se-derivatives $\hat\rho\pa_t,\hat\rho\pa_x$ restrict to $\eps=0$ as $r\pa_t$, $r\pa_x$ where $r=|x|$. These are \emph{edge vector fields} on $M_\circ$ \cite{MazzeoEdge}. Thus, an estimate for $L_\circ$ of the form
  \begin{equation}
  \label{EqIFMicroEdge}
    \|u\|_{H_\eop^{s,\ell}} \leq C\|L_\circ u\|_{H_\eop^{s-1,\ell-2}},\qquad s=\sfs_0+2,\ \ell=(\hat\alpha-\delta)-\alpha_\circ,
  \end{equation}
  where the weighted edge Sobolev space $H_\eop^{s,\ell}=r^\ell H_\eop^s$ is defined via testing with edge vector fields, can be used to improve the error term of~\eqref{EqIFMicroMhat} to $\|u\|_{H_{\seop,\eps}^{\sfs_0+3,\alpha_0-\delta,\hat\alpha-\delta}}\leq C\eps^\delta\|u\|_{H_{\seop,\eps}^{\sfs,\alpha_0,\hat\alpha}}$, which is thus \emph{small} for sufficiently small $\eps>0$. The proof of estimates of the form~\eqref{EqIFMicroEdge} is given in \cite{HintzConicWave}. We implement this option in~\S\ref{SsScUnif}.
\end{enumerate}
Either option produces the desired estimate~\eqref{EqIEst} for the $M_\circ$-weight $\alpha_\circ=0$ (and with more precise variable orders, as in point~\eqref{ItILoss} following the statement of Theorem~\ref{ThmI}).

\begin{rmk}[No adjoint estimates]
\label{RmkIFAdj}
  The solvability of the linear wave equations we study here is trivial, as for any fixed $\eps>0$ this follows from standard existence results on compact regions with spacelike boundaries; it is only the proof of \emph{uniform estimates} that is highly nontrivial. An important consequence of this observation is that it suffices to prove `direct' estimates for given solutions such as~\eqref{EqIEst}; one does not need to prove `adjoint' estimates for the adjoint problem $L_\eps^*u^*=f^*$ which would give solvability via duality. This will be particularly important in the proof of uniform s-regularity estimates below, where we use $\N_0$ degrees of s-regularity via testing with vector fields; the dual spaces are, at best, difficult to work with unless one has a pseudodifferential calculus available which allows for real se- \emph{and s-}regularity orders. (We do not develop such a calculus here.)
\end{rmk}

\subsubsection{Adiabatic (s-)regularity and tame estimates}
\label{SssIFs}

In order to explain the basic ideas and challenges succinctly, we start with the sharper variable order version of~\eqref{EqIEst} (but omit weights since they only play a minor role at this point)
\begin{equation}
\label{EqIFs}
  \|u\|_{H_{\seop,\eps}^\sfs(\Omega_\eps)} \leq C\|f\|_{H_{\seop,\eps}^\sfs(\Omega_\eps)},\qquad f=L_\eps u,
\end{equation}
valid for monotone variable orders satisfying threshold conditions (see Remark~\ref{RmkIFMicroThr}) and for $u$ vanishing in $t<0$.

Among the basic s-vector fields $\pa_t$, $\hat\rho\pa_x$, only $\pa_t$ is not already an se-vector field. In order to prove~\eqref{EqIEstHi}, a natural first step is to commute $L_\eps u=f$ with $\pa_t$, so
\[
  L_\eps(\pa_t u) = \pa_t f + [L_\eps,\pa_t]u.
\]
When $\wt L=(L_\eps)_{\eps\in(0,1)}$ has smooth (or just s-regular) coefficients on $\wt M$, the commutator $[\wt L,\pa_t]$ is of the same class $\hat\rho^{-2}\Diffse^2$ as $\wt L$ (cf.\ \eqref{EqIGenLeps}). Therefore, applying~\eqref{EqIFs} to this equation with $\sfs-2$ in place of $\sfs$ gives
\[
  \|\pa_t u\|_{H_{\seop,\eps}^{\sfs-2}} \leq C\bigl(\|\pa_t f\|_{H_{\seop,\eps}^{\sfs-2}} + \|u\|_{H_{\seop,\eps}^\sfs}\bigr).
\]
We have thus gained \emph{one} s-derivative at the expense of \emph{two} se-derivatives (two being the differential order of $[\wt L,\pa_t]$). Keeping in mind that the threshold conditions impose an absolute upper bound on $\sfs$, this is unacceptable.

The remedy is to revisit the microlocal se-regularity estimates from~\S\ref{SssIFMicro} on function spaces $H_{(\seop;\sop),\eps}^{(\sfs;k)}$ which encode $k$ degrees of s-regularity (in the argument thus far: $k=1$) and $\sfs$ degrees of se-regularity; see Definition~\ref{DefFsFn} and~\S\ref{SssFVarsse}. There is an important technical caveat, however: elements of the class $\tilde\Psi_\seop$ used thus far do not preserve (uniform) s-regularity; for example, acting with the 0-th order operator given by multiplication with $\chi\in\CIc(\R^4_{\hat t,\hat x})$ on a function $u$, we have $\pa_t(\chi u)=(\pa_t\chi)u+\chi(\pa_t u)$, with $\pa_t\chi$ blowing up like $\hat\rho^{-1}$ as $\eps\to 0$. Therefore, for the present task, we need to work with a class $\Psi_\seop$ of se-ps.d.o.s whose symbols are s-regular, i.e.\ $\hat\rho\pa_t$ in~\eqref{EqIFMicroSymb} is replaced with $\pa_t$. The standard geometric microlocal approach towards defining such a class of operators and prove their mapping and composition properties would be via a characterization of their Schwartz kernels on a suitable resolution of $[0,1)_\eps\times M\times M$; here, we instead take a simpler route and use the black box construction of \cite{HintzScaledBddGeo} to define $\Psi_\seop$ as an algebra of ps.d.o.s associated with a (parameterized) \emph{scaled bounded geometry structure}. See~\S\ref{SssFVarsse}.

Precisely because elements of $\Psi_\seop$ preserve s-regularity, they cannot localize sharply in $\hat t$. Instead, the error term of microlocal estimates, such as $\chi u$ in~\eqref{EqIEll}, must be measured on a set that is bigger by a unit amount in $t$. (The radial point and trapping estimates are proved in this setting in~\S\S\ref{SssEstRads}, \ref{SssEstHors}, and \ref{SssEstTraps}.) As an analogue of~\eqref{EqIFMicroPre}, we thus obtain
\[
  \|u\|_{H_{(\seop;\sop),\eps}^{(\sfs;k)}(\Omega_\eps)} \leq C\Bigl( \|L_\eps u\|_{H_{(\seop;\sop),\eps}^{(\sfs;k)}(\Omega^\sharp_\eps)} + \|u\|_{H_{(\seop;\sop),\eps}^{(\sfs_0;k)}(\Omega^\sharp_\eps)}\Bigr),
\]
where $\Omega_\eps^\sharp=\{(t,x)\colon 0\leq t\leq T+\delta,\ \eps(\bhm-\delta)\leq|x|+r_0+\delta+2(T-t)\}$ is larger than $\Omega_\eps$. In order to bridge the gap $T\leq t\leq T+\delta$, it is thus necessary to use the already established se-solvability theory on this time interval. See~\S\ref{SssEstRads} for details on how one can thus prove (via an iterative argument in the s-regularity order) the estimate~\eqref{EqIEstHi}.

\medskip

Finally, we briefly comment on how we prove tame estimates in~\S\ref{SsNTame}. The basic idea is to commute $\pa_t^k$ through $L_\eps u=f$ (and proceed inductively in $k$). The error term $[L_\eps,\pa_t^k]u$ can then be expanded into $j$-fold s-derivatives of the coefficients of $L_\eps$ acting on $(k-j)$-fold s-derivatives of $u$. Standard tame multiplication estimates (i.e.\ \emph{Moser estimates}) can be used to bound $L^2$-, or indeed fixed regularity $H_{\seop,\eps}$-norms of these terms. An important aspect of our arguments is that we only ever use at most a fixed finite number of se-derivatives; thus, we do not need to record tameness in the se-regularity order.

\subsection{Outline of the paper}
\label{SsIO}

The structure of the paper, and a list of the main results proved in each of its sections, is as follows.

\begin{itemize}
\item In~\S\ref{SF}, we develop the analytic toolbox for our analysis: we describe in detail the various classes of vector fields which arose above (in particular se-, s-, 3b-, and edge vector fields), as well as the associated classes of differential operators (and their model operators), Sobolev spaces, and pseudodifferential operators.
\item In~\S\ref{SGl}, we give the precise definition of glued spacetimes (Definition~\ref{DefGl}) and study the null-bicharacteristic flow in the se-cotangent bundle in detail.
\item In~\S\ref{SEst}, we prove se-microlocal estimates at (generalized) radial sets and at the trapped set; this culminates in Theorem~\ref{ThmEstStd} which provides uniform se-regularity control (i.e.\ a rigorous version of the estimate~\eqref{EqIFMicroReg}). We also collect information on the spectral family $\hat L(\sigma)$ of the Kerr model operator which is accessible without the assumption of mode stability.
\item In~\S\ref{SSc}, we study operators $\wt L=(L_\eps)_{\eps\in(0,1)}$ for which the Kerr model $L$ satisfies mode stability in $\Im\sigma\geq 0$, focusing on the scalar wave operator treated in Theorem~\ref{ThmI}. The main results are Theorems~\ref{ThmSc3b} (the 3b-estimate~\eqref{EqIFModelKerr} on Kerr), Theorems~\ref{ThmScUnif} and \ref{ThmScUnifSm} (uniform estimates on general or on small domains), and Theorem~\ref{ThmScS} (higher s-regularity as in~\eqref{EqIEstHi}). We illustrate the capabilities of our estimates in Theorem~\ref{ThmScSG} by solving linear equations with non-perturbative (i.e.\ $\cO(\eps^\infty)$) source terms.
\item In~\S\ref{SN}, we lay the groundwork for nonlinear problems by proving tame estimates (Theorem~\ref{ThmNTame}) and setting up a suitable version of the Nash--Moser iteration scheme (Theorem~\ref{ThmNTameNM}), including the relevant smoothing operators. As an illustration, we use this framework to solve a toy nonlinear wave equation (Theorem~\ref{ThmNToy}).
\item Appendix~\ref{SB} recalls fundamental notions of geometric singular analysis such as b-vector fields and operators, blow-ups, conormality and polyhomogeneity, and analysis on manifolds with bounded geometry.
\item Appendix~\ref{STrap} is devoted to the construction of extensions of (un)stable trapped sets---which is a purely dynamical result---as needed in our proof of uniform trapping estimates in~\S\ref{SsEstTrap}.
\end{itemize}

\section{Vector fields, function spaces, pseudodifferential operators}
\label{SF}

In \S\S\ref{SsFE}--\ref{SsF3}, we recall notions from edge and 3b-analysis from \cite{MazzeoEdge,Hintz3b}. The development of novel material on se-analysis begins in~\S\ref{SsFse}. The notion of s-regularity is introduced in~\S\ref{SsFs}. Up until this point, we only discuss Sobolev spaces with integer (or, more generally, constant real) orders. In our application, we need spaces with microlocally varying differentiability orders; these, and the underlying classes of pseudodifferential operators, are discussed in~\S\ref{SsFVar}. The reader unfamiliar with basic notions of geometric singular analysis should consult Appendix~\ref{SB} before proceeding.

\subsection{Towards diffraction by the gluing curve: edge analysis}
\label{SsFE}

Let $M_\circ$ be a manifold with boundary, and suppose the boundary $\pa M_\circ$ is the total space of a fibration $Z-\pa M_\circ\to Y$. Then the space $\Ve(M_\circ)$ of \emph{edge vector fields} \cite{MazzeoEdge} consists of all smooth vector fields on $M_\circ$ which are tangent to the fibers of $M_\circ$. This is a Lie algebra and $\CI(M_\circ)$-module. In local coordinates $r\geq 0$, $\omega\in\R^{n_Y}$, $t\in\R^{n_Z}$ in which the fibration $\pa M_\circ\to Y$ is given by the projection $(\omega,t)\mapsto\omega$, the space $\Ve(M_\circ)$ thus consists of all linear combinations, with smooth coefficients, of the vector fields
\[
  r\pa_{t_i}\ (1\leq i\leq n_Z),\quad r\pa_r,\quad \pa_{\omega_j}\ (1\leq j\leq n_Y).
\]
These vector fields are a local frame of the \emph{edge tangent bundle}
\[
  \Te M_\circ \to M_\circ,
\]
the smooth sections of which are thus precisely the edge vector fields. An \emph{edge differential operator} $A\in\Diffe^m(M_\circ)$ of order $m\in\N_0$ is the locally finite sum of up to $m$-fold compositions of elements of $\Ve(M_\circ)$. If $r\in\CI(M_\circ)$ is a boundary defining function, one can also consider spaces of weighted edge differential operators, denoted
\[
  \Diffe^{m,\ell}(M_\circ)=r^{-\ell}\Diffe^m(M_\circ)=\{r^{-\ell}P\colon P\in\Diffe^m(M_\circ)\}.
\]
Such operators define continuous linear maps on $\CIdot(M_\circ)$.

We shall encounter the edge setting in the following way: if $M$ is a smooth $(n+1)$-dimensional manifold without boundary and $\cC\subset M$ is a closed embedded $1$-dimensional submanifold, then $M_\circ:=[M;\cC]$ is a manifold with boundary, and the restriction of the blow-down map $\upbeta\colon[M;\cC]\to M$ to the front face $\pa M_\circ$ defines a fibration $\Sph^{n-1}-\pa M_\circ\to\cC$. In this case $n_Z=1$ and $n_Y=n-1$, so $\Ve(M_\circ)$ is spanned by $r\pa_t$ (with $t$ a coordinate along $\cC$), $r\pa_r$, and spherical vector fields $\pa_{\omega_j}$. Every element $P\in\Diffe^m(M_\circ)$ can thus locally be written in the form
\[
  P = \sum_{j+k+|\alpha|\leq m} a_{j k\alpha}(t,r,\omega) (r D_t)^j (r D_r)^k D_\omega^\alpha,\qquad a_{j k\alpha}\in\CI(\R_t\times[0,\infty)_r\times\Sph^{n-1}_\omega).
\]
The edge normal operator of $P$ at time $t_0\in\R$ is given by freezing coefficients at $t=t_0$ and $r=0$, so
\[
  N_{\eop,t_0}(P) := \sum_{j+k+|\alpha|\leq m} a_{j k\alpha}(t_0,0,\omega) (r D_t)^j (r D_r)^k D_\omega^\alpha;
\]
this is an operator on $\R_t\times[0,\infty)_r\times\Sph^{n-1}_\omega$ which is translation invariant in $t$ and dilation-invariant in $(t,r)$. If the coefficients $a_{j k\alpha}(t_0,0,\omega)$ are independent of the value $t_0$ of $t$, we drop the subscript `$t_0$'. (This property is sensitive to the particular choice of coordinates.) Formally passing to the Fourier transform in $t$ by replacing $D_t$ with multiplication by $-\sigma\in\C$, and subsequently setting $r':=r|\sigma|$ and $\hat\sigma=\frac{\sigma}{|\sigma|}$ produces the \emph{reduced normal operator}
\begin{equation}
\label{EqFERedNormOp}
\begin{split}
  \hat N_{\eop,t_0}(P,\hat\sigma) &:= \sum_{j+k+|\alpha|\leq m} a_{j k\alpha}(t_0,0,\omega) (-\hat\sigma r')^j (r' D_{r'})^k D_\omega^\alpha \\
    &\in \Diff_{\bop,\scop}^{m,0,m}([0,\infty]_{r'}\times\Sph^{n-1}) = (1+r')^m\Diff_{\bop,\scop}^m([0,\infty]_{r'}\times\Sph^{n-1}),
\end{split}
\end{equation}
where $\Diff_{\bop,\scop}^m([0,\infty]_{r'}\times\Sph^{n-1})$ is the space of differential operators which on $r'^{-1}([0,4))$ are b-operators and on $\rho'{}^{-1}([0,4))$, where $\rho'=r'{}^{-1}$, scattering operators of order $m$. (Note that $(1+r')^{-1}$ is a defining function of $\rho'=0$.) We write this as $\hat N_\eop(P,\hat\sigma)$ when $a_{j k\alpha}(t_0,0,\omega)$ is independent of $t_0$.

These notions extend in a straightforward manner to operators acting between spaces of sections of vector bundles $E,F\to M_\circ$, that is, $P\in\Diffe^m(M_\circ;E,F)$. Then $\hat N_{\eop,t_0}(P,\hat\sigma)$ is an element of $\Diff_{\bop,\scop}^{m,0,m}([0,\infty]_{r'}\times\Sph^{n-1};\pi_{t_0}^*E,\pi_{t_0}^*F)$, where $\pi_{t_0}\colon(r',\omega)\mapsto(t_0,0,\omega)\in\pa M_\circ$. Thus, the bundle $\pi_{t_0}^*E$ is the pullback of $E|_{\upbeta^{-1}(t_0)}$ to $[0,\infty]\times\Sph^{n-1}$, similarly for $\pi_{t_0}^*F$. It may happen that for an identification of the bundles $E|_{\upbeta^{-1}(t)}\to\upbeta^{-1}(t)=\Sph^{n-1}$ for varying $t$, similarly for $F$, the operators $\hat N_{\eop,t_0}(P,\hat\sigma)$ are $t_0$-independent; in this case we again drop `$t_0$' from the notation.

The associated Sobolev spaces are called \emph{weighted edge Sobolev spaces}. We first define them when $M_\circ$ is compact. Fixing a weighted b-density on $M_\circ$, such as $r^{n-1}|\dd t\,\dd r\,\dd g_{\Sph^{n-1}}|$ in local coordinates, to define $L^2(M_\circ)$, one defines $\He^{s,\ell}(M_\circ)=r^\ell\He^s(M_\circ)=\{r^\ell u\colon u\in\He^s(M_\circ)\}$ for $s\in\N_0$, $\ell\in\R$, where
\[
  \He^s(M_\circ) = \{ u \in L^2(M_\circ) \colon P u\in L^2(M_\circ)\ \forall\,P\in\Diffe^m(M_\circ) \}.
\]
The space $\He^{s,\ell}(M_\circ)$ can be given the structure of a Hilbert space, with squared norm given by the sum of squared $L^2$-norms of $r^{-\ell}P_j u$ where $\{P_j\}\subset\Diffe^m(M_\circ)$ is a finite subset spanning $\Diffe^m(M_\circ)$ over $\CI(M_\circ)$. Weighted edge Sobolev spaces $\He^{s,\ell}(M_\circ;E)$ of sections of vector bundles $E\to M_\circ$ are defined as distributions on $(M_\circ)^\circ$ which are tuples of elements of $\He^{s,\ell}(M_\circ)$ when multiplied by any smooth function on $M_\circ$ on whose support $E$ is trivialized. Thus, every $P\in\Diffe^{m,\ell}(M_\circ;E,F)$ defines a bounded operator $\He^{s,\ell'}(M_\circ;E)\to\He^{s-m,\ell'-\ell}(M_\circ;F)$ when $s\geq m$. Edge spaces of extendible and/or supported distributions on precompact subsets of manifolds $M_\circ$ with fibered boundary can be defined analogously to~\eqref{EqBHsupp}--\eqref{EqBHsuppext}.

\subsection{Towards scattering theory for the small black hole: 3b-analysis}
\label{SsF3}

We now work on the space
\begin{equation}
\label{EqF3cM}
  \cM := \bigl[\,\ol{\R^{1+n}}; \pa(\ol\R\times\{0\})\,\bigr]
\end{equation}
with interior $\cM^\circ=\R_t\times\R^n_x$. We denote the front face by $\cT=\cT_+\sqcup\cT_-$ where $\cT_\pm$ is the lift of $\{\pm\infty\}\times\{0\}$; and we denote the lift of $\pa\ol{\R^{1+n}}$ by $\cD$. Thus, one can use $\rho_\cD=\frac{1}{\la x\ra}$ and $\rho_\cT=\frac{\la x\ra}{\la(t,x)\ra}$ as defining functions of $\cD$ and $\cT$, respectively. Following \cite[\S3.1]{Hintz3b}, we define the space of \emph{3b-vector fields} by $\Vtb(\cM):=\rho_\cD^{-1}\Vtsc(\cM)$. This Lie algebra and $\CI(\cM)$-module is thus spanned over $\CI(\cM)$ by the vector fields
\begin{equation}
\label{EqF3VF}
  \la x\ra\pa_t,\ \la x\ra\pa_{x_1},\ \ldots,\ \la x\ra\pa_{x_n}.
\end{equation}
The corresponding class of (weighted) differential operators is denoted
\[
  \Difftb^{m,\ell_\cD,\ell_\cT}(\cM) = \rho_\cD^{-\ell_\cD}\rho_\cT^{-\ell_\cT}\Difftb^m(\cM).
\]
The associated Sobolev spaces, relative to a fixed choice of weighted b-density on $\cM$, are denoted
\[
  \Htb^{s,\alpha_\cD,\alpha_\cT}(\cM) = \rho_\cD^{\alpha_\cD}\rho_\cT^{\alpha_\cT}\Htb^s(\cM).
\]
When we consider spaces $\bar H_\tbop^{s,\alpha_\cD,\alpha_\cT}(U)$ of extendible distributions on an open subset $U\subset\cM$ (which in applications is parameter-dependent), we shall always work with fixed boundary defining functions $\rho_\cD,\rho_\cT$ and use the quotient norm, i.e.\ the smallest norm of any extension. Carefully note that even in the case that $\bar U\cap\pa\cM=\emptyset$, this norm depends on the weights $\alpha_\cD,\alpha_\cT$ at $\cD,\cT\subset\cM$. For example, in the case $s=0$, we have
\[
  \|u\|_{\bar H_\tbop^{0,\alpha_\cD,\alpha_\cT}(U)} = \| \rho_\cD^{-\alpha_\cD}\rho_\cT^{-\alpha_\cT}u \|_{L^2(U)}.
\]
Thus, when $U_\eps=\eps^{-1}U_1$ where $\emptyset\neq U_1\subset\cM$ and $\eps>0$ is a parameter, then the norm on $\bar H_\tbop^{0,\alpha_\cD,\alpha_\cT}(U_\eps)$ is uniformly equivalent to the $L^2(U_\eps)$ norm if and only if $\alpha_\cD=\alpha_\cT=0$.

Here and below, we no longer spell out the purely notational changes required for accommodating vector bundles.

\subsubsection{Spectral family and Fourier transform}

Consider a \emph{stationary} 3b-differential operator
\[
  P = \sum_{j+|\alpha|\leq m} a_{j\alpha}(x) (\la x\ra D_t)^j (\la x\ra D_x)^\alpha \in \Diff_{\tbop,\rm I}^{m,\ell_\cD,0}(\cM),
\]
where $a_{j\alpha}\in\la x\ra^{\ell_\cD}\CI(\ol{\R^n_x})$; the subscript `I' indicates the stationarity of $P$, i.e.\ the $t$-independence of its coefficients (or equivalently $[P,\pa_t]=0$). We can then define the \emph{spectral family} of $P$ by formally replacing $D_t$ with multiplication by $-\sigma\in\C$, which gives
\begin{equation}
\label{EqF3SpecFam}
  \hat P(\sigma) = \sum_{j+|\alpha|\leq m} a_{j\alpha}(x) (-\sigma\la x\ra)^j (\la x\ra D_x)^\alpha \in \Diff^m(\R^n).
\end{equation}
Note that $\hat P(\sigma)\in\Diffsc^{m,m+\ell_\cD}(\ol{\R^n})=\la x\ra^{m+\ell_\cD}\Diffsc^m(\ol{\R^n})$ with smooth dependence on $\sigma\in\C$, and for $\sigma=0$ one in fact has $\hat P(0)\in\Diffb^{m,\ell_\cD}(\ol{\R^n})$.

In order to record the structure of $\hat P(\sigma)$ as a family of operators more precisely, we recall from \cite{GuillarmouHassellResI,HintzKdSMS} (see also \cite[\S2.4]{Hintz3b}) for the low frequency regime the space
\[
  \cX_\scbtop := \bigl[ [0,1)_{|\sigma|}\times\cX; \{0\}\times\pa\cX \bigr],\qquad \cX:=\ol{\R^n},
\]
with boundary hypersurfaces denoted $\scface$ (`scattering face', the lift of $[0,1)\times\pa\cX$), $\tface$ (`transition face', the front face), and $\zface$ (`zero face', the lift of $\{0\}\times\cX$); and we further recall the Lie algebra $\Vscbt(\cX)$ of \emph{sc-b-transition vector fields}, consisting of all $V\in\rho_\scface\Vb(\cX_\scbtop)$ which are tangential to all $|\sigma|$-level sets (i.e.\ $\dd|\sigma|(V)=0$). Here $\rho_\scface\in\CI(\cX_\scbtop)$ is a defining function of $\scface$; a possible choice is $\rho_\scface=\frac{\rho}{\rho+|\sigma|}$ where $\rho=\la x\ra^{-1}$. Thus, sc-b-transition vector fields are spanned over $\CI(\cX_\scbtop)$ by $\rho_\scface\la x\ra\pa_x$. (A sc-b-transition vector field is thus a family of vector fields on $\cX^\circ$ with smooth dependence on the parameter $|\sigma|\in(0,1)$ and a specific degeneration near $|\sigma|=0$ and/or $\pa\cX$.) See Figure~\ref{FigF3scbt}.

\begin{figure}[!ht]
\centering
\includegraphics{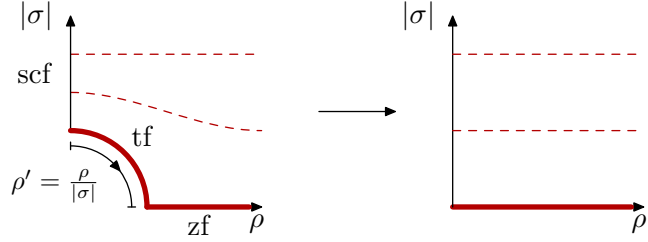}
\caption{\textit{On the left:} the sc-b-transition double space in $\rho=\la x\ra^{-1}<1$, with the factor $\Sph^{n-1}_\omega$, $\omega=\frac{x}{|x|}$, suppressed. \textit{On the right:} the product space $[0,1)_{|\sigma|}\times\ol{\R^n}$. The arrow is the blow-down map.}
\label{FigF3scbt}
\end{figure}

Since $\rho_\scface|\sigma|\la x\ra=\frac{\rho}{\rho+|\sigma|}|\sigma|\rho^{-1}=\frac{|\sigma|}{\rho+|\sigma|}\in\CI(\cX_\scbtop)$, we see from~\eqref{EqF3SpecFam} that for any fixed $\hat\sigma\in\C$, $|\hat\sigma|=1$,
\[
  \bigl(\rho^{\ell_\cD}\hat P(\hat\sigma|\sigma|)\bigr)_{|\sigma|\in[0,1)} \in \Diffscbt^{m,m,0,0}(\cX) = \rho_\scface^{-m}\Diffscbt^m(\cX),
\]
and therefore
\begin{equation}
\label{EqF3scbt}
  \bigl(|\sigma|^{\ell_\cD}\hat P(\hat\sigma|\sigma|)\bigr)_{|\sigma|\in[0,1)} \in \Diffscbt^{m,m+\ell_\cD,0,-\ell_\cD}(\cX) = \rho_\scface^{-(m+\ell_\cD)}\rho_\zface^{\ell_\cD}\Diffscbt^m(\cX).
\end{equation}
The second, third, and fourth orders refer to weights at $\scface$, $\tface$, and $\zface$, respectively, and $\rho_H\in\CI(\cX_\scbtop)$ is a defining function of $H\in\cM_1(\cX_\scbtop)$, with possible choices being $\rho_\scface=\frac{\rho}{\rho+|\sigma|}$ as above (with $\rho\in\CI(\cX)$ a boundary defining function), $\rho_\tface=\rho+|\sigma|$, and $\rho_\zface=\frac{|\sigma|}{\rho+|\sigma|}$.

This operator family has a \emph{transition face normal operator} $N_\tface(P,\hat\sigma)$, given by restriction to $\tface$: to define it, we first rewrite~\eqref{EqF3SpecFam} in inverse polar coordinates $\rho=|x|^{-1}$, $\omega=\frac{x}{|x|}$ as
\[
  \hat P(\sigma) = \sum_{j+k+|\alpha|\leq m} \rho^{-\ell_\cD}a^0_{j k\alpha}(\sigma,\rho,\omega) \Bigl(\frac{\sigma}{\rho}\Bigr)^j (\rho D_\rho)^k D_\omega^\alpha
\]
where $a^0_{j k\alpha}\in\CI(\cX)$. We then pass to $\rho'=\frac{\rho}{|\sigma|}$; thus
\begin{equation}
\label{EqF3scbttf}
  \tface=[0,\infty]_{\rho'}\times\Sph^{n-1}.
\end{equation}
Restricting the coefficients of $(|\sigma|^{\ell_\cD}\hat P(\hat\sigma|\sigma|))_{|\sigma|\in[0,1)}$ to $\sigma=0$ gives
\begin{align*}
  N_\tface(P,\hat\sigma) &= \sum_{k+|\alpha|\leq m} a_{\tface,k\alpha}(\hat\sigma,\hat\rho,\omega) (\rho' D_{\rho'})^k D_\omega^\alpha \\
    &\in \Diff_{\scop,\bop}^{m,m+\ell_\cD,-\ell_\cD}(\tface) = \Bigl(\frac{\rho'}{1+\rho'}\Bigr)^{-m-\ell_\cD}(1+\rho')^{-\ell_\cD}\Diff_{\scop,\bop}^m([0,\infty]_{\rho'}\times\Sph^{n-1}),
\end{align*}
where the coefficients $a_{\tface,k\alpha}\in(\frac{\rho'}{1+\rho'})^{-m-\ell_\cD}(1+\rho')^{-\ell_\cD}\CI(\tface)$ (with smooth dependence on $\hat\sigma$) are given by $a_{\tface,k\alpha}(\hat\sigma,\rho',\omega):=\sum_{j\leq m-k-|\alpha|} a^0_{j k\alpha}(0,0,\omega)\hat\sigma^j\rho'^{-j-\ell_\cD}$.

For the high frequency regime, we introduce $h=|\sigma|^{-1}$ and $z=\sigma/|\sigma|$; then
\begin{equation}
\label{EqF3sch}
  \hat P(h^{-1}z)\in\Diff_{\scop,\semi}^{m,m+\ell_\cD,m}(\cX)=h^{-m}\Diff_{\scop,\semi}^{m,m+\ell_\cD,0}(\cX)=h^{-m}\rho^{-m-\ell_\cD}\Diff_{\scop,\semi}^m(\cX)
\end{equation}
(where now $\rho\in\CI(\cX)$ is a global boundary defining function such as $\rho=\la x\ra^{-1}$), with smooth dependence on $z\in\C$, $|z|=1$; the space $\Diffsch^m(\cX)$ is defined in the usual way relative to the space of \emph{semiclassical scattering vector fields} $\Vsch(\cX)$ which are vector fields of the form $V=h W$ where $W\in\CI([0,1)_h;\Vsc(\cX))$. See \cite{VasyZworskiScl} and \cite[\S2.3]{Hintz3b} for details.

We next discuss function spaces. Fix on $\cX$ a weighted b-density $\mu$ and then on $\cM$ the (weighted b-)density $|\dd t|\otimes\mu$. By the Plancherel theorem, the Fourier transform, here with the convention
\[
  \hat u(\sigma,x) = (\cF u)(\sigma,x) := \int_\R e^{i\sigma t} u(t,x)\,\dd t,\qquad x\in\cX,
\]
induces an isomorphism $L^2(\cM)\cong L^2(\R_\sigma;L^2(\cX))$. Moreover, the Fourier transform intertwines the vector fields~\eqref{EqF3VF} with $-i\sigma\la x\ra$, $\la x\ra\pa_x$. For $\rho:=|x|^{-1}<1$ and $\omega=\frac{x}{|x|}$ we can replace these by $\sigma\rho^{-1}$, $\rho\pa_\rho$, $\pa_\omega$, i.e.\ $\rho_\scface^{-1}=\frac{\rho+|\sigma|}{\rho}$ times $\frac{\sigma}{\rho+|\sigma|}$, $\rho_\scface\rho\pa_\rho$, $\rho_\scface\pa_\omega$; on the other hand, for $|\sigma|>1$, we can replace them by $h^{-1}\rho^{-1}$ times $1$, $h\pa_x$. Thus, the Fourier transform gives an isomorphism
\begin{equation}
\label{EqF3FT}
\begin{split}
  \cF &\colon \Htb^{m,\alpha_\cD,0}(\cM) \to L^2\bigl(\R_\sigma; H_{\wh{\tbop},\sigma}^{m,\alpha_\cD}(\cX) \bigr), \\
  &\qquad
    H_{\wh{\tbop},\sigma}^{m,\alpha_\cD}(\cX) :=
      \begin{cases}
        H_{\scbtop,|\sigma|}^{m,m+\alpha_\cD,\alpha_\cD,0}(\cX), & |\sigma|\leq 1, \\
        H_{\scop,|\sigma|^{-1}}^{m,m+\alpha_\cD,m}(\cX), & |\sigma| > 1,
      \end{cases}
\end{split}
\end{equation}
where we set
\begin{align*}
  \|u\|_{H_{\scbtop,|\sigma|}^{m,r,l,b}(\cX)}^2 &:= \sum_{|\alpha|\leq m} \| \rho_\scface^{-r}\rho_\tface^{-l}\rho_\zface^{-b} (\rho_\scface\la x\ra\pa_x)^\alpha u\|_{L^2(\cX)}^2, \\
  \|u\|_{H_{\scop,h}^{m,r,b}(\cX)}^2 &:= \sum_{|\alpha|\leq m} \| \la x\ra^r h^{-b} (h\pa_x)^\alpha u\|_{L^2(\cX)}^2.
\end{align*}
See \cite[Proposition~4.24]{Hintz3b} for details.

The $H_{\scbtop,|\sigma|}^{m,r,l,b}(\cX)$-norms can in turn be described using standard scattering and/or b-norms (see Appendix~\ref{SB}). Write $\mu=\rho^{-w}\mu_\bop$ where $0<\mu_\bop\in\CI(\cX;\Omegab\cX)$ is an unweighted b-density; and use a density $\mu'$ on $\tface$ which is the product of an unweighted b-density $\mu'_\bop$ with $\rho'{}^{-w}$. Let $\chi\in\CIc([0,1))$ be equal to $1$ near $0$, and set $\chi_\zface=\chi(\frac{|\sigma|}{\rho+|\sigma|})$ and $\chi_\tface=\chi(|\sigma|+\rho)$; define moreover $\Psi_{|\sigma|}(\rho',\omega):=(|\sigma|,|\sigma|\rho',\omega)\in\cX_\scbtop$ for $|\sigma|\neq 0$ (using the coordinates $|\sigma|$, $\rho\geq 0$, $\omega\in\Sph^{n-1}$ on the right). Recalling~\eqref{EqF3scbttf}, we then have uniform (in $|\sigma|$) norm equivalences
\begin{equation}
\label{EqF3scbtNormzf}
\begin{split}
  \| \chi_\zface u \|_{H_{\scbtop,|\sigma|}^{m,r,l,b}(\cX;\mu)} &\sim |\sigma|^{-b} \| \chi_\zface u \|_{\Hb^{m,l-b}(\cX)}, \\
  \| \chi_\tface u \|_{H_{\scbtop,|\sigma|}^{m,r,l,b}(\cX)} = \|\chi_\tface u\|_{H_{\scbtop,|\sigma|}^{m,r+\frac{w}{2},l+\frac{w}{2},b}(\cX;\mu_\bop)} &\sim |\sigma|^{-l-\frac{w}{2}} \| \Psi_{|\sigma|}^*(\chi_\tface u) \|_{H_{\scop,\bop}^{m,r+\frac{w}{2},b-l-\frac{w}{2}}(\tface;\mu'_\bop)} \\
    & = |\sigma|^{-l-\frac{w}{2}} \| \Psi_{|\sigma|}^*(\chi_\tface u) \|_{H_{\scop,\bop}^{m,r,b-l}(\tface;\mu')}.
\end{split}
\end{equation}
See \cite[Proposition~2.21]{Hintz3b}.

\subsection{Combination; uniformity in \texorpdfstring{$\eps$}{epsilon}: se-analysis}
\label{SsFse}

Let $M$ be an $(n+1)$-dimensional manifold without boundary (which will play the role of the spacetime into which we wish to glue a black hole), and let $\cC\subset M$ be a closed embedded $1$-dimensional submanifold with orientable normal bundle. We then recall from \cite[Definitions~3.1 and 3.9]{HintzGlueLocI}:

\begin{definition}[Total gluing spacetime; vector fields]
\label{DefFse}
  The \emph{total gluing spacetime} is
  \[
    \wt M := \bigl[ \, [0,1)_\eps \times M ; \{0\} \times \cC\,\bigr],\qquad \text{with blow-down map}\ \wt\upbeta\colon \wt M \to [0,1)_\eps\times M.
  \]
   for $|\sigma|\neq 0$ (using the coordinates $|\sigma|$, $\rho\geq 0$, $\omega\in\Sph^{n-1}$ on the right) for $|\sigma|\neq 0$ (using the coordinates $|\sigma|$, $\rho\geq 0$, $\omega\in\Sph^{n-1}$ on the right)The boundary hypersurfaces of $\wt M$ are denoted $M_\circ=[M;\cC]$ (the lift of $\{0\}\times M$) and $\hat M=\ol{N}\cC\cong\cC\times\ol{\R^n}$ (the front face, i.e.\ lift of $\{0\}\times\cC$, which is the radially compactified normal bundle of $\cC$ in $M$). We write $\upbeta_\circ=\wt\upbeta|_{M_\circ}\colon M_\circ\to M$ and $\hat\upbeta=\wt\upbeta|_{\hat M}\colon\hat M\to\cC$. Furthermore, we set $M_\eps:=\{\eps\}\times M\subset\wt M$ for $\eps>0$. Defining the vertical tangent bundle $\wt T([0,1)\times M):=[0,1)\times T M\to[0,1)\times M$, we set
  \[
    \wt T\wt M := \wt\upbeta^* \wt T([0,1)\times M).
  \]
  With respect to the fibration $\hat\upbeta$ of $\hat M$, we moreover define the Lie algebra and $\CI(\wt M)$-module of \emph{se-vector fields} by
  \[
    \Vse(\wt M) = \{ V\in\Vb(\wt M) \colon \dd\eps(V)=0,\ V\ \text{is tangent to the fibers of $\hat M$} \}.
  \]
  The fibers of $\hat M$ are denoted $\hat M_p:=\hat\upbeta^{-1}(p)$ for $p\in\cC$; when an identification of $\cC$ with a subset of $\R_t$ is chosen, we also write $\hat M_t$ instead of $\hat M_p$ when $t\in\R$ corresponds to the point $p\in\cC$. Defining functions of $M_\circ$ and $\hat M$ (also local ones) are denoted $\rho_\circ$ and $\hat\rho$, respectively; we always require $\rho_\circ\hat\rho=\eps$. We write $\Diffse^m(\wt M)$ for the space of locally finite sums of up to $m$-fold compositions of elements of $\Vse(\wt M)$; 0-fold compositions are defined to be multiplication operators by elements of $\CI(\wt M)$. We finally write
  \[
    \Diffse^{m,\ell_\circ,\hat\ell}(\wt M) := \rho_\circ^{-\ell_\circ}\hat\rho^{-\hat\ell}\Diffse^m(\wt M).
  \]
\end{definition}

\begin{rmk}[Action of operators]
\label{RmkFseAction}
  If $P=(P_\eps)_{\eps\in(0,1)}\in\Diffse^m(\wt M)$ and $u\in\CI(M)$, we shall write $P u$ for the function $P_\eps u$ when the value of $\eps$ is clear from the context (typically from the subscript of an se-Sobolev norm, see below).
\end{rmk}

If $t\in\R$, $x\in\R^n$ are local coordinates near a point in $\cC$, with $\cC=\{x=0\}$, then one can take $\hat\rho=(\eps^2+|x|^2)^{1/2}$ and $\rho_\circ=\frac{\eps}{(\eps^2+|x|^2)^{1/2}}=\la\hat x\ra^{-1}$ where $\hat x=\frac{x}{\eps}$. A local spanning set of $\Vse(\wt M)$ is $\hat\rho\pa_t$, $\hat\rho\pa_{x_i}$ ($1\leq i\leq n$); this is thus a local frame for the \emph{se-tangent bundle}
\[
  \Tse\wt M \to \wt M,
\]
and we have
\begin{equation}
\label{EqFsewtTTse}
  \wt T\wt M = \hat\rho^{-1}\,\Tse\wt M,
\end{equation}
in the sense that the map $\CI(\wt M;\wt T\wt M)\ni V\mapsto\hat\rho V\in\CI(\wt M;\Tse\wt M)$ is an isomorphism of $\CI(\wt M)$-modules. An se-vector field $V\in\Vse(\wt M)$ is thus a smooth family, parameterized by $\eps\in(0,1)$, of smooth vector fields $V|_{M_\eps}$ on $M\cong\{\eps\}\times M=M_\eps$ which degenerate in a specific fashion near $\cC$ as $\eps\to 0$.

\begin{definition}[se-Sobolev spaces]
\label{DefFseSob}
  Suppose $M$ is compact. Fix a smooth positive density on $M$. Then for $\eps>0$ we define the \emph{se-Sobolev space} $H_{\seop,\eps}^m(M)=H^m(M)$ to have the $\eps$-dependent squared norm
  \[
    \|u\|_{H_{\seop,\eps}^m(M)}^2 := \sum_j \|P_j|_{M_\eps} u \|_{L^2(M)}^2,
  \]
  where $\{P_j\}\subset\Diffse^m(\wt M)$ is a \emph{fixed} (i.e.\ $\eps$-independent) finite spanning set.\footnote{Thus, the norm depends on $\cC$, even though we do not make this explicit in the notation.} The \emph{weighted se-Sobolev space} $H_{\seop,\eps}^{m,\alpha_\circ,\hat\alpha}(M)$, where $\alpha_\circ$, $\hat\alpha\in\R$, is equal to $H^m(M)$ for all $\eps>0$ as a vector space, but with norm
  \[
    \|u\|_{H_{\seop,\eps}^{m,\alpha_\circ,\hat\alpha}(M)} := \| \rho_\circ^{-\alpha_\circ}\hat\rho^{-\hat\alpha}u \|_{H_{\seop,\eps}^m(M)}
  \]
  where $\rho_\circ,\hat\rho$ are \emph{fixed} defining functions.
\end{definition}

If instead of a fixed positive density on $M$ one chooses a smooth positive section of the density bundle associated with $\wt T\wt M$, one obtains uniformly equivalent norms.

\begin{rmk}[Norm in local coordinates]
\label{RmkFseSobLoc}
  For $u$ with support in a fixed compact subset of a coordinate chart $\R_t\times\R^n_x$, the $H_{\seop,\eps}^m(M)$-norm is uniformly (for $\eps\in(0,1]$) equivalent to
  \begin{equation}
  \label{EqFseNormLoc}
    \|u\|_{H_{\seop,\eps}^m(M)}^2 = \sum_{j+|\alpha|\leq m} \bigl\| \bigl((\eps^2+|x|^2)^{1/2}\pa_t\bigr)^j \bigl((\eps^2+|x|^2)^{1/2}\pa_x\bigr)^\alpha u \bigr\|_{L^2(\R^{1+n}_{t,x})}^2.
  \end{equation}
\end{rmk}

These spaces are Hilbert spaces, and one can define them for $m\in\R$ via duality (with respect to $L^2(M)$) and interpolation (or better yet via testing with se-pseudodifferential operators, defined below). Every element $P\in\Diffse^{m,\ell_\circ,\hat\ell}(\wt M)$ defines, upon restriction to $M_\eps$, a continuous map $H^s\to H^{s-m}$, whose norm between the appropriate se-spaces is uniformly bounded: for all $s,\alpha_\circ,\hat\alpha$, there exists a constant $C$ \emph{which is independent of $\eps$} so that
\[
  \|u\|_{H_{\seop,\eps}^{s-m,\alpha_\circ-\ell_\circ,\hat\alpha-\hat\ell}(M)} \leq C\|P u\|_{H_{\seop,\eps}^{s,\alpha_\circ,\hat\alpha}(M)}
\]
for all $\eps\leq\frac12$.

While se-Sobolev spaces are convenient for estimates, the regularity of coefficients of se-operators is better measured in $L^\infty$-based spaces:

\begin{definition}[se-continuity]
\label{DefFseCont}
  For $\alpha_\circ,\hat\alpha\in\R$, we write
  \begin{equation}
  \label{EqFseCont}
    \cC_{\seop,\eps}^{k,\alpha_\circ,\hat\alpha}(M) := \cC^k(M), \qquad
    \|u\|_{\cC_{\sop,\eps}^{k,\alpha_\circ,\hat\alpha}(M)} := \sum_j \| \rho_\circ^{-\alpha_\circ}\hat\rho^{-\hat\alpha} P_j|_{M_\eps} u \|_{L^\infty(M)},
  \end{equation}
  where $P_j$ is as in Definition~\eqref{DefFseSob}. Furthermore, we set
  \[
    \cC_\seop^{k,\alpha_\circ,\hat\alpha}(\wt M) := \Bigl\{ \wt u \colon \wt M\setminus(M_\circ\cup\hat M)\to\C \colon \sup_{\eps\in(0,1)}\|\wt u|_{M_\eps}\|_{\cC_{\seop,\eps}^{k,\alpha_\circ,\hat\alpha}(M)} < \infty \Bigr\}.
  \]
\end{definition}

\subsubsection{Normal operators}
\label{SssFseNormal}

In local coordinates $t\in\R$ and $x\in\R^n$ as above, an element of $\Diffse^m(\wt M)$ can be written as
\begin{equation}
\label{EqFseNormalP}
  P = \sum_{j+|\alpha|\leq m} a_{j\alpha}(\eps,t,x) (\hat\rho D_t)^j (\hat\rho D_x)^\alpha,
\end{equation}
where the $a_{j\alpha}$ are smooth on $\wt M$. We fix the choice $\hat\rho=(\eps^2+|x|^2)^{1/2}$; restricting the coefficients to $M_\circ$ then gives $a_{\circ,j\alpha}=a_{j\alpha}|_{M_\circ}\in\CI(M_\circ)$ (i.e.\ smooth in $t$, $|x|$, $\frac{x}{|x|}$) and in view of $\hat\rho|_{M_\circ}=|x|$ the edge operator
\[
  N_{M_\circ}(P) = \sum_{j+|\alpha|\leq m} a_{\circ,j\alpha} (|x|D_t)^j (|x|D_x)^\alpha \in \Diffe^m(M_\circ).
\]
More invariantly, the restriction of a defining function of $\hat M$ to $M_\circ$ is a defining function of $\pa M_\circ$. Thus, there is a natural isomorphism of cotangent bundles
\begin{equation}
\label{EqFseBundleEdge}
  \Tse_{M_\circ}^*\wt M \cong \Te^* M_\circ
\end{equation}
which is the adjoint of $\Tse_{M_\circ}\wt M\cong\Te M_\circ$.

We next turn to normal operators at $\hat M$. For fixed $t_0$, let us set
\begin{equation}
\label{EqFseHatCoord}
  \hat t := \frac{t-t_0}{\eps},\qquad
  \hat x := \frac{x}{\eps}.
\end{equation}
Then $\hat\rho D_t=\eps\la\hat x\ra\cdot\eps^{-1}D_{\hat t}=\la\hat x\ra D_{\hat t}$ and $\hat\rho D_x=\la\hat x\ra D_{\hat x}$. Note that $\hat M_{t_0}=\ol{\R^n_{\hat x}}$ via continuous extension of $\R^n\ni\hat x\mapsto\lim_{\eps\to 0} (\eps,t_0,\eps\hat x)\in\hat M_{t_0}$. Defining thus $a_{j\alpha}(t_0,\cdot):=a_{j\alpha}|_{\hat M_{t_0}}\in\CI(\hat M_{t_0})=\CI(\ol{\R^n_{\hat x}})$, the operator $P$ induces the stationary 3b-operator
\begin{equation}
\label{EqFseNhatMt0}
  N_{\hat M_{t_0}}(P) = \sum_{j+|\alpha|\leq m} a_{j\alpha}(t_0,\hat x) (\la\hat x\ra D_{\hat t})^j (\la\hat x\ra D_{\hat x})^\alpha \in \Diff_{\tbop,\rm I}^m(\cM),
\end{equation}
where we set $\cM = [ \ol{\R^{1+n}_{\hat t,\hat x}}; \pa(\ol\R_{\hat t}\times\{0\}) ]$ as in~\eqref{EqF3cM}. (When $N_{\hat M_{t_0}}(P)$ is independent of $t_0$ for a fixed choice of local coordinates $t,x$, we drop `$t_0$' from the notation.) These considerations lead to a natural isomorphism
\begin{equation}
\label{EqFseBundle3b}
  \Ttb_{\hat z}^*\cM \cong \Tse_z^*\wt M,
\end{equation}
where the function $z=z(\hat z)$ is defined by continuous extension of $\hat z=(\hat t,\hat x)\mapsto(t_0,\hat x)=z$ (where on the right we use the coordinates $t$, $\hat x$ near $\hat M^\circ$).  In other words, smooth sections of $\Tse_{\hat M_{t_0}}^*\wt M$ are identified with smooth \emph{stationary} sections of $\Ttb^*\cM$.

Using the material from~\S\ref{SsF3}, the operator $N_{\hat M_{t_0}}(P)$ in turn has a spectral family
\[
  \hat N_{\hat M_{t_0}}(P,\sigma) = \sum_{j+|\alpha|\leq m} a_{j\alpha}(t_0,\hat x) (-\sigma\la\hat x\ra)^j (\la\hat x\ra D_{\hat x})^\alpha,
\]
which in turn has transition face normal operators. More generally, for $P\in\Diffse^{m,\ell_\circ,\hat\ell}(\wt M)$, one has normal operators
\begin{align*}
  N_{M_\circ}(\eps^{\ell_\circ}P) &\in \Diffe^{m,\hat\ell-\ell_\circ}(M_\circ), \\
  N_{\hat M_{t_0}}(\eps^{\hat\ell}P) &\in \Difftb^{m,\ell_\circ-\hat\ell,0}(\cM).
\end{align*}
In particular, using~\eqref{EqF3scbt} with $\ell_\cD=\ell_\circ-\hat\ell$, we then have
\begin{align}
\label{EqFseNMcircWeighted}
  (r')^{-(\hat\ell-\ell_\circ)}\hat N_{\eop,t_0}\bigl(r^{\hat\ell-\ell_\circ}N_{M_\circ}(\eps^{\ell_\circ}P),\hat\sigma\bigr) &\in \Diff_{\bop,\scop}^{m,\hat\ell-\ell_\circ,m-(\hat\ell-\ell_\circ)}([0,\infty]_{r'}\times\Sph^2), \\
\label{EqFseNhatMWeighted}
  \bigl( |\sigma|^{\ell_\circ-\hat\ell}\hat N_{\hat M_{t_0}}(\eps^{\hat\ell}P,\hat\sigma|\sigma|)\bigr)|_{|\sigma|\in[0,1)} &\in \Diffscbt^{m,m+(\ell_\circ-\hat\ell),0,-(\ell_\circ-\hat\ell)}(\cX),
\end{align}
where $\cX=\ol{\R^n_{\hat x}}$; here $r'=r|\sigma|$. See \cite[\S3]{HintzGlueLocI} for a more geometric perspective on the normal operators of se-differential operators.

\begin{lemma}[Identification of tf- and reduced normal operators]
\label{LemmaFseIdent}
  Let $P\in\Diffse^{m,\ell_\circ,\hat\ell}(\wt M)$. Let $\hat\sigma\in\C$, $|\hat\sigma|=1$. Write
  \[
    \hat N_{\eop,t_0}(P,\hat\sigma) \in \Diff_{\bop,\scop}^{m,\hat\ell-\ell_\circ,m-(\hat\ell-\ell_\circ)}([0,\infty]_{r'}\times\Sph^2)
  \]
  for the reduced normal operator~\eqref{EqFseNMcircWeighted}. Furthermore, write
  \begin{equation}
  \label{EqFseNhatMWeightedtf}
    \hat N_{\tface,t_0}(P,\hat\sigma) \in \Diff_{\scop,\bop}^{m,m+(\ell_\circ-\hat\ell),-(\ell_\circ-\hat\ell)}(\tface),\qquad \tface=[0,\infty]_{\rho'}\times\Sph^2,
  \end{equation}
  for the transition face normal operator of~\eqref{EqFseNhatMWeighted}; here $\rho'=\frac{1/\hat r}{|\sigma|}$ for $\hat r=|\hat x|=\frac{|x|}{\eps}$. Then the pullback of $\hat N_{\eop,t_0}(P,\hat\sigma)$ along the map $(\rho',\omega)\mapsto(r',\omega)=(\rho'{}^{-1},\omega)$ is equal to $\hat N_{\tface,t_0}(P,\hat\sigma)$.
\end{lemma}
\begin{proof}
  Consider first the unweighted case $\ell_\circ=\hat\ell=0$. Both normal operators only involve the coefficients of $P$ (as an se-operator) at $\pa\hat M_{t_0}$. It thus suffices to check the claim for a local frame of se-vector fields near $\pa\hat M_{t_0}$. We take $r\pa_t$, $r\pa_r$, $\pa_\omega$, the reduced normal operators of which are, in terms of $r'=r|\sigma|$ and $\hat\sigma=\frac{\sigma}{|\sigma|}$, given by
  \begin{equation}
  \label{EqFseIdent}
    {-}i r\sigma=- ir'\hat\sigma,\qquad
    r'\pa_{r'},\qquad
    \pa_\omega.
  \end{equation}
  On the other hand, writing $r=\eps\hat r$ and $\hat t=\frac{t-t_0}{\eps}$, these se-vector fields take the form $\hat r\eps\pa_t=\hat r\pa_{\hat t}$, $\hat r\pa_{\hat r}$, $\pa_\omega$, whose spectral families are $-i\sigma\hat r$, $\hat r\pa_{\hat r}$, $\pa_\omega$, and which thus have the transition face normal operators $-i\hat\sigma\rho'{}^{-1}$, $-\rho'\pa_{\rho'}$, $\pa_\omega$. Under the identification $\rho'=r'{}^{-1}$, these are equal to~\eqref{EqFseIdent}.

  To prove the weighted version of the Lemma, we only need to check the weights themselves. So for $w=\rho_\circ^{-\ell_\circ}\hat\rho^{-\hat\ell}$, where we may take $\rho_\circ=\frac{\eps}{r}=\hat r^{-1}$ and $\hat\rho=r$ near $\pa\hat M$, we have $N_{M_\circ}(\eps^{\ell_\circ}w)=r^{\ell_\circ-\hat\ell}$, so $\hat N_{\eop,t_0}(w,\hat\sigma)=(r')^{\ell_\circ-\hat\ell}$; on the other hand, $\hat N_{\hat M_{t_0}}(\eps^{\hat\ell}w,\hat\sigma|\sigma|)=\hat r^{\ell_\circ-\hat\ell}$, so $\hat N_{\tface,t_0}(w,\hat\sigma)=(\rho')^{-(\ell_\circ-\hat\ell)}$. The conclusion is now evident.
\end{proof}

The same calculations also show the equality of boundary spectra
\begin{equation}
\label{EqFseIdentSpecb}
  \specb(\hat N_{\eop,t_0}(P,\hat\sigma)) = \specb(\hat N_{\tface,t_0}(P,\hat\sigma)) = -\specb(\hat N_{\hat M_{t_0}}(P,0))
\end{equation}
and their independence of $\hat\sigma$; here, the boundary spectra are computed for the b-normal operators at $r'=0$, resp.\ $\rho'=0$, resp.\ $\hat r^{-1}=0$.

The functions $\hat t,\hat x$ defined by~\eqref{EqFseHatCoord} induce a diffeomorphism
\begin{equation}
\label{EqFsehatMpff}
  \ff[ \wt M; \hat M_{t_0} ] \cong \cM,
\end{equation}
where $\ff[X;Y]$ denotes the front face of $[X;Y]$. This follows from a local coordinate computation: replacing $M$ by $\R_t\times\R^n_x$, the space $[\wt M;\hat M_{t_0}]$ is the iterated blow-up
\[
  [ [0,1)\times\R_t\times\R^n_x; \{0\}\times\R_t\times\{0\}; \{0\} \times \{t_0\} \times \{0\} \bigr].
\]
Since $\{0\}\times\{t_0\}\times\{0\}\subset\{0\}\times\R_t\times\{0\}$, the order of the two blow-ups can be exchanged by \cite[Proposition~5.8.1]{MelroseDiffOnMwc}. Now, the front face of $[[0,1)\times\R_t\times\R^n_x;\{0\}\times\{t_0\}\times\{0\}]$ is diffeomorphic to $\ol{\R^{1+n}_{\hat t,\hat x}}$, and the lift of $\{0\}\times\R_t\times\{0\}$ meets this at the boundary at infinity where $\hat x/\hat t=0$, i.e.\ at $\pa{\ol{\R_{\hat t}}\times\{0\}}$; this gives~\eqref{EqFsehatMpff}. (See also \cite[\S3.1]{HintzGlueLocI}.) The operator $N_{\hat M_{t_0}}(P)$ then arises, in a more geometric fashion, as the restriction to $\ff[\wt M;\hat M_{t_0}]$ of $P$.

\subsubsection{Relationship of se- and edge-notions}

The relationship between se-vector fields on $\wt M$ and edge vector fields on $M_\circ$ has the following consequence for Sobolev spaces, which we phrase in local coordinates for simplicity.

\begin{lemma}[se- and edge Sobolev spaces]
\label{LemmaFseEdgeSob}
  Consider local coordinates $t\in\R$, $x\in\R^n$ on $M$, with $\cC=\{x=0\}$. Let $\chi\in\CIc([0,\infty))$, and write $\chi_\circ(\eps,t,x)=\chi(\frac{\eps}{|x|})$ which is thus equal to $1$ near $M_\circ$. Let $m\in\N_0$, $\alpha_\circ,\hat\alpha\in\R$. Then we have a uniform equivalence of norms
  \[
    \|\chi_\circ u\|_{H_{\seop,\eps}^{m,\alpha_\circ,\hat\alpha}(M)} \sim \eps^{-\alpha_\circ}\| (\chi_\circ u)(\eps,-) \|_{\He^{m,\hat\alpha-\alpha_\circ}(M_\circ)}
  \]
  for all $u$ with support in any fixed compact subset of $\R_t\times\R^n_x$. On the right, we regard $(\chi_\circ u)(\eps,-)$ as the function $\R_t\times(\R^n_x\setminus\{0\})\ni(t,x)\mapsto(\chi_\circ u)(\eps,t,x)$; and we use on $M_\circ$ the density $|\dd t\,\dd x|$.
\end{lemma}
\begin{proof}
  We have $(\eps^2+|x|^2)^{1/2}=|x|\eta(\eps,x)$ where $\eta(\eps,x)=(1+\frac{\eps^2}{|x|^2})^{1/2}$ and $\eta(\eps,x)^{-1}$ are smooth on $\supp\chi_\circ$. Thus, in the expression~\eqref{EqFseNormLoc} for $\chi_\circ u$ in place of $u$, the replacement of the derivatives on the right hand side by $(|x|\pa_t)^j(|x|\pa_x)^\alpha$ gives a uniformly equivalent norm.
\end{proof}

Define fiber-linear coordinates $\sigma_\eop,\xi_\eop\in\R$ and $\eta_\eop\in T^*\Sph^2$ on $\Te^*M_\circ$ by writing edge covectors in terms of $t\in I_\cC$, $r=|x|\geq 0$, $\omega=\frac{x}{|x|}\in\Sph^2$ as
\begin{equation}
\label{EqFseEdgeCoord}
  {-}\sigma_\eop\,\frac{\dd t}{r} + \xi_\eop\,\frac{\dd r}{r} + \eta_\eop,\qquad \eta_\eop\in T_\omega^*\Sph^2.
\end{equation}
Then the proof of Lemma~\ref{LemmaFseEdgeSob} also shows that $\sigma_\eop,\xi_\eop,\eta_\eop$ are smooth fiber-linear coordinates also on $\Tse^*\wt M$ in a neighborhood of $\pa M_\circ$, since $\frac{\dd t}{r}$, $\frac{\dd r}{r}$, and spherical 1-forms are a local frame of $\Tse^*\wt M$ near $\pa M_\circ$.

\subsubsection{Relationship of se- and 3b-notions}

In order to transfer bounds for distributions on $\cM$, with parametric $\eps$-dependence, in 3b-spaces to bounds in se-spaces, we first need to upgrade~\eqref{EqFsehatMpff} to a relationship between $[\wt M;\hat M_{t_0}]$ and $[0,1)_\eps\times\cM$. (The relationship on the level of function spaces is discussed in Lemma~\ref{LemmaFseHse3b} below.) We state this directly in local coordinates:

\begin{lemma}[Relationship between resolved spaces]
\label{LemmaFseRelGeo}
  Let $M=\R_t\times\R^n_x$, $\cC=\R_t\times\{0\}$, and define $\wt M=[[0,1)_\eps\times M;\{0\}\times\cC]$ and $\hat M_{t_0}\subset\wt M$ as in Definition~\usref{DefFse}. Define $\cM$ as in~\eqref{EqF3cM} with coordinates denoted $\hat t,\hat x$, and boundary hypersurfaces $\cT$ (front face) and $\cD$ (lift of the original boundary).
  \begin{enumerate}
  \item\label{ItFseRelGeoMap} The map
    \begin{subequations}
    \begin{equation}
    \label{EqFseRelGeoPsi0}
      (0,1)\times M \ni (\eps,t,x) \mapsto (\eps,\hat t,\hat x) = \Bigl(\eps,\frac{t-t_0}{\eps},\frac{x}{\eps}\Bigr) \in (0,1)\times\cM
    \end{equation}
    extends by continuity to a smooth map
    \begin{equation}
    \label{EqFseRelGeoPsi}
      \Psi \colon [\wt M;\hat M_{t_0}] \to \wt\cM := \bigl[\,[0,1)\times\cM; \{0\}\times\cT; \{0\}\times\cD\,\bigr]
    \end{equation}
    \end{subequations}
    which is a diffeomorphism onto its image (equal to the complement of the union of the lifts of $[0,1)\times\cT$ and $[0,1)\times\cD$). Under $\Psi$, the lift of $\hat M$, resp.\ $M_\circ$ gets mapped to the lift of $\{0\}\times\cT$, resp.\ $\{0\}\times\cD$.
  \item\label{ItFseRelGeoVF} Denote by ${}^\seop T[\wt M;\hat M_{t_0}]$ the pullback of $\Tse\wt M$, and denote by $\Ttb\wt\cM$ the pullback of $\Ttb\cM$ along the composition of maps $\wt M\to[0,1)\times\cM\to\cM$. Then pushforward along $\Psi$ induces a bundle isomorphism $\Tse[\wt M;\hat M_{t_0}]\cong\Ttb_{\Psi([\wt M;\hat M_{t_0}])}\wt\cM$. (That is, the space of pushforwards of se-vector fields under $\Psi$ is equal to the space of families of 3b-vector fields on $\cM$, parameterized by $\eps\in(0,1)$, whose coefficients are smooth on the image of $\Psi$.)
  \end{enumerate}
\end{lemma}

See Figure~\ref{FigFseRelGeo} for an illustration.

\begin{figure}[!ht]
\centering
\includegraphics{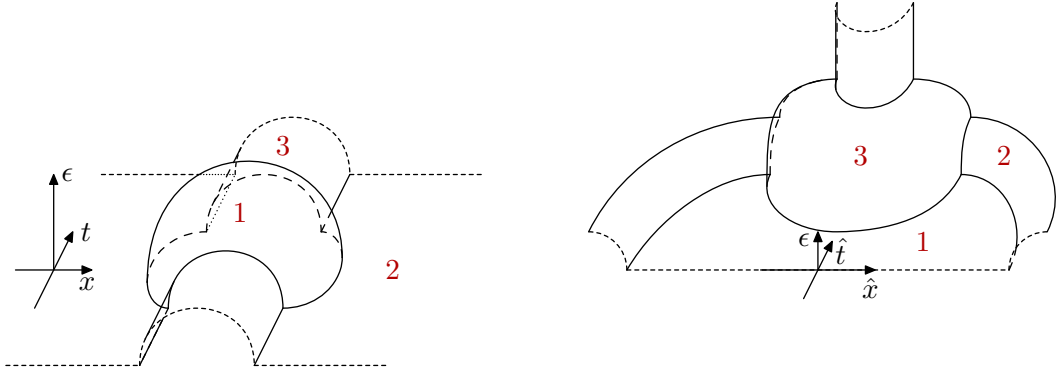}
\caption{Illustration of Lemma~\ref{LemmaFseRelGeo}. Boundary hypersurfaces which are mapped to each other under $\Psi$ are labeled by matching numbers on both sides. \textit{On the left:} the space $[\wt M;\hat M_{t_0}]$. \textit{On the right:} the subset of the space $[[0,1)\times\cM;\{0\}\times\cT;\{0\}\times\cD]$ where $\hat t\geq 0$.}
\label{FigFseRelGeo}
\end{figure}

\begin{proof}[Proof of Lemma~\usref{LemmaFseRelGeo}]
  We check part~\eqref{ItFseRelGeoMap} in local coordinates.

  Near the interior of the front face of $[\wt M;\hat M_{t_0}]$, a local coordinate system is $\eps$, $\frac{t-t_0}{\eps}$, $\frac{x}{\eps}$; and $\eps$, $\hat t$, $\hat x$ are local coordinates on $[0,1)\times\cM^\circ$.

  Next, near the interior of one of the two components of the intersection of the front face of $[\wt M;\hat M_{t_0}]$ with the lift of $\hat M$, we can use $t-t_0$, $\frac{\eps}{t-t_0}$, $\frac{x}{\eps}$ as local coordinates. These are equal to $\eps\hat t=\frac{\eps}{1/\hat t}$, $\frac{1}{\hat t}$, $\hat x$, which are local coordinates near the interior of one of the two components of the intersection of the front face of $[[0,1)\times\cM;\{0\}\times\cT]$ with the lift of $\{0\}\times\cM$ (which is disjoint from the lift of $\{0\}\times\cD$ and thus unaffected by the blow-up of the latter).

  We continue with a neighborhood of the interior of the intersection of the front face of $[\wt M;\hat M_{t_0}]$ with the lift of $M_\circ$. There, local coordinates are $|x|$, $\frac{x}{|x|}$, $\frac{t-t_0}{|x|}$, $\frac{\eps}{|x|}$. But these equal $\eps|\hat x|$, $\frac{\hat x}{|\hat x|}$, $\frac{\hat t}{|\hat x|}$, $\frac{1}{|\hat x|}$, with the last three being local coordinates on $\cM$ away from $\cT$, and the first one a local defining function of the lift of $\{0\}\times\cM$ to $[[0,1)\times\cM;\{0\}\times\cD]$ away from the lift of $\{0\}\times\cT$.

  Finally, we consider a neighborhood of the codimension $3$ corner of $[\wt M;\hat M_{t_0}]$. Starting with the local coordinates $t-t_0$, $\frac{\eps}{|x|}$, $|x|$, $\frac{x}{|x|}$ on $\wt M$ near $\hat M_{t_0}\cap M_\circ$, local coordinates on the space $[\wt M;\hat M_{t_0}]$ produced by blowing up $\hat M_{t_0}=\{t-t_0=|x|=0\}$ are $t-t_0$, $\frac{|x|}{t-t_0}$, $\frac{\eps}{|x|}$, $\frac{x}{|x|}$. In terms of $\hat t,\hat x$, these are equal to
  \begin{equation}
  \label{EqFseRelGeoCodim3}
    \eps\hat t,\ \frac{|\hat x|}{\hat t},\ \frac{1}{|\hat x|},\ \frac{\hat x}{|\hat x|}.
  \end{equation}
  On the other hand, starting with the local coordinates $\eps$, $\frac{|\hat x|}{\hat t}$, $\frac{1}{|\hat x|}$, $\frac{\hat x}{|\hat x|}$ on $[0,1)\times\cM$ near $[0,1)\times(\cT\cap\cD)$, we first blow up $\{0\}\times\cT=\{\eps=\frac{|\hat x|}{\hat t}=0\}$; in the region where $\frac{|\hat x|}{\hat t}\gtrsim\eps$ (which ends up containing a neighborhood of the relevant codimension $3$ corner), we can then use as local coordinates $\frac{\eps}{|\hat x|/\hat t}=\frac{\eps\hat t}{|\hat x|}$, $\frac{|\hat x|}{\hat t}$, $\frac{1}{|\hat x|}$, $\frac{\hat x}{|\hat x|}$, and the lift of $\{0\}\times\cD$ is given by $\{\frac{\eps\hat t}{|\hat x|}=\frac{1}{|\hat x|}=0\}$. In the region where $\frac{1}{|\hat x|}\gtrsim\frac{\eps\hat t}{|\hat x|}$, we can then introduce as local coordinates on the space on the right in~\eqref{EqFseRelGeoPsi}: $\frac{1}{|\hat x|}$, $\frac{\eps\hat t/|\hat x|}{1/|\hat x|}=\eps\hat t$, $\frac{|\hat x|}{\hat t}$, $\frac{\hat x}{|\hat x|}$. These match~\eqref{EqFseRelGeoCodim3}, as desired.

  We sketch an alternative proof based on commuting blow-ups: the map $(\eps,t,x)\mapsto(\eps,\hat t,\hat x)$ is easily seen to extend to a diffeomorphism $\Psi_0\colon[[0,1)\times\ol{\R_t\times\R^n_x};\{(0,0,0)\}]\to[[0,1)\times\ol{\R_{\hat t}\times\R^n_{\hat x}};\{0\}\times\pa\ol{\R^{1+n}}]$. We then blow up the lift of $\{0\}\times(\ol{\R_t}\times\{0\})$ in the domain (thus producing a compactification of $[\wt M;\hat M_{t_0}]$) and its image under $\Psi_0$, i.e.\ the lift of $\{0\}\times(\pa\ol{\R_{\hat t}}\times\{0\})$, in the target. The resulting diffeomorphism $\Psi_1$ restricts to $[\wt M;\hat M_{t_0}]$ to the map $\Psi$ in~\eqref{EqFseRelGeoPsi}. This uses that the image of this restriction $\Psi$ is disjoint from the lift of $[0,1)\times(\pa\ol{\R_{\hat t}}\times\{0\})$, so one may as well blow up this lift in the target without affecting the mapping properties of $\Psi$; but since $(\{0\}\times\pa\ol{\R^{1+n}})\cap([0,1)\times(\pa\ol\R\times\{0\})=\{0\}\times(\pa\ol\R\times\{0\})$, we can commute blow-ups to see that the resulting target space is
  \begin{align*}
    &\bigl[ [0,1)\times\ol{\R^{1+n}}; \{0\}\times\pa\ol{\R^{1+n}}; \{0\}\times(\pa\ol\R\times\{0\}); [0,1)\times(\pa\ol\R\times\{0\}) \bigr] \\
    &\qquad \cong \bigl[ [0,1)\times \ol{\R^{1+n}}; [0,1)\times(\pa\ol\R\times\{0\}); \{0\}\times(\pa\ol\R\times\{0\}); \{0\}\times\pa\ol{\R^{1+n}} \bigr] \\
    &\qquad = \bigl[ [0,1)\times\cM; \{0\}\times\cT; \{0\}\times\cD \bigr].
  \end{align*}

  Part~\eqref{ItFseRelGeoVF} follows from the fact that the space of se-vector fields on $\wt M$ is spanned by $\hat\rho\pa_t=\la\hat x\ra\pa_{\hat t}$ and $\hat\rho\pa_x=\la\hat x\ra\pa_{\hat x}$.
\end{proof}

\begin{lemma}[se- and 3b-Sobolev spaces]
\label{LemmaFseHse3b}
  Take $M=\R_t\times\R^n_x$, $\cC=\R_t\times\{0\}$ to define $\wt M$, and set $\cM=[\ol{\R^{1+n}_{\hat t,\hat x}};\pa(\ol\R\times\{0\})]$. Fix the densities $|\dd t\,\dd x|$ and $|\dd\hat t\,\dd\hat x|$ on $M$ and $\cM$. Define $\Psi$ as in~\eqref{EqFseRelGeoPsi0}{\rm--}\eqref{EqFseRelGeoPsi}. Let $\lambda>0$ and set $\Omega=\{|t|<\lambda,\ |x|<\lambda\}\subset M$, $\Omega_\eps=\{|\hat t|<\eps^{-1}\lambda,\ |\hat x|<\eps^{-1}\lambda\}\subset\cM$. Let $m\in\N_0$, $\alpha_\circ,\hat\alpha\in\R$. Then we have a uniform equivalence of norms
  \begin{equation}
  \label{EqFseHse3b}
    \|u\|_{H_{\seop,\eps}^{m,\alpha_\circ,\hat\alpha}} \sim \eps^{\frac{n+1}{2}-\hat\alpha}\|(\Psi_*u)(\eps,-)\|_{H_\tbop^{m,\alpha_\circ-\hat\alpha,0}}
  \end{equation}
  when $u$ has support in $\Omega$ (and thus $(\Psi_*u)(\eps,-)$ has support in $\Omega_\eps$).
\end{lemma}
\begin{proof}
  Since
  \[
    \|u\|_{\bar H_{\seop,\eps}^{m,\alpha_\circ,\hat\alpha}(\Omega)} = \eps^{-\hat\alpha} \|u\|_{\bar H_{\seop,\eps}^{m,\alpha_\circ-\hat\alpha,0}(\Omega)},
  \]
  it suffices to consider the case $\hat\alpha=0$. The norm equivalence~\eqref{EqFseHse3b} then follows for $m=\alpha_\circ=0$ from $|\dd t\,\dd x|=\eps^{n+1}|\dd\hat t\,\dd\hat x|$, for general $m$ from Lemma~\ref{LemmaFseRelGeo}\eqref{ItFseRelGeoVF}, and for general $\alpha_\circ$ from the relationship between $M_\circ$ and $\cD\subset\cM$ in Lemma~\ref{LemmaFseRelGeo}\eqref{ItFseRelGeoMap}. Note here that the closure of $\bigcup_{\eps\in(0,1)}\{\eps\}\times\Omega_\eps$ in $\wt\cM$ (see~\eqref{EqFseRelGeoPsi}) is disjoint from the `lateral' boundary hypersurfaces (i.e.\ the lifts of $[0,1)\times\cT$ and $[0,1)\times\cD$).
\end{proof}

For $\hat\alpha=0$ and $u$ supported in $\ol\Omega$, one can further characterize the 3b-norm on the right in~\eqref{EqFseHse3b} using the Fourier transform via~\eqref{EqF3FT}.

\subsection{Higher regularity: s-analysis}
\label{SsFs}

Note that se-regularity near $t=t_0$, $x=0$ amounts to uniform regularity in the `fast' variables $\hat t,\hat x$ from~\eqref{EqFseHatCoord} when $\hat t,\hat x$ are bounded. Regularity in the `slow' time variable $t$ is a stronger notion which we capture as follows.

\begin{definition}[s-vector fields]
\label{DefFs}
  In the notation of Definition~\usref{DefFse}, we define the Lie algebra and $\CI(\wt M)$-module of \emph{s-vector fields} by
  \[
    \Vs(\wt M) = \{ V\in\Vb(\wt M) \colon \dd\eps(V)=0 \}.
  \]
  The corresponding space of $m$-th order s-differential operators is denoted $\Diffs^m(\wt M)$.
\end{definition}

In local coordinates $t,x$ on $M$, with $\cC=\{x=0\}$, a local spanning set of $\Vs(\wt M)$ is given by $\pa_t$, $\hat\rho\pa_x$ where $\hat\rho=(\eps^2+|x|^2)^{1/2}$. Their restrictions to $M_\circ$ are $\pa_t$, $|x|\pa_x$ and thus span the space of b-vector fields on $M_\circ$. On the other hand, expressed in terms of $\hat t=\frac{t-t_0}{\eps}$, $\hat x=\frac{x}{\eps}$, these vector fields are $\eps^{-1}\pa_{\hat t}$, $\la\hat x\ra\pa_{\hat x}$.

\begin{rmk}[s-vector fields from the 3b-perspective]
\label{RmkFsLifts3b}
  The lifts of $\eps^{-1}\pa_{\hat t}$, $\la\hat x\ra\pa_{\hat x}$ to $[\wt M;\hat M_{t_0}]$ do \emph{not} restrict to the front face $\cM=\ff[\wt M;\hat M_{t_0}]$ to a subset, let alone a spanning set, of $\Vb(\cM)$. Thus, s-regularity on $\wt M$ translates into a notion \emph{different} from b-regularity on $\cM$. This may at first seem surprising since in the 3b-setting the strongest sensible notion of regularity for 3b-equations on $\cM$ is b-regularity on $\cM$ (see e.g.\ \cite[Corollary~6.16]{Hintz3b} or \cite[Corollary~5.35]{HintzNonstat}); but of course b-regularity on $\cM$ still measures, in any compact subset of $\cM^\circ$, only regularity in the \emph{fast} time $\hat t$.
\end{rmk}

\begin{definition}[Function spaces]
\label{DefFsFn}
  Let $M$ be compact, $m,k\in\N_0$, $\alpha_\circ,\hat\alpha\in\R$.
  \begin{enumerate}
  \item{\rm (Mixed (se;s)-Sobolev spaces.)} For $\eps\in(0,1)$, we define
    \begin{equation}
    \label{EqFsFnMixed}
      H_{(\seop;\sop),\eps}^{(m;k),\alpha_\circ,\hat\alpha}(M) := H^{m+k}(M)
    \end{equation}
    with norm
    \[
      \|u\|_{H_{(\seop;\sop),\eps}^{(m;k),\alpha_\circ,\hat\alpha}(M)}^2 := \sum_j \|Q_j u\|_{H_{\seop,\eps}^{m,\alpha_\circ,\hat\alpha}(M)}^2
    \]
    where $\{Q_j\}\subset\Diffs^k(\wt M)$ is a fixed finite spanning set of $\Diffs^k(\wt M)$ over $\CI(\wt M)$. (The $L^2$-space underlying the norm on $H_{\seop,\eps}$ is defined using a positive smooth density on $M$ unless noted otherwise.)
  \item{\rm (s-Sobolev spaces.)} In the special case $m=0$, we write
    \begin{equation}
    \label{EqFsFnPure}
      H_{\sop,\eps}^{k,\alpha_\circ,\hat\alpha}(M) = H^k(M),\qquad \|u\|_{H_{\sop,\eps}^{k,\alpha_\circ,\hat\alpha}(M)}^2 = \sum_j \|\rho_\circ^{-\alpha_\circ}\hat\rho^{-\hat\alpha}Q_j u\|_{L^2(M)}^2
    \end{equation}
  \item{\rm (s-continuity.)} We define $\cC_{\sop,\eps}^{k,\alpha_\circ,\hat\alpha}(M)$ and $\cC_\sop^{k,\alpha_\circ,\hat\alpha}(\wt M)$ analogously to Definition~\ref{DefFseCont}.
  \item\label{ItFsFnMixedC}{\rm (Mixed (se;s)-continuity.)} We define $\cC_{(\seop;\sop),\eps}^{(k;m),\alpha_\circ,\hat\alpha}(M)=\cC^{k+m}(M)$ but with norm
    \[
      \|u\|_{\cC_{(\seop;\sop),\eps}^{(k;m),\alpha_\circ,\hat\alpha}(M)} := \sum_{j,l} \|\rho_\circ^{-\alpha_\circ}\hat\rho^{-\hat\alpha}P_j Q_l u\|_{L^\infty(M)}
    \]
    where $\{P_j\}\subset\Diffse^k(\wt M)$ and $\{Q_l\}\subset\Diffs^m(\wt M)$ are fixed finite spanning sets over $\CI(\wt M)$. We then define $\cC_{\seop;\sop}^{(k;m),\alpha_\circ,\hat\alpha}(\wt M)$ analogously to Definition~\ref{DefFseCont}.
  \end{enumerate}
\end{definition}

We record the following version of Sobolev embedding.

\begin{prop}[s-Sobolev embedding]
\label{PropFsSobEmb}
  Suppose $k>\frac{\dim(M)}{2}+l=\frac{n+1}{2}+l$. Then
  \begin{equation}
  \label{EqFsSobEmb}
    H_{\sop,\eps}^{k,\alpha_\circ,\hat\alpha}(M) \hra \cC_{\sop,\eps}^{l,\alpha_\circ,\hat\alpha-\frac{n}{2}}(M),
  \end{equation}
  and the inclusion map has uniformly bounded operator norm. Conversely, for all $\eta>0$, the map
  \begin{equation}
  \label{EqFsIncl}
    \cC_{\sop,\eps}^{l,\alpha_\circ,\hat\alpha}(M) \hra H_{\sop,\eps}^{l,\alpha_\circ,\hat\alpha+\frac{n}{2}-\eta}(M)
  \end{equation}
  has uniformly bounded operator norm.
\end{prop}
\begin{proof}
  For any fixed $\eps>0$,~\eqref{EqFsSobEmb} is the standard Sobolev embedding $H^k(M)\hra\cC^l(M)$. It suffices to prove the uniform boundedness of~\eqref{EqFsSobEmb} for $l=0$. If $U\subset M$ is an open subset with $\ol U\cap\cC=\emptyset$, and $\chi\in\CIc(U)$, then $\|\chi u\|_{H_{\sop,\eps}^{k,\alpha_\circ,\hat\alpha}(M)}\sim\eps^{-\alpha_\circ}\|\chi u\|_{H^k(M)}$ (with `$\sim$' denoting uniform equivalence of norms), which controls $\eps^{-\alpha_\circ}\|\chi u\|_{L^\infty(M)}$. It thus suffices to work near the front face $\hat M\subset\wt M$ and in local coordinates $t$, $x$ as above. We shall reduce the Proposition to Sobolev embedding on $\R^{1+n}$ by passing to rescaled coordinates in which s-vector fields are uniformly equivalent to coordinate vector fields on $\R^{1+n}$. (This is thus an instance of a bounded geometry perspective, here for $\Vs(\wt M)$.)

  First, for $t_0<t_1$ and $\hat r_0>0$, consider a set
  \[
    U := \Bigl\{ (\eps,t,x) \colon \eps\in(0,1),\ t\in(t_0,t_1),\ \Bigl|\frac{x}{\eps}\Bigr|<\hat r_0 \Bigr\} \subset \wt M,\qquad
    U_\eps := U \cap M_\eps.
  \]
  On $U$, we may use $\hat\rho=\eps$ and $\rho_\circ=1$. Define the map $\Phi_\eps\colon(t,x)\mapsto(t,\frac{x}{\eps})$; then
  \[
    \Phi_\eps(U_\eps) = V := (t_0,t_1) \times B_{\hat r_0}(0) \subset \R_t \times \R^n_{\hat x},\qquad
    (\Phi_\eps)_*\pa_t = \pa_t, \quad
    (\Phi_\eps)_*(\eps\pa_x) = \pa_{\hat x}.
  \]
  Furthermore, $(\Phi_\eps)_*(|\dd t\,\dd x|)=\eps^n\,|\dd t\,\dd\hat x|$. Therefore, for $\chi\in\CIc(V)$, we have
  \[
    \|(\Phi_\eps^*\chi) u\|_{H_{\sop,\eps}^{k,\alpha_\circ,\hat\alpha}(M)} \sim \eps^{\frac{n}{2}-\hat\alpha} \| \chi((\Phi_\eps)_*u) \|_{H^k(\R^{1+n})} \gtrsim \eps^{\frac{n}{2}-\hat\alpha}\|\chi((\Phi_\eps)_*u)\|_{L^\infty} = \eps^{\frac{n}{2}-\hat\alpha}\|\chi u\|_{L^\infty}.
  \]

  Second, for $t_0<t_1$ and $r_0>0$, consider now
  \[
    U' := \bigl\{ (\eps,t,x) \colon \eps \in (0,1),\ t\in(t_0,t_1),\ \eps\leq|x|\leq r_0 \bigr\} \subset \wt M, \qquad
    U'_\eps := U \cap M_\eps.
  \]
  On this set, we may take $\hat\rho=|x|$ and $\rho_\circ=\frac{\eps}{|x|}$. Passing to polar coordinates $x=r\omega$, define $\Phi'\colon(t,r,\omega)\mapsto(t,-\log r,\omega)$; then s-vector fields are $\CI(\wt M)$-linear combinations of $\pa_t$, $\hat\rho\pa_r=r\pa_r$, $\pa_\omega$, and we record
  \begin{align*}
    &\Phi'(U'_\eps) = V'_\eps := (t_0,t_1) \times (-\log r_0,-\log\eps) \times \Sph^{n-1} \subset \R_t\times(-\log r_0,\infty)_R\times\Sph^{n-1}_\omega, \\
    &\qquad \Phi'_*\pa_t=\pa_t,\quad \Phi'_*(\hat\rho\pa_r)=-\pa_R,\quad \Phi'_*\pa_\omega=\pa_\omega.
  \end{align*}
  Passing on $U'$ from the density $|\dd t\,\dd x|=r^{n-1}|\dd t\,\dd r\,\dd\omega|$ on $M$ to $|\dd t\,\frac{\dd r}{r}\,\dd\omega|$ amounts to replacing $\hat\alpha$ by $\hat\alpha-\frac{n}{2}$; and then we note that $\Phi_*'(|\dd t\,\frac{\dd r}{r}\,\dd\omega|)=|\dd t\,\dd R\,\dd\omega|$. Shifting the orders so that $\alpha_\circ=\hat\alpha-\frac{n}{2}=0$, we then observe that
  \[
    \|\chi u\|_{H_{\sop,\eps}^{k,\alpha_\circ,\hat\alpha}(M;|\dd t\,\dd x|)} \sim \|\chi u\|_{H_{\sop,\eps}^k(M;|\dd t\,\frac{\dd r}{r}\,\dd\omega|)} \sim \| \Phi'_*(\chi u) \|_{H^k(\R_t\times\R_R\times\Sph^{n-1}_\omega)} \gtrsim \|\chi u\|_{L^\infty},
  \]
  completing the proof of~\eqref{EqFsSobEmb}.

  We verify~\eqref{EqFsIncl} near the corner $M_\circ\cap\hat M$ in coordinates $t\in I\Subset\R$, $c\eps\leq r\leq C$, and for $\alpha_\circ=\hat\alpha=0$ and $l=0$. We drop integration over $(n-1)$-spheres from the notation. Then
  \[
    \int_I\int_{c\eps}^C |r^{-\frac{n}{2}+\eta}u(\eps,t,r)|^2\,r^{n-1}\,\dd t\,\dd r \leq \|u\|_{L^\infty} |I| \int_{c\eps}^C r^{-1+2\eta}\,\dd r
  \]
  is uniformly bounded as $\eps\searrow 0$, as required.
\end{proof}

We next discuss the mapping properties of se- and s-differential operators on the mixed Sobolev spaces~\eqref{EqFsFnMixed}. We first note:

\begin{lemma}[Commutator s-vector fields]
\label{LemmaFsComm}
  Denote by
  \[
    \cV_{[\sop]}(\wt M)\subset\cV_\sop(\wt M)
  \]
  the space consisting of all s-vector fields $V$ with the property that $[V,W]\in\Vse(\wt M)$ for all $W\in\Vse(\wt M)$. Then $\cV_{[\sop]}(\wt M)$ spans $\cV_\sop(\wt M)$ over $\CI_\sop(\wt M)$.
\end{lemma}
\begin{proof}
  For $\chi\in\CI(\wt M)$ with $\supp\chi\cap\hat M=\emptyset$, we have $\chi\cV_{[\sop]}=\chi\cV_\sop$. Furthermore, $\Vse(\wt M)\subset\cV_{[\sop]}(\wt M)$.

  In local coordinates $t,x$ on $M$, a spanning set of $\Vs(\wt M)$ is given by the union of $\{\pa_t\}$ and $\Vse(\wt M)$. Note then that for all $\psi\in\CI(\wt M)$ with support in the coordinate chart for which $\psi|_{\hat M}=\hat\upbeta^*\phi$ for some $\phi\in\CI(\cC)$ (i.e.\ $\psi$ is fiber-constant on $\hat M$), and for $W\in\Vse(\wt M)$, we have
  \begin{equation}
  \label{EqFsOrderpatComm}
    [\psi\pa_t,W] = \psi[\pa_t,W]-(W\psi)\pa_t.
  \end{equation}
  But $[\pa_t,W]\in\Vse(\wt M)$, and $W\psi\in\eps\CI_\sop(\wt M)$; therefore $[\psi\pa_t,W]\in\Vse(\wt M)$. This proves $\psi\pa_t\in\cV_{[\sop]}(\wt M)$. Since every element of $\Vs(\wt M)$ can be written as a linear combination of $\psi\pa_t$ and elements of $\Vse(\wt M)$ with $\CI_\sop(\wt M)$-coefficients, the lemma is proved.
\end{proof}

The following simple consequence plays a crucial role in the regularity theory for se-operators.

\begin{cor}[Commutators of se-operators and commutator s-vector fields]
\label{CorFsComm}
  Let $P\in\Diffse^{m,\ell_\circ,\hat\ell}(\wt M)$. Let $\chi\in\CI(\wt M)$ be equal to $1$ near $\hat M$ and supported in the lift of the coordinate chart $\eps,t,x$ on $\wt M'$ around $\cC$. Then $[\chi\pa_t,P]\in\Diffse^{m,\ell_\circ,\hat\ell}(\wt M)$. If, moreover, the normal operators $N_{\hat M_{t_0}}(\eps^{\hat\ell}P)$ are independent of $t_0$, then $[\chi\pa_t,P]\in\Diffse^{m,\ell_\circ,\hat\ell-1}(\wt M)$.
\end{cor}
\begin{proof}
  The first statement follows from $\chi\pa_t\in\cV_{[\sop]}$ via induction on $m$. For the second part, it suffices to work near a compact subset of $\hat M^\circ$ and with $\hat\ell=0$; the claim then follows by direct differentiation of~\eqref{EqFseNormalP} where we take $\hat\rho=\eps$ and $a_{j\alpha}=a_{j\alpha}(\eps,t,\hat x)$, with $a_{j\alpha}$ independent of $t$ at $\eps=0$.
\end{proof}

\begin{cor}[Order of se- and s-operators]
\label{CorFsOrder}
  Given $P\in\Diffse^m(\wt M)$ and $Q\in\Diffs^q(\wt M)$, one can write $P Q=\sum_j Q_j P_j$ and $Q P=\sum_k P'_k Q'_k$ (finite sums) for suitable $P_j,P'_k\in\Diffse^m(\wt M)$ and $Q_j,Q'_k\in\Diffs^q(\wt M)$.
\end{cor}
\begin{proof}
  This follows from Lemma~\ref{LemmaFsComm} by a simple inductive argument in $q$.
\end{proof}

Therefore, when $M$ is compact, every $P\in\Diffse^m(\wt M)$ defines a uniformly bounded map
\[
  P \colon H_{(\seop;\sop),\eps}^{(s;k)}(M) \to H_{(\seop;\sop),\eps}^{(s-m;k)}(M),
\]
similarly for weighted spaces and operators. One can relax the regularity requirements to $P\in\cC_{\seop;\sop}^{(d_0;k)}\Diffse^m(\wt M)$ for some sufficiently large $d_0$ (depending only on $s,m$). Analogous mapping properties hold on the spaces $\cC_{\seop;\sop}^{(s;k)}$. Corollary~\ref{CorFsOrder} also implies that changing the order of the operators $P_j$ and $Q_l$ in Definition~\ref{DefFsFn}\eqref{ItFsFnMixedC} leads to a uniformly equivalent norm.

\subsection{Pseudodifferential operators; variable order spaces}
\label{SsFVar}

In our analysis we shall need to work with spaces $H_{\seop,\eps}^\sfs(M)$ whose differential order is a function on se-phase space. Such spaces are defined microlocally, i.e.\ via testing with variable order se-pseu\-do\-dif\-fer\-en\-tial operators. We define such operators in~\S\ref{SssFVarse} below, following a recapitulation (and rephrasing from a bounded geometry perspective) from \cite{Hintz3b,HintzConicWave} of the edge and 3b-settings in~\S\S\ref{SssFVarE}--\ref{SssFVar3b}. The ps.d.o.s (of class $\tilde\Psi_\seop$) defined in~\S\ref{SssFVarse} do not act uniformly boundedly on the mixed function spaces of Definition~\ref{DefFsFn}, since they do not preserve s-regularity. In~\S\ref{SssFVarsse}, we use the parameterized scaled bounded geometry perspective introduced in \cite{HintzScaledBddGeo} to define a more restrictive class of se-ps.d.o.s (denoted $\Psi_\seop$) whose elements \emph{do} act uniformly boundedly on mixed function spaces; these ps.d.o.s will be the key tools to track se-regularity microlocally in spaces encoding fixed integer amounts of s-regularity.

\subsubsection{Edge setting}
\label{SssFVarE}

In the edge setting, we only consider here the local coordinate description of (variable order) edge pseudodifferential operators in the special case of interest; we refer the reader to Mazzeo's original paper \cite{MazzeoEdge} and also to \cite[Appendix~A]{HintzConicWave} for details. We work in coordinates $t\in\R$, $r\geq 0$, $\omega\in\Sph^{n-1}$ on the manifold $M_\circ$ as in~\S\ref{SsFE}, and indeed in a coordinate patch on $\Sph^{n-1}$, so we regard $\omega\in\R^{n-1}$; for concreteness, we use projective coordinates $\omega=(\frac{x^j}{x^1})_{j=2,\ldots,n}$ on $\Sph^{n-1}\subset\R^n_x$. Edge vector fields are thus spanned by $r\pa_t$, $r\pa_r$, $\pa_\omega$. Given a variable order function $\sfs\in\CI(\Se^*M_\circ)$ and a symbol $a\in S^\sfs(\Te^*M_\circ)$ with support lying over a compact subset of the coordinate patch, we can define the quantization
\[
  \Op_\eop(a)\in\Psie^\sfs(M_\circ)
\]
as follows. Write edge covectors as $-\sigma_\eop\,\frac{\dd t}{r}+\xi_\eop\,\frac{\dd r}{r}+\eta_\eop\cdot\dd\omega$; fix a cutoff $\phi\in\CIc((-2,2))$ which is equal to $1$ on $[-1,1]$, and a cutoff $\chi\in\CIc(\R^{n-1})$ which equals $1$ near the $\omega$-support of $a$; then
\begin{equation}
\label{EqFVarE}
\begin{split}
  (\Op_\eop(a)u)(t,r,\omega) = (2\pi)^{-n-1}\iiint&\,\exp\Bigl[i\Bigl(-\sigma_\eop\frac{t-t'}{r'}+\xi_\eop\log\Bigl(\frac{r}{r'}\Bigr)+\eta_\eop\cdot(\omega-\omega')\Bigr)\Bigr] \\
    &\times \phi\Bigl(\Bigl|\frac{t-t'}{r'}\Bigr|\Bigr)\phi\Bigl(\log\Bigl(\frac{r}{r'}\Bigr)\Bigr)\chi(\omega')a(t,r,\omega;\sigma_\eop,\xi_\eop,\eta_\eop) \\
    &\times u(t',r',\omega')\,\dd\sigma_\eop\,\dd\xi_\eop\,\dd\eta_\eop\,\frac{\dd t'}{r'}\,\frac{\dd r'}{r'}\,\dd\omega'.
\end{split}
\end{equation}
The principal symbol of $A=\Op_\eop(a)$ is $a\bmod\bigcap_{\delta>0}S^{\sfs-1+2\delta}(\Te^*M_\circ)$, and one then defines the elliptic set $\Ell_\eop(A)$, resp.\ operator wave front set $\WF'_\eop(A)$ of $A$ to be the elliptic set, resp.\ essential support, of $a$ as usual; they can be regarded equivalently as conic subsets of $\Te^*M_\circ\setminus o$ (or their closures in $\ol{\Te^*}M_\circ\setminus o$) or as subsets of fiber infinity $\Se^*M_\circ$. Spaces of weighted ps.d.o.s are denoted $\Psie^{\sfs,\ell}(M_\circ)=r^{-\ell}\Psie^\sfs(M)$.

When $M_\circ$ is compact and equipped with a smooth weighted b-density, variable order edge spaces $\He^\sfs(M_\circ)$ can now be defined via testing with elliptic edge-ps.d.o.s: let $A\in\Psie^\sfs(M_\circ)$ be elliptic and fix any constant $s_0\leq\inf\sfs$, then we set
\[
  \|u\|_{\He^\sfs(M_\circ)}^2 := \|u\|_{\He^{s_0}(M_\circ)}^2 + \|A u\|_{L^2(M_\circ)}^2;
\]
similarly for weighted spaces $\He^{\sfs,\alpha}(M_\circ)$.

A more direct approach to variable order edge Sobolev spaces (discussed also in \cite[Appendix~A.2]{HintzConicWave}) is to regard $(M_\circ)^\circ$, equipped with any fixed smooth Riemannian edge metric, as a \emph{manifold with bounded geometry} \cite{ShubinBounded}. We shall freely use the material from Appendix~\ref{SB}. In local coordinates as above, and restricting further to bounded $r$, this amounts to covering a neighborhood of $\R_t\times(0,1]_r\times\R^{n-1}_\omega$ by the sets
\begin{equation}
\label{EqFVarEBdd}
  U_{j k} = \bigl((j-1)2^{-k},(j+1)2^{-k}\bigr)_t \times (2^{-k-1},2^{-k+1})_r \times \{ |\omega|<2 \},
\end{equation}
where $j\in\Z$, $k\in\N_0$ (or, more generally, $k\in\Z$ bounded from below, in order to cover a larger range of $r$). Under the maps
\begin{equation}
\label{EqFVarEBddMap}
  \phi_{j k}\colon U_{j k}\ni(t,r,\omega)\mapsto(2^k t-j,2^k r,\omega)\in U := (-1,1)_T\times\Bigl(\frac12,2\Bigr)_R\times\{|\omega|<2\},
\end{equation}
the vector of basic edge vector fields $(r\pa_t,r\pa_r,\pa_\omega)$ pushes forward to $A_{j k}(\pa_T,\pa_R,\pa_\omega)$ where the $(n+1)\times(n+1)$ matrix valued function $A_{j k}$ on $U$ and its inverse are uniformly (for all $j,k$, and with all derivatives) bounded. Therefore, elements of $S^s(\Te^*M_\circ)$ with compact support contained in $r\leq 1$ push forward to a uniformly bounded family of elements of $S^s(T^*U)$; when $M_\circ$ is compact, this implies that $S^s(\Te^*M_\circ)\subset S^s_{\rm uni}(T^*(M_\circ)^\circ)$. In fact, this remains true for symbols of class $\CI_\eop S^s(\Te^*M_\circ)$ where we write $\CI_\eop(M_\circ)\subset\CI((M_\circ)^\circ)$ for the space of all functions which remain uniformly bounded on $M_\circ$ upon application of any number of smooth edge vector fields. Symbols of class $\CI_\eop S^s(\Te^*M_\circ)$ thus have only edge regularity in the base variables.

We can then quantize symbols in $S^s_{\rm uni}(T^*(M_\circ)^\circ)$ to define the space $\Psi^s_{\rm uni}((M_\circ)^\circ)$. In the present context this is thus the space of $s$-th order edge pseudodifferential operators with edge-regular coefficients, which we denote $\CI_\eop\Psie^s(M_\circ)$; indeed one can show that for $a\in S^s(\Te^*M_\circ)$ the quantizations $\Op_\eop(a)$ (defined via~\eqref{EqFVarE}) and $\Op_{\rm uni}(a)$ agree modulo $\CI_\eop\Psie^{-\infty}(M_\circ)$, essentially since the cutoffs in~\eqref{EqFVarE} localize $(t,r,\omega)$ and $(t',r',\omega')$ to the same chart $U_{j k}$ or, more accurately, a uniformly bounded number of such charts. The same considerations apply to quantizations of symbols of variable order $\sfs\in\CI(\Se^*M_\circ)$ (and indeed the regularity $\sfs\in\CI_\eop(\Se^*M_\circ)$ is sufficient). One can then further equivalently define $\He^\sfs(M_\circ):=H_{\rm uni}^\sfs((M_\circ)^\circ)$, and similarly for weighted spaces; the $L^2$ space here is defined with respect to a uniformly positive and bounded density, which here means a smooth positive edge density.

Given $A\in\Psi^s_{\rm uni}((M_\circ)^\circ)$, the fact that $(M_\circ)^\circ$ is, by definition, the interior of the manifold $M_\circ$ allows one to define the elliptic set as a subset not only of $S^*(M_\circ)^\circ$, but of $\Se^*M_\circ$ in essentially the usual fashion: $A$ is elliptic at $\zeta\in\Se^*M_\circ$ if and only if there exists an open neighborhood $U\subset\Se^*M_\circ$ of $\zeta$ so that $A$ is uniformly elliptic on $U\cap S^*(M_\circ)^\circ$. One can also use compactifications of $S^*(M_\circ)^\circ$ other than $\Se^*M_\circ$. For example, if $Z\subset\pa M_\circ$ is a boundary fiber, one can define
\[
  \Ell_\eop(A) \subset \upbeta^*(\Se^*M_\circ),\qquad \upbeta\colon[M_\circ;Z]\to M_\circ,
\]
in the same fashion. This is useful since cutoffs like $\chi(t/r)$, with $\chi\in\CIc(\R)$, are on the one hand only elements of $\CI_\eop\subset\CI_\eop\Psie^0$ and not of $\CI\subset\Psie^0$, but on the other hand lift to be smooth on the blow-up of $\R_t\times[0,\infty)_r\times\Sph^{n-1}$ along $Z=\{0\}\times\{0\}\times\Sph^{n-1}$, so for example $\Ell_\eop(\chi(t/r))$, as a subset of $\Se^*M_\circ$, does not contain any points over $\pa M_\circ$, whereas as a subset of $\upbeta^*(\Se^*M_\circ)$ it does (namely the points where $\chi(t/r)\neq 0$, with $t/r$ being smooth near the interior of the front face of $[M_\circ;Z]$). --- The same considerations apply to the operator wave front set.

\subsubsection{3b-setting}
\label{SssFVar3b}

We now turn to spaces of 3b-pseudodifferential operators on the compactification $\cM$ of $\R_t\times\R^n_x$ defined in~\eqref{EqF3cM}. These spaces were first defined in \cite[Definition~4.3]{Hintz3b} via their Schwartz kernels; here, we instead use the corresponding bounded geometry perspective on $\R_t\times\R^n_x$ indicated in \cite[Definition~4.5]{Hintz3b}. (We drop the hats from the notation for better readability.) Concretely, recalling that 3b-vector fields are spanned by $\la x\ra\pa_t$, $\la x\ra\pa_x$, we cover $\R^{1+n}$ with two families of distinguished open sets. The first family covers the region $|x|<2$ and is given by
\begin{equation}
\label{EqFVar3b1}
  U^{(1)}_j = (j-1,j+1)_t \times B_{\R^n}(0,2)_x,\qquad B_{\R^n}(x_0,R)=\{x\in\R^n\colon|x-x_0|<R\},\quad j\in\Z.
\end{equation}
The second family covers the region $|x|>\frac12$. Working the region where $\pm x^1>\frac14\max_{j\neq 1}|x^j|$, we use local coordinates $\omega^j=\frac{x^j}{x^1}$, $j=2,\ldots,n$, on $\Sph^{n-1}$ (so $|\omega^j|<4$) and the sets
\begin{equation}
\label{EqFVar3b2}
  U^{(2)}_{\pm 1,j k} = \bigl(2^k(j-1),2^k(j+1)\bigr)_t \times (2^{k-1},2^{k+1})_r \times B_{\R^{n-1}}(0,2)_\omega, \qquad j\in\Z,\ k\in\N_0,
\end{equation}
which are subsets of $\R^{1+n}$ via $x^1=\pm r/\sqrt{1+|\omega|^2}$, $x^j=x^1\omega^j$ ($j=2,\ldots,n$); we define the sets $U_{\pm m,j k}^{(2)}$ for $m=2,\ldots,n$ analogously. See Figure~\ref{FigFVar3b}. We can then map $U_j^{(1)}\to U^{(1)}:=(-1,1)_T\times B_{\R^n}(0,2)_X$ via $(t,x)\mapsto(t-j,x)$, and the basic 3b-vector fields $\pa_t$, $\pa_x$ push forward to $\pa_T$, $\pa_X$; similarly, we map
\[
  U_{\pm 1,j k}^{(2)} \ni (t,r,\omega) \mapsto (T,R,\omega) \in U^{(2)}:=(-1,1)_T\times\Bigl(\frac12,2\Bigr)_R\times B_{\R^{n-1}}(0,2)_\omega
\]
via $(T,R,\omega)=(2^{-k}t-j,2^{-k}r,\omega)$; the pushforwards of $r\pa_t$, $r\pa_r$, $\pa_\omega$ under this map are $R\pa_T$, $R\pa_R$, $\pa_\omega$. Note then that $R,R^{-1}$ are bounded on $U^{(2)}$. We conclude that a 3b-vector field on $\cM$ is an element of $\CI_{\rm uni}(\R^{1+n};T\R^{1+n})$, and indeed we have $\CI_{\rm uni}(\R^{1+n};T\R^{1+n})=\CI_\tbop\Vtb(\cM)$ where $\CI_\tbop(\cM)$ is the space of bounded smooth functions on $\cM^\circ=\R^{1+n}$ which remain bounded upon application of any number of 3b-vector fields on $\cM$.

\begin{figure}[!ht]
\centering
\includegraphics{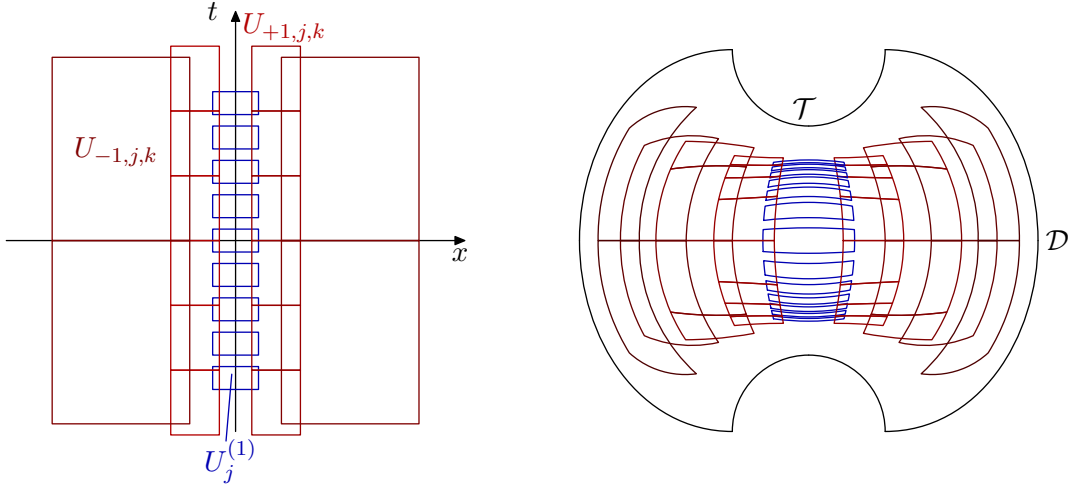}
\caption{Illustration of some of the coordinate charts for the bounded geometry perspective on $\Vtb(\cM)$, here in the 2-dimensional case $\cM=[\ol{\R_t\times\R_x},\pa(\ol\R_t\times\{0\})]$. Charts of type~\eqref{EqFVar3b1} are drawn in blue, charts of type~\eqref{EqFVar3b2} in red. \textit{On the left:} picture in $t,x$ coordinates. \textit{On the right:} the same picture (with some additional charts), drawn in a compactified fashion.}
\label{FigFVar3b}
\end{figure}

We can now define (variable order, weighted) 3b-Sobolev spaces $\Htb^{\sfs,\alpha_\cD,\alpha_\cT}(\cM)$ as weighted (with weight $w=\rho_\cD^{\alpha_\cD}\rho_\cT^{\alpha_\cT}$) uniform Sobolev spaces on $\cM^\circ$; here $\sfs\in\CI(\Stb^*\cM)$ (or even $\sfs\in\CI_\tbop(\Stb^*\cM)$, although we do not need this level of generality). The underlying volume density is a positive 3b-density; usage of other weighted b-densities can be accommodated by shifting the weights appropriately. This defines the same space (up to equivalence of norms) as in \cite[\S{4.5}]{Hintz3b} where smooth coefficient 3b-ps.d.o.s are utilized instead of the uniform ps.d.o.s here (which are 3b-ps.d.o.s with 3b-regular coefficients), for the smooth coefficient 3b-ps.d.o.s, upon localization of their Schwartz kernel to a collar neighborhood of the 3b-diagonal, are special instances of uniform ps.d.o.s, and the notions of ellipticity coincide.

We now recall from \cite[Proposition~4.29(1)]{Hintz3b} the variable order version of the isomorphism~\eqref{EqF3FT} for variable orders $\sfs\in\CI(\Stb^*\cM)$ which are \emph{stationary}, i.e.\ invariant under the lift of $t$-translations to $\Ttb^*\cM$. (In coordinates $-\sigma_\tbop\frac{\dd t}{\la x\ra}+\xi_\tbop\frac{\dd x}{\la x\ra}$, the lift of the translation $(t,x)\mapsto(t+a,x)$ is given simply by $(t,x,\sigma_\tbop,\xi_\tbop)\mapsto(t+a,x,\sigma_\tbop,\xi_\tbop)$.) Namely, using a density $|\dd t|\otimes\mu$ on $\cM$ where $\mu$ is a weighted b-density on $\cX=\ol{\R^n_x}$, we have the isomorphism
\begin{equation}
\label{EqFVar3bFT}
\begin{split}
  \cF &\colon \Htb^{\sfs,\alpha_\cD,0}(\cM) \to L^2\bigl(\R_\sigma; H_{\wh{\tbop},\sigma}^{\sfs,\alpha_\cD}(\cX)\bigr), \\
  &\qquad
    H_{\wh{\tbop},\sigma}^{\sfs,\alpha_\cD}(\cX) :=
      \begin{cases}
        H_{\scbtop,|\sigma|}^{\sfs_\infty,\sfs_\scop+\alpha_\cD,\alpha_\cD,0}(\cX), & |\sigma|\leq 1, \\
        H_{\scop,|\sigma|^{-1}}^{\sfs_\infty,\sfs_\scop+\alpha_\cD,\sfs_\semi}(\cX), & |\sigma| > 1,
      \end{cases}
\end{split}
\end{equation}
where for $|\sigma|\leq 1$ the orders $\sfs_\infty$ (variable sc-b-transition regularity order) and $\sfs_\scop$ (variable sc-decay order at $\scface\subset\cX_\scbtop$) are defined in terms of $\sfs$ via pullback along identifications of the sc-b-transition phase space and 3b-phase space as in \cite[Lemmas~4.28(1) and 3.12]{Hintz3b}; and similarly for $|\sigma|>1$, the orders $\sfs_\infty$ (semiclassical scattering regularity order), $\sfs_\scop$ (semiclassical scattering decay order), and $\sfs_\semi$ (semiclassical order, i.e.\ power of $h$) are defined in terms of $\sfs$ as in \cite[Lemma~4.28(2)]{Hintz3b}. The norms on variable order sc-b-transition and semiclassical scattering spaces here are defined as usual via testing with elliptic elements in the corresponding class of variable order pseudodifferential operators. (In the sc-b-transition case, they can be described in terms of the more standard variable order b- \cite[Appendix~A]{BaskinVasyWunschRadMink} and variable order scattering \cite{VasyMinicourse} spaces completely analogously to~\eqref{EqF3scbtNormzf}.) We refer the reader to \cite[\S{4.3}]{HintzScaledBddGeo} for a vast generalization of~\eqref{EqFVar3bFT}.

\subsubsection{se-setting: bounded geometry perspective}
\label{SssFVarse}

We consider the se-setting first in local coordinates, so we take
\begin{equation}
\label{EqFVarseLoc}
  M=\R_t\times\R^n_x,\qquad
  \cC=\R_t\times\{0\}\subset M,\qquad
  \wt M=[[0,1)_\eps\times M;\{0\}\times\cC]
\end{equation}
as in Definition~\ref{DefFse}. In the region $|\hat x|<2$ (with $\hat x=\frac{x}{\eps}$), we define for $\eps\in(0,1)$ the sets
\begin{equation}
\label{EqFVarse1}
  U^{(1)}_{\eps,j} := \{\eps\} \times \bigl((j-1)\eps,(j+1)\eps\bigr)_t \times B_{\R^n}(0,2)_{\hat x},\qquad j\in\Z.
\end{equation}
In $|\hat x|>\frac12$, so $|x|>\frac{\eps}{2}$, we work in the region where $\pm x^1>\frac14\max_{j\neq 1}|x^j|$. There, we introduce local coordinates $\omega^j=\frac{x^j}{x^1}$ ($j=2,\ldots,n$) on $\Sph^{n-1}$, and define for $j\in\Z$, $k\in\N_0$ the sets
\begin{equation}
\label{EqFVarse2}
  U^{(2)}_{\eps,\pm 1,j k} := \{\eps\} \times \bigl(2^k(j-1)\eps,2^k(j+1)\eps\bigr)_t \times \bigl(2^{k-1}\eps,2^{k+1}\eps\bigr)_r \times B_{\R^{n-1}}(0,2)_\omega,
\end{equation}
regarded as a subset of $\wt M$ via $x^1=\pm r/\sqrt{1+|\omega|^2}$, and similarly $U^{(2)}_{\eps,\pm m,j k}$ when $x^m$ and $x^1$ exchange roles. Roughly speaking, $U^{(2)}_{\eps,\pm 1,j k}$ lives in the near-field regime (bounded $|\hat x|$) when $k$ is bounded while $\eps$ decreases, and it lives in the far-field regime ($|x|\gtrsim 1$) when $2^k\eps\geq 1$; see Figure~\ref{FigFVarse}. Define the maps
\begin{align*}
  \phi^{(1)}_{\eps,j} \colon U^{(1)}_{\eps,j} &\to U^{(1)}:=(-1,1)_T \times B_{\R^n}(0,2)_X, \\
    (\eps,t,\hat x) &\mapsto (T,X)=(\eps^{-1}t-j,\hat x), \\
  \phi^{(2)}_{\eps,\pm 1,j k} \colon U^{(2)}_{\eps,\pm 1,j k} &\to U^{(2)} := (-1,1)_T \times \Bigl(\frac12,2\Bigr)_R \times B_{\R^{n-1}}(0,2)_\omega, \\
    (\eps,t,r,\omega) &\mapsto (T,R,\omega) = (2^{-k}\eps^{-1}t-j,2^{-k}\eps^{-1}r,\omega).
\end{align*}
We then compute the pushforwards of the basic se-vector fields $\hat\rho\pa_t$, $\hat\rho\pa_x$ (or of the splitting of $\hat\rho\pa_x$ into radial and spherical vector fields), with the choice of local defining function $\hat\rho=\eps$ for $|\hat x|\lesssim 1$ and $\hat\rho=r$ for $|\hat x|\gtrsim 1$, to be
\begin{equation}
\label{EqFVarsePushfwd}
\begin{alignedat}{2}
  (\phi_{\eps,j}^{(1)})_* &\colon
    \eps\pa_t \mapsto \pa_T,&\qquad&
    \eps\pa_x = \pa_{\hat x} \mapsto \pa_X, \\
  (\phi_{\eps,\pm 1,j k}^{(2)})_* &\colon
    r\pa_t \mapsto R\pa_T, &\qquad&
    r\pa_r \mapsto R\pa_R, \qquad\quad
    \pa_\omega \mapsto \pa_\omega.
\end{alignedat}
\end{equation}
Note that $R,R^{-1}\in(\frac12,2)$ are bounded away from $0$ and $\infty$.

\begin{figure}[!ht]
\centering
\includegraphics{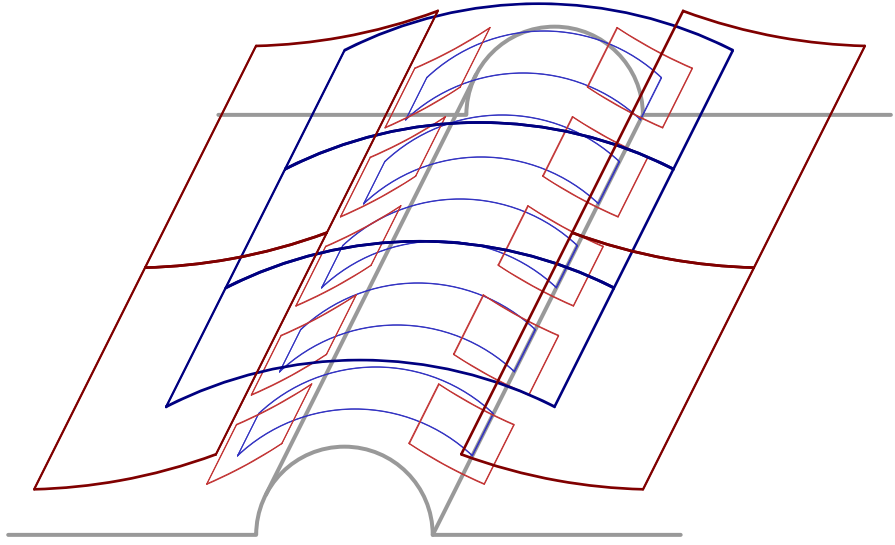}
\caption{Coordinate charts for the bounded geometry perspective on $\Vse(\wt M)$; charts of type~\eqref{EqFVarse1} are drawn in blue, charts of type~\eqref{EqFVarse2} in red. The manifold $\wt M$ is drawn in gray.}
\label{FigFVarse}
\end{figure}

For each $\eps>0$, the covering $\fC(\eps)$ of $M$ by the charts $(U^{(1)}_{\eps,j},\phi^{(1)}_{\eps,j})$ and $(U^{(2)}_{\eps,\pm m,j k},\phi^{(2)}_{\eps,\pm m,j k})$ gives $M$ the structure of a manifold with (parameterized) bounded geometry. We can thus define variable order se-Sobolev spaces by
\begin{equation}
\label{EqFVarsetildeH}
  \tilde H_{\seop,\eps}^\sfs(M) := H_{\rm uni,\fC(\eps)}^\sfs(M)
\end{equation}
for variable order functions $\sfs\in\CI_{\rm uni,\fC(\eps)}(S^*M)$, where we made the covering explicit in the notation of the uniform Sobolev space. We shall only work with variable order functions $\sfs\in\CI(\Sse^*\wt M)$, and thus with the order function $\sfs|_{S^*M_\eps}$ in~\eqref{EqFVarsetildeH}. Similarly, we define
\[
  \tilde\Psi_\seop^{\sfs,\ell_\circ,\hat\ell}(\wt M) := \{ \wt P=(\wt P_\eps)_{\eps\in(0,1)} \colon \wt P_\eps \in \rho_\circ^{-\ell_\circ}\hat\rho^{-\hat\ell}\Psi_{\rm uni,\fC(\eps)}^\sfs(M)\ \text{with uniform (in $\eps$) bounds} \},
\]
where the uniform bounds refer to uniform (in $\eps$) bounds on the symbol seminorms of the underlying uniform (on $M$ with covering $\fC(\eps)$) symbols, as well as to uniform (in $\eps$) $\CI$ bounds on residual operators in the above charts for $\fC(\eps)$. We use the notation $\tilde\Psi$ here to indicate the lack of smoothness of the coefficients of such ps.d.o.s, both in the base variables and in $\eps$; the allowed level of regularity in the base variables is, by definition, se-regularity, which will be useful for sharp localizations near fibers of $\hat M$, while the lack of smoothness in $\eps$ is of no concern since $\eps$ is merely a parameter (much as in the semiclassical setting, see e.g.\ \cite[\S{4.4.1}]{ZworskiSemiclassical}).

The reason for the tilde in~\eqref{EqFVarsetildeH} is the following.
\begin{lemma}[Uniform Sobolev spaces and se-Sobolev spaces]
\label{LemmaFVarseSob}
  Fix on $M=\R_t\times\R^n_x$ the Lebesgue measure $|\dd t\,\dd x|$, and set $\hat\rho=(\eps^2+|x|^2)^{1/2}$. Then have $\tilde H_{\seop,\eps}^0(M)=\hat\rho^{\frac{n+1}{2}}L^2(M)$, in the sense that the norms of both spaces are uniformly equivalent.
\end{lemma}
\begin{proof}
  In view of~\eqref{EqFVarsePushfwd}, the norm on $\tilde H_{\seop,\eps}^0(M)$ is the $L^2$ norm with respect to a smooth positive se-density, i.e.\ with respect to $a|\frac{\dd t}{\hat\rho}\frac{\dd x}{\hat\rho^n}|$ where $a,a^{-1}\in\CI_{\rm uni,\fC(\eps)}(M)$; this is due to the fact that the quotients of $\hat\rho$ and $\eps$, resp.\ $r$ are uniformly bounded in $\CI$ in $|\hat x|<2$, resp.\ $|\hat x|>\frac12$.
\end{proof}

\begin{definition}[Variable order weighted se-Sobolev spaces]
\label{DefFVarseSob}
  For $\sfs\in\CI(\Sse^*\wt M)$, $\alpha_\circ,\hat\alpha\in\R$, and $\rho_\circ=\frac{\eps}{(\eps^2+|x|^2)^{1/2}}$, $\hat\rho=(\eps^2+|x|^2)^{1/2}$, we set
  \[
    H_{\seop,\eps}^{\sfs,\alpha_\circ,\hat\alpha}(M) := \rho_\circ^{\alpha_\circ}\hat\rho^{\hat\alpha-\frac{n+1}{2}}\tilde H_{\seop,\eps}^\sfs(M).
  \]
\end{definition}

For constant integer orders $\sfs$, Lemma~\ref{LemmaFVarseSob} shows that this is consistent with our earlier Definition~\ref{DefFseSob}. Elements of $\tilde\Psi_\seop(\wt M)$ act boundedly between these spaces, with uniform (in $\eps$) operator norm bounds.

Lemmas~\ref{LemmaFseEdgeSob} and \ref{LemmaFseHse3b} remain valid also for variable order se-spaces:

\begin{prop}[se- and edge or 3b-Sobolev spaces]
\label{PropFVarseSobRel}
  Let $\sfs\in\CI(\Sse^*\wt M)$, $\alpha_\circ,\hat\alpha\in\R$. We fix the density $|\dd t\,\dd x|$ on $M$ and $M_\circ$.
  \begin{enumerate}
  \item\label{ItFVarseSobRele}{\rm (Edge relationship.)} Let $\sfs_\circ:=\sfs|_{\Sse^*_{M_\circ}\wt M}\in\CI(\Se^*M_\circ)$ (using~\eqref{EqFseBundleEdge}). Let $K\subset\R_t\times\R^n_x$ be compact. Then for all $\delta>0$ there exist $c_0>0$ and $C$ so that the following holds. Let $\chi\in\CIc([0,c_0))$, $\chi_\circ=\chi(\frac{\eps}{|x|})$. Then for all $\eps<c_0$ and all $u$ with $\supp u\subset K$, we have
    \begin{equation}
    \label{EqFVarseSobRele}
    \begin{split}
      \|\chi_\circ u\|_{H_{\seop,\eps}^{\sfs,\alpha_\circ,\hat\alpha}(M)} &\leq C\eps^{-\alpha_\circ}\|(\chi_\circ u)(\eps,-)\|_{\He^{\sfs_\circ+\delta,\hat\alpha-\alpha_\circ}(M_\circ)}, \\
      \eps^{-\alpha_\circ}\|(\chi_\circ u)(\eps,-)\|_{\He^{\sfs_\circ-\delta,\hat\alpha-\alpha_\circ}(M_\circ)} &\leq C\|\chi_\circ u\|_{H_{\seop,\eps}^{\sfs,\alpha_\circ,\hat\alpha}(M)}.
    \end{split}
    \end{equation}
    for all $u$ with $\supp u\subset K$. If $\sfs$ is $\eps$-independent on $([0,c_0)\times K)\cap\supp\chi_\circ$, one can take $\delta=0$, i.e.\ one obtains a uniform equivalence of norms.
  \item\label{ItFVarseSobRel3b}{\rm (3b relationship.)} Let $t_0\in\R$, and $\hat\sfs:=\sfs|_{\Sse^*_{\hat M_{t_0}}\wt M}$, which we regard (via~\eqref{EqFseBundle3b}) as an element $\hat\sfs\in\CI_{\rm I}(\Stb^*\cM)$; here the subscript {\rm `I'} restricts to stationary functions on $\cM=[\ol{\R_{\hat t}\times\R^n_{\hat x}};\pa(\ol{\R_{\hat t}}\times\{0\})]$. Fix the density $|\dd\hat t\,\dd\hat x|$ on $\cM$. Then for all $\delta>0$ there exist $c_0>0$ and $C$ so that the following holds. Let $\lambda\in(0,c_0)$, and define $\Omega=\{|t-t_0|<\lambda,\ |x|<\lambda\}$ and $\hat\Omega_\eps=\{|\hat t|<\eps^{-1}\lambda,\ |\hat x|<\eps^{-1}\lambda \}\subset\cM$. Write $\Psi_\eps(t,x)=(\hat t,\hat x)=(\frac{t-t_0}{\eps},\frac{x}{\eps})$. Then for all $u$ with support in $\Omega$ (and thus $(\Psi_\eps)_*u$ with support in $\hat\Omega_\eps$), we have
    \begin{equation}
    \label{EqFVarseSobRel3b}
    \begin{split}
      \|u\|_{H_{\seop,\eps}^{\sfs,\alpha_\circ\hat\alpha}(M)} &\leq C\eps^{\frac{n+1}{2}-\hat\alpha}\|(\Psi_\eps)_*u\|_{\Htb^{\hat\sfs+\delta,\alpha_\circ-\hat\alpha,0}(\cM)}, \\
      \eps^{\frac{n+1}{2}-\hat\alpha}\|(\Psi_\eps)_*u\|_{\Htb^{\hat\sfs-\delta,\alpha_\circ-\hat\alpha,0}(\cM)} &\leq C\|u\|_{H_{\seop,\eps}^{\sfs,\alpha_\circ\hat\alpha}(M)}.
    \end{split}
    \end{equation}
  If $\sfs$ is invariant under the lift of dilations $(\eps,t_0+\tau,x)\mapsto(a\eps,t_0+a\tau,a x)$ on $\Omega$, one can take $\delta=0$.
  \end{enumerate}
\end{prop}

In the coordinates $\hat t,\hat x$, the dilations in part~\eqref{ItFVarseSobRel3b} are given by $(\eps,\hat t,\hat x)\mapsto(a\eps,\hat t,\hat x)$ simply.

\begin{proof}[Proof of Proposition~\usref{PropFVarseSobRel}]
  It suffices to consider the case $\alpha_\circ=\hat\alpha=0$.

  We begin with part~\eqref{ItFVarseSobRele}. Since $\sfs$ differs from the $\eps$-independent extension $\sfs_\circ'$ of $\sfs_\circ$ by a function which is smooth on $\supp\chi_\circ$ and vanishes over $M_\circ$, we have $\sfs_\circ'+\delta>\sfs>\sfs_\circ'-\delta$ on $([0,c_0)\times K)\cap\supp\chi_\circ$ when $c_0$ is sufficiently small; it therefore suffices (upon replacing $\sfs$ by $\sfs_\circ'\pm\delta$) to prove the uniform equivalence of $\|\chi_\circ u\|_{H_{\seop,\eps}^\sfs}$ and $\eps^{-\alpha_\circ}\|(\chi_\circ u)(\eps,-)\|_{\He^\sfs}$ for $\eps$-independent $\sfs$. We may further replace the densities on $M$ and $M_\circ$ by a positive se-density and a positive edge-density, respectively; indeed, this amounts to multiplication of the density $|\dd t\,\dd x|$ by $\hat\rho^{-n}$ and $r^{-n}$, respectively, where we may take as the local defining function $\hat\rho$ of $\hat M$ on $\supp\chi_\circ$ the function $\hat\rho=r$.

  At this point, we can compute the $H_{\seop,\eps}^\sfs$ and $\He^{\sfs_\circ}$ norms using their characterizations as uniform Sobolev spaces. Note that the coverings $\fC(\eps)$ and $\fC(\eps')$ are compatible in a uniform sense when $\eps/\eps'\in[\frac12,2]$; that is, their union is still a bounded geometry covering, and the $\CI$ bounds on the transition functions are uniformly bounded in $\eps,\eps'$. Thus, the norms on $H_{{\rm uni},\fC(\eps)}$ and $H_{{\rm uni},\fC(\eps')}$ are uniformly equivalent. Let now $\eps\in(0,1)$ and set $\eps'=2^{-\ell}$ for a value $\ell\in\N_0$ for which $\eps/\eps'\in[\frac12,2]$. Then the projection of the set $U_{2^{-\ell},+1,j k}^{(2)}$ in~\eqref{EqFVarse2} to the $t,r,\omega$ coordinates is equal to the set $U_{j k'}$ in~\eqref{EqFVarEBdd} for $k'=\ell-k$. (On $\supp\chi_\circ$ we have $\frac{r}{\eps}\geq\eps c_0^{-1}>c_0^{-1}$, so $2^k\geq\frac{1}{2 c_0}$, and on $K$ we have $r\lesssim 1$, so $k-\ell\lesssim 1$, and therefore overall $k'$ is bounded uniformly from below, and by a constant plus $\ell$ from above.) This implies the claim.

  The arguments for part~\eqref{ItFVarseSobRel3b} are similar. First, replacing $\sfs$ by the dilation-invariant extension of $\hat\sfs$ causes an error which on $\Omega$ is less than $\delta$ in absolute value if the size $\lambda$ of $\Omega$ is sufficiently small. The weights are dealt with as in the proof of Lemma~\ref{LemmaFseHse3b}. By shifting the $t$ coordinate, we may assume $t_0=0$. One then observes that the projection to $t,\hat x$ coordinates of the image under $\Psi(2^{-\ell},-)$ of the set $U^{(1)}_{2^{-\ell},j}$ in~\eqref{EqFVarse1} is equal to the set $U^{(1)}_j$ in~\eqref{EqFVar3b1}, and similarly for $U^{(2)}_{2^{-\ell},+1,j k}$ in~\eqref{EqFVarse2} and $U^{(2)}_{+1,j k'}$ in~\eqref{EqFVar3b2} for $k'=k-\ell$.
\end{proof}

When $M$ is compact, we can define $H_{\seop,\eps}^{\sfs,\alpha_\circ,\hat\alpha}(M)$ using a partition of unity; on $[0,1)_\eps\times V$ where $V\subset M$ is open with closure disjoint from $\cC$, one uses an ordinary Sobolev norm on $M$ (independently of $\eps$); similarly one can define $\tilde\Psi_\seop^{\sfs,\ell_\circ,\hat\ell}(\wt M)$. Such ps.d.o.s will arise in this paper as uniform quantizations of elements of (weighted versions of) $S^\sfs(\Tse^*\wt M)$ or $\CI_\seop S^\sfs(\Tse^*\wt M)$, where $\CI_\seop(\wt M)$ denotes smooth functions on $\wt M^\circ$ which remain uniformly bounded upon application of any number of se-vector fields (smoothness in $\eps$ is not needed here). We shall then have occasion to consider the elliptic set $\Ell_\seop(A)$ of $A\in\tilde\Psi_\seop^\sfs$ as a subset
\[
  \Ell_\seop(A) \subset \Sse^*\wt M\qquad\text{or}\qquad \Ell_\seop(A)\subset\upbeta^*\Sse^*\wt M
\]
where $\upbeta\colon[\wt M;Z]\to\wt M$ is the blow-down map for the blow-up of $\wt M$ along some p-submanifold $Z$ of $\wt M$ (typically a union of fibers of $\hat M$); similarly for the operator wave front set $\WF'_\seop(A)$.

\subsubsection{Mixed function spaces}
\label{SssFVarsse}

We need to extend Definition~\ref{DefFsFn} to the case of variable order se-differential orders. Thus, for compact $M$ and for $\sfs\in\CI(\Sse^*\wt M)$, $k\in\N_0$, and $\alpha_\circ,\hat\alpha\in\R$, we define the mixed (se;s)-Sobolev spaces
\begin{equation}
\label{EqFVarsseSob}
  H_{(\seop;\sop),\eps}^{(\sfs;k),\alpha_\circ,\hat\alpha}(M) := H^{\sfs|_{M_\eps}+k}(M), \qquad
  \|u\|_{H_{(\seop;\sop),\eps}^{(\sfs;k),\alpha_\circ,\hat\alpha}(M)}^2 := \sum_j \|P_j u\|_{H_{\seop,\eps}^{\sfs,\alpha_\circ,\hat\alpha}(M)}^2,
\end{equation}
where $\{P_j\}\subset\Diffs^k(\wt M)$ is a fixed finite spanning set of $\Diffs^k(\wt M)$ over $\CI(\wt M)$, and the underlying $L^2$-space is defined using a positive smooth density on $M$ unless stated otherwise.

Carefully note that elements of $\tilde\Psi_\seop^0(\wt M)$ do \emph{not} define bounded linear operators on such spaces when $k\geq 1$. As a concrete example, multiplication by a function which is se-regular but not s-regular, such as $\chi(\frac{t-t_0}{\hat\rho})$ where $\chi\in\CIc(\R)$ is nonzero, does not preserve s-regularity. Therefore, in order to do se-microlocal analysis on mixed (se;s)-Sobolev spaces, we need to work with se-ps.d.o.s whose coefficients are s-regular. To this end, we employ the scaled bounded geometry perspective introduced in \cite{HintzScaledBddGeo}. Concretely, working in local coordinates $t,x$ as in~\eqref{EqFVarseLoc}, we define unit cells \emph{for the space $\Vs(\wt M)$ of s-vector fields} as follows. In the region $|\hat x|<2$, we define
\begin{subequations}
\begin{equation}
\label{EqFVarsseBG1}
\begin{split}
  &U_{\eps,j}^{(1)} := \{\eps\} \times (j-1,j+1)_t \times B_{\R^n}(0,2)_{\hat x},\qquad j\in\Z, \\
  &\phi_{\eps,j}^{(1)} \colon U_{\eps,j}^{(1)} \ni (\eps,t,\hat x) \mapsto (T,X)=(t-j,\hat x) \in U^{(1)} := (-1,1)_T \times B_{\R^n}(0,2)_X.
\end{split}
\end{equation}
In $|\hat x|>\frac12$, and in the region where $\pm x^1>\frac14\max_{j\neq 1}|x^j|$, we introduce $\omega^j=\frac{x^j}{x^1}$ and define for $j\in\Z$, $k\in\N_0$ with $2^k\eps\leq 1$ the sets
\begin{equation}
\label{EqFVarsseBG2}
\begin{split}
  &U_{\eps,\pm 1,j k}^{(2)} := \{\eps\} \times (j-1,j+1)_t \times (2^{k-1}\eps,2^{k+1}\eps)_r \times B_{\R^{n-1}}(0,2)_\omega, \\
  &\qquad \phi_{\eps,\pm 1,j k}^{(2)} \colon U_{\eps,\pm 1,j k}^{(2)} \ni (\eps,t,r,\omega) \mapsto (T,R,\omega)=(t-j,2^{-k}\eps^{-1}r,\omega) \\
  &\qquad \hspace{14em} \in U^{(2)} := (-1,1)_T \times \Bigl(\frac12,2\Bigr)_R \times B_{\R^{n-1}}(0,2)_\omega,
\end{split}
\end{equation}
\end{subequations}
with similar definitions for sets $U_{\eps,\pm m,j k}$ and maps $\phi_{\eps,\pm m,j k}$ in the region where $\pm x^m>\frac14\max_{j\neq m}|x^j|$ for $m=2,\ldots,n$. The pushforwards of the basic s-vector fields $\pa_t$, $\eps\pa_{x^j}$ in $|\hat x|<2$, resp.\ $\pa_t$, $r\pa_r$, $\pa_\omega$ in $|\hat x|>\frac12$ under the map $\phi_{\eps,j}^{(1)}$, resp.\ $\phi_{\eps,\pm 1,j k}^{(2)}$ are $\pa_T$, $\pa_{\hat x^j}$, resp.\ $\pa_T$, $R\pa_R$, $\pa_\omega$, and thus uniformly equivalent to the coordinate vector fields on $U^{(1)}$, resp.\ $U^{(2)}$. In the region $|x|>\frac14$, which is not yet fully covered by the above two families of unit cells, we use the unit cells $\{\eps\}\times(j-1,j+1)_t\times \frac14(k+(-2,2)^n)_x$ for $j\in\Z$, $k=(k_1,\ldots,k_n)\in\Z^n$ with $\min_{i=1,\ldots,n}|k_i|\geq 3$, which we map to $(-2,2)_T\times(-2,2)^n_X$ via $(\eps,t,x)\mapsto(t-j,4 x-k)$. On such unit cells, which are thus far from $\cC$, the coordinate vector fields $\pa_T$, $\pa_X$ are uniformly equivalent to $\pa_t$, $\pa_x$.

In the terminology of \cite{HintzScaledBddGeo}, we conclude that $\Vs(\wt M)$ is the \emph{coefficient Lie algebra} for the parameterized (by $\eps\in(0,1)$) bounded geometry structure~\eqref{EqFVarsseBG1}--\eqref{EqFVarsseBG2}.\footnote{This structure does not quite match the conditions in \cite{HintzScaledBddGeo} since the sets $U^{(1)}$ and $U^{(2)}$ are not equal to $(-2,2)^{1+n}$. This is easily remedied, at the expense of leading to sets which are not as closely related to those used in~\S\ref{SssFVarse} and notationally more cumbersome. We leave the necessary modifications to the reader.} As the scaling on $U_{\eps,j}^{(1)}$, we take $\rho_{\eps,j,T}=\eps$, $\rho_{\eps,j,X^i}=1$ for $i=1,\ldots,n$; and as the scaling on $U_{\eps,\pm m,j k}^{(2)}$, we take $\rho_{\eps,(\pm m,j k),T}=\eps 2^k$ (which is thus comparable to $r$ on this set), $\rho_{\eps,(\pm m,j k),X^i}=1$ for $i=1,\ldots,n$. The \emph{operator Lie algebra} corresponding to this scaling is $\CI_\sop\Vse(\wt M)$; indeed, its elements are those vector fields which, in the above charts, have uniformly bounded smooth coefficients when expressed as linear combinations of $\rho_{\eps,\alpha,T}\pa_T$ and $\rho_{\eps,\alpha,X^i}\pa_{X^i}$ for $\alpha=j$ or $\alpha=(\pm m,j k)$; the latter vector fields pull back to the basic se-vector fields $\eps\pa_t$, $\eps\pa_{x^j}$ in $|\hat x|<2$, resp.\ $r\pa_t$, $r\pa_r$, $\pa_\omega$ in $|\hat x|>\frac12$, $|x|<1$.

In analogy with the notation introduced in \cite[Definition~3.59]{HintzScaledBddGeo}, we shall now write
\[
  \Psi_\seop^{\sfs,\alpha_\circ,\hat\alpha}(\wt M) = \rho_\circ^{-\alpha_\circ}\hat\rho^{-\hat\alpha}\Psi_\seop^\sfs(\wt M)
\]
for the corresponding space of parameterized scaled bounded geometry pseudodifferential operators; more precisely, in view of Lemma~\ref{LemmaFsComm}, we may work with the more precise version given in \cite[Proposition~4.22]{HintzScaledBddGeo}, and shall do so exclusively in the sequel. Roughly speaking, and considering the case $\alpha_\circ=\hat\alpha=0$ for brevity, elements of $\Psi_\seop^\sfs(\wt M)$ are families $(P_\eps)_{\eps\in(0,1)}$ of operators on $M$, with each $P_\eps$ being the sum of local quantizations on the above s-unit cells $U_{\eps,\alpha}$ where $\alpha=j$ or $\alpha=(\pm m,j k)$; the local quantizations have Schwartz kernels $(\int \exp(i(T-T')\sigma/\rho_{\eps,\alpha,T}+i(X-X')\cdot\xi) a(T,X,\sigma,\xi)\,\dd\sigma\,\dd\xi)\frac{|\dd T'\dd X'|}{\rho_{\eps,\alpha,T}}$ where $|\pa_{T,X}^\beta\pa_{\sigma,\xi}^\gamma a|\leq C_{\beta\gamma\delta}(1+|\sigma|+|\xi|)^{\sfs-(1-2\delta)|\gamma|}$ for all $\beta,\gamma$ and all $\delta>0$, with the constants $C_{\beta\gamma\delta}$ being uniform over all $\alpha$. (For the full definition, including the space of residual operators, i.e.\ those of se-differential order $-\infty$, see \cite[Definitions~3.31 and~3.34, equation~(4.27)]{HintzScaledBddGeo}.)

For the purposes of the present paper, we need to record two key property of elements $P=(P_\eps)_{\eps\in(0,1)}\in\Psi_\seop^\sfs(\wt M)$.
\begin{enumerate}
\item The maps
  \begin{equation}
  \label{EqFVarsseMap}
    P_\eps \colon H_{(\seop;\sop),\eps}^{(\sfs';k)}(M) \to H_{(\seop;\sop),\eps}^{(\sfs'-\sfs;k)}(M)
  \end{equation}
  are uniformly (in $\eps\in(0,1)$) bounded for all $\sfs'\in\CI(\Sse^*\wt M)$ and $k\in\N_0$; similarly for weighted operators and spaces. This generalizes Corollary~\ref{CorFsOrder}. (It is also a crucial improvement over the mapping properties of elements of $\tilde\Psi_\seop^\sfs(\wt M)$, for which~\eqref{EqFVarsseMap} in general fails unless $k=0$.)
\item For commutator s-vector fields $W\in\cV_{[\sop]}(\wt M)$ (see Lemma~\ref{LemmaFsComm}), $[W,P]\in\Psi_\seop^\sfs(\wt M)$.
\end{enumerate}

We use such ps.d.o.s to prove the following elliptic regularity result. (We continue assuming for notational simplicity that $M$ is compact.)

\begin{lemma}[se-microlocal elliptic regularity]
\label{LemmaFVarsseEll}
  Let $k,m\in\N_0$, and let $\wt L\in\cC_{\seop;\sop}^{(d_0;k)}\Psi_\seop^m(\wt M)$. (That is, the operator $\wt L$ is a linear combination of products of elements of $\cC_{\seop;\sop}^{(d_0;k)}(\wt M)$ with elements of $\Psi_\seop^m(\wt M)$; recall here Definition~\usref{DefFsFn}\eqref{ItFsFnMixedC}.) Let $\sfs\in\CI(\Sse^*\wt M)$ and $N\in\R$, and assume that $d_0=d_0(m,\sfs,N)$ is sufficiently large. Let $\chi\in\CI_\sop(\wt M)$, and let $B,G\in\Psi_\seop^0(\wt M)$ be such that $B=\chi B\chi$, $G=\chi G\chi$, and $\WF_\seop'(B)\subset\Ell_\seop(G)\cap\Ell_\seop(\wt L)$. Then there exists $C$ so that
  \begin{equation}
  \label{EqFVarsseEll}
    \|B u\|_{H_{(\seop;\sop),\eps}^{(\sfs;k)}(M)} \leq C\Bigl( \|G\wt L u\|_{H_{(\seop;\sop),\eps}^{(\sfs-m;k)}(M)} + \|\chi u\|_{H_{(\seop;\sop),\eps}^{(-N;k)}(M)}\Bigr).
  \end{equation}
  Analogous statements hold for weighted operators and Sobolev spaces.
\end{lemma}
\begin{proof}
  For $k=0$ and $d_0=\infty$, one can prove~\eqref{EqFVarsseEll} using a microlocal elliptic parametrix near $\WF_\seop'(B)$ of $\wt L$ of class $\tilde\Psi_\seop^{-m}$. Since operator norm bounds of se-ps.d.o.s on a fixed range of se-Sobolev spaces only depend on some (finite) seminorm of the coefficients of $\wt L$, the constant $C$ in~\eqref{EqFVarsseEll} only depends on the $\cC_\seop^{d_0}$-norm of the coefficients of $\wt L$ for some sufficiently large $d_0$.

  We now argue by induction on $k$. Apply~\eqref{EqFVarsseEll} with $k-1$ in place of $k$ to $V u$ where $V\in\cV_{[\sop]}(\wt M)$. Now, using $[B,V]\in\Psi_\seop^0$, we have
  \[
    \|B V u\|_{H_{(\seop;\sop),\eps}^{(\sfs;k-1)}} \geq \|V B u\|_{H_{(\seop;\sop),\eps}^{(\sfs;k-1)}} - \|[B,V] u\|_{H_{(\seop;\sop),\eps}^{(\sfs;k-1)}} \geq \|V B u\|_{H_{(\seop;\sop),\eps}^{(\sfs;k-1)}} - C\|\tilde B u\|_{H_{(\seop;\sop),\eps}^{(\sfs;k-1)}}
  \]
  where $\tilde B\in\Psi_\seop^0(\wt M)$ is elliptic on $\WF'_\seop(B)$. We similarly estimate $\|G\wt L V u\|_{H_{(\seop;\sop),\eps}^{(\sfs-m;k-1)}}\leq C\|G\wt L u\|_{H_{(\seop;\sop),\eps}^{(\sfs-m;k)}}+C\|\tilde G u\|_{H_{(\seop;\sop),\eps}^{(\sfs-m;k-1)}}$ where $\tilde G\in\Psi_\seop^0(\wt M)$ is elliptic on $\WF'_\seop(G)$. Here, we use that $\wt L$ and $[\wt L,V]\in\cC_{\seop;\sop}^{(d_0;k-1)}\Psi_\seop^m(\wt M)$ preserve $k$ and $k-1$ orders of s-regularity, respectively. Summing over a collection of $V$ which span $\cV_\sop(\wt M)$ over $\CI(\wt M)$, we thus obtain~\eqref{EqFVarsseEll} (with slightly enlarged $G$ and $\chi$).
\end{proof}

\begin{lemma}[Real principal type propagation]
\label{LemmaFVarsseProp}
  Let $k,m\in\N_0$ and $\delta>0$. Suppose the operator
  \[
    \wt L=\wt L_0+\wt L_1,\qquad \wt L_0\in\Psi_\seop^m(\wt M),\quad \wt L_1\in\cC_{\seop;\sop}^{(d_0;k),\delta,\delta}\Psi_\seop^m(\wt M),
  \]
  has a real homogeneous principal symbol. Let $\sfs\in\CI(\Sse^*\wt M)$ be monotonically decreasing along the null-bicharacteristic flow of $\wt L$. Suppose $B,E,G\in\Psi_\seop^0(\wt M)$ and $\chi\in\CI_\sop(\wt M)$ are such that all backwards null-bicharacteristics starting at $\WF_\seop'(B)$ reach $\Ell_\seop(E)$ in finite time while remaining in $\Ell_\seop(G)$. Suppose moreover that $\WF_\seop'(B)\subset\Ell_\seop(G)$, and $B=\chi B\chi$, etc. Let $N\in\R$. Then for sufficiently large $d_0=d_0(m,\sfs,N)$, there exists a constant $C$ so that
  \[
    \|B u\|_{H_{(\seop;\sop),\eps}^{(\sfs;k)}(M)} \leq C\Bigl( \|G\wt L u\|_{H_{(\seop;\sop),\eps}^{(\sfs-m+1;k)}(M)} + \|E u\|_{H_{(\seop;\sop),\eps}^{(\sfs;k)}(M)} + \|\chi u\|_{H_{(\seop;\sop),\eps}^{(-N;k)}(M)}\Bigr).
  \]
  Analogous statements hold for weighted operators and Sobolev spaces.
\end{lemma}
\begin{proof}
  For $k=0$, this can be proved by a standard positive commutator argument utilizing $\tilde\Psi_\seop$; see \cite[\S8]{HintzMicro} and \cite{BaskinDatchevPropagation} for expository accounts. (The structure of $\wt L$ ensures that the null-bicharacteristic flow is continuous down to $\Sse^*_{M_\circ\cup\hat M}\wt M$.)

  For $k\geq 1$, one can argue by induction similarly to the proof of Lemma~\ref{LemmaFVarsseProp}. For example, for $k=1$ and $V\in\cV_{[\sop]}(\wt M)$, one estimates
  \[
    \|G\wt L V u\|_{H_{(\seop;\sop),\eps}^{(\sfs-m+1;k-1)}} \leq \|G\wt L u\|_{H_{(\seop;\sop),\eps}^{(\sfs-m+1;k)}} + \|\tilde G u\|_{H_{(\seop;\sop),\eps}^{(\sfs+1;k-1)}}
  \]
  since $[\wt L,V]\in(\CI+\cC_{\seop;\sop}^{(d_0;k-1),\delta,\delta})\Psi_\seop^m(\wt M)$; and the second term can be estimated using the inductive hypothesis, with se-regularity order increased by $1$.

  An alternative argument which avoids this se-regularity increase (and can be adapted to radial point estimates in~\S\ref{SssEstRads} and tame estimates in~\S\ref{SsNTame}) proceeds as follows. Fix a spanning set $V_1,\ldots,V_N\in\cV_{[\sop]}(\wt M)$ of $\cV_\sop(\wt M)$ over $\CI_\sop(\wt M)$; concretely, one can take $V_1=\chi\pa_t$ and choose $V_2,\ldots,V_N\in\Vse(\wt M)$ to span $\Vse(\wt M)$ over $\CI(\wt M)$. Set $V_0=I$. We can then write $[\wt L,V_j]=\sum_{k=0}^N A_{j,l}V_l$ where $A_{j,l}\in(\CI+\cC_{\seop;\sop}^{(d_0;k-1),\delta,\delta})\Psi_\seop^{m-1}(\wt M)$. Set $f:=\wt L u$ and $u^{(1)}=(V_j u)_{j=0,\ldots,N}$, $f^{(1)}=(V_j f)_{j=0,\ldots,N}$, and $(\wt L^{(1)})_{j l}=(\delta_{j l}\wt L-A_{j,l})$; then
  \begin{equation}
  \label{EqFVarssePropInd}
    \wt L^{(1)}u^{(1)}=f^{(1)}.
  \end{equation}
  But $\wt L^{(1)}$ is principally scalar, with the same principal symbol as $\wt L$, and we can apply the inductive hypothesis to equation~\eqref{EqFVarssePropInd}, \emph{with the same $\sfs$} and with $\wt L^{(1)}$, $k-1$ in place of $\wt L$, $k$, to complete the inductive step.
\end{proof}

\section{Geometry and phase space dynamics of glued spacetimes}
\label{SGl}

\subsection{Glued spacetimes}
\label{SsGl}

Glued spacetimes are families of spacetimes (not necessarily solving the field equations) describing Kerr black holes of mass $0<\eps\ll 1$ moving along timelike geodesics in a background spacetime $(M,g)$. As a prerequisite for stating the precise definition, we recall from \cite{KerrKerr,BoyerLindquistKerr}:

\begin{definition}[Kerr in Boyer--Lindquist coordinates]
\label{DefGlKerrBL}
  For $\bhm>0$ and $a\in\R$ satisfying the \emph{subextremality condition} $|a|<\bhm$, let $\hat r_{\bhm,a}=\bhm+\sqrt{\bhm^2-a^2}$ and set $\mu(\hat r)=\hat r^2-2\bhm\hat r+a^2$, $\varrho^2(\hat r,\theta)=\hat r^2+a^2\cos^2\theta$; then the Kerr metric $\hat g_{\bhm,a}$ on the manifold $\R_{\hat t_{\rm BL}}\times(\hat r_{\bhm,a},\infty)_{\hat r}\times\Sph^2_{\theta,\phi}$ is the (Ricci-flat) metric
  \begin{equation}
  \label{EqGlKerrBL}
    \hat g_{\bhm,a} := -\frac{\mu}{\varrho^2}(\dd\hat t_{\rm BL}-a\sin^2\theta\,\dd\phi)^2 + \varrho^2\Bigl(\frac{\dd\hat r^2}{\mu}+\dd\theta^2\Bigr) + \frac{\sin^2\theta}{\varrho^2}\bigl((\hat r^2+a^2)\dd\phi-a\,\dd\hat t_{\rm BL}\bigr)^2.
  \end{equation}
\end{definition}

\begin{lemma}[Horizon-penetrating coordinates]
\label{LemmaGlCoord}
  (See \cite[\S3.1]{HintzKdSMS}.) For any pair of smooth functions $\tilde T,\tilde\Phi\colon[0,\infty)\to\R$ which are analytic on $[0,4\bhm]$, define $T,\Phi\colon(\hat r_{\bhm,a},\infty)\to\R$ (up to arbitrary additive constants) by
  \begin{equation}
  \label{EqGlCoord}
    T'(\hat r)=-\frac{\hat r^2+a^2}{\mu(\hat r)}+\tilde T(\hat r),\qquad
    \Phi'(\hat r)=-\frac{a}{\mu(\hat r)}+\tilde\Phi(\hat r).
  \end{equation}
  Set $\hat t=\hat t_{\rm BL}-T(\hat r)$, $\phi_*=\phi-\Phi(\hat r)$, and $\hat r_{\bhm,a}^-:=\bhm-\sqrt{\bhm^2-a^2}$. Then $\hat g_{\bhm,a}$ extends analytically from $\hat r>\hat r_{\bhm,a}$ to a (Ricci-flat) metric on $\R_{\hat t}\times(\hat r_{\bhm,a}^-,\infty)_{\hat r}\times\Sph^2_{\theta,\phi_*}$. One can moreover choose $\tilde T,\tilde\Phi$ to satisfy $\tilde T=\tilde\Phi=0$ for large $\hat r$, to depend smoothly on $\bhm,a$ near fixed subextremal parameters $\bhm_0,a_0$, and so that $\dd\hat t$ is everywhere past timelike.
\end{lemma}

The level sets of $\hat t$ are then transversal to the future event horizon $\cH^+=\{\hat r=\hat r_{\bhm,a}\}$.

\begin{definition}[Kerr model]
\label{DefGlKerr}
  (See \cite[Definition~3.25]{HintzGlueLocI}.) For parameters $\bhm>0$ and $\bha\in\R^3$ for which $\bhm,a:=|\bha|$ are subextremal, define
  \begin{gather*}
    \hat K_{\bhm,\bha} := \{\hat x\in\R^3 \colon \hat r=|\hat x|\leq\bhm \}, \\
    \hat M_{\bhm,\bha}^\circ := \R_{\hat t}\times\hat X_{\bhm,\bha}^\circ,\qquad
    \hat X_{\bhm,\bha}^\circ := \R^3_{\hat x} \setminus \hat K_{\bhm,\bha}^\circ.
  \end{gather*}
  Define polar coordinates $\hat r=|\hat x|$, $\theta\in(0,\pi)$, $\phi_*\in(0,2\pi)$ on $\R^3_{\hat x}\supset\hat X_{\bhm,\bha}^\circ$ so that $\bha$ (when nonzero) points in the direction of the north pole $\theta=0$. We then define $\hat g_{\bhm,\bha}$ on $\hat M_{\bhm,\bha}^\circ$ to be equal to $g_{\bhm,a}$ in the coordinates from Lemma~\ref{LemmaGlCoord}.
\end{definition}

We shall identify $\hat M_{\bhm,\bha}^\circ$ with an open subset of
\[
  \cM = \bigl[\,\ol{\R^{1+3}_{\hat t,\hat x}}; \pa\ol{\R_{\hat t}\times\{0\}}\,\bigr].
\]
Then $\hat g_{\bhm,\bha}$ extends to the closure $\hat M_{\bhm,\bha}$ of $\hat M_{\bhm,\bha}^\circ$ as a smooth stationary Lorentzian 3-body-scattering metric, that is,
\[
  \hat g_{\bhm,\bha} \in \CI(\hat X_{\bhm,\bha};S^2\,\Ttsc^*_{\hat X_{\bhm,\bha}}\hat M_{\bhm,\bha}),\qquad \hat X_{\bhm,\bha} := \ol{\R^3_{\hat x}} \setminus \hat K_{\bhm,\bha}^\circ.
\]

\begin{definition}[Glued spacetime]
\label{DefGl}
  Fix the following data:
  \begin{itemize}
  \item $(M,g)$: a smooth $4$-dimensional globally hyperbolic Lorentzian manifold;
  \item\label{ItGlGeod} $\cC\subset M$: a smooth inextendible timelike curve,\footnote{In the present paper, it is irrelevant whether or not $\cC$ is a geodesic; this only matters when we solve the gluing problem for the Einstein equations in \cite{HintzGlueLocI,HintzGlueLocIII}.} parameterized by arc length by $c\colon I_\cC\subseteq\R\to M$;
  \item Fermi normal coordinates
  \begin{equation}
  \label{EqGlFermi}
    z=(z^0,z^1,z^2,z^3)=(t,x)\in U_{\rm Fermi}\subset I_\cC\times\R^3
  \end{equation}
  around $\cC$, with $\pa_t$ future timelike. That is, $U_{\rm Fermi}$ is an open neighborhood of $I_\cC\times\{0\}$, and $g_{\mu\nu}(t,x)=g^{\rm Mink}_{\mu\nu}+\cO(|x|)$ where $g^{\rm Mink}=-\dd t^2+\dd x^2$ is the Minkowski metric and $\cO(|x|)$ denotes a smooth function on $M$ which vanishes at $\cC=\{x=0\}$;
  \item subextremal Kerr parameters $\bhm>0$, $\bha\in\R^3$.
  \end{itemize}
  A \emph{glued spacetime} is then a pair $(\wt M,\wt g)$ where
  \begin{enumerate}
  \item $\wt M:=[[0,1)_\eps\times M;\{0\}\times\cC]$ is the total gluing spacetime from Definition~\usref{DefFse};
  \item\label{ItGlMetric} $\wt g$ is a section of $S^2\wt T^*\wt M$ over $\wt M\setminus\wt K^\circ$ where $\wt K=\{|\hat x|\leq\bhm\}$, $\hat x=\frac{x}{\eps}$. We assume that:
    \begin{enumerate}
    \item\label{ItGlwtg} $\wt g$ is of class $\CI+\cC_\seop^{\infty,1,1}=\CI+\rho_\circ\hat\rho\CI_\seop$;
    \item $\wt g|_{M_\circ}=\upbeta_\circ^*g$, so in coordinates $\eps$, $t$, $x\neq 0$, we have $\wt g_{\mu\nu}(0,t,x)=g_{\mu\nu}(t,x)$;
    \item\label{ItGlwtgKerr} in the coordinates $\eps$, $t$, $\hat x$, we have $\wt g_{\mu\nu}(0,t,\hat x)=(\hat g_{\bhm,\bha})_{\hat\mu\hat\nu}(\hat x)$ where the left, resp.\ right hand side is a metric coefficient in the coordinates $z=(t,x)$ on $M$, resp.\ $\hat z=(\hat t,\hat x)$ on $\hat M_{\bhm,\bha}^\circ$.
    \end{enumerate}
  \end{enumerate}
  We write $g_\eps:=\wt g|_{M_\eps}\in\CI(M_\eps\setminus\wt K^\circ;S^2 T^*M_\eps)$, where we recall that $M_\eps=\{\eps\}\times M$ is the $\eps$-level set of $\wt M$ for $\eps>0$.
\end{definition}

Recall that Fermi normal coordinates are unique up to replacing $x$ by $R x$ where $R\in O(3)$ is constant; see e.g.\ \cite[Lemma~3.14]{HintzGlueLocI}. The choice of spatial coordinates $x$ thus amounts to a choice of the axis of rotation of the small Kerr black hole. Recall moreover that for any fixed precompact open set $V\subset M$ there exists $\eps(\ol V)>0$ so that $\wt g$ is a Lorentzian signature section of $S^2\wt T^*\wt M)$ on $\wt\upbeta^{-1}([0,\eps(\ol V))\times V)\subset\wt M$ (see the discussion preceding \cite[Notation~3.17]{HintzGlueLocI}); this follows from the continuity of $\wt g$ and the fact that both $\wt g|_{M_\circ}$ and $\wt g|_{\hat M}$ are Lorentzian.

\begin{rmk}[Regular set]
\label{RmkGlRegular}
  We shall not make the open neighborhood of $M_\circ\cup(\hat M\setminus\wt K^\circ)$ of $\wt M$ on which $\wt g$ is Lorentzian explicit. For quasilinear problems, one only needs the simple fact that if $\wt g$ is Lorentzian on a compact subset $\wt C\subset\wt M$, then also $\wt g+\wt h$ is Lorentzian on $\wt C$ provided that $\wt h$ is sufficiently small in $L^\infty$ as a section of $S^2\wt T^*\wt M$ (equipped with any fixed smooth positive definite fiber metric) over $\wt C$.
\end{rmk}

\begin{rmk}[se, 3b, and edge perspectives]
\fakephantomsection
\label{RmkGlSe3be}
  \begin{enumerate}
  \item\label{ItGlSe3be3b} In view of the identification~\eqref{EqFsewtTTse}, we can regard a glued spacetime metric $\wt g$ as a Lorentzian se-metric,
    \begin{equation}
    \label{EqGlMetricse}
      \wt g \in \hat\rho^2(\CI+\cC^{\infty,1,1}_\seop)(\wt M\setminus\wt K^\circ;S^2\,\Tse^*\wt M).
    \end{equation}
    In local coordinates, this arises from $\dd z=\hat\rho\frac{\dd z}{\hat\rho}$, with $\frac{\dd z}{\hat\rho}$ being a local frame of $\Tse^*M$. Therefore, $\eps^{-2}\wt g\in\rho_\circ^{-2}(\CI+\cC_\seop^{\infty,1,1})(\wt M\setminus\wt K^\circ;S^2\,\Tse^*\wt M)$ can be restricted to any fiber $\hat M_{t_0}$ of $\hat M$ to yield, via~\eqref{EqFseBundle3b}, a stationary metric of class
    \[
      \rho_\circ^{-2}\CI(\hat M_{\bhm,\bha};S^2\,\Ttb^*\hat M_{\bhm,\bha}) = \CI(\hat M_{\bhm,\bha};S^2\,\Ttsc^*\hat M_{\bhm,\bha})
    \]
    which is, in view of Definition~\ref{DefGl}\eqref{ItGlwtgKerr}, the Kerr metric $\hat g_{\bhm,\bha}$ (independently of $t_0$).
  \item\label{ItGlSe3bee} Restriction of~\eqref{EqGlMetricse} to $M_\circ$ recovers $g$ as a weighted edge metric:
  \[
    \wt g|_{M_\circ} = g \in \upbeta_\circ^*\CI(M;S^2 T^*M) \subset |x|^2\CI(M_\circ;S^2\,\Te^*M_\circ).
  \]
  \end{enumerate}
\end{rmk}

\begin{rmk}[Other model metrics at $\hat M$]
\label{RmkGlOther}
  While the assumption on $\wt g|_{\hat M}$ in Definition~\ref{DefGl} is the relevant setting for applications to black hole gluing, one can easily consider other settings as well. For example, one can require $\wt g_{\mu\nu}(0,t,\hat x)=\hat g_{\hat\mu\hat\nu}(\hat x)$ where $\hat g$ is a stationary asymptotically flat metric. The setting where $\hat g$ is equal to the Minkowski metric (or a general asymptotically flat metric without trapping and horizons) is of particular interest if one wishes to study the wave equation $\Box_{g_\eps}+V_\eps$ on (possibly degenerating) spacetimes coupled to sharply localized potentials $V_\eps(x)=\eps^{-2}V(\frac{x}{\eps})$ where $V$ is Schwartz (or has inverse cubic decay, say). For further comments on this particular case, see Remarks~\ref{RmkGlDynOther}, \ref{RmkEstOther}, \ref{RmkEstOtherUnif}, and \ref{RmkScSGV}.
\end{rmk}

\subsection{Phase space dynamics}
\label{SsGlDyn}

\emph{For notational simplicity, we shall assume from now on that
\begin{equation}
\label{EqGlDynwtgCI}
  \wt g \in \CI(\wt M\setminus\wt K^\circ;S^2\wt T^*\wt M),
\end{equation}
and thus $\wt g\in\hat\rho^2\CI(\wt M\setminus\wt K^\circ;S^2\,\Tse^*\wt M)$} unless otherwise noted. That is, we assume that the se-regular remainder terms of $\wt g$ in Definition~\ref{DefGl}\eqref{ItGlwtg} are, in fact, smooth. We stress, however, that the results in this section go through under the original assumption $\wt g\in\CI+\cC_\seop^{\infty,1,1}$ (and indeed under the even weaker assumption $\wt g\in(\CI+\cC_\seop^{\infty,\delta,\delta})(\wt M\setminus\wt K^\circ;S^2\wt T^*\wt M)$ for some $\delta>0$). (When this is not obvious, we will point this out explicitly.) We shall study the null-bicharacteristic dynamics in $\Tse^*\wt M\setminus o$ in stages:
\begin{enumerate}
\item \S\ref{SssGlDynMc} recalls results from \cite{HintzConicWave} regarding the dynamics over $M_\circ$;
\item \S\ref{SssGlDynKerr} describes the structure of the null-geodesic flow (lifted to phase space) on subextremal Kerr spacetimes following \cite{DyatlovWaveAsymptotics};
\item \S\ref{SssGlDynTr} concerns properties of the parallel transport of vectors along the trapped set of Kerr and follows \cite{MarckParallelNull};
\item \S\ref{SssGlDynStd} gives a description of the se-phase space dynamics uniformly for small $\eps>0$.
\end{enumerate}

For nonlinear applications, it is useful to be more precise as far as the \emph{amount} of se-regularity of the remainder term of $\wt g$ is concerned. This only plays a role in~\S\ref{SssGlDynStd} (since in the earlier sections the remainder term does not enter at all, as it vanishes on $M_\circ\cup\hat M$).

\begin{rmk}[Other settings]
\label{RmkGlDynOther}
  When $\wt g|_{\hat M_t}=\hat g$ is an asymptotically flat nontrapping spacetime without horizons (in the context of Remark~\ref{RmkGlOther}), the dynamics of the null-bi\-char\-ac\-ter\-is\-tic flow over $\hat M^\circ$ is much simpler than in the Kerr setting considered here: it has a simple source-to-sink structure where the source and sink are~\eqref{EqGlDynKerrIn} and \eqref{EqGlDynKerrOut}, respectively.
\end{rmk}

We denote by
\[
  \wt G \colon \Tse^*_z\wt M \ni \zeta \mapsto \wt g|_z^{-1}(\zeta,\zeta)
\]
the dual metric function. Using~\eqref{EqGlDynwtgCI}, the rescaling $\hat\rho^{-2}\wt g\in\CI(\wt M\setminus\wt K^\circ;S^2\,\Tse^*\wt M)$ is a nondegenerate Lorentzian signature section down to $\hat M\cup M_\circ$, and correspondingly\footnote{We recall that $P^{[2]}$ refers to spaces of fiber-wise homogeneous polynomials of degree $2$ with smooth dependence on the base point. For se-regular metric remainder terms,~\eqref{EqGlDynwtG} gets replaced by $(\CI+\cC_\seop^{\infty,1,1})(\wt M\setminus\wt K^\circ)P^{[2]}(\Tse^*_{\wt M\setminus\wt K^\circ}\wt M)$.}
\begin{equation}
\label{EqGlDynwtG}
  \hat\rho^2\wt G \in P^{[2]}(\Tse^*_{\wt M\setminus\wt K^\circ}\wt M).
\end{equation}
We then write
\[
  \wt\Sigma := (\Tse^*\wt M \setminus o) \cap \{ \hat\rho^2\wt G = 0 \},\qquad
  \pa\wt\Sigma \subset \Sse^*\wt M,
\]
for the characteristic set and its boundary at fiber infinity, respectively. We fix as the time orientation of $\wt g$ the unique one for which $\hat\rho\pa_t=\rho_\circ^{-1}\pa_{\hat t}$ is future timelike at $M_\circ\cap\hat M$. (This is consistent with the usual time orientation on Kerr in which $\pa_{\hat t}$ is future timelike for $\hat r\gg\bhm$, and also with the conventions in Lemma~\ref{LemmaGlCoord} and Definition~\ref{DefGl}.) We then write
\[
  \wt\Sigma = \wt\Sigma^+ \sqcup \wt\Sigma^-
\]
for the future (`$+$') and past (`$-$') components of $\wt\Sigma$.

We shall describe the structure of the flow of the rescaled Hamiltonian vector field
\[
  \hat\rho^2 H_{\wt G} \in \Vse(\Tse^*_{\wt M\setminus\wt K^\circ}\wt M)
\]
in $\wt\Sigma$, where $\hat\rho^2 H_{\wt G}=H_{\hat\rho^2\wt G}$. Here we write $\Vse(\Tse^*\wt M)$ for the space of smooth vector fields which, in terms of the fiber-linear coordinates defined by writing se-covectors as $-\sigma_\seop\,\frac{\dd t}{\hat\rho}+\xi_\seop\,\frac{\dd r}{\hat\rho}+\eta_\seop$, $\eta_\seop\in T^*\Sph^2$, are spanned over $\CI(\Tse^*\wt M)$ by lifts of se-vector fields on $\wt M$ and the vertical vector fields $\pa_{\sigma_\seop}$, $\pa_{\xi_\seop}$, $\pa_{\eta_\seop}$. In $\eps>0$, the $\hat\rho^2 H_{\wt G}$-flow is the same as the lift of the null-geodesic flow on $(M_\eps,g_\eps)=(M_\eps,\wt g|_{M_\eps})$ to the cotangent bundle, up to reparametrization.

\subsubsection{Far field regime: dynamics over \texorpdfstring{$M_\circ$}{the lift of the original spacetime}}
\label{SssGlDynMc}

Over the interior of $M_\circ$ (where $\Tse^*_{(M_\circ)^\circ}\wt M=T^*(M\setminus\cC)$), the vector field $H_{\wt G}$ is equal to the restriction of $H_G$ to $T^*_{M\setminus\cC}M$. Near $\cC$ on the other hand, we may choose $\hat\rho$ to be equal to $|x|$, and then
\[
  \hat\rho^2 H_{\wt G}|_{\Tse^*_{M_\circ}\wt M}=|x|^2 H_G\in\Ve(\Te^*M_\circ).
\]
Therefore, we are in effect regarding $\cC\subset M$ as a timelike curve of conic singularities (which happen to be smooth) and study the phase space dynamics in the edge cotangent bundle; the analysis of general such settings is presented in \cite{HintzConicWave}. We shall show here that the only radial sets for the $|x|^2 H_G$-flow in $\wt\Sigma\cap\Te^*M_\circ$ are the sets defined as follows:

\begin{definition}[Radial sets over $\pa M_\circ$]
\label{DefGlDynInOut}
  Write edge covectors as $-\sigma_\eop\frac{\dd t}{r}+\xi_\eop\frac{\dd r}{r}+\eta_\eop$ as in~\eqref{EqFseEdgeCoord}. We then define the \emph{incoming} and \emph{outgoing radial sets over $\pa M_\circ$} by $\cR_{\rm in}=\bigsqcup_\pm\cR_{\rm in}^\pm$, $\cR_{\rm out}=\bigsqcup_\pm\cR_{\rm out}^\pm\subset\Te^*_{\pa M_\circ}M_\circ=\Tse^*_{\pa M_\circ}\wt M$, where
  \begin{equation}
  \label{EqGlDynInOutSets}
  \begin{alignedat}{2}
    \cR_{\rm in}^\pm &= \{ (t,r,\omega;\sigma_\eop,\xi_\eop,\eta_\eop) = (t,0,\omega;\sigma_\eop,-\sigma_\eop,0) \colon {\pm}\sigma_\eop > 0 \} && \subset \wt\Sigma^\pm \cap \Tse^*_{\pa M_\circ}\wt M, \\
    \cR_{\rm out}^\pm &= \{ (t,r,\omega;\sigma_\eop,\xi_\eop,\eta_\eop) = (t,0,\omega;\sigma_\eop,\sigma_\eop,0) \colon {\pm}\sigma_\eop > 0 \} && \subset \wt\Sigma^\pm \cap \Tse^*_{\pa M_\circ}\wt M.
  \end{alignedat}
  \end{equation}
\end{definition}

In the local coordinates~\eqref{EqFseEdgeCoord}, note that $G\in P^{[2]}(T^*M)$ is equal to the dual metric function of the Minkowski metric $\ubar g=-\dd t^2+\dd x^2$ plus terms vanishing over $\cC$. Therefore,
\begin{equation}
\label{EqGlDynGe}
  G_\eop := r^2 G = -\sigma_\eop^2 + \xi_\eop^2 + |\eta_\eop|^2 + \tilde G_\eop,\qquad \tilde G_\eop\in r P^{[2]}(\Te^*M).
\end{equation}
The time orientation is such that $-\dd t$ is future timelike; since $\ubar g^{-1}(-\dd t,-\dd t\pm\dd r)=\ubar g^{-1}(-\dd t,-\dd t)<0$, the 1-forms $-\dd t\pm\dd r$ are future lightlike for $\ubar g$, and thus so are $-\sigma_\eop\frac{\dd t}{r}+\xi_\eop\frac{\dd r}{r}$ for $r^{-2}g$ at $r=0$ where $\sigma_\eop=1$, $\xi_\eop=\pm 1$; this justifies the signs in~\eqref{EqGlDynInOutSets}. We then record that
\begin{align}
  H_{G_\eop} &= -(\pa_{\sigma_\eop}G_\eop)r\pa_t + (\pa_{\xi_\eop}G_\eop)(r\pa_r+\sigma_\eop\pa_{\sigma_\eop}) + (\pa_{\eta_{\eop,j}}G_\eop)\pa_{\omega^j} \nonumber\\
    &\qquad + r(\pa_t G_\eop)\pa_{\sigma_\eop} - \bigl((r\pa_r+\sigma_\eop\pa_{\sigma_\eop})G_\eop\bigr)\pa_{\xi_\eop} - (\pa_{\omega^j}G_\eop)\pa_{\eta_{\eop,j}} \nonumber\\
\label{EqGlDynHamEdge}
    &= 2\sigma_\eop r\pa_t + 2\xi_\eop(r\pa_r+\sigma_\eop\pa_{\sigma_\eop}) + 2\slg^{j k}\eta_{\eop,k}\pa_{\omega^j} + 2\sigma_\eop^2\pa_{\xi_\eop} - (\pa_{\omega^j}\slg^{k l})\eta_{\eop,k}\eta_{\eop,l}\pa_{\eta_{\eop,j}} + H_{\tilde G_\eop},
\end{align}
where $\slg$ denotes the standard metric on $\Sph^2$. We have $H_{\tilde G_\eop}\in r\Ve(\Te^*M)$. Consider now $H_{G_\eop}$ on the characteristic set over $r=0$: when $\eta_\eop\neq 0$, it has a non-vanishing $\pa_\omega$-component; when $\eta_\eop=0$, we must have $-\sigma_\eop^2+\xi_\eop^2=0$, so we are in $\cR_{\rm in}\cup\cR_{\rm out}$.

We analyze the null-bicharacteristic dynamics near the radial sets using projective coordinates on $\ol{\Te^*}M_\circ$. For brevity, we only consider the radial sets $\cR_{\rm in}^+$, $\cR_{\rm out}^+$. Near their boundaries $\pa\cR_{\rm in}^+$, $\pa\cR_{\rm out}^+$ at fiber infinity $\Se^*M_\circ$, we introduce projective coordinates
\begin{equation}
\label{EqGlDynInOutCoord2}
  \rho_\infty = \frac{1}{\sigma_\eop},\quad
  \hat\xi_\eop = \frac{\xi_\eop}{\sigma_\eop},\quad
  \hat\eta_\eop = \frac{\eta_\eop}{\sigma_\eop},
\end{equation}
and compute $\rho_\infty^2 G_\eop\equiv-1+\hat\xi_\eop^2+|\hat\eta_\eop|^2\bmod r\CI$ and
\begin{equation}
\label{EqGlDynInOutHam}
  \rho_\infty H_{G_\eop} \equiv 2\hat\xi_\eop(r\pa_r-\rho_\infty\pa_{\rho_\infty}-\hat\eta_\eop\pa_{\hat\eta_\eop}) + 2(1-\hat\xi_\eop^2)\pa_{\hat\xi_\eop} + \rho_\infty H_{|\eta_\eop|^2} \bmod r\Vb(\Te^*M);
\end{equation}
here $\rho_\infty H_{|\eta_\eop|^2}=2\slg^{j k}\hat\eta_{\eop,k}\pa_{\omega^j}-(\pa_{\omega^j}\slg^{k l})\hat\eta_{\eop,k}\hat\eta_{\eop,l}\pa_{\hat\eta_{\eop,j}}$. Along integral curves of $\rho_\infty H_{G_\eop}$ within the characteristic set over $r=0$, $\hat\xi_\eop$ either remains constant at $\pm 1$ (which happens at the radial sets), or it tends to $\pm 1$ as the affine parameter tends to $\pm\infty$. Thus, in the characteristic set over $\pa M_\circ$, $\cR_{\rm in}^+$ is a source and $\cR_{\rm out}^+$ is a sink for the $\rho_\infty H_{G_\eop}$ flow. (This behavior is already discussed in the ultrastatic conic context in \cite{MelroseWunschConic}.)

On the other hand, the coefficient of $r\pa_r$ in~\eqref{EqGlDynInOutHam} is positive at $\pa\cR_{\rm out}^+$ and negative at $\pa\cR_{\rm in}^+$, so $\pa\cR_{\rm in}^+$ is a sink in the radial direction and $\pa\cR_{\rm out}^+$ is a source. We can explicitly construct the unstable manifold
\[
  \pa W_{t_0,\rm out}^+\subset\wt\Sigma^+\cap\Se^*M_\circ
\]
of the $\rho_\infty H_{G_\eop}$-invariant set $\pa\cR_{t_0,\rm out}^+=\pa\cR_{\rm out}^+\cap\{t=t_0\}$. One possibility is to appeal to the stable/unstable manifold theorem \cite{HirschPughShubInvariantManifolds}. More simply, one can exploit the smoothness of $(M,g)$: define $W_{M,t_0,\rm out}^+$ as the forward flow-out under $H_G$ of the set $\cR_{M}^+:=\{(t,x,\sigma,\xi)\colon x=0,\ \sigma>0,\ |\xi|=\sigma\}\subset T^*_\cC M$ (where we write covectors as $-\sigma\,\dd t+\xi\,\dd x$) inside of $T^*M\setminus\cR_M^+$. Then $\pa W_{t_0,\rm out}^+$ is the union of $\pa\cR_{t_0,\rm out}^+$ with the boundary at fiber infinity of the closure of $W_{M,\rm out}^+\cap T^*(M\setminus\cC)$ inside of $\ol{\Te^*_{(M_\circ)^\circ}}M_\circ\setminus o$. Since a backward null-bicharacteristic of $H_{G_\eop}$ in $\wt\Sigma^+$ starting over $(M_\circ)^\circ=M\setminus\cC$ tends to $\{t=t_0\}\cap\Te^*_{\pa M_\circ}M_\circ$ if and only if the backward null-bicharacteristic of $H_G$ starting at the same point intersects $T^*_\cC M$ (and thus necessarily the dual light cone of $(M,g)$ over $T^*_\cC M$) at $t=t_0$, this indeed produces the unstable manifold. In an analogous fashion, one can construct the stable manifold
\[
  \pa W_{t_0,\rm in}^+\subset\wt\Sigma^+\cap\Se^*M_\circ
\]
of $\pa\cR_{t_0,\rm in}^+=\pa\cR_{\rm out}^+\cap\{t=t_0\}$.

\begin{definition}[Stable and unstable manifolds]
\label{DefGlDynStableUn}
  For $t_0\in I_\cC$, we define $\pa W_{t_0,\rm out}^\pm$, resp.\ $\pa W_{t_0,\rm in}^\pm$ as the unstable, resp.\ stable manifold of $\pa\cR_{t_0,\rm out}^\pm=\pa\cR_{\rm out}^\pm\cap\{t=t_0\}$, resp.\ $\pa\cR_{t_0,\rm in}^\pm=\pa\cR_{\rm in}^\pm\cap\{t=t_0\}$ for the flow of $\pm\rho_\infty H_{G_\eop}$ inside of $\Se^*M_\circ$. Furthermore, $W_{t_0,\rm out}^\pm\subset\Te^*M_\circ\setminus o$ is the conic extension of $\pa W_{t_0,\rm out}^\pm$, and $W_{t_0,\rm out}=\bigsqcup_\pm W_{t_0,\rm out}^\pm$. Finally, we write $W_{\rm out}=\bigcup_{t_0}W_{t_0,\rm out}$; similarly for `in' in place of `out'.
\end{definition}

\begin{definition}[Standard domains in $M$]
\label{DefGlDynStdM}
  Let $T\in\cV(M)$ denote a future timelike vector field. A compact submanifold with corners $\Omega\subset M$ is a \emph{standard domain} if satisfies one of the following two alternatives.
  \begin{enumerate}
  \item\label{ItGlDynStdMC} $\Omega$ is contained in the Fermi normal coordinate chart $U_{\rm Fermi}$ and given by
    \[
      \Omega = \Omega_{t_0,t_1,r_0} := \{ (t,x) \colon t_0\leq t\leq t_1,\ |x|\leq r_0 + 2(t_1-t) \}
    \]
    for some $t_0<t_1$ and $r_0>0$, and moreover:
    \begin{enumerate}
    \item\label{ItGlDynStdIni} $X_{t_0,t_1,r_0}=\{t=t_0\}\cap\Omega_{t_0,t_1,r_0}$ is an initial spacelike hypersurface (i.e.\ $T t>0$ on $X_{t_0,t_1,r_0}$);
    \item\label{ItGlDynStdFin} $Y^-_{t_0,t_1,r_0}=\{t=t_1\}\cap\Omega_{t_0,t_1,r_0}$ is a final spacelike hypersurface (i.e.\ $T t>0$ on $Y^-_{t_0,t_1,r_0}$), and also $Y^\wedge_{t_0,t_1,r_0}=\{|x|=r_0+2(t_1-t)\}\cap\Omega_{t_0,t_1,r_0}$ is a final spacelike hypersurface (i.e.\ $T(|x|-(r_0+2(t_1-t)))>0$ on $Y^\wedge_{t_0,t_1,r_0}$);
    \item\label{ItGlDynStdNonrf} $\Omega$ is \emph{non-refocusing} in that $W_{t_-,\rm out}^+\cap W_{t_+,\rm in}^+=\emptyset$ for all $t_-,t_+\in[t_0,t_1]$;
    \item\label{ItGlDynStdOrder} in the coordinates~\eqref{EqFseEdgeCoord}, we have $\sigma_\eop>0$ on $\wt\Sigma^+\cap\Te^*_\Omega M_\circ$, and there exists $c_0>0$ so that for the function $\hat\xi_\eop=\frac{\xi_\eop}{\sigma_\eop}$ we have
      \begin{equation}
      \label{EqGlDynStdMonotone}
        \sigma_\eop^{-1} H_{G_\eop}\hat\xi_\eop\geq c_0\ \ \text{on}\ \wt\Sigma^+\ \text{whenever}\ \hat\xi_\eop\in\Bigl[-\frac34,\frac34\Bigr].
      \end{equation}
    \end{enumerate}
  \item\label{ItGlDynStdMNonC} $\Omega$ is \emph{disjoint} from $\cC$, has only spacelike boundary hypersurfaces, and only one boundary hypersurface $X\subset\Omega$ is initial (i.e.\ $T$ points into $\Omega$ everywhere in $X^\circ$).
  \end{enumerate}
  We moreover abuse notation and write $\Omega$ also for its lift $\upbeta_\circ^*\Omega$ to $M_\circ$.
\end{definition}

Condition~\eqref{ItGlDynStdNonrf} is an instance of \cite[Definition~2.5]{HintzConicWave}. Condition~\eqref{ItGlDynStdOrder} will be used to aid the construction of edge regularity order functions which are monotonically decreasing along the flow of $\rho_\infty H_{G_\eop}$; see~\S\ref{SsEstStd}.

\begin{rmk}[More general domains]
\label{RmkGlDynStdMore}
  The requirement that only one boundary hypersurface be initial is made solely for notational convenience later on. Similarly, the exact form of standard domains intersecting $\cC$ is chosen so as to minimize notational complications (see \cite[Definition~2.2]{HintzConicWave} for a more general definition). In our gluing applications, we can freely choose the domains on which to work, and thus we leave the modifications required for analysis on more general domains to the reader.
\end{rmk}

\begin{lemma}[Existence of standard domains]
\label{LemmaGlDynStdEx}
  Fix $t'\in I_\cC$. Then there exists $\eta>0$ so that for all $t'-\eta<t_0<t_1<t'+\eta$ and $r_0\in(0,\eta)$, the domain $\Omega_{t_0,t_1,r_0}$ is a standard domain.
\end{lemma}
\begin{proof}
  The conditions~\eqref{ItGlDynStdIni} and \eqref{ItGlDynStdFin} in Definition~\ref{DefGlDynStdM} are true when $g$ is the Minkowski metric in Fermi normal coordinates. Therefore, they remain true for domains $\Omega_{t_0,t_1,r_0}$ which are contained in a sufficiently small neighborhood of $\cC$, which is guaranteed for all small $\eta>0$. Similarly, if $g$ is the Minkowski metric, then all domains $\Omega_{t_0,t_1,r_0}$ are non-refocusing; the continuous dependence of geodesics on the metric implies the non-refocusing property (i.e.\ condition~\eqref{ItGlDynStdNonrf}) for small $\eta>0$. For condition~\eqref{ItGlDynStdOrder} finally, fix a number $c'_0>0$ for which~\eqref{EqGlDynStdMonotone} is valid on Minkowski space with $c'_0$ in place of $c_0$; then for $c_0=\frac12 c'_0$ the inequality~\eqref{EqGlDynStdMonotone} remains valid when $\eta>0$ is sufficiently small, since then the contribution of the remainder term in~\eqref{EqGlDynInOutHam} is smaller than $\frac12 c'_0$.
\end{proof}

\begin{lemma}[Null-bicharacteristic flow on standard domains: $M_\circ$]
\label{LemmaGlDynStdFlowM}
  Let $\Omega$ be a standard domain. Denote by $X\subset\Omega$ its initial boundary hypersurface, and by $Y$ the union of its final boundary hypersurfaces. Let $\gamma\colon I\subseteq\R\to\pa\wt\Sigma^\pm\cap\Se^*_{\Omega}M_\circ$, $0\in I$, be a maximal integral curve of $\pm\sigma_\eop^{-1}r^2 H_G$. Let $s_-=\inf I$ and $s_+=\sup I$. Then one the following possibilities occurs.
  \begin{enumerate}
  \item\label{ItGlDynStdFlowMRad} $\gamma$ is contained in the radial sets $\pa\cR_{\rm in}\cup\pa\cR_{\rm out}$ from Definition~\usref{DefGlDynInOut} and thus constant;
  \item\label{ItGlDynStdFlowMInOut} $\gamma\subset\Se^*_{\pa M_\circ}M_\circ\setminus(\pa\cR_{\rm in}\cup\pa\cR_{\rm out})$, in which case $s_\pm=\pm\infty$, and $\lim_{s\searrow -\infty}\gamma(s)\in\pa\cR_{\rm in}$ and $\lim_{s\nearrow\infty}\gamma(s)\in\pa\cR_{\rm out}$;
  \item\label{ItGlDynStdFlowMOutY} $\gamma\subset S^*(M\setminus\cC)$, and $\gamma(s)\to\pa\cR_{\rm out}$ as $s\searrow s_-=\infty$, while $s_+<\infty$ and $\gamma(s_+)\in\Se^*_Y M_\circ$;
  \item\label{ItGlDynStdFlowMXIn} $\gamma\subset S^*(M\setminus\cC)$, and $s_->-\infty$, $\gamma(s_-)\in\Se^*_X M_\circ$, while $\gamma(s)\to\pa\cR_{\rm in}$ as $s\nearrow s_+=\infty$;
  \item\label{ItGlDynStdFlowMXY} $|s_\pm|<\infty$, and $\gamma(s_-)\in\Se^*_X M_\circ$, $\gamma(s_+)\in\Se^*_Y M_\circ$.
  \end{enumerate}
\end{lemma}
\begin{proof}
  When $\Omega\cap\cC=\emptyset$, only possibility~\eqref{ItGlDynStdFlowMXY} can occur. Consider thus the case that $\Omega\cap\cC\neq\emptyset$. When $\gamma$ lies over $\pa M_\circ$, only possibilities~\eqref{ItGlDynStdFlowMRad} and \eqref{ItGlDynStdFlowMInOut} occur, as already discussed.

  Consider next the case that $\gamma$ does not lie in $\Te^*_{\pa M_\circ}M_\circ$; then it is disjoint from it. If $\gamma(0)\in\pa W_{t_0,\rm in}^+$ with $(t_0,0)\in\Omega$, then $\gamma(s)\to\pa\cR_{\rm in}^+$ as $s\nearrow\infty$; otherwise, $\gamma(s)$ hits the future boundary $Y$ of $\Omega$ at a finite value $s=s_+<\infty$. Depending on whether $\gamma(0)\in\pa W_{t_0,\rm out}^+$ with $(t_0,0)\in\Omega$ or not, we then similarly have $\gamma(s)\to\pa\cR_{\rm out}^+$ or $\gamma(s_-)\in\Se^*_X M_\circ$. This corresponds to possibilities~\eqref{ItGlDynStdFlowMOutY}--\eqref{ItGlDynStdFlowMXY}.
\end{proof}

See Figure~\ref{FigGlDynStdFlowM}.

\begin{figure}[!ht]
\centering
\includegraphics{FigGlDynStdFlowM}
\caption{Illustration of Lemma~\ref{LemmaGlDynStdFlowM}. \textit{On the left:} a standard domain of the type in Definition~\ref{DefGlDynStdM}\eqref{ItGlDynStdMC} (not to scale in the $x$ direction). In the edge cosphere bundle $\Se^*M_\circ$ over the green point, the incoming (red) and outgoing (blue) null-geodesics, lifted to $\Se^*M_\circ$, limit to $\pa\cR_{\rm in}$ and $\pa\cR_{\rm out}$, respectively. \textit{On the right:} a standard domain of the type in Definition~\ref{DefGlDynStdM}\eqref{ItGlDynStdMNonC}; the green lines are null-geodesics.}
\label{FigGlDynStdFlowM}
\end{figure}

The following covering result will be used to pass from local results on $\wt M$ to semiglobal ones; see~\S\ref{SsScSG}.

\begin{prop}[Covering with standard domains]
\label{PropGlDynCover}
  Let $M=\R_\cT\times X_0$, $g=-\beta\,\dd\cT^2+h(\cT)$, be a metric splitting of $(M,g)$, so $0<\beta\in\CI(M)$, $h$ is a smooth (in $\cT$) family of Riemannian metrics on $X_0$, and every level set $X_\cT$ of $\cT$ is a Cauchy hypersurface of $(M,g)$. Let $K\subset J^+(X_1)$ (with $J^+$ denoting the causal future) be compact. Let $\{\fp\}=X_0\cap\cC$, and let $X'_0$ be a Cauchy hypersurface of $(M,g)$ contained in $\{-1<\cT<1\}$ which, outside of any fixed neighborhood of $\fp$ is equal to $X_0$ but which near $\fp$ is equal to a level set of the (Fermi normal coordinate) $t$.\footnote{The existence of such $X'_{-1}$ was shown in \cite[Lemma~3.34]{HintzGlueLocI}.} Let $\eta>0$. Then there exist $J\in\N_0$ and a collection $\Omega_j$, $0\leq j\leq J$, of standard domains in $M$ with the following properties:
  \begin{enumerate}
  \item\label{ItGlDynCoverCov} $K\subseteq\bigcup_{j=0}^J\Omega_j$;
  \item\label{ItGlDynCoverIni} for all $j$, the initial Cauchy surface of $\Omega_j$ is contained in $X'_0$ or in $\bigcup_{k\leq j-1}\Omega_k^\circ$;
  \item\label{ItGlDynCoverSmall} all $\Omega_k$ which intersect $\cC$ are of the form $\Omega_{t_0,t_1,r_0}$ for $k$-dependent values $t_0,t_1,r_0$ satisfying $|t_1-t_0|,r_0<\eta$.
  \end{enumerate}
\end{prop}

The existence of a metric splitting was proved by Bernal--S\'anchez \cite{BernalSanchezTimeFn}. See Figure~\ref{FigGlDynCover} for an illustration of Proposition~\ref{PropGlDynCover}.

\begin{figure}[!ht]
\centering
\includegraphics{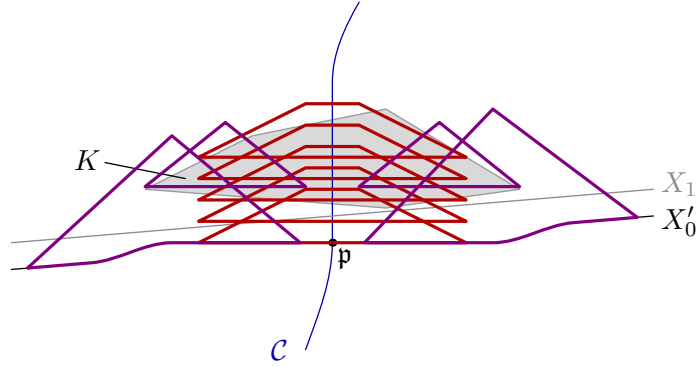}
\caption{Illustration of Proposition~\ref{PropGlDynCover}. The standard domains $\Omega_j$ intersecting $\cC$ are drawn in red, those disjoint from $\cC$ in purple. (The purple domains constructed in the proof of the Proposition are truncated at the top.)}
\label{FigGlDynCover}
\end{figure}

\begin{proof}[Proof of Proposition~\usref{PropGlDynCover}]
  By reparameterizing $\cT$, we may assume that $\cT=t$ on $\cC$; then $M=I_\cC\times X_0$. Let $V_0\subset V_1\subset X'_0$ be precompact open subsets with $X'_0\cap J^-(K)\subset V_0$ and $\ol{V_0}\subset V_1$. Denote by $D_j=D^+(V_j)\subset J^+(X'_0)$ the future domain of dependence of $V_j$ for $j=0,1$; so in particular $D_j\cap X_\cT$ is precompact for all $\cT$, and for $\cT^+:=\sup_K\cT<\infty$ we have $K\subset D_0([0,\cT^+]):=D_0\cap\cT^{-1}([0,\cT^+])$.

  For all sufficiently small $r_0>0$ and $\delta>\tau>0$, the set $[0,\cT^++\tau+\delta]_t\times\{|x|\leq r_0+2\delta\}$ is contained in $U_{\rm Fermi}$, and furthermore $X_0'$ and $X_0$ agree on $\{|x|\leq r_0+2\delta\}$. Set $t_0^{(q)}:=q\tau$ for $q\in\N_0$, and let $Q\in\N$ be the smallest value with $t_0^{(Q)}\geq\cT^+$; then $t_0^{(Q)}+\delta<\sup I_\cC$. Moreover, the union of the sets
  \[
    \Omega^{(q)}:=\Omega_{t_0^{(q)},t_0^{(q)}+\delta,r_0},\qquad q=0,\ldots,Q,
  \]
  covers $\cC\cap\cT^{-1}([0,\cT^+])$. Set $X^{(q)}:=X_{t_0^{(q)},t_0^{(q)}+\delta,r_0}$ (see Definition~\ref{DefGlDynStdM}\eqref{ItGlDynStdIni}). We claim that if we reduce the size of $r_0,\delta,\tau$ even further (if necessary), we can moreover arrange that
  \begin{equation}
  \label{EqGlDynCoverX}
    J^-(X^{(q+1)}\setminus(\Omega^{(q)})^\circ)\cap\{|x|<r_0/2,\ \cT=t_0^{(q)}\}=\emptyset,\qquad q=0,\ldots,Q.
  \end{equation}
  Indeed, since for any $t_0\in I_\cC$ the set $X_{t_0+\tau,t_0+\tau+\delta,r_0}\setminus(\Omega_{t_0,t_0+\delta,r_0})^\circ$ converges, as $\tau\searrow 0$, to the coordinate sphere $\{t_0\}\times\{|x|=r_0+2\delta\}\subset\cT^{-1}(t_0)$, which is disjoint from $\cC$, it suffices to choose $\tau>0$ sufficiently small.

  We set $\Omega'_0=\Omega^{(0)}$. In view of~\eqref{EqGlDynCoverX}, the union of $\Omega'_0$ and the future domain of dependence of $X'_0\setminus\{|x|\leq r_0/4\}$ contains $D_0([0,t_0^{(1)}])$. Taking $\Omega'_1=D^+(K'_0)\cap\{\cT\leq t_0^{(1)}\}$ for some sufficiently large compact $K'_0\subset D_1\cap\cT^{-1}(0)\setminus\{|x|\leq r_0/4\}$, we thus have $D_0([0,t_0^{(1)}])\subset\Omega'_0\cup\Omega'_1$. We continue with $\Omega'_2:=\Omega^{(1)}$ and $\Omega'_3=D^+(K_1)\cap\{\cT\leq t_0^{(2)}\}$ for a sufficiently large compact $K_1\subset D_1\cap\cT^{-1}(t_0^{(1)})\setminus\{|x|\leq r_0/4\}$ so that $D_0([0,t_0^{(2)}])\subset\bigcup_{j=0}^3\Omega'_j$; and so on.

  To finish the construction, we need to modify the thus constructed sets $\Omega'_0,\ldots,\Omega'_J$ (with $J=2 Q+1$) slightly for odd indices $j=2 l+1$: the lateral boundary hypersurfaces of $\Omega'_j$, by which we mean those not contained in $\cT^{-1}(t_0^{(l)})$ and $\cT^{-1}(t_0^{(l+1)})$, are lightlike (and smooth if one chooses $X_j$ appropriately and $\tau>0$ so small that no focal points develop along future null-geodesics emanating from $\pa X_j$ inside $\cT^{-1}([t_0^{(l)},t_0^{(l+1)}])$). We thus perturb $\Omega'_j$ to an appropriate slightly larger domain whose lateral boundary hypersurfaces are spacelike, thus obtaining $\Omega_j$ for odd $j$. For even $j$, we set $\Omega_j=\Omega'_j$. This finishes the construction.
\end{proof}

\subsubsection{Near field regime: Kerr dynamics over \texorpdfstring{$\hat M_{t_0}$}{a fiber of the front face}}
\label{SssGlDynKerr}

We now turn to the null-bi\-char\-ac\-ter\-is\-tic flow of $\hat\rho^2 H_{\wt G}=\rho_\circ^{-2}\eps^2 H_{\wt G}=\rho_\circ^{-2}H_{\eps^2\wt G}$ over $\hat M$. Since this vector field is an se-vector field, i.e.\ tangent to the fibers of $\hat M$, it suffices to study it over $\hat M_{t_0}$ for any fixed $t_0\in I_\cC$. By Remark~\ref{RmkGlSe3be}\eqref{ItGlSe3be3b}, $\eps^2\wt G|_{\hat M_{t_0}}\in\rho_\circ^2\CI(\hat M_{t_0};\Tse^*_{\hat M_{t_0}}\wt M)$ can be identified with the dual metric function $\hat G_b\in\hat r^{-2}P^{[2]}(\Ttb^*_{\hat X_b}\hat M_b)=P^{[2]}(\Ttsc^*_{\hat X_b}\hat M_b)$ of the Kerr metric $\hat g_b$ via the identification~\eqref{EqFseBundle3b}; here we write
\[
  b:=(\bhm,\bha).
\]
We conclude that the se-vector field $H_{\eps^2\wt G}$ has as its \emph{normal operator} at $\Tse^*_{\hat M_{t_0}}\wt M$ (given by restriction to the front face of the lift of $\Tse^*\wt M$ to $[\wt M;\hat M_{t_0}]$, which is equal to $\Ttb^*\cM$) the Hamiltonian vector field $H_{\hat G_b}\in\hat r^{-2}\Vtb(\Ttb^*\hat M_b)$. The \emph{restriction} of $H_{\eps^2\wt G}$ to $\Tse^*_{\hat M_{t_0}}\wt M$ is therefore equal to the zero energy operator of $H_{\hat G_b}$, which we denote
\[
  H^0_{\hat G_b} \in \hat r^{-2}\Vb(\Ttb^*_{\hat X_b}\hat M_b);
\]
this is formally obtained by removing the $\pa_{\hat t}$-derivatives in the expression for $H_{\hat G_b}$.

In summary, the null-bicharacteristic dynamics of $\rho_\circ^{-2}H_{\eps^2\wt G}$ in $\Tse^*_{\hat M_{t_0}}\wt M$ are those of
\begin{equation}
\label{EqGlDynKerrH0}
  \hat r^2 H_{\hat G_b}^0\in\Vb(\Ttb^*_{\hat X_b}\hat M_b),
\end{equation}
i.e.\ those of the subextremal Kerr metric $\hat g_b$, quotiented out by $\hat t$-translations. (The $\pa_{\hat t}$-component of $H_{\hat G_b}$ does play a role later on when we work on domains with initial and final Cauchy hypersurfaces; see~\S\ref{SssGlDynStd} below.)

\begin{rmk}[Local coordinate description]
\label{RmkGlDynKerrLoc}
  Local coordinates on $\Tse^*\wt M$ near the interior of $\hat M$ are $t\in\R$, $\hat x\in\R^3$ and $\sigma_\seop\in\R$, $\xi_\seop\in\R^3$, where we write se-covectors as
  \begin{equation}
  \label{EqGlDynKerrLocCoordse}
    {-}\sigma_\seop\,\frac{\dd t}{\eps} + \xi_\seop\,\frac{\dd x}{\eps},
  \end{equation}
  corresponding to $\eps\pa_t=\pa_{\hat t}$, $\eps\pa_x=\pa_{\hat x}$ being a local frame of $\Tse^*\wt M$ there. But then the Hamiltonian vector field of $p\in\CI(\Tse^*\wt M)$ is
  \begin{align}
    H_p &= -(\pa_{\sigma_\seop}p)\eps\pa_t + (\pa_{\xi_\seop}p)\eps\pa_x + \eps(\pa_t p)\pa_{\sigma_\seop} - \eps(\pa_x p)\pa_{\xi_\seop} \nonumber\\
  \label{EqGlDynKerrLocRestr}
      &= \bigl((\pa_{\xi_\seop}p)\pa_{\hat x} - (\pa_{\hat x}p)\pa_{\xi_\seop}\bigr) + \eps\bigl( -(\pa_{\sigma_\seop}p)\pa_t + (\pa_t p)\pa_{\sigma_\seop}\bigr) \\
  \label{EqGlDynKerrLocNormOp}
      &= -(\pa_{\sigma_\seop}p)\pa_{\hat t} + (\pa_{\xi_\seop}p)\pa_{\hat x} + (\pa_{\hat t}p)\pa_{\sigma_\seop} - (\pa_{\hat x}p)\pa_{\xi_\seop}.
  \end{align}
  The expression~\eqref{EqGlDynKerrLocNormOp} is the Hamiltonian vector field of $p=p(\hat t,\hat x,\sigma_\seop,\xi_\seop)$ since $\sigma_\seop,\xi_\seop$ are the canonical fiber-linear coordinates on $T^*\R^{1+3}_{\hat t,\hat x}$; indeed,~\eqref{EqGlDynKerrLocCoordse} is equal to $-\sigma_\seop\,\dd\hat t+\xi_\seop\,\dd\hat x$. (Since $p$ is smooth in $t$, we have $\pa_{\hat t}p\in\eps\CI$, so the $\pa_{\sigma_\seop}$-component vanishes at $\hat M$.) The expression~\eqref{EqGlDynKerrLocRestr}, which involves the smooth coordinates $t,\hat x,\sigma_\seop,\xi_\seop$ near $\Tse^*_{\hat M^\circ}\wt M$, is the most relevant one for studying the null-bicharacteristic dynamics of $H_p$ on $\Tse^*\wt M$; restricting it as a smooth vector field to $\hat M$ amounts to dropping the $\eps(\cdots)$ term, leaving one with the zero energy operator of the Hamiltonian vector field of the $\hat t$-independent function $p=p(t_0,\hat x,\sigma_\seop,\xi_\seop)$, with smooth parametric dependence on $t_0\in I_\cC$ (which in our setting is trivial since $p=\eps^2\wt G$ is independent of $t_0$ over $\hat M$, namely given by the dual metric function of the Kerr metric $\hat g_b$).
\end{rmk}

The global structure of the $\hat r^2 H_{\hat G_b}^0$-flow is described in a manner that is convenient for present purposes by Dyatlov \cite[\S{3.2}]{DyatlovWaveAsymptotics} (building on the crucial observation by Wunsch--Zworski \cite{WunschZworskiNormHypResolvent} regarding the normally hyperbolic nature of the trapped set), near the horizon in the closely related Kerr--de~Sitter setting by Vasy \cite[\S{6}]{VasyMicroKerrdS}, and near spatial infinity in \cite{DyatlovWaveAsymptotics} and (in a different asymptotically flat setting) in Vasy--Zworski \cite{VasyZworskiScl}. We recall here its key features.

\medskip

(i) \textit{Flow near spatial infinity.} The description of the flow near $\pa\hat X_b$ is the only part where the (3)b-structures in~\eqref{EqGlDynKerrH0} are relevant. First, note that by the discussion leading up to~\eqref{EqGlDynKerrH0}, the restriction $\hat r^2 H_{\hat G_b}^0|_{\Ttb^*_{\pa\hat X_b}\hat M_b}$ is equal to $\hat\rho^2 H_{\wt G}|_{\Tse^*_{\pa\hat M_{t_0}}\wt M}$ (under the usual identification of $\hat X_b$ with $\hat M_{t_0}\cap\{|\hat x|\geq\bhm\}$); but the source ($\cR_{\rm in}$) to sink ($\cR_{\rm out}$) nature of the flow of the latter vector field in $\wt\Sigma^+$ was already described in the previous section after~\eqref{EqGlDynInOutHam}.

Furthermore, recall that the fiber-linear coordinates~\eqref{EqFseEdgeCoord} are also smooth fiber-linear coordinates in an open neighborhood of $\Tse^*_{\pa M_\circ}\wt M$; as local coordinates on the base we use $t,r,\omega$ and $\rho_\circ=\frac{\eps}{r}=\hat r^{-1}$, with $r=0$ on $\hat M$. Thus, the expressions~\eqref{EqGlDynHamEdge}--\eqref{EqGlDynInOutHam} remain valid in $\Tse^*_{\pa M_\circ}\wt M$, \emph{mutatis mutandis}; to wit, in the projective fiber coordinates~\eqref{EqGlDynInOutCoord2}, which we now denote by $\sigma_\seop=\sigma_\eop$ etc., we have
\begin{equation}
\label{EqGlDynInOutMhat}
  \rho_\infty H_{\hat r^2\hat G_b}^0 \equiv 2\hat\xi_\eop(-\rho_\circ\pa_{\rho_\circ}-\rho_\infty\pa_{\rho_\infty}-\hat\eta_\eop\pa_{\hat\eta_\eop}) + 2(1-\hat\xi_\eop^2)\pa_{\hat\xi_\eop} + H_{|\hat\eta_\eop|^2}\bmod \rho_\circ\Vb(\Ttb^*_{\hat X_b}\hat M_b).
\end{equation}
This arises from the fact that $r\pa_r$ in $\eps,t,r,\omega$ coordinates reads $-\rho_\circ\pa_{\rho_\circ}+r\pa_r$ in $t,r,\omega,\rho_\circ$ coordinates, and the restriction of this to $r=0$ is $-\rho_\circ\pa_{\rho_\circ}$. This shows that
\begin{equation}
\label{EqGlDynKerrOut}
  \pa\hat\cR_{\rm out}^+ := \pa\cR_{\rm out}^+ \cap \Stb^*_{\hat X_b}\hat M_b
\end{equation}
(where $\hat\xi_\eop=1$) is a sink and
\begin{equation}
\label{EqGlDynKerrIn}
  \pa\hat\cR_{\rm in}^+ := \pa\cR_{\rm in}^+ \cap \Stb^*_{\hat X_b}\hat M_b
\end{equation}
(where $\hat\xi_\eop=-1$) is a source for the flow of~\eqref{EqGlDynInOutMhat} in the closure in $\ol{\Ttb^*_{\hat X_b}}\hat M_b$ of the characteristic set
\begin{equation}
\label{EqGlDynKerrChar}
  \hat\Sigma_b := (\Ttb^*_{\hat X_b}\hat M_b\setminus o) \cap \{ \hat r^2\hat G_b=0 \}.
\end{equation}
The characteristic set has two components $\hat\Sigma_b^\pm$, with $\pa\hat\cR_\bullet^+\subset\pa\hat\Sigma_b^+$ for $\bullet={\rm in},{\rm out}$. We define $\pa\hat\cR_\bullet^-=-\pa\hat\cR_\bullet^+$ (multiplication by $-1$ in the fibers).

\medskip

(ii) \textit{Trapping.} For the study of the flow away from $\pa\hat X_b$, we can dispense of all rescalings and rescaled bundles and simply study $H_{\hat G_b}^0$ as a vector field on $T^*_{\hat X_b^\circ}\hat M_b^\circ$.

\begin{definition}[Trapping]
\label{DefGlDynKerrTrap}
  In Boyer--Lindquist coordinates (see Definition~\usref{DefGlKerrBL}), introduce the phase space variables $\sigma,\xi,\eta_\theta,\eta_\phi$ by writing covectors as $-\sigma\,\dd\hat t_{\rm BL}+\xi\,\dd\hat r+\eta_\theta\,\dd\theta+\eta_\phi\,\dd\phi$. The \emph{trapped set} is then the subset
  \[
    \hat\Gamma_b^\pm := \{ (\hat r,\theta,\phi;\sigma,\xi,\eta_\theta,\eta_\phi) \in \hat\Sigma_b^\pm \colon \hat r>\hat r_b,\ \xi=H_{\hat G_b}\xi=0 \}
  \]
  of the component $\hat\Sigma_b^\pm$ of the characteristic set $\hat\Sigma_b$ (see~\eqref{EqGlDynKerrChar}). The \emph{stable (forward)/unstable (backward) trapped sets} are denoted
  \[
    \hat\Gamma_b^{{\rm s/u},\pm}\subset\hat\Sigma_b^\pm\cap T^*_{\hat X_{b,\rm BL}^\circ}\hat M_b^\circ,
  \]
  where $\hat X_{b,\rm BL}^\circ:=\hat X_b\cap\hat r^{-1}((\hat r_b,\infty))$.
\end{definition}

Thus, $\hat\Gamma_b^\pm$ is a conic smooth codimension $2$ submanifold of $\hat\Sigma_b^\pm$ by \cite[Proposition~3.3]{DyatlovWaveAsymptotics} which is a symplectic submanifold of $T^*\hat M_b^\circ$ (see \cite[Assumption~(7)]{DyatlovWaveAsymptotics} and its verification after \cite[Proposition~3.5]{DyatlovWaveAsymptotics}) and whose quotient by fiber dilations, or equivalently whose boundary $\pa\hat\Gamma_b^\pm\subset S^*\hat M_b^\circ$ at fiber infinity, is compact. Moreover, by \cite[Assumption~(5) as verified around equation~(3.18)]{DyatlovWaveAsymptotics}, we have\footnote{The reference uses the symbol $\tau$ for what we call $-\sigma$ here.}
\begin{equation}
\label{EqGlDynKerrTrapSigma}
  {\pm}\sigma > 0 \quad\text{on}\quad\hat\Gamma_b^\pm.
\end{equation}
Furthermore, \cite[Proposition~3.5]{DyatlovWaveAsymptotics} shows that $\hat\Gamma_b^{{\rm s},\pm}$ and $\hat\Gamma_b^{{\rm u},\pm}$ are smooth conic codimension $1$ submanifolds of $\hat\Sigma_b^\pm$ which only intersect at $\hat\Gamma_b^\pm=\hat\Gamma^{{\rm s},\pm}_b\cap\hat\Gamma^{{\rm u},\pm}_b$. They admit smooth defining functions
\begin{equation}
\label{EqGlDynKerrTrapDef}
  \varphi^{\rm s/u}_b \in \CI(S^*_{\hat X_{b,\rm BL}^\circ}\hat M_b^\circ)
\end{equation}
which we identify with their homogeneous degree $0$ extensions to $T^*_{\hat X_{b,\rm BL}^\circ}\hat M_b^\circ\setminus o$, and which satisfy
\begin{equation}
\label{EqGlDynKerrInvDefFn}
  {\pm}\sigma^{-1}H_{\hat G_b}\varphi^{\rm s}_b = w^{\rm s}_b\varphi^{\rm s}_b, \qquad
  \pm\sigma^{-1}H_{\hat G_b}\varphi^{\rm u}_b = -w^{\rm u}_b\varphi^{\rm u}_b\quad\text{on}\ \hat\Sigma_b^\pm, \qquad
  \sigma\cdot\{ \varphi^{\rm s},\varphi^{\rm u} \} \neq 0 \quad \text{at}\ \hat\Gamma_b;
\end{equation}
the functions $w^{\rm s/u}\in\CI(S^*_{\hat X_{b,\rm BL}^\circ}\hat M_b^\circ)$ are positive at $\pa\hat\Gamma_b$. Moreover, the $\sigma^{-1}H_{\hat G_b}$-flow in the cosphere bundle is $r$-normally hyperbolic (for every $r$) by \cite[Propositions~3.6 and 3.7]{DyatlovWaveAsymptotics}. (We recall this notion in~\S\ref{STrap}.)

\medskip

(iii) \textit{Event horizon and red-shift.} The phase space dynamics near the event horizon $\hat r=\hat r_b$ were described with an eye towards microlocal estimates by Vasy \cite{VasyMicroKerrdS}; see also \cite{HaberVasyPropagation,GannotHorizons}.

\begin{definition}[Generalized radial sets over the event horizon]
\label{DefGlDynKerrHor}
  We work in the coordinates $\hat t,\hat r,\theta,\phi_*$ on $\hat M_b^\circ$ from Lemma~\ref{LemmaGlCoord}, and write covectors as $-\sigma\,\dd\hat t+\xi\,\dd\hat r+\eta_\theta\,\dd\theta+\eta_{\phi_*}\,\dd\phi_*$. Then we define
  \begin{equation}
  \label{EqGlDynKerrHor}
  \begin{split}
    &\hat\cR_{\cH^+} := \{\hat t=0\}\cap N^*\{\hat r=\hat r_b\}\setminus o \subset T^*_{\hat X_b^\circ}\hat M_b^\circ\setminus o,  \qquad \hat\cR_{\cH^+}^\pm = \hat\cR_{\cH^+} \cap \{ \pm\xi > 0 \}.
  \end{split}
  \end{equation}
\end{definition}

Note here that $\hat\cR_{\cH^+}^\pm\subset\hat\Sigma_b^\pm$; indeed, the vector field dual to $\dd\hat r$ is $\hat g_b^{-1}(\dd\hat r,-)=\frac{\mu}{\varrho^2}\pa_{\hat r}$ in Boyer--Lindquist coordinates in $\hat r>\hat r_b$, which in the coordinates $\hat t,\hat r,\theta,\phi_*$ reads $\varrho^{-2}(\mu\pa_{\hat r}-\mu T'\pa_{\hat t}-\mu\Phi'\pa_{\phi_*})$; using~\eqref{EqGlCoord}, this can be restricted to $\hat r=\hat r_b$ (where $\mu=0$) and gives $\frac{\hat r_b^2+a^2}{\varrho^2}(\pa_{\hat t}+\frac{a}{\hat r_b^2+a^2}\pa_{\phi_*})$, which is indeed \emph{future} null.

To describe the flow near $\R_{\hat t}\times\pa\hat\cR_{\cH^+}\subset S^*\hat M_b^\circ$ quantitatively, we define $\hat t_0=\hat t_{\rm BL}-T_0$, $\phi_0=\phi-\Phi_0$ where $T_0'(\hat r)=-\frac{\hat r^2+a^2}{\mu(\hat r)}$ and $\Phi_0'(\hat r)=-\frac{a}{\mu(\hat r)}$, i.e.\ they satisfy~\eqref{EqGlCoord} with $\tilde T=\tilde\Phi=0$. Starting with the expression for the dual metric function in Boyer--Lindquist coordinates (using~\eqref{EqGlKerrBL})
\[
  \varrho^2\hat g_b^{-1} = -\frac{1}{\mu}\bigl((\hat r^2+a^2)\pa_{\hat t_{\rm BL}}+a\pa_\phi\bigr)^2 + \mu\pa_r^2 + \pa_\theta^2 + \frac{1}{\sin^2\theta}(a\sin^2\theta\,\pa_{\hat t_{\rm BL}}+\pa_\phi)^2,
\]
we change variables by plugging in $\pa_{\hat t_0}$, $\pa_{\hat r}+\frac{\hat r^2+a^2}{\mu}\pa_{\hat t_0}+\frac{a}{\mu}\pa_{\phi_0}$, $\pa_{\phi_0}$ for $\pa_{\hat t_{\rm BL}}$, $\pa_{\hat r}$, $\pa_\phi$, respectively. In coordinates defined by writing covectors as
\begin{equation}
\label{EqGlDynKerrHor0}
  {-}\sigma_0\,\dd\hat t_0+\xi_0\,\dd\hat r+\eta_\theta\,\dd\theta+\eta_{\phi_0}\,\dd\phi_0,
\end{equation}
we then find the dual metric function to be
\begin{equation}
\label{EqGlDynKerrHamHor}
\begin{split}
  \varrho^2\hat G_b &= \mu\xi_0^2 - 2\xi_0\bigl((\hat r^2+a^2)\sigma_0 - a\eta_{\phi_0}\bigr) + \tilde\sC, \\
  &\qquad \tilde\sC = \eta_\theta^2 + \frac{1}{\sin^2\theta}(\eta_{\phi_0}-a\sin^2\theta\,\sigma_0)^2.
\end{split}
\end{equation}
Therefore, we have
\begin{equation}
\label{EqGlDynKerrHam}
\begin{split}
  H_{\varrho^2\hat G_b} &= 2\bigl(\mu\xi_0 - (\hat r^2+a^2)\sigma_0 + a\eta_{\phi_0}\bigr)\pa_{\hat r} - (\mu'\xi_0 - 4\hat r\sigma_0)\xi_0\pa_{\xi_0} \\
    &\qquad + 2(\hat r^2+a^2)\xi_0\pa_{\hat t_0} + 2 a\xi_0\pa_{\phi_0} + H_{\tilde\sC}.
\end{split}
\end{equation}
Note that $H_{\tilde\sC}$ is a linear combination of $\pa_{\hat t_0}$, $\pa_\theta$, $\pa_{\phi_0}$, and $\pa_{\eta_\theta}$, with the coefficient $\pa_\theta\tilde\sC$ of $\pa_{\eta_\theta}$ vanishing quadratically at $\hat\cR_{\cH^+}$. Consider now projective coordinates near $\R_{\hat t_0}\times\pa\hat\cR_{\cH^+}^+$,
\begin{equation}
\label{EqGlDynKerrHamProj}
  \rho_\infty = \xi_0^{-1},\quad
  \hat\sigma_0 = \frac{\sigma_0}{\xi_0},\quad
  \hat\eta_\theta = \frac{\eta_\theta}{\xi_0},\quad
  \hat\eta_{\phi_0} = \frac{\eta_{\phi_0}}{\xi_0}.
\end{equation}
Then we have
\begin{equation}
\label{EqGlDynKerrHamLot}
\begin{split}
  \rho_\infty H_{\varrho^2\hat G_b} &\equiv 2\bigl(\mu'(\hat r_b)(\hat r-\hat r_b) - (\hat r^2+a^2)\hat\sigma_0 + a\hat\eta_{\phi_0}\bigr)\pa_{\hat r} \\
    &\qquad + \mu'(\hat r_b)\bigl(\rho_\infty\pa_{\rho_\infty} + \hat\sigma_0\pa_{\hat\sigma_0} + \hat\eta_\theta\pa_{\hat\eta_\theta} + \hat\eta_{\phi_0}\pa_{\hat\eta_{\phi_0}}\bigr) + \tilde H,
\end{split}
\end{equation}
where $\tilde H$ is the sum of vector fields whose coefficients vanish at least quadratically at $\R_{\hat t_0}\times\pa\hat\cR_{\cH^+}$ and linear combinations (with smooth coefficients) of the vector fields $\pa_{\hat t_0}$, $\pa_\theta$, $\pa_{\phi_0}$ which are tangent to $\R_{\hat t_0}\times\pa\hat\cR_{\cH^+}^+$. Therefore, for the local quadratic defining function
\begin{equation}
\label{EqGlDynKerrHorrho2}
  \rho_{\cH^+}^2 = \rho_\infty^2 + \hat\sigma_0^2 + \hat\eta_\theta^2 + \hat\eta_{\phi_0}^2 + \delta(\hat r-\hat r_b)^2
\end{equation}
of $\R\times\pa\hat\cR_{\cH^+}^+$, which is smooth on $\ol{T^*}\hat M_b^\circ$ near $\R\times\pa\hat\cR_{\cH^+}^+$, we have $\rho_\infty H_{\varrho^2\hat G_b}\rho_{\cH^+}^2\geq \mu'(\hat r_b)\rho_{\cH^+}^2$ in some $\hat t_0$-translation invariant neighborhood of $\R\times\pa\hat\cR_{\cH^+}^+$ in $S^*\hat M_b^\circ$ when $\delta>0$ is sufficiently small. This shows that $\pa\hat\cR_{\cH^+}^+$ is a source for the $\rho_\infty H_{\varrho^2\hat G_b}$-flow, and thus also for the $\rho_\infty H_{\hat G_b}$-flow within the characteristic set $\hat\Sigma_b$. The analysis near $\hat\cR_{\cH^+}^-$ is completely analogous, and one can indeed use the same quadratic defining function~\eqref{EqGlDynKerrHorrho2} for the same purpose.

We further record that $|\dd\hat r|^2_{\hat g_b^{-1}}=\frac{\mu}{\varrho^2}$, so $\dd\hat r$ is timelike in $\hat r<\hat r_b$, and indeed future timelike by continuity from $\hat r=\hat r_b$ where it is future causal. Therefore, $\hat r$ is strictly monotonically decreasing in $\hat r<\hat r_b$ along future null-geodesics.

\medskip

(iv) \textit{Qualitative global dynamics.} Finally, we recall how the various subsets of phase space discussed thus far are connected.

\begin{lemma}[Null-bicharacteristic flow mod $\pa_{\hat t}$ on Kerr]
\label{LemmaGlDynKerrFlow}
  Denote by $\rho_\infty\in\CI(\ol{\Ttb^*_{\hat X_b}}\hat M_b)$ a defining function of fiber infinity $\Stb^*_{\hat X_b}\hat M_b$. Let $I\ni s\mapsto\gamma(s)$ denote a maximally extended null-bicharacteristic of $\rho_\infty\hat r^2\hat G_b$, projected off $\hat t$, inside of $\pa\hat\Sigma_b^+\subset\Stb^*_{\hat X_b}\hat M_b$. Then:
  \begin{enumerate}
  \item in the forward direction, $\gamma$ either tends to the trapped set $\pa\hat\Gamma_b^+$ (see Definition~\usref{DefGlDynKerrTrap}) or the radial set $\pa\hat\cR_{\rm out}^+$ over spatial infinity (see~\eqref{EqGlDynKerrIn}), or $\gamma$ enters $\hat r<\hat r_b$ and escapes through $\hat r=\bhm$ in finite time;
  \item in the backward direction, $\gamma$ either tends to $\pa\hat\Gamma_b^+$, $\pa\hat\cR_{\rm in}^+$ (see~\eqref{EqGlDynKerrOut}), or the radial set over the event horizon $\pa\hat\cR_{\cH^+}^+$ (see Definition~\usref{DefGlDynKerrHor}) as $s\searrow-\infty$.
  \end{enumerate}
  Moreover, $\gamma$ cannot tend to $\pa\hat\Gamma_b^+$ in \emph{both} the forward and backward directions unless it is contained in $\pa\hat\Gamma_b^+$. The same statements hold true for the projections off $\hat t$ of maximally extended null-bicharacteristics of $-\rho_\infty\hat r^2\hat G_b$ in $\pa\hat\Sigma_b^-$ upon replacing the superscripts `$+$' by `$-$'.
\end{lemma}
\begin{proof}
  This is essentially proved in \cite[\S3.4]{DyatlovWaveAsymptotics}; we sketch the argument here. We work only in $\pa\hat\Sigma_b^+$. The concavity, resp.\ convexity of $\hat r$ along the flow for $0<\hat r-\hat r_b\leq\delta_0$, resp.\ $\hat r\geq\delta_0^{-1}$ proved for sufficiently small $\delta_0>0$ in \cite[Proposition~3.1]{DyatlovWaveAsymptotics} implies that forward null-bicharacteristics crossing $\hat r=\delta_0^{-1}$ in the outward direction must tend to $\hat r=\infty$ and thus to $\pa\hat\cR_{\rm out}^+$ by the sink nature of this radial set. For those which cross $\hat r=\hat r_b+\delta_0$ in the direction of decreasing $\hat r$, the value of $\hat r$ must continue decreasing while it is larger than $\hat r_b$; but since $\pa\hat\cR_{\cH^+}^+$ is a \emph{source} (in the direction normal to $\pa\hat\cR_{\cH^+}^+$), they in fact must reach and cross $\hat r=\hat r_b$ in finite time, after which they cross $\hat r=\bhm$ in finite time.

  Similarly, backward null-bicharacteristics crossing $\hat r=\delta_0^{-1}$ in the outward direction must tend to $\pa\hat\cR_{\rm in}^+$ due to the source nature of this radial set. If they cross $\hat r=\hat r_b+\delta_0$ in the direction of decreasing $\hat r$, they must tend to $\pa\hat\cR_{\cH^+}^+$ (which is a source for the $\rho_\infty H_{\hat G_b}$ flow), as they cannot cross the event horizon $\hat r=\hat r_b$.

  Forward, resp.\ backward null-bicharacteristics which in the forward direction stay in a compact subset of $\{\hat r_b<\hat r<\infty\}$ are forward, resp.\ backward trapped, i.e.\ for all large $s$ they lie in $\pa\hat\Gamma_b^{{\rm s},+}$, resp.\ $\pa\hat\Gamma_b^{{\rm u},+}$. The fact that $\hat\Gamma_b^{\rm s,+}\cap\hat\Gamma_b^{\rm u,+}=\hat\Gamma_b^+$ implies the final statement.
\end{proof}

\subsubsection{Parallel transport along the trapped set}
\label{SssGlDynTr}

This section (which will not be used in the present paper and can thus be skipped until the reader studies \cite{HintzGlueLocIII}) establishes properties of parallel transport of (co)vectors along trapped null-geodesics in subextremal Kerr spacetimes which are essential for verifying a subprincipal symbol condition for the propagation of microlocal regularity into the trapped set for tensorial wave equations. (This verification, for the linearized gauge-fixed Einstein equations, is carried out in \cite{HintzGlueLocIII}.) Roughly speaking, we show, following Marck \cite{MarckParallelNull}, that parallel transport along trapped null-geodesics has a nilpotent structure when expressed in a suitable frame; see Proposition~\ref{PropGlDynTrNabla} for the precise statement.

Write $\pi\colon T^*\hat M_b^\circ\to\hat M_b^\circ$ for the base projection, and denote by $\nabla^{\pi^*T^*\hat M_b^\circ}$ the pullback of the Levi-Civita connection on $(\hat M_b^\circ,\hat g_b)$. Write
\begin{equation}
\label{EqGlDynTrDb}
  \cD_b := \nabla^{\pi^*T^*\hat M_b^\circ}_{H_{\hat G_b}} \in \Diff^1(T^*\hat M_b^\circ;\pi^*T^*\hat M_b^\circ).
\end{equation}
This is a (principally scalar) transport operator. Therefore, given $(z_0,\zeta_0)\in T^*\hat M_b^\circ$ and $e_0\in(\pi^*T^*\hat M_b^\circ)_{(z_0,\zeta_0)}=T^*_{z_0}\hat M_b^\circ$, we can solve $\cD_b e=0$ along the $H_{\hat G_b}$-integral curve $\gamma(s)$ through $(z_0,\zeta_0)$, with initial condition $e(0)=e_0$ at $\gamma(0)=(z_0,\zeta_0)$. This defines a notion of parallel transport
\[
  (\pi^*T^*\hat M_b^\circ)_{(z_0,\zeta_0)}\to(\pi^*T^*\hat M_b^\circ)_{\gamma(s)},\qquad e_0\mapsto e(s).
\]
Identifying $e(s)$ with an element of $T^*_{\pi(\gamma(s))}\hat M_b^\circ$, we compute
\[
  \nabla_{(\pi\circ\gamma)'} e(s) = \frac12\nabla_{\pi_*H_{\hat G_b}} e(s) = \frac12\nabla_{H_{\hat G_b}}^{\pi^*T^*\hat M_b^\circ} e(s) = 0,
\]
where we use that the pushforward of $H_{\hat G_b}|_{(z,\zeta)}$ along $\pi$ is twice the vector dual to $\zeta$. Therefore, the operator $\cD_b$ captures parallel transport along all geodesics on $\hat M_b^\circ$ at once, and allows us to choose suitable frames to describe this parallel transport depending on the geodesic \emph{in phase space}. We proceed to study this on Kerr, following Marck's work \cite{MarckParallelNull} and pointing out additional properties of his construction along the way.

Following Carter \cite[\S4]{CarterHamiltonJacobiEinstein}, we define the tetrad
\begin{equation}
\label{EqGlDynTrCarter}
\begin{alignedat}{2}
  \omega^{(0)} &= \frac{\sqrt\mu}{\varrho}(\dd\hat t_{\rm BL}-a\,\sin^2\theta\,\dd\phi), &\qquad
  \omega^{(1)} &= \frac{\varrho}{\sqrt\mu}\dd\hat r, \\
  \omega^{(2)} &= \varrho\,\dd\theta, &\qquad
  \omega^{(3)} &= \varrho^{-1}\sin\theta\,(a\,\dd\hat t_{\rm BL}-(\hat r^2+a^2)\dd\phi).
\end{alignedat}
\end{equation}
The covectors $\omega^{(0)}$ and $\omega^{(1)}$ are smooth on $\hat M_b^\circ$, while $\omega^{(2)}$ and $\omega^{(3)}$ are not. The dual tetrad $\omega_{(\mu)}$ is given by
\begin{alignat*}{2}
  \omega_{(0)} &= \frac{1}{\varrho\sqrt\mu}\bigl((a^2+\hat r^2)\pa_{\hat t_{\rm BL}}+a\pa_\phi\bigr), &\qquad
  \omega_{(1)} &= \frac{\sqrt\mu}{\varrho}\pa_{\hat r}, \\
  \omega_{(2)} &= \varrho^{-1}\pa_\theta, &\qquad
  \omega_{(3)} &= -\frac{1}{\varrho\sin\theta}(a\sin^2\theta\,\pa_{\hat t_{\rm BL}}+\pa_\phi).
\end{alignat*}
In this frame, the Kerr metric is given by the Minkowski matrix $\hat g_b(\omega_{(\mu)},\omega_{(\nu)})=\ubar g_{\mu\nu}$ (i.e.\ $-1$ for $\mu=\nu=0$; $+1$ for $\mu=\nu\geq 1$; and $0$ otherwise), likewise for the dual metric. The Carter constant \cite{CarterGlobalKerr}
\[
  \sC = \varrho^2(\omega_{(2)}^2+\omega_{(3)}^2) - a^2\cos^2\theta\,\hat G_b \in P^{[2]}(T^*\hat M_b^\circ),
\]
which is quadratic in the momenta, is a constant of motion, i.e.\ $H_{\hat G_b}\sC=0$. On the characteristic set, it is equal to
\[
  \tilde\sC=\varrho^2(\omega_{(2)}^2+\omega_{(3)}^2) = \eta_\theta^2 + \frac{1}{\sin^2\theta}(\eta_\phi-a\sin^2\theta\,\sigma)^2
\]
in the coordinates used in Definition~\ref{DefGlDynKerrTrap}; this quantity already appeared earlier in~\eqref{EqGlDynKerrHamHor}. We next recall the Killing--Yano tensor \cite{PenroseNakedSing}, which is the 2-form
\[
  KY = r\omega^{(2)}\wedge\omega^{(3)} + a\cos\theta\,\omega^{(0)}\wedge\omega^{(1)};
\]
it satisfies $KY_{\mu\nu;\lambda}+KY_{\mu\lambda;\nu}=0$. It is smooth since the 2-form
\[
  \omega^{(2)}\wedge\omega^{(3)} = -a\,\dd(\cos\theta)\wedge\dd\hat t_{\rm BL} - (\hat r^2+a^2)\dd\theta\wedge\sin\theta\,\dd\phi
\]
is smooth.

Consider now a null covector $(z,\zeta)\in\hat M_b^\circ$. We then define
\begin{subequations}
\begin{equation}
\label{EqGlDynTre03}
  e^0 = \zeta,\qquad
  e^3 = KY(\zeta) \in (\pi^*T^*\hat M_b^\circ)_{(z,\zeta)}
\end{equation}
(i.e.\ $e^3_\mu=\hat g_b^{\nu\kappa}KY_{\mu\nu}\zeta_\kappa$); their coefficients in the tetrad~\eqref{EqGlDynTrCarter} are
\[
  e^0=(\omega_{(0)},\omega_{(1)},\omega_{(2)},\omega_{(3)}),\qquad
  e^3=(a\cos\theta\,\omega_{(1)},a\cos\theta\,\omega_{(0)},r\omega_{(3)},-r\omega_{(2)}).
\]
We then put
\begin{equation}
\label{EqGlDynTre12}
  e^1 := (r\omega_{(1)},r\omega_{(0)},-a\cos\theta\,\omega_{(3)},a\cos\theta\,\omega_{(2)}),\qquad
  e^2 := \frac{\varrho^2}{2}(\omega_{(0)},\omega_{(1)},-\omega_{(2)},-\omega_{(3)}).
\end{equation}
\end{subequations}

\begin{lemma}[Frame]
\label{LemmaGlDynTrFrame}
  The vectors $e^\mu$, $\mu=0,1,2,3$, defined by~\eqref{EqGlDynTre03}--\eqref{EqGlDynTre12} are smooth stationary sections of $\pi^*T^*\hat M_b^\circ$. They form a quasi-orthonormal frame of $\pi^*T^*\hat M_b^\circ$ at all null covectors $\zeta\in T^*_z\hat M_b^\circ$ for which $\sC\neq 0$; concretely,
  \[
    (\hat g_b^{-1}(e^\mu,e^\nu))_{\mu,\nu=0,\ldots,3} = \sC\begin{pmatrix} 0 & 0 & -1 & 0 \\ 0 & 1 & 0 & 0 \\ -1 & 0 & 0 & 0 \\ 0 & 0 & 0 & 1 \end{pmatrix}.
  \]
\end{lemma}
\begin{proof}
  The stationarity of $e^\mu$ follows from its definition and the stationarity of the Carter tetrad. The smoothness of $e^0$ and $e^3$ is clear. That of $e^1$ and $e^2$ can be verified by direct calculation; alternatively, one notes that $e^1=T e^0$, $e^2=S e^0$ where
  \[
    T=a\cos\theta\,\omega^{(3)}\wedge\omega^{(2)} + r\omega^{(0)}\wedge\omega^{(1)},\qquad
    S=\frac{\varrho^2}{2}\bigl(-I + 2((\omega^{(1)})^2-(\omega^{(0)})^2)\bigr)
  \]
  are smooth. Lastly, since $e^0$ is null, we have $-\omega_{(0)}^2+\omega_{(1)}^2=-\omega_{(2)}^2-\omega_{(3)}^2$. Thus,
  \[
    \hat g_b^{-1}(e^3,e^3) = -a^2\cos^2\theta(\omega_{(1)}^2-\omega_{(0)}^2) + r^2(\omega_{(2)}^2+\omega_{(3)}^2) = \varrho^2(\omega_{(2)}^2+\omega_{(3)}^2) = \sC.
  \]
  We omit the straightforward calculation of the other inner products.
\end{proof}

\begin{prop}[Form of the pullback connection]
\label{PropGlDynTrNabla}
  Define the frame $\sfe^\mu:=\frac{1}{\sqrt\sC}e^\mu$ of $(\pi^*T^*\hat M_b^\circ)_{(z,\zeta)}$ at all null covectors $\zeta\in T^*_z\hat M_b^\circ$ at which $\sC>0$; denote the set of these by $\sC_+\subset T^*_z\hat M_b^\circ\setminus o$. Then in the splitting
  \[
    (\pi^*T^*\hat M_b^\circ)|_{\sC_+} = \la\sfe^0\ra \oplus \la\sfe^1\ra \oplus \la\sfe^2\ra \oplus \la\sfe^3\ra,
  \]
  the operator $\cD_b$ in~\eqref{EqGlDynTrDb} takes the form
  \[
    \cD_b = H_{\hat G_b} + 2\sigma \begin{pmatrix} 0 & 1 & 0 & 0 \\ 0 & 0 & 1 & 0 \\ 0 & 0 & 0 & 0 \\ 0 & 0 & 0 & 0 \end{pmatrix}.
  \]
  This in particular applies at $\hat\Gamma_b$, i.e.\ $\hat\Gamma_b\subset\sC_+$.
\end{prop}
\begin{proof}
  For the integral curve $\gamma(s)$ of $H_{\hat G_b}$ with $\gamma(0)=(z,\zeta)\in\sC_+$ the covector $e^0=\dot\gamma(s)^\flat$ is parallel along the geodesic $\pi\circ\gamma$. Since $H_{\hat G_b}\sC=0$, this implies that $\cD_b\sfe^0=\sC^{-1/2}\cD_b e^0=0$. By the properties of $KY$, also $e^3$ and therefore $\sfe^3$ are parallel, i.e.\ $\cD_b\sfe^3=0$.

  Next, write $\cD_b\sfe^1=:\alpha_\mu\sfe^\mu$. Then Lemma~\ref{LemmaGlDynTrFrame} gives $\alpha_1=\hat g_b^{-1}(\cD_b\sfe^1,\sfe^1)=H_{\hat g_b}(|\sfe^1|_{\hat g_b^{-1}}^2)=0$, furthermore $\alpha_2=-\hat g_b^{-1}(\cD_b\sfe^1,\sfe^0)=\hat g_b^{-1}(\sfe^1,\cD_b\sfe^0)=0$, and also $\alpha_3=\hat g_b^{-1}(\cD_b\sfe^1,\sfe^3)=-\hat g_b^{-1}(\sfe^1,\cD_b\sfe^3)=0$. Therefore,
  \[
    \cD_b\sfe^1 = \alpha_0\sfe^0,\qquad \alpha_0=-\hat g_b^{-1}(\cD_b\sfe^1,\sfe^2).
  \]
  Similarly, one finds that for $\cD_b\sfe^2=:\beta_\mu\sfe^\mu$ one has $\beta_0=\beta_2=\beta_3=0$, whereas $\beta_1=\hat g_b^{-1}(\cD_b\sfe^2,\sfe^1)=-\hat g_b^{-1}(\sfe^2,\cD_b\sfe^1)=\alpha_0$. A lengthy direct computation gives $\alpha_0=2\sigma$.

  Finally, we show that $\sC>0$ at the trapped set: if $\hat G_b=0$ and $\sC=0$, then $\omega_{(2)}=\omega_{(3)}=0$; but also $\xi=0$ at $\hat\Gamma_b$, so $\omega_{(1)}=0$; and since $\hat\Gamma_b\subset\hat\Sigma_b$, this forces $\omega_{(0)}=0$ as well. But this means we are at the zero section, which however is disjoint from $\hat\Sigma_b$. Since $\sC=\tilde\sC\geq 0$ everywhere on $\hat\Sigma_b$, this means that $\sC>0$ on $\hat\Gamma_b$.
\end{proof}

Proposition~\ref{PropGlDynTrNabla} exhibits the nilpotent structure of parallel transport along null-geodesics which are trapped or, more generally, along null-geodesics for which $\sC\neq 0$: the coefficients of a parallel section along such a null-geodesic, in the frame $\sfe^\mu$, grow at most quadratically (and in particular not exponentially) with the affine parameter.

\begin{rmk}[Schwarzschild case]
\label{RmkGlDynTrSchw}
  In the special case $\bha=0$ of the Schwarzschild metric, and writing covectors as $-\sigma\,\dd t+\xi\,\dd r+\eta$, $\eta\in T^*\Sph^2$, we have, at the trapped set $r=3\bhm$, $\xi=0$:
  \[
    e^0 = -\sigma\,\dd t + \eta,\qquad
    e^1 = -9\bhm\sigma\,\dd r,\qquad
    e^2 = \frac{9\bhm^2}{2}(\sigma\,\dd t+\eta),\qquad
    e^3 = 3\bhm(\slstar\eta),
  \]
  with $\slstar$ denoting the Hodge star operator on $\Sph^2$ with the standard metric. The frame used in \cite[\S{9.1}]{HintzVasyKdSStability} on Schwarzschild--de~Sitter space is related to this by a constant linear transformation.
\end{rmk}

\subsubsection{Global dynamics on standard domains}
\label{SssGlDynStd}

We now combine the results from~\S\S\ref{SssGlDynMc}--\ref{SssGlDynKerr} to study the null-bicharacteristic flow in the se-characteristic set $\wt\Sigma$. Unless specified otherwise, we only use the se-regularity $\cC_\seop^{2,1,1}=\eps\cC_\seop^2$ of the lower order term of $\wt g$ in Definition~\ref{DefGl}\eqref{ItGlwtg}. (This suffices to ensure that integral curves of $H_{\wt G}$ are unique.)

\begin{definition}[Standard domains in $\wt M$]
\label{DefGlDynStd}
  A standard domain $\wt\Omega\subset\wt M\setminus\wt K^\circ$ is a submanifold with corners of $\wt M$ of the form $\wt\Omega=\wt\upbeta^{-1}([0,1)_\eps\times\Omega)\setminus\wt K^\circ$ where $\Omega\subset M$ is a standard domain in the sense of Definition~\usref{DefGlDynStdM}. We write $\Omega_\eps=\wt\Omega\cap M_\eps$. The boundary hypersurfaces of $\wt\Omega$ are denoted as follows.
  \begin{enumerate}
  \item\label{ItGlDynStdC} If $\Omega=\Omega_{t_0,t_1,r_0}$, then we set
    \begin{align*}
      \wt X_{t_0,t_1,r_0} &:= \wt\upbeta^{-1}([0,1)_\eps\times X_{t_0,t_1,r_0}) \setminus \{ |\hat x|<\bhm \}, \\
      \wt Y^-_{t_0,t_1,r_0} &:= \wt\upbeta^{-1}([0,1)_\eps\times Y^-_{t_0,t_1,r_0}) \setminus \{ |\hat x|<\bhm \}, \\
      \wt Y^\wedge_{t_0,t_1,r_0} &:= [0,1)_\eps\times Y^\wedge_{t_0,t_1,r_0}, \\
      \wt Y^|_{t_0,t_1,r_0} &:= \{ (\eps,t,\hat x) \colon \eps\in[0,1),\ t\in[t_0,t_1],\ |\hat x|=\bhm \}.
    \end{align*}
  \item\label{ItGlDynStdNotC} If $\Omega\cap\cC=\emptyset$, and the boundary hypersurfaces of $\Omega$ are $X$ (initial) and $Y_1,\ldots,Y_J$ (final), then we set $\wt X:=[0,1)_\eps\times X$ and $\wt Y_j:=[0,1)_\eps\times Y_j$.
  \end{enumerate}
\end{definition}

See Figure~\ref{FigGlDynStd}. In the second case of Definition~\ref{DefGlDynStd}, we have $\wt\Omega\cap\eps^{-1}([0,\eps_0])=[0,\eps_0]\times\Omega$ for all sufficiently small $\eps_0>0$ (chosen so that $\wt K$ is disjoint from $[0,\eps_0]\times\Omega$). If one shrinks $\eps_0>0$ further, then one can ensure that for each initial, resp.\ final boundary hypersurface $X\subset\Omega$ also $\{\eps\}\times X=\wt X\cap M_\eps$ will be an initial, resp.\ final spacelike hypersurface for $g_\eps=\wt g|_{M_\eps}$.

In the first case, one can argue similarly for $\wt Y_{t_0,t_1,r_0}^\wedge$ since $Y_{t_0,t_1,r_0}^\wedge\cap\cC=\emptyset$. Since $\dd\hat r$ is (future) timelike at $\hat r=\bhm$ for the Kerr metric, the se-covector $\dd\hat r=\eps^{-1}\dd r$ is future timelike for $\wt g$ in a neighborhood of $\{\hat r=\bhm\}\subset\hat M$, so also $\wt Y^|_{t_0,t_1,r_0}\cap M_\eps$ is a final boundary hypersurface of $\Omega_\eps$ for all sufficiently small $\eps>0$. Analogous reasoning using the se-covector $\hat\rho^{-1}\dd t=\hat\rho_\circ\,\dd\hat t$, which is past timelike at $\hat M^\circ$ in view of the final part of Lemma~\ref{LemmaGlCoord} and also at $M_\circ$ by Definition~\ref{DefGl}, shows that also the intersections of $M_\eps$ with $\wt X_{t_0,t_1,r_0}$ and $\wt Y^-_{t_0,t_1,r_0}$ are spacelike for all sufficiently small $\eps>0$.

\begin{figure}[!ht]
\centering
\includegraphics{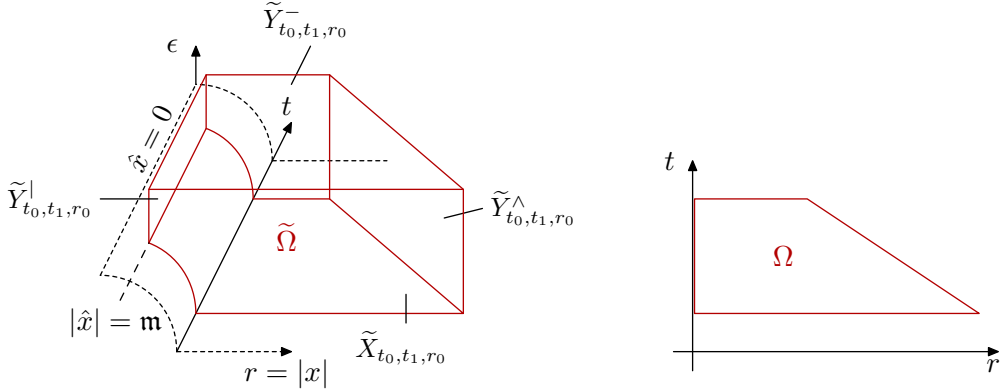}
\caption{A standard domain in $\wt M$, its boundary hypersurfaces, and (on the right) the corresponding standard domain in $M$.}
\label{FigGlDynStd}
\end{figure}

Having already defined the radial sets $\cR_{\rm in}$, $\cR_{\rm out}$ in Definition~\ref{DefGlDynInOut}, we define
\begin{equation}
\label{EqGlDynStdRadHor}
  \cR_{\cH^+} := I_\cC\times\hat\cR_{\cH^+}
\end{equation}
in terms of~\eqref{EqGlDynKerrHor} as the set of all covectors $\xi_\seop\,\dd\hat r$ over $\hat r=\bhm$ where $\xi_\seop\neq 0$; its two components are denoted $\cR_{\cH^+}^\pm=\cR_{\cH^+}\cap\wt\Sigma^\pm$ as usual. Carefully note that $\cR_{\cH^+}\subset\Tse^*_{\hat M}\wt M\setminus o$ lives only at $\eps=0$: the Kerr event horizon plays no special role (and indeed has no geometric meaning anymore) for any $\eps>0$. Similarly, we define the subsets
\begin{equation}
\label{EqGlDynStdTrap}
  \Gamma := I_\cC\times\hat\Gamma_b,\qquad
  \Gamma^{\rm s/u} := I_\cC\times\hat\Gamma_b^{\rm s/u},
\end{equation}
of $\Tse^*_{\hat M}\wt M\setminus o$, with future/past components denoted $\Gamma^\pm$, $\Gamma^{\rm s/u,\pm}$.

\begin{lemma}[Null-bicharacteristic flow on standard domains: $\wt M$]
\label{LemmaGlDynStdFlow}
  Let $\wt\Omega$ be a standard domain. Fix a defining function $\rho_\infty\in\CI(\ol{\Tse^*}\wt M)$ of fiber infinity, and let $\sfH:=\rho_\infty\hat\rho^2 H_{\wt G}\in\Vse(\Tse^*_{\wt M\setminus\wt K^\circ}\wt M)$. Then there exists $\eps_0>0$ so that the following statements hold.
   \begin{enumerate}
  \item\label{ItGlDynStdFlowSpace}{\rm (Causal character.)} The boundary hypersurfaces of $\Omega_\eps$ are spacelike for all $\eps\in(0,\eps_0]$; furthermore, if $\Omega\cap\cC\neq\emptyset$, then $\hat\rho^{-1}\dd t$ is past timelike on $\Omega_\eps$.
  \item\label{ItGlDynStdFlow}{\rm (Flow.)} Let $I\ni s\mapsto \gamma(s)\in\pa\wt\Sigma^\pm\cap\Sse^*_{\wt\Omega}\wt M$ be a maximal integral curve of $\pm\sfH$ over $\wt\Omega\cap\eps^{-1}([0,\eps_0])$. Then exactly one of the following possibilities must occur:
    \begin{enumerate}
    \item\label{ItGlDynStdFlowXY} $\gamma$ lies over $M_\eps$ for some $\eps>0$, starts at an initial hypersurface and ends at a final hypersurface of $\Omega_\eps$;
    \item $\gamma$ is contained in one of the invariant sets $\pa\cR_{\rm out}$, $\pa\cR_{\rm in}$, $\pa\cR_{\cH^+}$, $\pa\Gamma$;
    \item $\gamma$ lies over $\pa M_\circ$ and tends to $\pa\cR_{\rm in}$ in the backward and $\pa\cR_{\rm out}$ in the forward direction;
    \item $\gamma$ lies over $(M_\circ)^\circ$ and in the forward direction tends to $\pa\cR_{\rm in}$ or crosses a final boundary hypersurface in finite time, and in the backward direction tends to $\pa\cR_{\rm out}$ or crosses the initial boundary hypersurface in finite time, but it does not tend to $\pa\cR_{\rm in}\cup\pa\cR_{\rm out}$ in \emph{both} the forward and backward directions;
    \item $\gamma$ lies over $\hat M^\circ$ and in the forward direction tends to $\pa\Gamma$, $\pa\cR_{\rm out}$, or ends at $\wt Y^|_{t_0,t_1,r_0}$ (i.e.\ $\hat r=\bhm$) in finite time, and in the backward direction tends to $\pa\Gamma$, $\pa\cR_{\rm in}$, or $\pa\cR_{\cH^+}$, but it does not tend to $\pa\Gamma$ in \emph{both} directions.
    \end{enumerate}
  \end{enumerate}
\end{lemma}

\begin{proof}[Proof of Lemma~\usref{LemmaGlDynStdFlow}]
  Part~\eqref{ItGlDynStdFlowSpace} has already been established. Part~\eqref{ItGlDynStdFlow} is a combination of Lemmas~\ref{LemmaGlDynStdFlowM} and \ref{LemmaGlDynKerrFlow} when $\eps=0$. For $\eps>0$ on the other hand, note that $\sfH$ is tangent to the $\eps$-level sets inside $\Sse^*\wt M$; and at points $(z,\zeta)\in\wt\Sigma^+\cap\Sse^*_{\wt\Omega}\wt M$ where $\eps>0$, we have $\sfH t>0$ by the past timelike nature of $\dd t$. Thus, \eqref{ItGlDynStdFlowXY} is the only possibility in this case.
\end{proof}

\begin{figure}[!ht]
\centering
\includegraphics{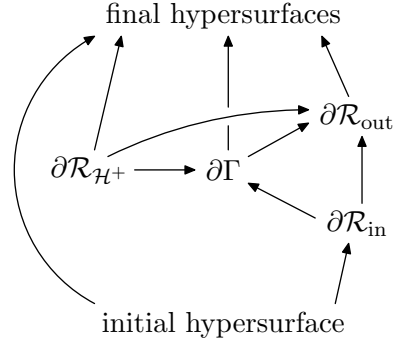}
\caption{Illustration of the future null-bicharacteristic flow, as described in Lemma~\ref{LemmaGlDynStdFlow}.}
\label{FigGlDynFlow}
\end{figure}

\begin{lemma}[Monotone functions]
\label{LemmaGlDynStdMono}
  Let $\wt\Omega$ be a standard domain associated with $\Omega$ where $\Omega\cap\cC\neq\emptyset$. Write se-covectors in $\wt\Omega$ as $-\sigma_\seop\,\frac{\dd t}{\hat\rho}+\xi_\seop\,\frac{\dd r}{\hat\rho}+\eta_\seop$ where $\eta_\seop\in T^*\Sph^2$. Then for all sufficiently large $\hat r_0>\hat r_b$ and all sufficiently small $\eps_0>0$,\footnote{In this region, we may take $\hat\rho=r$.} the following statements hold for all $\eps\in[0,\eps_0]$.
  \begin{enumerate}
  \item\label{ItGlDynStdMonoSigma} $\pm\sigma_\seop>0$ on $\wt\Sigma^\pm$ over $\Omega_\eps\cap\{\hat r\geq\hat r_0\}$.
  \item\label{ItGlDynStdMonoFn} Let $\chi_1,\chi_2\in\CI(\R)$ with
    \begin{alignat*}{3}
      \chi_1|_{(-\infty,\frac12]}&=0,\quad&
      \chi_1'&\geq 0,\quad&
      \chi_1|_{[\frac34,\infty)}&=1, \\
      \chi_2|_{(-\infty,\hat r_0]}&=0,\quad&
      \chi_2'&\geq 0,\quad&
      \chi_2|_{[2\hat r_0,\infty)}&=1.
    \end{alignat*}
    Write $\hat\xi_\seop=\frac{\xi_\seop}{\sigma_\seop}$. Then
    \begin{equation}
    \label{EqGlDynStdMono}
      {\pm}\sigma_\seop^{-1}\hat\rho^2 H_{\wt G}(\chi_1(\hat\xi_\seop)\chi_2(\hat r))\geq 0\quad\text{on}\quad\pa\wt\Sigma^\pm\cap\Sse^*_{\Omega_\eps}\wt M.
    \end{equation}
    If moreover $\sqrt{\chi_j},\sqrt{\chi_j'}\in\CI(\R)$ for $j=1,2$, then on $\pa\wt\Sigma^\pm\cap\Sse^*_{\Omega_\eps}\wt M$, \eqref{EqGlDynStdMono} can be written as the sum of squares of smooth functions on $\pa\wt\Sigma^\pm$.
  \end{enumerate}
\end{lemma}

See Figure~\ref{FigGlDynStdMono} for an illustration of the monotonicity property~\eqref{EqGlDynStdMono}.

\begin{figure}[!ht]
\centering
\includegraphics{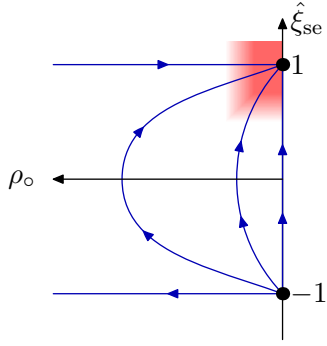}
\caption{Projection to the coordinates $\hat r$ and $\hat\xi_\seop$ (which equals $\xi/\sigma$ if one writes covectors as $-\sigma\,\dd\hat t+\xi\,\dd\hat r+\eta$, $\eta\in T^*\Sph^2$) of the null-bicharacteristic flow, here drawn in the compactified coordinate $\rho_\circ=\hat r^{-1}$. The shading of red indicates the size of the function $\chi_1(\hat\xi_\seop)\chi_2(\hat r)$ in~\eqref{EqGlDynStdMono}.}
\label{FigGlDynStdMono}
\end{figure}

\begin{proof}
  The positivity of $\sigma_\seop$ on $\wt\Sigma$ is ensured over $\wt\Omega\cap M_\circ$ by Definition~\ref{DefGlDynStdM}\eqref{ItGlDynStdOrder}, and thus follows in an open neighborhood thereof by continuity. For the proof of part~\eqref{ItGlDynStdMonoSigma}, it then suffices to note that $\lim_{\hat r_0\to\infty}\lim_{\eps\to 0}(\wt\Omega\cap\{\hat r\geq\hat r_0,\ \eps\leq\eps_0\})=\wt\Omega\cap M_\circ$.

  For part~\eqref{ItGlDynStdMonoFn}, we compute, with $\sfH:=\sigma_\seop^{-1}\hat\rho^2 H_{\wt G}$,
  \[
    \sfH(\chi_1(\hat\xi_\seop)\chi_2(\hat r)) = \chi_1'\chi_2\sfH\hat\xi_\seop + \chi_1\chi_2'\sfH\hat r.
  \]
  But on $\supp\chi_1'$ we have $\hat\xi_\seop\in[\frac12,\frac34]$ and thus $\sfH\hat\xi_\seop$ has a positive lower bound on the characteristic set over $\wt\Omega\cap M_\circ$ in view of~\eqref{EqGlDynStdMonotone}. This continues to hold on a sufficiently small open neighborhood. Thus, the first term is nonnegative, and indeed has a smooth square root on the characteristic set when $\sqrt{\chi_1'}$ is smooth. For the second term, we note that on $\supp(\chi_1\chi_2')$ we have $\hat\xi_\seop\geq\frac12$ and $\hat r\geq\hat r_0$; in view of (the first term of) \eqref{EqGlDynInOutMhat} (with $\rho_\circ=\hat r^{-1}$) this implies that $\sfH\hat r>0$ on $\hat M$ when $\hat r_0$ is sufficiently large; this persists over $\Omega_\eps$ when $\eps$ is sufficiently small.
\end{proof}

For later use, we record the saddle point structure of the null-bicharacteristic flow near the radial sets $\pa\cR_{\rm in}$ and $\pa\cR_{\rm out}$ from Definition~\ref{DefGlDynInOut} in a quantitative way: the combination of~\eqref{EqGlDynInOutHam} and \eqref{EqGlDynInOutMhat} (see also the discussion following~\eqref{EqGlDynInOutMhat}) is
\begin{equation}
\label{EqGlDynStdInOut}
\begin{split}
  \rho_\infty H_{r^2\wt G} &\equiv 2\hat\xi_\seop(r\pa_r-\rho_\circ\pa_{\rho_\circ}-\rho_\infty\pa_{\rho_\infty}-\hat\eta_\seop\pa_{\hat\eta_\seop}) + 2(1-\hat\xi_\seop^2)\pa_{\hat\xi_\seop} \\
    &\qquad + H_{|\hat\eta_\seop|^2} \bmod (\cI+\eps\cC_\seop^1)\Vb(\Sse^*\wt M),
\end{split}
\end{equation}
where $\cI=\hat\rho\CI(\wt M)+\rho_\circ\CI(\wt M)\subset\CI(\wt M)$ is the module of smooth functions vanishing at $\hat M\cap M_\circ$.

Finally, we record the following technical result on the $H_{\eps^2\wt G}$-flow in a full neighborhood of the trapped set inside of $\wt\Sigma$, which we only state near the future characteristic set. The proof is given in Appendix~\ref{STrap}.

\begin{prop}[Extensions of stable and unstable defining functions]
\label{PropTrap}
  Suppose that the second (remainder) term in Definition~\usref{DefGl}\eqref{ItGlwtg} lies in $\eps\cC^{d_0}_\seop=\cC_\seop^{d_0,1,1}$ where $2\leq d_0\in\N\cup\{\infty\}$. Let $I\subset I_\cC$ be a bounded open interval with $I\subset\bar I\subset I_\cC$. Then there exist a neighborhood $\cU\subset\Sse^*_{\wt M\setminus M_\circ}\wt M$ of $\pa\Gamma^+\cap t^{-1}(\bar I)$ and functions $\wt\varphi^{\rm u},\wt\varphi^{\rm s}$, defined on $\cU$, with the following properties:
  \begin{enumerate}
  \item $\wt\varphi^\bullet|_{\hat M_{t_0}}=\varphi_b^\bullet$ for $\bullet={\rm u,s}$ and for all $t_0\in I$. More precisely, in the coordinates $\eps,t,\hat x$ and relative to coordinates on the fibers $\Sph^3$ of $\Sse^*\wt M$ defined in terms of the fiber-linear coordinates $-\sigma_\seop\,\dd\hat t+\xi_\seop\,\dd\hat x$,
  \[
    \wt\varphi^\bullet-\varphi^\bullet \in \eps\cC^{d_0}_\seop\bigl(\cU\cap([0,\eps_0]\times I_t\times\R^3_{\hat x}\times\Sph^3)\bigr),
  \]
  i.e.\ $\eps^{-1}(\wt\varphi^\bullet-\varphi^\bullet)$ remains uniformly bounded upon application of up to $d_0$ many of the vector fields $\eps\pa_t$, $\pa_{\hat x}$, and vector fields on $\Sph^3$;
  \item on $\cU\cap\wt\Sigma^+$, we have
    \begin{equation}
    \label{EqTrapHamPhisu}
      \rho_\infty\eps^2 H_{\wt G}\wt\varphi^{\rm s}=\wt w^{\rm s}\wt\varphi^{\rm s}, \qquad
      \rho_\infty\eps^2 H_{\wt G}\wt\varphi^{\rm u}=-\wt w^{\rm u}\wt\varphi^{\rm u},
    \end{equation}
    where $\wt w^\bullet$ are functions with $\wt w^\bullet|_{\hat M_{t_0}}=w^\bullet_b$ for $\bullet={\rm u,s}$ and for all $t_0\in I$. More precisely, $\wt w^\bullet-w^\bullet\in\eps\cC^{d_0-1}_\seop$;
  \item $\wt w^{\rm s/u}>0$ on $\cU$.
  \end{enumerate}
\end{prop}

\section{Estimates for linear tensorial wave equations on glued spacetimes}
\label{SEst}

We fix a glued spacetime $(\wt M,\wt g)$ associated with $(M,g)$, the inextendible timelike geodesic $\cC\subset M$, subextremal Kerr black hole parameters $b=(\bhm,\bha)$, and Fermi normal coordinates $t\in I_\cC\subset\R$, $x\in\R^3$ around $\cC$ as in Definition~\ref{DefGl}; we furthermore write $\hat x=\frac{x}{\eps}$. Local defining functions of $\hat M$ and $M_\circ$ are thus $\hat\rho=(\eps^2+|x|^2)^{1/2}$ and $\rho_\circ=\la\hat x\ra^{-1}$. We shall prove estimates for $\eps$-dependent wave operators $L_\eps=\wt L|_{M_\eps}$ on $M_\eps$ associated with $g_\eps=\wt g|_{M_\eps}$.

\textbf{Operator class.} We specify the class of wave operators by the following data.
\begin{enumerate}[label=(4.$\wt L$.\alph*),ref=4.$\wt L$.\alph*]
\item\label{ItEstLBundle} $\wt E\to\wt M\setminus\wt K^\circ$: a smooth vector bundle. We write $\hat E:=\wt E|_{\hat M}\to\hat M$ and assume that
  \begin{equation}
  \label{EqEstLBundle}
    \hat E=\hat\pi^*\hat\cE
  \end{equation}
  for a vector bundle $\hat\cE\to\hat X_b$ where, in the coordinates $t\in I_\cC$, $\hat x\in\ol{\R^3}$ on $\hat M$, we write $\hat\pi(t,\hat x)=\hat x$. (In other words, we are given a compatible family of isomorphisms $\hat E|_{(t,\hat x)}\cong\hat E|_{(t',\hat x)}$ for all $t,t'\in I_\cC$, with smooth dependence on $t,t',\hat x$.)
\item\label{ItEstLOp} $\wt L\in\hat\rho^{-2}(\CI+\cC_\seop^{\infty,1,1})\Diffse^2(\wt M\setminus\wt K^\circ;\wt E)$: a principally scalar operator, with principal symbol equal to the dual metric function $\wt G\in\hat\rho^{-2}(\CI+\cC_\seop^{\infty,1,1})P^{[2]}(\Tse^*_{\wt M\setminus\wt K^\circ}\wt M)$ of $\wt g$. We furthermore assume that the normal operators $N_{\hat M_{t_0}}(\eps^2\wt L)$ are of class $\Diff_{\tbop,\rm I}^{2,-2,0}(\hat M_b;\hat E|_{\hat M_{t_0}})=\rho_\circ^2\Diff_{\tbop,\rm I}^2(\hat M_b;\hat E|_{\hat M_{t_0}})$ and independent of $t_0$; that is,
  \begin{equation}
  \label{EqEstL}
    L := N_{\hat M_{t_0}}(\eps^2\wt L) \in \rho_\circ^2\Diff_{\tbop,\rm I}^2(\hat M_b;\hat\pi^*\hat\cE),
  \end{equation}
  where $\hat\pi\colon\hat M_b\to\hat X_b$ now denotes the projection $(\hat t,\hat x)\mapsto\hat x$.
\end{enumerate}

\textbf{Model operators associated with $\wt L$.} \emph{For notational simplicity only}, we shall, for the remainder of this section, make the stronger assumption~\eqref{EqGlDynwtG} on the metric and
\[
  \wt L \in \hat\rho^{-2}\Diffse^2(\wt M\setminus\wt K^\circ;\wt E) = \Diffse^{2,0,2}(\wt M\setminus\wt K^\circ;\wt E)
\]
on the operator (unless noted otherwise). We leave the (purely notational) modifications to our arguments below in the general case to the reader. Given $\wt L$, we further recall the notation $\wt\Sigma=\bigsqcup_\pm\wt\Sigma^\pm$ for its characteristic set; and we set
\begin{equation}
\label{EqEstLepsLcirc}
\begin{split}
  L_\eps &= \wt L|_{M_\eps\setminus\wt K^\circ} \in \Diff^2(M_\eps\setminus\wt K^\circ; \wt E),\qquad \eps>0, \\
  L_\circ &:= N_{M_\circ}(\wt L) \in \hat\rho^{-2}\Diffe^2(M_\circ;E_\circ),\qquad E_\circ:=\wt E|_{M_\circ}.
\end{split}
\end{equation}
Thus, $L_\circ$ is principally scalar with principal symbol given by the dual metric function $G$ of $g$. Also the operator $L$ in~\eqref{EqEstL} is principally scalar, and indeed a wave operator on the Kerr spacetime $(\hat M_b^\circ,\hat g_b)$ which commutes with time translations; we denote its spectral family (see~\eqref{EqF3SpecFam}) by
\[
  \hat L(\sigma) \in \Diff^2(\hat X_b^\circ;\hat\cE).
\]
(In the notation of~\S\ref{SssFseNormal}, we have $\hat L(\sigma)=\hat N(\eps^2\wt L,\sigma)$.) From the general discussion in~\S\S\ref{SsF3}--\ref{SsFse}, we have, for $\hat\sigma,z\in\C$ with $|\hat\sigma|=|z|=1$,
\begin{equation}
\label{EqEstLRegimes}
  (\hat L(\hat\sigma|\sigma|))_{|\sigma|\in[0,1]} \in |\sigma|^2\Diffscbt^{2,0,0,2}(\hat X_b;\hat\cE)=\Diffscbt^{2,0,-2,0}, \qquad
  \hat L(h^{-1}z) \in h^{-2}\Diff_{\scop,\semi}^2(\hat X_b;\hat\cE).
\end{equation}
The low energy spectral family thus has transition face normal operators (see also~\eqref{EqFseNhatMWeighted} and~\eqref{EqFseNhatMWeightedtf} with $m=2$, $\ell_\circ=0$, $\hat\ell=2$) which we denote
\begin{equation}
\label{EqLtf}
  L_\tface(\hat\sigma) = N_\tface\bigl(|\sigma|^{-2}(\hat L(\hat\sigma|\sigma|))_{|\sigma|\in[0,1]}\bigr) \in \Diff_{\scop,\bop}^{2,0,2}(\tface;\pi_\tface^*\hat\cE|_{\pa\hat X_b}),\qquad \tface=[0,\infty]_{\rho'}\times\Sph^2,
\end{equation}
where $\rho'=(\hat r|\sigma|)^{-1}$ and $\pi_\tface\colon\tface\to\Sph^2=\pa\hat X_b$ is the projection $(\rho',\omega)\mapsto\omega$. This can be identified with the ($t_0$-independent) reduced normal operator of $L_\circ$ in the manner described in Lemma~\ref{LemmaFseIdent}.

The basic example of this setup is the operator $\wt L=(L_\eps)_{\eps\in(0,1)}$ defined by $L_\eps=\Box_{g_\eps}$; this is discussed in detail in~\S\ref{SSc}.

\bigskip

\textbf{Plan for the remainder of this section.} We will begin by proving estimates controlling the regularity of solutions of $\wt L u=f$ on standard domains $\wt\Omega$ (as in Definition~\ref{DefGlDynStd}\eqref{ItGlDynStdC}) in weighted se-Sobolev spaces, with constants that are uniform in $\eps$. These estimates take the form
\begin{equation}
\label{EqEst}
  \|u\|_{H_{\seop,\eps}^{\sfs,\alpha_\circ,\hat\alpha}(\Omega_\eps)^{\bullet,-}} \leq C\Bigl( \|\wt L u\|_{H_{\seop,\eps}^{\sfs,\alpha_\circ,\hat\alpha-2}(\Omega_\eps)^{\bullet,-}} + \|u\|_{H_{\seop,\eps}^{\sfs_0,\alpha_\circ,\hat\alpha}(\Omega_\eps)^{\bullet,-}} \Bigr)
\end{equation}
for suitable orders $\sfs_0<\sfs\in\CI(\Sse^*_{\wt\Omega}\wt M)$ and weights $\alpha_\circ,\hat\alpha\in\R$; see~\S\ref{SsEstStd} for the proof of this estimate, and \S\ref{SsEstFn} for a precise definition of the norms involved. (Recall from Remark~\ref{RmkFseAction} that we write $\wt L u$ for $L_\eps u$, with the value of $\eps$ indicated in the subscript of the se-Sobolev norms.) To prove this estimate, we combine the following (micro)local estimates:
\begin{itemize}
\item microlocal propagation estimates near the incoming and outgoing radial sets $\cR_{\rm in}$, $\cR_{\rm out}$ over $\pa\hat M$ from Definition~\ref{DefGlDynInOut} (see~\S\ref{SsEstRad});
\item a microlocal propagation estimate near the generalized radial set $\cR_{\cH^+}$ over the event horizon from Definition~\ref{DefGlDynKerrHor} and~\eqref{EqGlDynStdRadHor} (see~\S\ref{SsEstHor});
\item a microlocal propagation estimate near the normally hyperbolic trapping from Definition~\ref{DefGlDynKerrTrap} and~\eqref{EqGlDynStdTrap} (see~\S\ref{SsEstTrap});
\item an energy estimate near the initial and final hypersurfaces of $\wt\Omega$ (see~\S\ref{SsEstEn}).
\end{itemize}

The radial point estimates are proved in the standard manner using positive commutator arguments (utilizing only the principal symbol map, here on the algebra of se-pseudodifferential operators with se-regular coefficients) as in \cite{MelroseEuclideanSpectralTheory,VasyMicroKerrdS}. Similarly, we shall be brief on the details in the energy estimates. The trapping estimate roughly follows the two-step commutator proof of \cite{DyatlovSpectralGaps}, though the implementation is rather delicate, and hence we provide more details here. We also use elliptic and real principal type \cite{DuistermaatHormanderFIO2} estimates, see Lemmas~\ref{LemmaFVarsseEll} and \ref{LemmaFVarsseProp}.

\begin{rmk}[Other settings]
\label{RmkEstOther}
  In the context of Remark~\ref{RmkGlOther}, one can skip \S\S\ref{SsEstHor}--\ref{SsEstTrap} if $L$ is a wave operator on an asymptotically flat nontrapping spacetime without horizons (though possibly non-scalar and with lower order terms such as potentials, with e.g.\ a singular potential term $\eps^{-2}V(\frac{x}{\eps})$ of $\wt L$ contributing the potential term $V(\hat x)$ to $L$).
\end{rmk}

The estimate~\eqref{EqEst} does not yet give uniform control of $u$ by $\wt L u$ since the second term on the right is not \emph{small}, even when $\eps$ is small. The strategy for improving the error is to localize it to neighborhoods of the two boundary hypersurfaces $M_\circ$ and $\hat M$ and invert the respective normal operators.
\begin{itemize}
\item To control the localization $\hat\chi u$ of $u$ to a neighborhood of $\hat M$, we replace $\wt L$ by ($\eps^{-2}$ times) the Kerr model $L$. Using the Fourier transform in $\hat t$ and estimates for the resolvent $\hat L(\sigma)^{-1}$, one expects to be able to estimate $\hat\chi u$ in 3b-Sobolev spaces (cf.\ Lemma~\ref{LemmaFseHse3b}) by $L(\hat\chi u)$. This would allow one to replace the orders on the error term in~\eqref{EqEst} by $\alpha_\circ,\hat\alpha-1$. The resolvent analysis, however, depends crucially on subtle information on $\hat L(\sigma)$ (mode stability and suitable behavior near $\sigma=0$) which is strongly problem-dependent, unlike the previous arguments. In~\S\ref{SsEstFT}, we thus only prove those (Fredholm and high energy) estimates on $\hat L(\sigma)$ which do not require this information.
\item The localization $\chi_\circ u$ of $u$ to a neighborhood of $M_\circ$ can be controlled by $L_\circ(\chi_\circ u)$ on edge Sobolev spaces; this uses the main result from \cite{HintzConicWave} which we recall in~\S\ref{SsEstMc}. Since $\wt L$ differs from $L_\circ$ by an se-operator whose coefficients vanish at $M_\circ$, this can be used to improve the weights of the error term in~\eqref{EqEst} to $\alpha_\circ-1,\hat\alpha$.
\end{itemize}

Regarding the first point, the resolvent analysis \emph{for scalar waves} on Kerr has been performed in \cite{HintzKdSMS} in function spaces compatible with the setup of the present paper; we recall this in~\S\ref{SsScUnif}. The case of linearized gravity (which presents additional difficulties due to the existence of zero energy bound states and resonances) is discussed in \cite{HintzGlueLocIII}.

\emph{If one can control both localizations}, say in the simple setting of scalar waves, then one can improve the error term in~\eqref{EqEst} to
\[
  \|u\|_{H_{\seop,\eps}^{\sfs',\alpha_\circ-\delta,\hat\alpha-\delta}}\leq C\eps^\delta\|u\|_{H_{\seop,\eps}^{\sfs,\alpha_\circ,\hat\alpha}}
\]
for some $\sfs'\leq\sfs$ and $\delta>0$. When $\eps$ is so small that $C\eps^\delta<\frac12$, one can therefore absorb this error into the left hand side of~\eqref{EqEst} and obtain the desired uniform control of $u$.

We remark that uniform estimates on standard domains disjoint from $\hat M$, as in Definition~\ref{DefGlDynStd}\eqref{ItGlDynStdNotC} (or also in Definition~\ref{DefGlDynStd}\eqref{ItGlDynStdC} but restricted to $\eps\geq\eps_1>0$) are standard estimates for wave equations with smooth coefficients. Thus, \textit{unless stated otherwise, all standard domains $\wt\Omega\subset\wt M$ will, for the remainder of this section, be associated with the standard domain
\begin{equation}
\label{EqEstStdDomain}
  \Omega=\Omega_{t_0,t_1,r_0}\subset U_{\rm Fermi}\subset M
\end{equation}
for some fixed choices of} $t_0,t_1\in I_\cC$, $r_0>0$. We will need to localize sharply near $t=t_1$; we phrase this geometrically using the blow-up $[\wt M;\hat M_{t_1}]$ and the associated blow-down map
\begin{equation}
\label{EqEstBlowup}
  \upbeta\colon[\wt M;\hat M_{t_1}]\to\wt M.
\end{equation}
(See the left part of Figure~\ref{FigFseRelGeo}.)

\subsection{Extendible and supported se-Sobolev spaces}
\label{SsEstFn}

Since we aim to prove uniform estimates such as~\eqref{EqEst} on $\eps$-dependent domains in $\eps$-dependent supported/extendible Sobolev spaces, we carefully define here the relevant norms.

Set $T:=t_1-t_0+10$ and fix $R>r_0+2(t_1-t_0)+10$; then $\Omega\subset(t_0-1,t_0+T+1)\times\{|x|<R\}$. Therefore, we can regard $\Omega$, and more generally any $\Omega_{t'_0,t'_1,r'_0}$ with $t_0-2<t'_0<t'_1<t_1+2$ and $r'_0<r_0+1$, as a subset of the \emph{compact} manifold (without boundary)
\[
  M' := \Sph^1_T \times \TT^3_{2 R},
\]
where $\Sph^1_T=[t_0-5,t_1+5]/(t_0-5\sim t_1+5)$ is a circle of circumference $T$, and where $\TT^3_{2 R}:=[-R,R]^3/\sim$, where $\sim$ identifies opposite faces. Furthermore, $\cC\cap\Omega\subset M$ is then equal to $\cC'\cap\Omega\subset M'$ where
\[
  \cC' := \Sph^1_T \times \{0\}.
\]
Setting $\wt M'=[[0,1)_\eps\times M';\{0\}\times\cC']$ as in Definition~\ref{DefFse}, we then define weighted se-Sobolev spaces $H_{\seop,\eps}^{\sfs,\alpha_\circ,\hat\alpha}(M')$ with underlying density $|\dd t\,\dd x|$ as in~\S\ref{SssFVarse}; here $\sfs\in\CI(\Sse^*\wt M')$.

We now turn to supported/extendible spaces; for this purpose, we write $M,\wt M,\cC$ for $M',\wt M',\cC'$; and we drop the weights $\alpha_\circ,\hat\alpha$ from the notation for brevity. Given a closed subset $\cK\subset M_\eps$, we equip $\dot H_{\seop,\eps}^\sfs(\cK)\subset H_{\seop,\eps}^\sfs(M)$ with the induced norm and $\bar H_{\seop,\eps}^\sfs(\cK^\circ)=H_{\seop,\eps}^\sfs(M)/\dot H_{\seop,\eps}^\sfs(M\setminus\cK^\circ)$ with the quotient norm. Of primary interest to us are mixed spaces on standard domains $\wt\Omega\subset\wt M$ associated with the standard domain~\eqref{EqEstStdDomain}, denoted
\[
  H_{\seop,\eps}^\sfs(\Omega_\eps)^{\bullet,-}.
\]
In this space, the norm of an element $u=\tilde u|_{\Omega_\eps^\circ}$, with $\tilde u\in\dot H_{\seop,\eps}^\sfs(\{t\notin(t_0-1,t_0)\})$,\footnote{The notation $\dot H_{\seop,\eps}^\sfs(\{t\geq t_0\})$ makes no sense here since the time coordinate is wrapped up on a circle.} is the quotient norm in $\dot H_{\seop,\eps}^\sfs(\{t\notin(t_0-1,t_0)\})/\dot H_{\seop,\eps}^\sfs(\cK_\eps)$ where $\cK_\eps$ is the closure in $M_\eps\cong M$ of the complement of $\Omega_\eps$ in $\{t\notin(t_0-1,t_0)\}$. Thus, by definition, for every $u\in H_{\seop,\eps}^\sfs(\Omega_\eps)^{\bullet,-}$ there exists a (unique) $\tilde u\in\dot H_{\seop,\eps}^\sfs(\{t\notin(t_0-1,t_0)\})$ realizing the quotient norm, i.e.\ the norms of $u$ and $\tilde u$ are equal.
\begin{rmk}[Working locally]
\label{RmkEstFnLoc}
  In practice, we will only work locally near $\wt\Omega$ and only specify a fixed order function $\sfs$ over (a neighborhood of) $\wt\Omega$; when writing $H_{\seop,\eps}^\sfs(\Omega_\eps)^{\bullet,-}$, we then tacitly \emph{fix} an extension of $\sfs$ to a globally defined element of $\CI(\Sse^*\wt M)$ to define the norm.
\end{rmk}

The uniform relationships between se- and edge or 3b-Sobolev norms stated in Proposition~\ref{PropFVarseSobRel} remain valid for supported/extendible spaces. To state this precisely, we need to define these spaces in the edge setting on $M_\circ=[M;\cC]$ where $M=\Sph^1_T\times\TT_{2 R}^3$: for the standard domain $\Omega=\Omega_{t_0,t_1,r_0}$, we define
\[
  \He^\sfs(\Omega)^{\bullet,-}
\]
as the space of restrictions to $\Omega^\circ$ of elements of $\He^\sfs(M_\circ)$ which vanish for $t\in[t_0-1,t_0]$. In the 3b-setting, we note that the image of $\Omega_\eps$ under the map $\Psi_\eps(t,x)=(\frac{t-t_0}{\eps},\frac{x}{\eps})$ is equal to
\[
  \hat\Omega_\eps = \bigl\{ (\hat t,\hat x)\in\hat M_b^\circ \colon 0\leq\hat t\leq\eps^{-1}t_1,\ \bhm\leq|\hat x|\leq \eps^{-1}(r_0+2(t_1-\eps\hat t)) \bigr\}.
\]

\begin{prop}[Supported/extendible se- and edge or 3b-Sobolev spaces]
\label{PropEstFnSobRel}
  Let $\sfs\in\CI(\Sse^*\wt M)$, $\alpha_\circ,\hat\alpha\in\R$. Then the uniform norm bounds~\eqref{EqFVarseSobRele} and \eqref{EqFVarseSobRel3b} remain valid when using the spaces $H_{\seop,\eps}^{\sfs,\alpha_\circ,\hat\alpha}(\Omega_\eps)^{\bullet,-}$, $\He^{\sfs,\ell}(\Omega)^{\bullet,-}$, and $\Htb^{\sfs,\alpha_\cD,0}(\hat\Omega_\eps)^{\bullet,-}$ with the orders stated there.
\end{prop}
\begin{proof}
  For notational simplicity, we only discuss the case $\alpha_\circ=\hat\alpha=0$. The analogue of the first inequality in~\eqref{EqFVarseSobRele} then reads
  \begin{equation}
  \label{EqEstFnSobRele1}
    \|\chi_\circ u\|_{H_{\seop,\eps}^\sfs(\Omega_\eps)^{\bullet,-}} \leq C\|\chi_\circ u\|_{\He^{\sfs_\circ+\delta,0}(\Omega)^{\bullet,-}},\qquad \sfs_\circ=\sfs|_{\Sse^*_{M_\circ}\wt M};
  \end{equation}
  here $\delta>0$ is fixed, and $\chi_\circ=\chi(\frac{\eps}{|x|})$ where $\chi\in\CIc([0,c_0))$ with $c_0>0$ depending on $\delta$, while $\eps\in(0,c_0]$. Shrinking $c_0$, we thus have $\sfs\leq\sfs_\circ+\delta$ for $\eps,\frac{\eps}{|x|}<2 c_0$. To prove~\eqref{EqEstFnSobRele1}, denote by $u'\in\He^{\sfs_\circ+\delta}(M_\circ)$ the minimal norm extension of $\chi_\circ u$ with $u'|_{[t_0-1,t_0]}=0$. Let $\chi^\sharp\in\CIc([0,2 c_0))$ be equal to $1$ on $[0,c_0]$, and let $\chi^\sharp_\circ=\chi^\sharp(\frac{\eps}{|x|})$. Note that $\supp\dd\chi_\circ^\sharp$ is contained in the set $\frac{\eps}{2 c_0}\leq|x|\leq\frac{\eps}{c_0}$; moreover, the pushforward of $\chi_\circ^\sharp$ under the map $(t,r,\omega)\mapsto(T,R,\omega):=(\frac{c_0}{\eps}t,\frac{c_0}{\eps}r,\omega)$ (cf.\ \eqref{EqFVarEBddMap}) is given by $(T,R,\omega)\mapsto\chi_\circ^\sharp(\frac{c_0}{R})$ (which is uniformly bounded in $\CI(U)$ in the notation of~\eqref{EqFVarEBddMap}). Therefore, multiplication by $\chi_\circ^\sharp$ is uniformly (in $\eps$) bounded on every edge Sobolev space; this implies the uniform bound
  \[
    \|\chi^\sharp_\circ u'\|_{\He^{\sfs_\circ+\delta}(M_\circ)} \leq C^\sharp\|\chi_\circ u\|_{\He^{\sfs_\circ+\delta}(\Omega)^{\bullet,-}}.
  \]
  We now apply the estimate~\eqref{EqFVarseSobRele} to $\chi^\sharp_\circ u'$ to deduce, using $\chi_\circ\chi_\circ^\sharp=\chi_\circ$, the estimate
  \[
    \|\chi_\circ u'\|_{H_{\seop,\eps}^\sfs(M)} \leq C\|\chi_\circ u'\|_{\He^{\sfs_\circ+\delta}(M_\circ)} \leq C'\|\chi_\circ^\sharp u'\|_{\He^{\sfs_\circ+\delta}(M_\circ)} \leq C'C^\sharp\|\chi_\circ u\|_{\He^{\sfs_\circ+\delta}(\Omega)^{\bullet,-}},
  \]
  where in the second inequality we use the uniform boundedness of multiplication by $\chi_\circ$ on every edge Sobolev space. Since $\chi_\circ u'$ is an extension of $\chi_\circ u$ to $M$, this proves~\eqref{EqEstFnSobRele1}. The analogue of the second inequality in~\eqref{EqFVarseSobRele} is proved similarly: one cuts off the minimal norm extension $u'$ of $\chi_\circ u$ in $H_{\seop,\eps}^\sfs(M)$ using the uniformly bounded multiplication operator $\chi_\circ^\sharp$ and applies~\eqref{EqFVarseSobRele} to $\chi_\circ^\sharp u'$.

  Next, we prove the analogue of the 3b-estimate~\eqref{EqFVarseSobRel3b},
  \[
    \|u\|_{H_{\seop,\eps}^\sfs(\Omega_\eps)^{\bullet,-}} \leq C\eps^{\frac{n+1}{2}}\|v\|_{H_\tbop^{\hat\sfs+\delta}(\hat\Omega_\eps)^{\bullet,-}},\qquad v=(\Psi_\eps)_*u.
  \]
  Denote by $v'\in\Htb^{\hat\sfs+\delta}(\cM)$ the minimal norm extension of $v$ which vanishes for $\hat t\leq 0$. We now note that multiplication by $\chi(\eps\hat z)$ (with $\hat z=(\hat t,\hat x)$), where $\chi\in\CIc(\R^4)$, is uniformly bounded on every 3b-Sobolev space; this follows from the uniform (in $\eps$) boundedness in $\CI$ of the pushforwards of $\chi(\eps\hat z)$ to the 3b-unit cells~\eqref{EqFVar3b1}--\eqref{EqFVar3b2}, uniformly in the parameters $j,k$ of the cells. We choose $\chi$ to be equal to $1$ on $\Omega$ and $0$ outside a small neighborhood thereof. Then the norm of $\chi(\eps\hat z)v'\in\Htb^{\hat\sfs+\delta}(\cM)$ is uniformly (in $\eps$) bounded by $\|v\|_{H_\tbop^{\hat\sfs+\delta}(\hat\Omega_\eps)^{\bullet,-}}$; and $u':=\Psi_\eps^*(\chi v')$ restricts to $u=\Psi_\eps^*v$ on $\Omega_\eps$. By~\eqref{EqFVarseSobRel3b}, we thus have the uniform estimate
  \[
    \|u\|_{H_{\seop,\eps}^\sfs(\Omega_\eps)^{\bullet,-}} \leq \|u'\|_{H_{\seop,\eps}^\sfs(M)} \leq C\|\chi v'\|_{H_\tbop^{\hat\sfs+\delta}(\cM)} \leq C'\|v'\|_{H_\tbop^{\hat\sfs+\delta}(\cM)} = C'\|v\|_{H_\tbop^{\hat\sfs+\delta}(\hat\Omega_\eps)^{\bullet,-}},
  \]
  as claimed. The analogue of the second inequality in~\eqref{EqFVarseSobRel3b} is proved similarly.
\end{proof}

\subsection{Microlocal propagation near the radial sets over \texorpdfstring{$\pa M_\circ$}{the boundary of the front face}}
\label{SsEstRad}

Near the radial sets $\cR_{\rm in/out}$ from Definition~\ref{DefGlDynInOut} and near $t^{-1}([t_0,t_1])$, we work with the coordinates
\[
  t,\qquad \hat\rho:=r=|x|,\qquad \rho_\circ:=\frac{\eps}{r},\qquad \omega=\frac{x}{|x|}\in\Sph^2,
\]
and with the fiber-linear coordinates on $\Tse^*\wt M$ given by writing se-covectors as $-\sigma_\seop\,\frac{\dd t}{r}+\xi_\seop\,\frac{\dd r}{r}+\eta_\seop$ where $\eta_\seop\in T^*\Sph^2$. These constitute local coordinates near $\wt\Sigma^+$. (They are also local coordinates near $\wt\Sigma^-$, but we shall only work near $\wt\Sigma^+$ below. The analogous estimates in $\wt\Sigma^-$ are then proved in the same manner upon replacing $\wt G$ by $-\wt G$.) Thus, we work in $\sigma_\seop>0$, and we can use $\rho_\infty=\sigma_\seop^{-1}$ as a defining function of fiber infinity, and the projective fiber coordinates $\hat\xi_\seop=\frac{\xi_\seop}{\sigma_\seop}$, $\hat\eta_\seop=\frac{\eta_\seop}{\sigma_\seop}$ from~\eqref{EqGlDynInOutCoord2}.

To state the sharp localization near $t=t_1$ efficiently, we introduce
\[
  \tau = \frac{t-t_1}{r}.
\]
Thus, $\tau$ is a smooth coordinate near the interior of the front face of~\eqref{EqEstBlowup} in $\hat r>0$.

\begin{prop}[Radial point estimate near $\cR_{\rm in}^+$]
\label{PropEstRadIn}
  Fix a positive definite fiber inner product on $\wt E$ which over $\hat M$ is the pullback of a fiber inner product on $\hat\cE$. Use the volume density $|\dd g_\eps|$ on $M_\eps\setminus\wt K^\circ$. Define the quantity
  \begin{equation}
  \label{EqEstRadInTheta}
    \vartheta_{\rm in} := \sup_{\pa\cR_{\rm in}^+} {\rm spec}\Bigl[\rho_\infty\sigmase^1\Bigl(\hat\rho^2\frac{\wt L-\wt L^*}{2 i}\Bigr)\Bigr] \in \R,
  \end{equation}
  where ${\rm spec}$ denotes the spectrum. Let $s,s_0,N,\alpha_\circ,\hat\alpha\in\R$, and suppose that $s>s_0$ and
  \begin{equation}
  \label{EqEstRadInThr}
    s_0 + \alpha_\circ - \hat\alpha > \frac12(-1+\vartheta_{\rm in}).
  \end{equation}
  Define $\cK:=\{\rho_\infty=\hat\rho=\rho_\circ=0,\ \hat\xi_\seop=-1,\ t\geq t_0,\ \tau\leq 0\}\subset\upbeta^*(\Sse^*\wt M)$. Then for all neighborhoods $\cU\subset\upbeta^*(\Sse^*\wt M)$ of $\cK$ and all $\chi\in\CIc([\wt M;\hat M_{t_1}])$ equal to $1$ near the base projection $\upbeta^{-1}(\pa M_\circ)\cap\{t\geq t_0,\ \tau\leq 0\}$ of $\cK$, there exist operators $B,E,G\in\tilde\Psi_\seop^0(\wt M)$ with $\chi B\chi=B$ etc.\ so that
  \begin{itemize}
  \item the operator wave front sets of $B,E,G$ are contained in $\cU$,
  \item $B$ is elliptic at $\cK$,
  \item $\WF'_\seop(E)\subset\{r>0\}$,
  \end{itemize}
  and a constant $C$ so that the estimate
  \begin{equation}
  \label{EqEstRadIn}
    \|B u\|_{H_{\seop,\eps}^{s,\alpha_\circ,\hat\alpha}} \leq C\Bigl( \|G\wt L u\|_{H_{\seop,\eps}^{s-1,\alpha_\circ,\hat\alpha-2}} + \|E u\|_{H_{\seop,\eps}^{s,\alpha_\circ,\hat\alpha}} + \|G u\|_{H_{\seop,\eps}^{s_0,\alpha_\circ,\hat\alpha}} + \|\chi u\|_{H_{\seop,\eps}^{-N,\alpha_\circ,\hat\alpha}}\Bigr)
  \end{equation}
  holds for all $u$ with support in $t\geq t_0$, in the strong sense that if the right hand side is finite, then so is the left hand side (with the stated bound).
\end{prop}

Due to the independence of the normal operators of $\wt L$ at $\hat M_{t_0}$ on $t_0$, the supremum in~\eqref{EqEstRadInTheta} can be computed over the compact set $\pa\cR_{\rm in}^+\cap\Sse^*_{\hat M_{t_0}}\wt M$. The value of $\vartheta_{\rm in}$ typically depends on the choice of fiber inner product, and to obtain an optimal statement one should take the infimum over all choices, as done e.g.\ in \cite[Definition~4.3]{HintzNonstat}.

The wave front set condition on $E$ together with the localization to a small neighborhood of the set $\cK$ means that $E u$ can only control $u$ microlocally in a punctured neighborhood of $\pa\cR_{\rm in,t'}^+$ inside the stable manifold of $\pa\cR_{\rm in,t'}^+$ for all $t'\in[t_0,t_1]$ (and this is what the choice of $E$ arising from the proof indeed does). We note moreover that $\cK$ is the intersection of the preimage of $\pa\cR_{\rm in}^+$ under $\upbeta$ with the lift of $\Sse^*_{\wt\Omega}\wt M$.\footnote{One can compute that for the lift of the null-bicharacteristic flow to $\upbeta^*(\Sse^*\wt M)$ the part of $\cK$ over the interior of the front face does not contain any radial points \emph{except} over $\tau=-1$, which from the perspective of the front face of $[\wt M;\hat M_{t_1}]$ is past null infinity of the Kerr black hole at $\hat M_{t_1}$. This is partially discussed in \cite[Remark~3.8]{HintzConicWave}.} See Figure~\ref{FigEstRadIn}.

We also remark that there is a version of the estimate~\eqref{EqEstRadIn} without the support condition on $u$: one only needs to add a further term $\|E'u\|_{H_{\seop,\eps}^{s,\alpha_\circ,\hat\alpha}}$ where $E'$ involves a cutoff in $\frac{t-t_0}{r}$. This follows from a simple modification of the proof below.

\begin{figure}[!ht]
\centering
\includegraphics{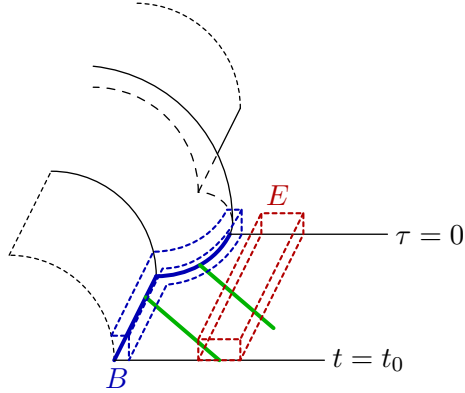}
\caption{Illustration of the estimate~\eqref{EqEstRadIn}: the term $E u$ controls $u$ microlocally near incoming null-bicharacteristics (indicated in green), and $B u$ gives control at $\cK$ (the projection of which to $\wt M$ is the thick blue line).}
\label{FigEstRadIn}
\end{figure}

\begin{proof}[Proof of Proposition~\usref{PropEstRadIn}]
  We fix on $\wt E$ a positive definite fiber inner product which restricts to the chosen one over $\hat M$. We use a standard positive commutator argument as in \cite{MelroseEuclideanSpectralTheory} and \cite[\S{2.4}]{VasyMicroKerrdS}, though here on a manifold with corners as in \cite[\S{4.2}]{HintzVasyScrieb} and indeed similar to \cite[\S{3.4.1}]{HintzKdSMS} as far as the saddle point dynamics go; there is only the minor subtlety that the rather sharp localization near $t=t_1$ necessitates working with non-smooth (but still se-regular) symbols there, analogously to \cite[Proposition~3.4]{HintzConicWave}. The uniformity in $\eps$ of our estimate in se-estimate Sobolev spaces will simply be a consequence of the uniform boundedness of se-ps.d.o.s. Thus, we shall be brief.

  We consider an operator $A=A^*\in\tilde\Psi_\seop^{2 s-1,2\alpha_\circ,2\hat\alpha-2}$ (see~\S\ref{SssFVarse}) with principal symbol
  \begin{equation}
  \label{EqEstRadInSymb}
    a = \check a^2,\qquad
    \check a = \rho_\infty^{-s+\frac12}\rho_\circ^{-\alpha_\circ}\hat\rho^{-(\hat\alpha-1)}\chi_\circ(\rho_\circ)\hat\chi(\hat\rho)\chi_\cR(|\hat\eta_\seop|^2) \chi_\Sigma(\hat\xi_\seop+1) \chi_-(t-t_0)\chi_+\Bigl(\frac{t-t_1}{r}\Bigr).
  \end{equation}
  The cutoff functions are as follows: $\chi_\circ\in\CIc([0,\delta_0))$ is equal to $1$ near $0$, and satisfies $\chi_\circ'\leq 0$ and $\sqrt{-\chi_\circ\chi_\circ'}\in\CI$; similarly for $\hat\chi$ and $\chi_\cR$. Here $\delta_0>0$ will serve to localize the support of $\check a$ near $\cR_{\rm in}^+$. Moreover, we choose $\chi_\Sigma\in\CIc((-2\delta_0,2\delta_0))$ to be equal to $1$ on $[-\delta_0,\delta_0]$, which for small enough $\delta_0$ ensures that $\supp(\chi_\cR(|\hat\eta_\seop|^2)\chi_\Sigma'(\hat\xi_\seop+1))\cap\wt\Sigma^+=\emptyset$ since the characteristic set over $\hat\rho=\rho_\circ=0$ is given by $\hat\xi_\seop^2+|\hat\eta_\seop|^2=1$. Finally, we choose $\chi_\pm$ with
  \begin{equation}
  \label{EqEstRadInChipm}
  \begin{alignedat}{5}
    \chi_-&\in\CI(\R), &\quad& \chi_-'&\geq 0, &\quad& \chi_-|_{(-\infty,-\delta_0]}=0, &\quad& \chi_-|_{[0,\infty)}=1, \\
    \chi_+&\in\CI(\R), &\quad& \chi_+'&\leq 0, &\quad& \chi_+|_{(-\infty,0]}=1, &\quad& \chi_+|_{[\delta_0,\infty)}=0, \\
    \sqrt{\mp\chi_\pm\chi_\pm'}&\in\CI(\R).
  \end{alignedat}
  \end{equation}

  Write $\sfH=\rho_\infty\hat\rho^2 H_{\wt G}=\rho_\infty H_{\hat\rho^2\wt G}-\rho_\infty\wt G H_{\hat\rho^2}$. Using~\eqref{EqGlDynStdInOut}, we now compute
  \[
    c':=\rho_\infty^{2 s}\rho_\circ^{2\alpha_\circ}\hat\rho^{2\hat\alpha}H_{\wt G}a = 2\rho_\infty^{2 s-1}\rho_\circ^{2\alpha_\circ}\hat\rho^{2(\hat\alpha-1)}\check a\sfH\check a.
  \]
  At $\pa\cR_{\rm in}^+$ (where $\hat\xi_\seop=-1$), only derivatives of the weights enter, and one finds
  \begin{equation}
  \label{EqEstRadIncp}
    c'|_{\pa\cR_{\rm in}^+} = 2\Bigl( 2(\hat\alpha-1) - 2\alpha_\circ - 2\Bigl(s-\frac12\Bigr)\Bigr) = -2(2 s+2\alpha_\circ-2\hat\alpha+1).
  \end{equation}
  The signs of those terms of $c'$ arising from differentiating one of the cutoffs $\chi_\circ,\hat\chi,\chi_\cR$ are determined by the saddle point dynamics at $\pa\cR_{\rm in}^+$. Concretely, if one fixes $\chi_\cR$ (thus $\rho_\infty H_{\hat\rho^2\wt G}|\hat\eta_\seop|^2$ has a positive lower bound $c_0>0$ on $\supp(\chi_\Sigma\dd\chi_\cR)$ over $\pa\hat M$), then upon shrinking the supports of $\chi_\circ,\hat\chi$ further we can ensure that $\rho_\infty H_{\hat\rho^2\wt G}|\hat\eta_\scop|^2\geq\frac12 c_0$ on $\supp\check a$ for all sufficiently small $\eps\geq 0$. Therefore, the term in $c'$ involving $\chi_\cR\sfH\chi_\cR$ is then the negative of the square of a smooth function. Similarly, $\chi_\circ\sfH\chi_\circ=\chi_\circ\chi_\circ'\sfH\rho_\circ$ is the negative of the square of a smooth function on $\supp\check a$ when $\delta_0$ is sufficiently small (since then $\sfH\rho_\circ$ is negative on $\supp\check a\cap\supp\dd\chi_\circ$), whereas the term $\hat\chi\sfH\hat\chi$ is the square of a smooth function (ultimately giving rise to the a priori control term $E u$ in~\eqref{EqEstRadIn}). By construction, derivatives falling on $\chi_\Sigma$ produce symbols supported away from $\wt\Sigma^+$ which are thus smooth multiples of (weighted versions of) $\wt G$.

  Next, we consider the term $\chi_-\sfH\chi_-=\chi_-\chi_-'\sfH t$; due to the past timelike nature of $\dd t$, $\sfH t\in\hat\rho\CI$ is non-negative, and indeed $\hat\rho^{-1}\sfH t>0$ on $\wt\Sigma^+$; therefore the contribution of this term to $c'$ is a square (of a symbol with $\hat M$-weight $\hat\alpha-\frac12$ even). Finally, to treat the term $\chi_+\sfH\chi_+$, we go back to~\eqref{EqGlDynHamEdge}. This implies that, on the characteristic set, $\sfH\frac{t-t_1}{r}=2-2\frac{t-t_1}{r}\hat\xi_\seop$ up to corrections which vanish at $\pa\hat M$; but since $\frac{t-t_1}{r}\in[0,\frac12]$ on $\supp\chi_+'$, this is strictly positive, and so $\chi_+\sfH\chi_+$ is the negative of a square on $\supp\check a$.

  The positive commutator argument proceeds by computing the $L^2$ pairing (on $M_\eps$)
  \begin{equation}
  \label{EqEstRadInComm}
    2\Im\la\wt L u,A u\ra = \la \sC u,u\ra,\qquad \sC=i[\wt L,A]+2\frac{\wt L-\wt L^*}{2 i}A.
  \end{equation}
  The principal symbol of $\sC\in\tilde\Psi_\seop^{2 s,2\alpha_\circ,2\hat\alpha}$ is equal to $\rho_\infty^{-2 s}\rho_\circ^{-2\alpha_\circ}\hat\rho^{-2\hat\alpha}$ times $c:=c'+2\ell_1\tilde\chi^2$ where $\ell_1=\rho_\infty\sigmase^1(\hat\rho^2\frac{\tilde L-\tilde L^*}{2 i})$ is a smooth function on $\Sse^*\wt M$ near $\pa\cR_{\rm in}^+$, and $\tilde\chi$ is the product of all cutoffs in the definition of $\check a$. Combined with~\eqref{EqEstRadIncp}, we thus see that $c<0$ at $\pa\cR_{\rm in}^+$ (and thus nearby if the supports of the cutoffs are sufficiently small) provided the threshold condition~\eqref{EqEstRadInThr} holds.

  From here on, the proof proceeds in the usual fashion by writing $c$ as a sum $-\eta\tilde\chi^2$ (for some fixed $\eta\in(0,s+\alpha_\circ-\hat\alpha-\frac12(-1+\vartheta_{\rm in}))$) plus a sum of terms which are (positive or negative) squares of smooth functions, quantizing the symbols appearing in such a formula, and using the ellipticity of the main term (coming from~\eqref{EqEstRadIncp} and $\vartheta_{\rm in}$) to obtain the desired control of a microlocalization of $u$ in $H_{\seop,\eps}^{s,\alpha_\circ,\hat\alpha}$, uniformly for all sufficiently small $\eps>0$, in terms of the microlocal $H_{\seop,\eps}^{s-\frac12,\alpha_\circ,\hat\alpha}$ norm of $u$ on a slightly larger set, which can be controlled inductively while $s-\frac12$ still exceeds the threshold regularity.

  The proof of the strong version of~\eqref{EqEstRadIn} can be done using standard regularization arguments \cite[\S{5.4.7}]{VasyMinicourse}, now carried out in the se-pseudodifferential algebra. Alternatively, one simply notes that for positive $\eps$, the propagation of microlocal $H^s$ regularity from the elliptic set of $E$ to that of $B$ over $M_\eps$ is the standard propagation of singularities statement \cite{DuistermaatHormanderFIO2}, so $B u$ does lie in $H^s$; this is sufficient to justify the above positive commutator argument without further regularization.
\end{proof}

\begin{prop}[Radial point estimate near $\cR_{\rm out}^+$]
\label{PropEstRadOut}
  Define, analogously to~\eqref{EqEstRadInTheta}, the quantity
  \begin{equation}
  \label{EqEstRadOutTheta}
    \vartheta_{\rm out} := \sup_{\pa\cR_{\rm out}^+} {\rm spec}\Bigl[\rho_\infty\sigmase^1\Bigl(\hat\rho^2\frac{\wt L-\wt L^*}{2 i}\Bigr)\Bigr] \in \R.
  \end{equation}
  Let $s,N,\alpha_\circ,\hat\alpha\in\R$, and suppose that
  \begin{equation}
  \label{EqEstRadOutThr}
    s+\alpha_\circ-\hat\alpha<\frac12(-1-\vartheta_{\rm out}).
  \end{equation}
  Define $\cK=\{\rho_\infty=\hat\rho=\rho_\circ=0,\ \hat\xi_\seop=+1,\ t\geq t_0,\ \tau\leq 0\}\subset\upbeta^*(\Sse^*\wt M)$. Then for all neighborhoods $\cU\subset\upbeta^*(\Sse^*\wt M)$ of $\cK$ and all $\chi\in\CIc([\wt M;\hat M_{t_1}])$ equal to $1$ near the base projection of $\cK$, there exist operators $B,E,G\in\tilde\Psi_\seop^0(\wt M)$ with $\chi B\chi=B$ etc.\ so that
  \begin{itemize}
  \item the operator wave front sets of $B,E,G$ are contained in $\cU$,
  \item $B$ is elliptic at $\cK$,
  \item $\WF_\seop'(E)$ is disjoint from $\upbeta^{-1}(\pa W_{\rm out}^+)$ (see Definition~\usref{DefGlDynStableUn}),
  \end{itemize}
  and a constant $C$ so that the estimate
  \begin{equation}
  \label{EqEstRadOut}
    \|B u\|_{H_{\seop,\eps}^{s,\alpha_\circ,\hat\alpha}} \leq C\Bigl( \|G\wt L u\|_{H_{\seop,\eps}^{s-1,\alpha_\circ,\hat\alpha-2}} + \|E u\|_{H_{\seop,\eps}^{s,\alpha_\circ,\hat\alpha}} + \|\chi u\|_{H_{\seop,\eps}^{-N,\alpha_\circ,\hat\alpha}}\Bigr)
  \end{equation}
  holds (in the strong sense) for all $u$ with support in $t\geq t_0$.
\end{prop}
\begin{proof}
  The proof is completely analogous to that of Proposition~\ref{PropEstRadIn}, and uses the commutant~\eqref{EqEstRadInSymb} with $\hat\xi_\seop-1$ in place of $\hat\xi_\seop+1$. The only further change then is that due to $\hat\xi_\seop$ now being equal to $+1$ at $\pa\cR_{\rm out}^+$, there is a sign switch in the first (main) term in the expression~\eqref{EqGlDynStdInOut}, and correspondingly the sign of the expression~\eqref{EqEstRadIncp} switches. The contribution $\ell_1$ from the imaginary part of $\wt L$ enters however with the same sign as before, so overall the principal symbol of $\sC$ in~\eqref{EqEstRadInComm} is $\rho_\infty^{-2 s}\rho_\circ^{-2\alpha_\circ}\hat\rho^{-2\hat\alpha}$ times a self-adjoint bundle endomorphism of $\wt E$ that is negative at $\pa\cR_{\rm out}^+$ if $2(2 s+2\alpha_\circ-2\hat\alpha+1)+2\ell_1<0$ there. This leads to the threshold condition~\eqref{EqEstRadOutThr}. Since this condition allows for arbitrarily negative $s$, we may iterate the resulting radial point estimate arbitrarily often and thus remove the $G u$ term in~\eqref{EqEstRadIn}.
\end{proof}

\subsubsection{Higher s-regularity}
\label{SssEstRads}

We shall need suitable analogues of Propositions~\ref{PropEstRadIn}--\ref{PropEstRadOut} on mixed (se;s)-Sobolev spaces (see~\eqref{EqFVarsseSob}). Such analogues cannot feature sharp localizations in $\tau$, but only localizations in $t$, cf.\ the discussion at the beginning of~\S\ref{SssFVarsse}. In order for solutions of $\wt L u=f$ to have s-regularity, we need the coefficients of $\wt L$ to have s-regularity; thus, we shall replace assumption~\eqref{ItEstLOp} by
\[
  \wt L \in \hat\rho^{-2}(\CI+\cC_{\seop;\sop}^{(d_0;k),1,1})\Diffse^2(\wt M\setminus\wt K^\circ;\wt E),
\]
where the conditions on $d_0\geq 2$ and $k\in\N_0$ are specified in the statements below. For the arguments below, it will be important to relax the assumptions on $\wt L$ even further and allow for pseudodifferential lower order terms, i.e.
\begin{equation}
\label{EqEstRadsOp}
  \wt L \in \hat\rho^{-2}(\CI+\cC_{\seop;\sop}^{(d_0;k),1,1})(\Diffse^2(\wt M\setminus\wt K^\circ;\wt E)+\Psi_\seop^1(\wt M;\wt E)).
\end{equation}
We shall tacitly always require the Schwartz kernel of the pseudodifferential term to be localized to a small neighborhood of the diagonal, but will not spell this out below.

\begin{prop}[Radial point estimate near $\cR_{\rm in}^+$: s-regularity]
\label{PropEstRadsIn}
  We use the notation of Proposition~\usref{PropEstRadIn} and assume~\eqref{EqEstRadInThr}, except now we define $\cK:=\{\rho_\infty=\hat\rho=\rho_\circ=0,\ \hat\xi_\seop=-1,\ t_0\leq t\leq t_1\}\subset\Sse^*\wt M$. Let $k\in\N_0$, and assume that $\wt L$ satisfies~\eqref{EqEstRadsOp} for some sufficiently large $d_0=d_0(s,s_0,N)$. Then for all neighborhoods $\cU\subset\Sse^*\wt M$ of $\cK$ and all $\chi\in\CIc(\wt M)$ equal to $1$ near the base projection $\pa M_\circ\cap\{t_0\leq t\leq t_1\}$ of $\cK$, there exist operators $B,E,G\in\Psi_\seop^0(\wt M)$ so that $\chi B\chi=B$ etc., the operator wave front sets of $B,E,G$ are contained in $\cU$, furthermore $B$ is elliptic at $\cK$, and $\WF'_\seop(E)\subset\{r>0\}$, and such that, for some constant $C$, the estimate
  \begin{equation}
  \label{EqEstRadsIn}
  \begin{split}
    \|B u\|_{H_{(\seop;\sop),\eps}^{(s;k),\alpha_\circ,\hat\alpha}} &\leq C\Bigl( \|G\wt L u\|_{H_{(\seop;\sop),\eps}^{(s-1;k),\alpha_\circ,\hat\alpha-2}} + \|E u\|_{H_{(\seop;\sop),\eps}^{(s;k),\alpha_\circ,\hat\alpha}} \\
      &\quad\hspace{5em} + \|G u\|_{H_{(\seop;\sop),\eps}^{(s_0;k),\alpha_\circ,\hat\alpha}} + \|\chi u\|_{H_{(\seop;\sop),\eps}^{(-N;k),\alpha_\circ,\hat\alpha}}\Bigr)
  \end{split}
  \end{equation}
  holds (in the strong sense) for all $u$ with support in $t\geq t_0$.
\end{prop}
\begin{proof}
  Consider first the case $k=0$. One can then give a proof using the same argument as in the proof of Proposition~\ref{PropEstRadIn}, except one replaces the factor $\chi_+(\frac{t-t_1}{r})$ in~\eqref{EqEstRadInSymb} by $\chi_+(t-t_1)$ and quantizes $a$ to a formally self-adjoint operator in the stronger operator class $\Psi_\seop^{2 s-1,2\alpha_\circ,2\hat\alpha-2}$ from~\S\ref{SssFVarsse}. The Hamiltonian derivative of $\chi_+$ still has the same sign as the main term, and can thus again be dropped. One can also directly deduce~\eqref{EqEstRadsIn} from~\eqref{EqEstRadIn} by enlarging the elliptic sets of $G,E$ (which one now requires to be of class $\Psi_\seop$) slightly. The fact that $\wt L$ now features a pseudodifferential lower order term does not affect the argument in any way. We moreover note that for any fixed orders $s,s_0,N,\alpha_\circ,\hat\alpha$, the proof of~\eqref{EqEstRadsIn} only uses some large but finite se-regularity $d_0$ of the coefficients of $\wt L$.

  For $k\geq 1$, one can argue by induction (reducing the case $s,k$ to $s+1,k-1$) as in the first version of the proof of Lemma~\ref{LemmaFVarsseProp}. A more precise argument, which reduces the inductive step to the case $s,k-1$ (with $d_0$ unchanged), will be given in the case of the outgoing radial point estimate in Proposition~\ref{PropEstRadsOut} below.
\end{proof}

Near $\cR_{\rm out}^+$, the issue with an inductive step which reduces the case of orders $s,k$ to that of $s+1,k-1$ is that the se-regularity order is subject to an \emph{upper bound}, and thus one would only be able to prove $H_{(\seop;\sop),\eps}^{(s;k)}$-estimates for $s+k$ below this bound. A standard method to circumvent this restriction is to exhibit a spanning set of s-vector fields with the property that their commutators with $\wt L$ can be written as compositions of these vector fields with operators whose principal symbols vanish at the radial set; see e.g.\ \cite[Lemma~3.20]{HintzConicWave}, which in turn was inspired by \cite[\S{4}]{BaskinVasyWunschRadMink} and \cite{HassellMelroseVasySymbolicOrderZero}. We give an alternative argument here by \emph{only commuting $\wt L$ with $\pa_t$} and using the second part of Corollary~\ref{CorFsComm}. The key observation is that near the characteristic set of $\wt L$ over $\pa M_\circ$, the operator $\hat\rho\pa_t=r\pa_t$ is elliptic; indeed, its principal symbol is $i\sigma_\eop$, which does not vanish on $G_\eop^{-1}(0)\setminus o$ near $\pa M_\circ$ in view of~\eqref{EqGlDynGe}. If $B\in\Psi_\seop^0(\wt M)$ satisfies $\WF_\seop'(B)\cap\Char_\seop(r\pa_t)=\emptyset$ and has Schwartz kernel supported in $r>0$ in both factors, we claim that we have uniform norm equivalences
\begin{equation}
\label{EqEstRadsNorm}
  \|B u\|_{H_{(\seop;\sop),\eps}^{(s;k)}} \sim \|B u\|_{H_{(\seop;\sop),\eps}^{(s;k-1)}} + \|\pa_t B u\|_{H_{(\seop;\sop),\eps}^{(s;k-1)}} \sim \sum_{j=0}^k \|\pa_t^j B u\|_{H_{\seop,\eps}^s};
\end{equation}
similarly for weighted spaces. It suffices to prove the direction `$\lesssim$' in the first norm equivalence; we do this for $k=0$. If $V_1,\ldots,V_N\in\Vse(\wt M)$ spans $\Vse(\wt M)$ over $\CI(\wt M)$, then we have
\[
  \|B u\|_{H_{\seop,\eps}^{(s;1)}}\sim\|\pa_t B u\|_{H_{\seop,\eps}^s}+\sum_{l=1}^N\|V_l B u\|_{H_{\seop,\eps}^s}+\|B u\|_{H_{\seop,\eps}^s}.
\]
But we can write $V_l=W_l\circ r\pa_t+R_l$ where $W_l\in\Psi_\seop^0$ and $R_l\in\Psi_\seop^1$, with $\WF_\seop'(R_l)\cap\WF_\seop'(B)=\emptyset$. Since $W_l\circ r\in\Psi_\seop^0$ and $R_l B\in\Psi_\seop^{-\infty}$ are uniformly bounded in $H_{\seop,\eps}^s$, we can further bound this by a uniform constant times $\|\pa_t B u\|_{H_{\seop,\eps}^s}+\|B u\|_{H_{\seop,\eps}^s}$, as claimed.

\begin{lemma}[Commutator with $\pa_t$]
\label{LemmaEstRadsComm}
  Let $K\subset\Sse^*\wt M$ be a compact subset which is disjoint from $\Char_\seop(\hat\rho\pa_t)$. Then we can write
  \[
    [\wt L,\pa_t] = \sum_{j=1}^N a_j(A_j\hat\rho\pa_t + R_j)
  \]
  where $a_j\in\hat\rho^{-2}(\hat\rho\CI+\cC_{\seop;\sop}^{(d_0;k-1),1,1})$, $A_j\in\Psi_\seop^1(\wt M)$, $R_j\in\Psi_\seop^2(\wt M)$, $\WF_\seop'(R_j)\cap K=\emptyset$.
\end{lemma}
\begin{proof}
  Since the $\hat M$-normal operators of $\eps^2\wt L$ are $t$-independent, we have $[\wt L,\pa_t]\in\hat\rho^{-2}(\hat\rho\CI+\cC_{\seop;\sop}^{(d_0;k-1),1,1})\Diffse^2$ by Corollary~\ref{CorFsComm}, which is thus of the form $\sum_{j=1}^N a_j B_j$ where $a_j$ are as in the statement of the lemma and $B_j\in\Diffse^2$. It then remains to write $B_j=A_j\hat\rho\pa_t+R_j$ using a microlocal elliptic parametrix construction in $\Psi_\seop$.
\end{proof}

\begin{prop}[Radial point estimate near $\cR_{\rm out}^+$: s-regularity]
\label{PropEstRadsOut}
  We use the notation of Proposition~\usref{PropEstRadOut} and assume~\eqref{EqEstRadOutThr}, except now we define $\cK:=\{\rho_\infty=\hat\rho=\rho_\circ=0,\ \hat\xi_\seop=1,\ t_0\leq t\leq t_1\}\subset\Sse^*\wt M$. Let $k\in\N_0$, and assume that $\wt L$ satisfies~\eqref{EqEstRadsOp} for some sufficiently large $d_0=d_0(s,s_0,N)$. Then for all neighborhoods $\cU\subset\Sse^*\wt M$ of $\cK$ and all $\chi\in\CIc(\wt M)$ equal to $1$ near the base projection $\pa M_\circ\cap\{t_0\leq t\leq t_1\}$ of $\cK$, there exist operators $B,E,G\in\Psi_\seop^0(\wt M)$ so that $\chi B\chi=B$ etc., the operator wave front sets of $B,E,G$ are contained in $\cU$, furthermore $B$ is elliptic at $\cK$, and $\WF'_\seop(E)$ is disjoint from $\pa W_{\rm out}^+$, and such that, for some constant $C$, the estimate
  \begin{align*}
    \|B u\|_{H_{(\seop;\sop),\eps}^{(s;k),\alpha_\circ,\hat\alpha}} &\leq C\Bigl( \|G\wt L u\|_{H_{(\seop;\sop),\eps}^{(s-1;k),\alpha_\circ,\hat\alpha-2}} + \|E u\|_{H_{(\seop;\sop),\eps}^{(s;k),\alpha_\circ,\hat\alpha}} \\
      &\quad\hspace{5em} + \|G u\|_{H_{(\seop;\sop),\eps}^{(s_0;k),\alpha_\circ,\hat\alpha}} + \|\chi u\|_{H_{(\seop;\sop),\eps}^{(-N;k),\alpha_\circ,\hat\alpha}}\Bigr)
  \end{align*}
  holds (in the strong sense) for all $u$ with support in $t\geq t_0$.
\end{prop}
\begin{proof}
  As in the proof of Proposition~\ref{PropEstRadsIn}, the case $k=0$ is a direct consequence of~\eqref{EqEstRadOut}, or alternatively it can be proved by replacing the sharp temporal cutoff in $\frac{t-t_1}{r}$ in the proof of Proposition~\ref{PropEstRadOut} by a cutoff in $t-t_1$.

  For $k=1$, we use the notation of Lemma~\ref{LemmaEstRadsComm} and note that
  \begin{equation}
  \label{EqEstRadsOutL1}
    \wt L^{(1)}(\pa_t u) = f^{(1)} := \pa_t f + R u,\qquad \wt L^{(1)}=\wt L-\sum_{j=1}^N a_j A_j\hat\rho,\quad R = \sum_{j=1}^N a_j R_j.
  \end{equation}
  Since $\wt L^{(1)}$ differs from $\wt L$ by an operator whose coefficients vanish (relative to $\hat\rho^{-2}$) at $\pa M_\circ$ and thus at $\pa\cR_{\rm out}^+$, the inductive hypothesis applies to this equation. (The presence of a pseudodifferential first order term in $\wt L^{(1)}$ is the reason for working with~\eqref{EqEstRadsOp}.) We choose $\cU$ so small that $\hat\rho\pa_t$ is elliptic on $\cU$. By~\eqref{EqEstRadsNorm} and using $[B,\pa_t]\in\Psi_\seop^0$, we have
  \[
    \|B u\|_{H_{(\seop;\sop),\eps}^{(s;1),\alpha_\circ,\hat\alpha}} \leq C\Bigl(\|\tilde B u\|_{H_{\seop,\eps}^{s,\alpha_\circ,\hat\alpha}} + \|B\pa_t u\|_{H_{\seop,\eps}^{s,\alpha_\circ,\hat\alpha}}\Bigr)
  \]
  where $\tilde B\in\Psi_\seop^0$ is elliptic near $\WF_\seop'(B)$. Similarly,
  \[
    \|G f^{(1)}\|_{H_{\seop,\eps}^{s-1,\alpha_\circ,\hat\alpha-2}} \leq \|\pa_t G f\|_{H_{\seop,\eps}^{s-1,\alpha_\circ,\hat\alpha-2}} + \|[\pa_t,G]f\|_{H_{\seop,\eps}^{s-1,\alpha_\circ,\hat\alpha-2}} + C\|\chi u\|_{H_{\seop,\eps}^{-N,\alpha_\circ,\hat\alpha-2}},
  \]
  where the last term arises from $G R\in\Psi_\seop^{-\infty}(\wt M)$ provided we arrange $G R=\chi G R\chi$, which indeed holds if $R$ has Schwartz kernel supported sufficiently close to the diagonal. The first two terms are bounded by $\|G f\|_{H_{(\seop;\sop),\eps}^{(s-1;1),\alpha_\circ,\hat\alpha-2}}+\|\tilde G f\|_{H_{\seop,\eps}^{s-1,\alpha_\circ,\hat\alpha-2}}$ where $\tilde G\in\Psi_\seop^0$ is elliptic near $\WF_\seop'(G)$. Arguing similarly for the terms $E(\pa_t u)$, $G(\pa_t u)$, and $\chi\pa_t u$, one finishes the inductive step.
\end{proof}

\subsection{Microlocal propagation near the event horizon}
\label{SsEstHor}

Turning now to the future component $\cR_{\cH^+}^+$ of the generalized radial set over the event horizon (see~\eqref{EqGlDynStdRadHor}), we work with $t\in I_\cC$ and the radial coordinate $\hat r$ on $\hat X_b\cong\hat M_t$. As fiber coordinates on $\Tse^*\wt M$ near $\cR_{\cH^+}^+$, we may take $\sigma_0,\xi_0,\eta_\theta,\eta_{\phi_0}$ from~\eqref{EqGlDynKerrHor0} via the identification~\eqref{EqFseBundle3b}. Thus, $\hat r\approx\hat r_b$ and $\xi_0>0$ near $\cR_{\cH^+}^+$. We denote (by a mild abuse of notation) by $\rho_{\cH^+}^2$ the quadratic defining function of $\pa\cR_{\cH^+}^+$, smooth on $\Sse^*\wt M$ near $\pa\cR_{\cH^+}^+$, which is defined by the expression~\eqref{EqGlDynKerrHorrho2}. We moreover set
\[
  \rho_\infty := \xi_0^{-1}
\]
as in~\eqref{EqGlDynKerrHamProj}.

Due to the (normal) source structure of $\pa\hat\cR_{\cH^+}^+\subset\Sse^*_{\hat M_t}\wt M$ for each $t$, also the set $\pa\cR_{\cH^+}^+\subset\Sse^*_{\hat M}\wt M$ is a normal source for the $\hat\rho^2 H_{\wt G}$ flow (which is, after all, tangent to each fiber of $\hat M$). Thus, the radial point estimates of \cite[\S{2.4}]{VasyMicroKerrdS} are almost directly applicable with $t$ as a parameter, except as soon as $\eps>0$ the function $t$ is not constant (and indeed increasing) along the null-bicharacteristic flow in $\wt\Sigma^+$. Similarly to Proposition~\ref{PropEstRadIn}, we thus insert a factor which localizes in $t$, and indeed sharply so near $t=t_1$ using~\eqref{EqEstBlowup}.

\begin{prop}[Generalized radial point estimate near $\cR_{\cH^+}^+$]
\label{PropEstHor}
  With respect to a positive definite inner product on $\wt E$ which over $\hat M$ is the pullback of a fiber inner product on $\hat\cE$, define the quantity
  \begin{equation}
  \label{EqEstHorTheta}
  \begin{split}
    \vartheta_{\cH^+} &:= \sup_{\pa\cR_{\cH^+}^+} \spec\Bigl[\frac{\varrho^2}{2(\hat r_b-\bhm)}\rho_\infty\sigmase^1\Bigl(\frac{\wt L-\wt L^*}{2 i}\Bigr)\Bigr] \\
      &= \sup_{\pa\hat\cR_{\cH^+}^+} \spec\Bigl[\frac{\varrho^2}{2(\hat r_b-\bhm)}\rho_\infty\sigma^1\Bigl(\frac{L-L^*}{2 i}\Bigr)\Bigr] \in \R.
  \end{split}
  \end{equation}
  Let $s,s_0,N,\alpha_\circ,\hat\alpha\in\R$, and suppose that
  \begin{equation}
  \label{EqEstHorThr}
    s > s_0 > \frac12(1+\vartheta_{\cH^+}).
  \end{equation}
  Define $\cK:=\upbeta^{-1}(\pa\cR_{\cH^+}^+)\cap\{t\geq t_0,\ \frac{t-t_1}{\eps}\leq 0\}$. Then for all neighborhoods $\cU\subset\upbeta^*(\Sse^*\wt M)$ of $\cK$ and all $\chi\in\CIc([\wt M;\hat M_{t_1}])$ equal to $1$ near the base projection of $\cK$, there exist operators $B,G\in\tilde\Psi_\seop^0(\wt M)$ with $\chi B\chi=B$ and $\chi G\chi=G$, and so that
  \begin{itemize}
  \item the operator wave front sets of $B,G$ are contained in $\cU$,
  \item $B$ is elliptic at $\cK$,
  \end{itemize}
  and a constant $C$ so that the estimate
  \begin{equation}
  \label{EqEstHor}
    \|B u\|_{H_{\seop,\eps}^{s,\alpha_\circ,\hat\alpha}} \leq C\Bigl(\|G\wt L u\|_{H_{\seop,\eps}^{s-1,\alpha_\circ,\hat\alpha-2}} + \|G u\|_{H_{\seop,\eps}^{s_0,\alpha_\circ,\hat\alpha}} + \|\chi u\|_{H_{\seop,\eps}^{-N,\alpha_\circ,\hat\alpha}}\Bigr)
  \end{equation}
  holds (in the strong sense) for all $u$ with support in $t\geq t_0$.
\end{prop}

In~\eqref{EqEstHor}, the weight $\alpha_\circ$ is irrelevant since the supports of all distributions involved are disjoint from $M_\circ$, so we may take $\alpha_\circ=0$. Similarly, the weight $\hat\alpha$ is irrelevant since near $\hat r=\hat r_b$ we may take $\hat\rho=\eps$; but multiplication by functions of $\eps$ commutes with every se-operator, and so we may simply divide $u$ by $\hat\rho^{\hat\alpha}$ to reduce to $\hat\alpha=0$, say.

\begin{proof}[Proof of Proposition~\usref{PropEstHor}]
  Since the proof is very similar to that of Proposition~\ref{PropEstRadIn}, we shall be very brief. We use the expression~\eqref{EqGlDynKerrHamLot} for the vector field $\sfH'=\rho_\infty H_{\varrho^2\hat G_b}$, which on the characteristic set over $\eps=0$ equals $\varrho^2\sfH$ where we set $\sfH=\eps^2\rho_\infty H_{\wt G}$. With $\alpha_\circ=\hat\alpha=0$, we take the commutant to have symbol $a=\check a^2$ where
  \[
    \check a = \rho_\infty^{-s+\frac12}\chi(\rho_{\cH^+}^2)\chi_-(t-t_0)\chi_+\Bigl(\frac{t-t_1}{\eps}\Bigr).
  \]
  Here, $\chi_\pm$ are as in~\eqref{EqEstRadInChipm}. Moreover, $\chi\in\CIc([0,\delta_0))$ is equal to $1$, with $\delta_0>0$ so small that $\sfH\rho_{\cH^+}^2$ is larger than a positive multiple of $\rho_{\cH^+}^2$ on $\supp\check a\cap\{\eps=0\}$ (cf.\ the discussion after~\eqref{EqGlDynKerrHorrho2}); with $\chi$ fixed, the function $\sfH\rho_{\cH^+}^2$ therefore has a positive lower bound on $\supp\check a\cap\supp\dd\chi$ for all sufficiently small $\eps\geq 0$. As in the proof of Proposition~\ref{PropEstRadIn}, the past timelike nature of $\dd t$ implies that the derivatives $\sfH\chi_-$ and $\sfH\chi_+$ are positive, resp.\ negative squares near $\supp\check a$; this uses the final part of Lemma~\ref{LemmaGlCoord}. The principal symbol $c$ of the operator $\varrho^2\rho_\infty^{2 s}\sC$, with $\sC$ defined by~\eqref{EqEstRadInComm} evaluates at $\pa\cR_{\cH^+}^+$ to
  \[
    \mu'(\hat r_b)(-2 s+1) + 2\ell_1
  \]
  where $\ell_1=\varrho^2\rho_\infty\sigma^1(\frac{L-L^*}{2 i})$. When this is negative (thus the derivative of $\chi_+$ has the same sign, and is dropped in~\eqref{EqEstHor}), we obtain the desired estimate~\eqref{EqEstHor} for all sufficiently small $\eps$. From this calculation, we read off the threshold quantity in~\eqref{EqEstHorThr} from $\mu'(\hat r_b)=2(\hat r_b-\bhm)$.
\end{proof}

\begin{rmk}[From the perspective of the small Kerr black hole: horizon]
\label{RmkEstHorKerr}
  Recall from Lemma~\ref{LemmaFseRelGeo} the relationship between $[\wt M;\hat M_{t_0}]$ (minus the lift of $\wt K^\circ$) and $[0,1)_\eps$ times the Kerr spacetime manifold $\hat M_b$, given in the coordinates $(\eps,t,\hat x)$ near $\hat M^\circ$ by $(\eps,t,\hat x)\mapsto(\eps,\hat t,\hat x)=(\eps,\frac{t-t_0}{\eps},\hat x)$. Then the time interval on which Proposition~\ref{PropEstHor} applies for fixed $\eps>0$ is $0\leq\hat t\leq\eps^{-1}t_1$ (or $\hat t\leq\eps^{-1}t_1+\delta$ where $\delta>0$ captures the size of $\cU$, which extends a bit past $\frac{t-t_1}{\eps}=0$). Recall that for all $\eps>0$ all null-bicharacteristics over $\Omega_\eps$ exit the domain in the forward and backward direction in a finite amount affine time by Lemma~\ref{LemmaGlDynStdFlow}\eqref{ItGlDynStdFlowXY}. But this amount scales like $\eps^{-1}$. The uniform (in $\eps$) control of $u$ microlocally near $\{0\leq\hat t\leq\eps^{-1}t_1\}\cap N^*\{\hat r=\hat r_b\}\setminus o$ (in $L^2$, together with 3b-derivatives, which near there are simply coordinate derivatives $\pa_{\hat t}$, $\pa_{\hat x}$) thus ultimately needs to take into account the null-bicharacteristic dynamics near the event horizon from $\hat t=0$ to $\hat t=\infty$; and indeed in the compactification $\hat M_b$ of the Kerr spacetime manifold the appropriate notion of null-bicharacteristic flow (namely, the flow in $\Ttb^*\hat M_b\setminus o$) has an invariant (generalized radial) set over $\hat r=\hat r_b$, $\hat t^{-1}=0$. In terms of the \emph{exponential} compactification $[0,1)_{e^{-\hat t}}\times\hat X_b^\circ$ of spatially bounded regions in time, this set is precisely the generalized radial set at future infinity already studied (in the de~Sitter and Kerr--de~Sitter context, near their cosmological and event horizons) from this perspective in \cite[Proposition~2.1]{HintzVasySemilinear}. While we do not use such an exponential compactification in the present asymptotically flat setting, one may nonetheless regard Proposition~\ref{PropEstHor} as a version of \cite[Proposition~2.1]{HintzVasySemilinear} giving an estimate only up to time $\eps^{-1}$, but with uniform norm bounds as $\eps\searrow 0$. An important difference however is that the metric itself depends on $\eps$ and is, in fact, only $\eps$-close to the Kerr metric ($\eps^\delta$-close, for any $\delta>0$, would be sufficient).
\end{rmk}

\subsubsection{Higher s-regularity}
\label{SssEstHors}

We now proceed as in~\S\ref{SssEstRads} and prove an analogue of Proposition~\ref{PropEstHor} on mixed (se;s)-Sobolev spaces, under the regularity assumption~\eqref{EqEstRadsOp} on $\wt L$. Since $\eps\pa_t$ is \emph{not} elliptic at $\pa\cR_{\cH^+}$, we cannot quite argue as in Proposition~\ref{PropEstRadsOut} using Lemma~\ref{LemmaEstRadsComm}; but a simple modification will suffice: the observation is that $\pa_{\hat r}$ \emph{is} elliptic at $\pa\cR_{\cH^+}$, and commuting $\wt L$ with $\pa_{\hat r}$ results in a subprincipal term ($A'_0$ in Lemma~\ref{LemmaEstHorsComm} below) with the correct sign so as to not affect the threshold condition~\eqref{EqEstHorThr}. This can be regarded as a manifestation of the \emph{enhanced red-shift effect}, which first appeared in \cite{DafermosRodnianskiKerrBoundedness}, and which (in form of the monotonicity of the principal symbol $\xi_0$ of $\pa_{\hat r}$ under the $H_{\hat G_b}$-flow) underlies the microlocal radial point estimate proved above or in \cite[Equation~(2.3)]{VasyMicroKerrdS}.

Thus, analogously to~\eqref{EqEstRadsNorm}, we have the following norm equivalences for $B\in\Psi_\seop^0(\wt M)$ with $\WF_\seop'(B)\cap\Char_\seop(\pa_{\hat r})=\emptyset$ and Schwartz kernel supported in a region of bounded $\hat r$ in both factors:
\begin{equation}
\label{EqEstHorsNorm}
  \|B u\|_{H_{(\seop;\sop),\eps}^{(s;k)}} \sim \|B u\|_{H_{(\seop;\sop),\eps}^{(s;k-1)}} + \|\pa_t B u\|_{H_{(\seop;\sop),\eps}^{(s;k-1)}} + \|\pa_{\hat r}B u\|_{H_{(\seop;\sop),\eps}^{(s;k-1)}} \sim \sum_{j+l\leq k} \|\pa_t^j\pa_{\hat r}^l B u\|_{H_{\seop,\eps}^s};
\end{equation}
similarly for weighted spaces. The analogue of Lemma~\ref{LemmaEstRadsComm} is:

\begin{lemma}[Commutators with $\pa_t$ and $\pa_{\hat r}$]
\label{LemmaEstHorsComm}
  Let $K\subset\Sse^*\wt M$ be a compact subset which is disjoint from $\Char_\seop(\pa_{\hat r})$. Then we can write
  \begin{align*}
    [\wt L,\pa_t] &= \sum_{j=1}^N a_j(A_j\pa_{\hat r} + R_j), \\
    [\wt L,\pa_{\hat r}] &= \hat\rho^{-2}(i^{-1}A'_0\pa_{\hat r} + R'_0) + a'_0\wt L +  \sum_{j=1}^N a'_j(A'_j\pa_{\hat r} + R'_j)
  \end{align*}
  where $a_j,a'_j\in\hat\rho^{-2}(\hat\rho\CI+\cC_{\seop;\sop}^{(d_0;k-1),1,1})$, $a'_0\in\CI$, $A_j,A'_0,A'_j\in\Psi_\seop^1(\wt M)$, $R_j,R'_0,R'_j\in\Psi_\seop^2(\wt M)$, the operator wave front sets of $R_j,R'_0,R'_j$ are disjoint from $K$, and
  \[
    \sigmase^1(A'_0)|_{\cR_{\cH^+}^+}<0.
  \]
\end{lemma}
\begin{proof}
  The commutator $[\wt L,\pa_t]$ as well as the commutator with $\pa_{\hat r}$ of any contribution to $\wt L$ of class $\hat\rho^{-2}(\hat\rho\CI+\cC_{\seop;\sop}^{(d_0;k),1,1})\Diffse^2$ is handled as in the proof of Lemma~\ref{LemmaEstRadsComm}, with $\hat\rho\pa_t$ replaced by $\pa_{\hat r}$; these give rise to $a_j,A_j,R_j$ and $a'_j,A'_j,R'_j$. The terms $a_0',A'_0,R'_0$ only arise from the $\hat M$-normal operator of $\wt L$ (and thus have smooth coefficients). Concretely, using~\eqref{EqGlDynKerrHam} and recalling that $\sigmase^1(\pa_{\hat r})=i\xi_0$, we compute
  \[
    H_{\hat G_b}\xi_0 = \varrho^{-2}H_{\varrho^2\hat G_b}\xi_0 + \varrho^2\hat G_b H_{\varrho^{-2}}\xi_0 = -\varrho^{-2}(\mu'\xi_0-4\hat r\sigma_0)\xi_0 + \varrho^2\hat G_b H_{\varrho^{-2}}\xi_0.
  \]
  Since $\eps^2\wt G=\hat G_b$ over $\hat M$ and $\mu'(\hat r_b)=2(\hat r_b-\bhm)>0$, this implies the claim: the term $a'_0$ arises from $\varrho^2 H_{\varrho^{-2}}\xi_0$, while $A'_0$ (with principal symbol $-\varrho^{-2}(\mu'\xi_0-4\hat r\sigma_0)$, which equals $2\varrho^{-2}(\hat r_b-\bhm)\xi_0<0$ at $\cR_{\cH^+}^+$) and $R'_0$ arise via a microlocal parametrix construction for $\pa_{\hat r}$ on $K$.
\end{proof}

\begin{prop}[Generalized radial point estimate near $\cR_{\cH^+}^+$: s-regularity]
\label{PropEstHors}
  We use the notation of Proposition~\usref{PropEstHor}, except now we define $\cK=\pa\cR_{\cH^+}^+\cap\{t_0\leq t\leq t_1\}$. Let $k\in\N_0$, and assume that $\wt L$ satisfies~\eqref{EqEstRadsOp} for some sufficiently large $d_0=d_0(s,s_0,N)$. Then for all neighborhoods $\cU\subset\Sse^*\wt M$ of $\cK$ and all $\chi\in\CIc(\wt M)$ equal to $1$ near the base projection $\{\hat r=\hat r_b,\ t_0\leq t\leq t_1\}$ of $\cK$, there exist operators $B,G\in\Psi_\seop^0(\wt M)$ with $\chi B\chi=B$ and $\chi G\chi=G$ and with operator wave front sets contained in $\cU$ and with $B$ elliptic at $\cK$, so that, for some constant $C$, the estimate
  \[
    \|B u\|_{H_{(\seop;\sop),\eps}^{(s;k),\alpha_\circ,\hat\alpha}} \leq C\Bigl(\|G\wt L u\|_{H_{(\seop;\sop),\eps}^{(s-1;k),\alpha_\circ,\hat\alpha-2}} + \|G u\|_{H_{(\seop;\sop),\eps}^{(s_0;k),\alpha_\circ,\hat\alpha}} + \|\chi u\|_{H_{(\seop,\sop),\eps}^{(-N;k),\alpha_\circ,\hat\alpha}}\Bigr)
  \]
  holds (in the strong sense) for all $u$ with support in $t\geq t_0$.
\end{prop}
\begin{proof}
  The case $k=0$ follows from Proposition~\ref{PropEstHor} as in the proof of Proposition~\ref{PropEstRadsIn}. For $k\geq 1$, we use induction. We use the notation of Lemma~\ref{LemmaEstHorsComm} to deduce from $\wt L u=f$ the equation $\wt L^{(1)}u^{(1)}=f^{(1)}$ where
  \begin{align*}
    \wt L^{(1)} = 
      \begin{pmatrix}
        \wt L & -\sum_{j=1}^N a_j A_j \\
        0 & \wt L + i\hat\rho^{-2}A'_0 - \sum_{j=1}^N a'_j A'_j
      \end{pmatrix}, \quad
    &\wt u^{(1)} = \begin{pmatrix} \pa_t u \\ \pa_{\hat r}u \end{pmatrix}, \\
    &\wt f^{(1)} = \begin{pmatrix} \pa_t f + \sum_{j=1}^N a_j R_j u \\ \pa_{\hat r}f + a'_0 f + \hat\rho^{-2}R'_0 u + \sum_{j=1}^N a'_j R'_j u \end{pmatrix}.
  \end{align*}
  This is an equation for sections of $\wt E\oplus\wt E$. Using the fixed inner product on $\wt E$ on both summands, we now observe that the threshold quantity~\eqref{EqEstHorTheta} for $\wt L^{(1)}$ is equal to that of $\wt L$ due to the negativity of $\sigmase^1(A'_0)$ at $\cR_{\cH^+}^+$. Applying the inductive hypothesis to this equation thus completes the inductive step for $\wt L$ (upon enlarging the elliptic set of $G$ and the support of $\chi$ slightly).
\end{proof}

\subsection{Microlocal estimates near normally hyperbolic trapping}
\label{SsEstTrap}

We next turn to the most delicate microlocal aspect of the present paper: the proof of uniform se-estimates governing the propagation of quantitative se-regularity from the stable manifold $\Gamma^{\rm s,+}\setminus\Gamma^+$ into the trapped set $\Gamma^+=\Gamma^{\rm s,+}\cap\Gamma^{\rm u,+}$ (using the notation~\eqref{EqGlDynStdTrap}). (Again we only focus on the future half $\wt\Sigma^+$ of the characteristic set, the arguments in $\wt\Sigma^-$ being completely analogous.) Recall from~\eqref{EqGlDynKerrTrapSigma} that the fiber variable $\sigma$ in Definition~\ref{DefGlDynKerrTrap}, which is also a smooth fiber-linear function on $\Tse^*\wt M$ near $\Gamma^+$, is positive on $\Gamma^+$; we shall thus use
\[
  \rho_\infty = \sigma^{-1}
\]
for the purpose of shifting differential orders of symbols. We use $\wt\varphi^{\rm s/u}$, $\wt w^{\rm s/u}$ from Proposition~\ref{PropTrap}, and use the same notation for their extensions to homogeneous degree $0$ symbols on $\Tse^*\wt M\setminus o$. By multiplying $\wt\varphi^{\rm s/u}$ and $\varphi^{\rm s/u}$ by a sufficiently large constant, we may assume that the conclusions of Proposition~\ref{PropTrap} as well as the positivity
\begin{equation}
\label{EqEstTrapPoisson}
  \rho_\infty^{-1}\{\wt\varphi^{\rm u},\wt\varphi^{\rm s}\}>0
\end{equation}
(which can be arranged by replacing $\wt\varphi^{\rm s}$ by $-\wt\varphi^{\rm s}$ if necessary) are valid in a neighborhood
\[
  \pa\wt\Sigma^+ \cap \{ |\wt\varphi^{\rm u}| < 6,\,|\wt\varphi^{\rm s}| < 6 \}
\]
of the trapped set. In the spirit of \cite{DyatlovSpectralGaps} and \cite[\S{3.2.2}]{HintzPolyTrap}, we recall~\eqref{EqGlDynKerrInvDefFn} and set
\begin{equation}
\label{EqEstTrapNuMin}
  \nu_{\rm min} := \min\Bigl\{\,\inf_{\pa\Gamma^+} w_b^{\rm s}, \inf_{\pa\Gamma^+} w_b^{\rm u} \Bigr\} > 0.
\end{equation}

The validity of estimates at normally hyperbolic trapping for the operator $L$ on Kerr (and therefore also for $\wt L$) requires a condition on the subprincipal symbol of $L$ at $\pa\Gamma^+$. Since $L$ acts on sections of the vector bundle $\hat\pi^*\hat\cE\to\hat M_b$ (with $\hat\pi\colon\hat M_b\to\hat X_b$ as in~\eqref{ItEstLOp}), we shall phrase this using the \emph{subprincipal operator} introduced in \cite{HintzPsdoInner} (see also \cite{DenckerPolarization}); in any local trivialization of $\hat\cE$, and trivializing the half-density bundle using the metric density $|\dd\hat g_b|$, this is defined as
\begin{equation}
\label{EqEstTrapSsub}
  S_{\rm sub}(L) = -i H_{\hat G_b} + \sigma_{\rm sub}(L)
\end{equation}
where $\sigma_{\rm sub}(L)$ is the matrix of subprincipal symbols \cite[\S5.2]{DuistermaatHormanderFIO2}. It is well-defined as an operator
\[
  S_{\rm sub}(L) \in \Diff^1(T^*\hat M_b^\circ;\pi^*(\hat\pi^*\hat\cE)),
\]
where $\pi\colon T^*\hat M_b^\circ\to\hat M_b^\circ$ is the base projection, and commutes with time translations.

For better readability, we introduce
\[
  \cL := \eps^2\wt L \in \Diffse^{2,-2,0}(\wt M\setminus\wt K^\circ;\wt E),\qquad
  \cG := \eps^2\wt G;
\]
these are unweighted at $\hat M$. We again set $\tau:=\frac{t-t_1}{\eps}$, and recall~\eqref{EqEstBlowup}.

\begin{prop}[Uniform estimate near the trapped set]
\label{PropEstTrap}
  Suppose there exists a stationary positive definite fiber inner product on the pullback bundle $\pi^*(\hat\pi^*\hat\cE)\to T^*_{\R_{\hat t}\times\hat\Gamma_b^+}\hat M_b^\circ$ so that $\frac{1}{2 i}(S_{\rm sub}(L)-S_{\rm sub}(L)^*)<\frac12\nu_{\rm min}$.\footnote{Here we use that $S_{\rm sub}(L)$ is a transport operator in the direction of $H_{\hat G_b}$, and thus one can compute its adjoint given only the inner product on $\pi^*(\hat\pi^*\hat\cE)$ over the trapped set. Moreover, the scalar operator $-i H_{\hat G_b}$ is formally self-adjoint, so $\frac{1}{2 i}(S_{\rm sub}(L)-S_{\rm sub}(L)^*)$ is of order $0$ and indeed a stationary self-adjoint bundle endomorphism of $\pi^*(\hat\pi^*\hat\cE)$. See \cite{HintzPsdoInner} for details.} Let $s,N,\alpha_\circ,\hat\alpha\in\R$. Then for all neighborhoods $\cU\subset\upbeta^*(\Sse^*\wt M)$ of $\upbeta^*(\pa\Gamma^+)\cap\{t\geq t_0,\ \tau\leq 0\}$ and all $\chi\in\CIc([\wt M;\hat M_{t_1}])$ equal to $1$ near the base projection of $\upbeta^*(\pa\Gamma^+)$, there exist operators $B_\Gamma,B^{\rm s},G\in\tilde\Psi_\seop^0$ with $B_\Gamma=\chi B_\Gamma\chi$ etc.\ so that
  \begin{itemize}
  \item the operator wave front sets of $B_\Gamma,B^{\rm s},G$ are contained in $\cU$,
  \item $B_\Gamma$ is elliptic at $\Gamma^+\cap\{t\geq t_0,\ \tau\leq 0\}$,
  \item the operator wave front set of $B^{\rm s}$ is disjoint from $\Gamma^{\rm u,+}$,
  \end{itemize}
  and a constant $C$ so that the estimate
  \begin{equation}
  \label{EqEstTrap}
    \|B_\Gamma u\|_{H_{\seop,\eps}^{s,\alpha_\circ,\hat\alpha}} \leq C\Bigl( \|G\wt L u\|_{H_{\seop,\eps}^{s,\alpha_\circ,\hat\alpha-2}} + \|B^{\rm s}u\|_{H_{\seop,\eps}^{s+1,\alpha_\circ,\hat\alpha}} + \|\chi u\|_{H_{\seop,\eps}^{-N,\alpha_\circ,\hat\alpha}}\Bigr)
  \end{equation}
  holds (in the strong sense) for all $u$ with support in $t\geq t_0$.
\end{prop}

The elliptic set of $B^{\rm s}$ will control $u$ microlocally in a punctured neighborhood of $\Gamma^+$ in $\Gamma^{\rm s,+}$ via standard real principal type propagation. The estimate~\eqref{EqEstTrap} then extends this control into the trapped set (with a loss of two derivatives relative to elliptic estimates).

Since $\Gamma^+$ lies over the interior of $\hat M^\circ$, we only need to consider the case $\alpha_\circ=\hat\alpha=0$ (as in~\S\ref{SsEstHor}). Since $\wt L=\eps^{-2}\cL$, we thus need to prove
\begin{equation}
\label{EqEstTrap2}
    \|B_\Gamma u\|_{H_{\seop,\eps}^s} \leq C\Bigl( \|G\cL u\|_{H_{\seop,\eps}^s} + \|B^{\rm s}u\|_{H_{\seop,\eps}^{s+1}} + \|\chi u\|_{H_{\seop,\eps}^{-N}}\Bigr).
\end{equation}

\begin{rmk}[From the perspective of the small Kerr black hole: trapping]
\label{RmkEstTrapKerr}
  Analogously to Remark~\ref{RmkEstHorKerr}, one can regard this as a uniform estimate for long time intervals $0\leq\hat t\leq\eps^{-1}t_1$ which mirrors the global-in-time trapping estimate proved in \cite[Theorem~3.9]{HintzPolyTrap}. The reference uses weighted cusp Sobolev spaces, which in present notation measure $L^2$-regularity with respect to $\pa_{\hat t}$, $\pa_{\hat x}$, which is thus again compatible with 3b-estimates in bounded subsets of $\hat M_b$ and thus with se-estimates on subsets of $\wt M$ on which $\hat x$ is bounded. Unlike in the reference, however, we are working here on size $\cO(\eps)$ perturbations of the stationary Kerr spacetime, while on the other hand we are not localizing to $\hat t^{-1}\ll 1$.
\end{rmk}

Our proof of Proposition~\ref{PropEstTrap} will closely follow Dyatlov's strategy \cite{DyatlovSpectralGaps}, albeit in a version more akin to the spacetime form described in \cite[\S{3}]{HintzPolyTrap}. The novel aspects in our proof are, firstly, the need for careful time localization in the propagation estimate for the auxiliary equation for $\Op(\wt\varphi^{\rm u})u$ (see~\eqref{EqEstTrapChimp}), and, secondly, the usage of the monotonicity of $t$ (in addition to the unstable nature of $\Gamma^{\rm u}$) along the $H_\cG$-flow in the propagation estimate for $\cL$ from a small neighborhood of the stable trapped set to a size $1$ part of the unstable trapped set (see in particular~\eqref{EqEstTrapFinalTime} and its subsequent combination with~\eqref{EqEstTrapLogProp}).

As in \cite[\S{3}]{HintzPolyTrap}, the main technical tools are quantitative real principal type propagation statements of the type \cite[(3.5)]{HintzPolyTrap} akin to the G\aa{}rding inequality:

\begin{lemma}[Quantitative bound]
\label{LemmaEstTrapQuant}
  Let $B,B'\in\tilde\Psi_\seop^s(\wt M)$, and suppose that, as subsets of $\upbeta^*\Sse^*\wt M$, we have $\WF'_\seop(B')\subset\Ell_\seop(B)$. Suppose moreover that the principal symbols $b,b'$ satisfy $|b'|\leq b$ on $\WF'_\seop(B')$. Then for all $\delta>0$, there exists a constant $C$ so that, for all $\eps>0$,
  \[
    \|B' u\|_{H_{\seop,\eps}^0}^2 \leq (1+\delta)\|B u\|_{H_{\seop,\eps}^0}^2 + C\|u\|_{H_{\seop,\eps}^{-\frac12}}^2.
  \]
\end{lemma}
\begin{proof}
  (See \cite[Lemma~3.7]{HintzPolyTrap}.) Since the principal symbol of $(1+\delta)^2 B^*B-(B')^*B'$ has a smooth positive square root $e$, we have $R:=(1+\delta)B^*B-(B')^*B^*-E^*E\in\tilde\Psi_\seop^{2 s-1}$ where $E\in\tilde\Psi_\seop^s$ is a quantization of $e$. The uniform boundedness of $R\colon H_{\seop,\eps}^{s-\frac12}\to H_{\seop,\eps}^{-s+\frac12}=(H_{\seop,\eps}^{s-\frac12})^*$ implies the desired estimate.
\end{proof}

\begin{proof}[Proof of Proposition~\usref{PropEstTrap}]
  Fix any positive definite fiber inner product on $\wt E$ which over $\hat M$ is the pullback of a fiber inner product on $\hat\cE$; all adjoints in the argument below are computed with respect to this inner product.

  \pfstep{Pseudodifferential conjugation of $\cL$.} In order to exploit the subprincipal symbol condition, we use \cite[Proposition~3.12]{HintzPsdoInner} to find a time-translation invariant elliptic operator $Q\in\Psi_{\tbop,{\rm I}}^0(\hat M_b;\hat\pi^*\hat\cE)$ so that $\rho_\infty\sigmatb^1(\frac{1}{2 i}(Q L Q^--(Q L Q^-)^*))<\frac12\nu_{\rm min}$ on $\pa\hat\Gamma_b$, where $Q^-$ is a parametrix of $Q$. We can then quantize the 3b-principal symbol of $Q$ as a $t$-independent se-principal symbol via~\eqref{EqFseBundle3b}, or directly pull back $Q$ (as an $\eps$-independent operator) along the map~\eqref{EqFseRelGeoPsi0}--\eqref{EqFseRelGeoPsi} to obtain an elliptic operator $\cQ\in\tilde\Psi_\seop^0$ with the property that
  \[
    \rho_\infty\sigmase^1\Bigl(\frac{1}{2 i}\Bigl(\cQ\cL\cQ^--(\cQ\cL\cQ^-)^*\Bigr)\Bigr)<\frac12\nu_{\rm min}
  \]
  at $\pa\Gamma^+$; here $\cQ^-$ is a parametrix of $\cQ$ with $\cQ^-\cQ=I+R$, $R\in\tilde\Psi_\seop^{-\infty}$. Note that for any fixed $\delta_0>0$, we can choose $\cQ,\cQ^-,\cR$ (which are uniform ps.d.o.s) to enlarge supports of distributions supported near a compact subset of $\hat M^\circ$ (here: the base projection of the trapped set) only by amounts $\leq\delta_0\eps$ in $t$ and $\leq\delta_0$ in $\hat x$. Furthermore, $\cQ\cL\cQ^-(\cQ u)=\cQ\cL u+\cQ\cL R u$, and the $H_{\seop,\eps}^s$ norm of $\cQ\cL R u$ on some set can be bounded by the $H_{\seop,\eps}^{-N}$ norm of $u$ on a slightly larger set. Altogether, we may thus replace $\cL$ and $u$ by $\cQ\cL\cQ^-=\cL+\cQ[\cL,\cQ^-]$ and $\cQ u$, respectively; the goal is still to prove the estimate~\eqref{EqEstTrap2}, but now $\cL\in\tilde\Psi_\seop^2$ (in fact $\cL\in\Diffse^2+\tilde\Psi_\seop^1$) is pseudodifferential, has scalar principal symbol $\cG$, and satisfies
  \begin{equation}
  \label{EqEstTrapL1}
    \rho_\infty \ell_1 < \frac12\nu_{\rm min}\ \ \text{at}\ \ \pa\Gamma,\qquad \ell_1:=\sigmase^1(\cL_1),\ \ \cL_1:=\frac{1}{2 i}(\cL-\cL^*);
  \end{equation}
  and $u$ vanishes in $t\leq t_0-\delta_0\eps$ and thus a fortiori in $t\leq t_0-\delta_0$.

  \pfstep{Equation for $\wt\Phi^{\rm u}u$.} In view of~\eqref{EqTrapHamPhisu}, we have $\rho_\infty H_\cG\wt\varphi^{\rm u}=-\wt w^{\rm u}\wt\varphi^{\rm u}+\rho_\infty^2\cG\wt\psi$ where the smooth function $\wt\psi$ on $\Sse^*\wt M$ (also regarded as a homogeneous degree $0$ function on $\Tse^*\wt M\setminus o$), defined in a neighborhood of $\{|\wt\varphi^{\rm u}|<6,\,|\wt\varphi^{\rm s}|<6\}\cap\wt\Sigma$, accounts for the fact that~\eqref{EqTrapHamPhisu} is only valid on the characteristic set $\cG^{-1}(0)$. Let now
  \[
    \wt\Phi^{\rm u} \in \tilde\Psi_\seop^0,\qquad
    \wt W^{\rm u} \in \tilde\Psi_\seop^1,\qquad
    \wt\Psi \in \tilde\Psi_\seop^{-1},
  \]
  be formally self-adjoint quantizations of $\wt\chi\wt\varphi^{\rm u}$, $\rho_\infty^{-1}\wt\chi\wt w^{\rm u}$, and $\rho_\infty\wt\chi\wt\psi$, respectively, where $\wt\chi\in\CI(\Sse^*\wt M)$ equals $1$ near $\{|\wt\varphi^{\rm u}|<5,\ |\wt\varphi^{\rm s}|<5\}\cap\pa\wt\Sigma$ and is supported in the set where $\wt\varphi^{\rm u}$, $\wt w^{\rm u}$, $\wt\psi$ are defined. Then\footnote{If we took $\wt\varphi^{\rm u}$ to be the $\eps$-independent extension of $\varphi^{\rm u}$, then $\rho_\infty H_\cG\wt\varphi^{\rm u}$ would equal $-w^{\rm u}\varphi^{\rm u}$ plus an error term of size $\eps$. In the subsequent equation for $\cL'u'$, this would cause $f'$ to have an additional term of the schematic form $\eps\tilde\Psi_\seop^1 u$. In the estimate~\eqref{EqEstTrapPhiu} then, one would have an extra term $\eps\|u\|_{H_{\seop,\eps}^{s+1}([0,5]\times[0,5])}$ on the right hand side, which is too strong in the differential order sense to close the estimate (no matter the power of $\eps$). Remark that if $\ell_1\leq 0$ (e.g.\ for scalar waves), one can close the estimate if one uses a further commutator argument using the monotonicity of $t$ (which gives a bound of $\eps\|u\|_{H_{\seop,\eps}^{s+1}}$ in terms of $\|u\|_{H_{\seop,\eps}^s}$). But since in our application this inequality on $\ell_1$ does not hold, we do not discuss this special case further.}
  \[
    i[\cL,\wt\Phi^{\rm u}] = -\wt W^{\rm u}\wt\Phi^{\rm u} + \wt\Psi\cL + \wt R + \wt R',\qquad \wt R\in\tilde\Psi_\seop^0,\ \wt R'\in\tilde\Psi_\seop^2,\quad \WF_\seop'(\wt R')\subset\supp\dd\wt\chi,
  \]
  and therefore, writing $\cL u=:f$, we get the secondary equation
  \begin{equation}
  \label{EqEstTrapSecondary}
    \cL'u'=f',\qquad \cL':=\cL-i\wt W^{\rm u},\quad u':=\wt\Phi^{\rm u}u,\quad f':=(\wt\Phi^{\rm u}-i\wt\Psi)f - i\wt R u - i\wt R'u.
  \end{equation}
  In view of~\eqref{EqEstTrapNuMin} and \eqref{EqEstTrapL1}, we have $\ell_1':=\rho_\infty\sigmase^1(\cL'_1)<-\frac12\nu_{\rm min}$ for $\cL'_1:=\frac{1}{2 i}(\cL'-(\cL')^*)$.

  We then run a positive commutator argument for~\eqref{EqEstTrapSecondary}, the commutant being
  \[
    A=\Op_\seop(a),\qquad a:=\rho_\infty^{-2(s+1)+1}\chi(\wt\varphi^{\rm u})^2\chi(\wt\varphi^{\rm s})^2\chi_\Sigma(\rho_\infty^2\cG)\chi_-(t-t_0)^2\chi_+\Bigl(\frac{t-t_1}{\eps}\Bigr)^2.
  \]
  Here, $\chi_\Sigma\in\CIc(\R)$ is equal to $1$ near $0$, and $\chi\in\CIc((-4,4))$ is fixed with $\chi|_{[-3,3]}=1$ and $\sqrt{-x\chi(x)\chi'(x)}\in\CI(\R_x)$; and furthermore $\chi_-,\chi_+$ are as in~\eqref{EqEstRadInChipm} except we presently require $\chi_-|_{(-\infty,-2\delta_0]}=0$, $\chi_-|_{[-\delta_0,\infty)}=1$, $\chi_+|_{(-\infty,\delta_0]}=1$, and $\chi_+|_{[2\delta_0,\infty)}=0$. Consider $2\Im\la\cL'u',A u'\ra=\la\sC u',u'\ra$ where $\sC=i[\cL',A]+2\cL'_1 A\in\tilde\Psi_\seop^{2(s+1)}$ has principal symbol
  \[
    H_\cG a+2\ell'_1 a.
  \]
  The derivative along $H_\cG$ of the weight $\rho_\infty$ produces a symbol with pointwise bound by $C\eps\rho_\infty^{-1}a$ since $H_{\hat G_b}\rho_\infty=-\sigma^{-2}H_{\hat G_b}\sigma=0$. On the characteristic set, the derivative of $\chi(\wt\varphi^{\rm u})^2$ is $-2\chi\chi'\wt w^{\rm u}\wt\varphi^{\rm u}$, which is thus a positive square (and therefore necessitates an a priori control term); the derivative of $\chi(\wt\varphi^{\rm s})^2$ on the other hand is a negative square (and will be dropped in the estimate below). The derivatives of $\chi_-$, resp.\ $\chi_+$ have the usual signs, i.e.\ they are positive, resp.\ negative. The \emph{main} term of $\sC$ is $2\cL'_1 A$, whose principal symbol is the product of $a$ with a negative self-adjoint bundle endomorphism (namely, one whose eigenvalues are bounded above by $-\nu_{\rm min}$).

  We write the resulting estimate in a schematic form. We write $I\times J\times T$ for a neighborhood of $\{|\wt\varphi^{\rm s}|\in I,\ |\wt\varphi^{\rm u}|\in J\}$ inside $\pa\wt\Sigma^+\cap t^{-1}(T)$. We moreover write $\|u\|_{H_{\seop,\eps}^s(I\times J\times T)}$ for the $H_{\seop,\eps}^0$ norm of $B u$ for an operator $B\in\tilde\Psi_\seop^s$ whose whose principal symbol equals $\rho_\infty^{-s}$ on $I\times J\times T$ and is bounded in absolute value by $\rho_\infty^{-s}$ everywhere, and whose wave front set is contained in a small neighborhood of $I\times J\times T$ (as measured on $\upbeta^*\Sse^*\wt M$). Finally, we shall not write out the error terms $\|u\|_{H_{\seop,\eps}^{s-\frac12}([0,6]\times[0,6]\times[t_0-\delta_0,t_1+5\delta_0\eps])}+\|\chi u\|_{H_{\seop,\eps}^{-N}}$, which need to be added to the right hand side of every estimate below. Then for $u$ vanishing in $t<t_0-\delta_0$, we have
  \begin{align*}
    &\| \wt\Phi^{\rm u}u \|_{H_{\seop,\eps}^{s+1}([0,3]\times[0,3]\times[t_0-2\delta_0,t_1+\delta_0\eps])} \\
    &\qquad \leq C\Bigl( \|u\|_{H_{\seop,\eps}^{s+1}([0,4]\times[3,4]\times[t_0-\delta_0,t_1+2\delta_0\eps])} + \|f'\|_{H_{\seop,\eps}^s([0,4]\times[0,4]\times[t_0-1,t_1+2\delta_0\eps])} \\
    &\qquad \hspace{16em} + \|\wt\Phi^{\rm u}u\|_{H_{\seop,\eps}^{s+\frac12}([0,4]\times[0,4]\times[t_0-1,t_1+2\delta_0\eps])} \Bigr) \\
    &\qquad\leq C\Bigl( \|u\|_{H_{\seop,\eps}^{s+1}([0,4]\times[3,4]\times[t_0-\delta_0,t_1+2\delta_0\eps])} + \|f\|_{H_{\seop,\eps}^s([0,4]\times[0,4]\times[t_0-1,t_1+3\delta_0\eps])} \\
    &\qquad \hspace{5em}+ \|u\|_{H_{\seop,\eps}^s([0,4]\times[0,4]\times[t_0-\delta_0,t_1+3\delta_0\eps])} + \|\wt\Phi^{\rm u}u\|_{H_{\seop,\eps}^{s+\frac12}([0,4]\times[0,4]\times[t_0-1,t_1+2\delta_0\eps])} \Bigr).
  \end{align*}
  One can iterate this estimate once to improve the final, error, term on the right (with slightly enlarged cutoffs), and thus ultimately obtain (cf.\ \cite[(3.34)]{HintzPolyTrap})
  \begin{equation}
  \label{EqEstTrapPhiu}
  \begin{split}
    &\| \wt\Phi^{\rm u}u \|_{H_{\seop,\eps}^{s+1}([0,3]\times[0,3]\times[t_0-2\delta_0,t_1+\delta_0\eps])} \\
    &\qquad \leq C\Bigl( \|u\|_{H_{\seop,\eps}^{s+1}([0,4]\times[3,4]\times[t_0-\delta_0,t_1+2\delta_0\eps])} + \|f\|_{H_{\seop,\eps}^s([0,4]\times[0,4]\times[t_0-1,t_1+3\delta_0\eps]])} \\
    &\qquad \hspace{17em} + \|u\|_{H_{\seop,\eps}^s([0,4]\times[0,4]\times[t_0-\delta_0,t_1+3\delta_0\eps])} \Bigr).
  \end{split}
  \end{equation}

  \pfstep{Quantitative propagation for $\wt\Phi^{\rm u}$.} Using the fact that $H_{\wt\varphi^{\rm u}}(\rho_\infty^{-1}\wt\varphi^{\rm s})\neq 0$, we can now estimate $u$ in $H_{\seop,\eps}^s$ in terms of $\wt\Phi^{\rm u}u$ in $H_{\seop,\eps}^{s+1}$ via propagation along the $H_{\wt\varphi^{\rm u}}$-flow using Lemma~\ref{LemmaEstTrapQuant} and carefully chosen commutants as in \cite[Step~2 of the proof of Theorem~3.9]{HintzPolyTrap}. In the present setting, the commutants need to contain also time localizations; we use for this purpose the function
  \begin{equation}
  \label{EqEstTrapChimp}
    \chi_{-+} := \chi_-(t-t_0)\chi_+\Bigl(\frac{t-t_1}{\eps}+C_+\wt\varphi^{\rm s}\Bigr)
  \end{equation}
  where we fix $C_+$ so that $\rho_\infty^{-1}H_{\wt\varphi^{\rm u}}(\frac{t-t_1}{\eps}+C_+\wt\varphi^{\rm s})>0$ at $\Gamma^+$; here $\chi_-,\chi_+$ are the functions used already in the previous step. Define now the rescaling $\wt\varphi^\bullet_{\rm new}:=\max((\delta_0/60)^{-1}C_+,1)\wt\varphi^\bullet$ for $\bullet={\rm s,u}$; the effect is that on the set $\{|\wt\varphi_{\rm new}^\bullet|<6,\ \bullet={\rm s,u}\}$ (which is contained in the set $\{|\wt\varphi^\bullet|<6\}$), the level sets of $\frac{t-t_1}{\eps}$ and of $\frac{t-t_1}{\eps}+C_+\wt\varphi^{\rm s}$ are within time distance $\frac{\delta_0}{10}\eps$ of each other. We now rename $\wt\varphi_{\rm new}^\bullet$ as $\wt\varphi^\bullet$. (The estimate~\eqref{EqEstTrapPhiu} remains valid under this rescaling.) The derivative of $\chi_-$ along $H_{\wt\varphi^{\rm u}}$ is irrelevant (since we only prove estimates for $u$ supported in $t\geq t_0-\delta_0$), and we have arranged also for the derivative of $\chi_+$ to be negative. This gives, for $\delta>0$ to be chosen later,
  \begin{align*}
    &\| u \|_{H_{\seop,\eps}^s([0,2\delta]\times[0,\delta_1]\times[t_0-\delta_0,t_1])} \\
    &\qquad \leq C\Bigl( \delta^{\frac12} \| u \|_{H_{\seop,\eps}^s([2\delta,1]\times[0,2]\times[t_0-\delta_0,t_1+\delta_0\eps])} + \| u \|_{H_{\seop,\eps}^s([0,1] \times [1,2] \times [t_0-\delta_0,t_1+\delta_0\eps])} \\
    &\qquad \hspace{16em} + \delta^{\frac12} \| \wt\Phi^{\rm u}u \|_{H_{\seop,\eps}^{s+1}([0,3]\times[0,3]\times[t_0-2\delta_0,t_1+\delta_0\eps])} \Bigr).
  \end{align*}
  Here $\delta_1>0$ is fixed (independently of how small $\eps,\delta$ are) so that the $\rho_\infty^{-1}H_{\wt\varphi^{\rm u}}$ flow starting in $[0,2\delta]\times[0,\delta_1]$ remains in $|\wt\varphi^{\rm u}|<2$ as long as $|\wt\varphi^{\rm s}|<1$; and the second term on the right arises from localizing to $|\wt\varphi^{\rm u}|<2$. (Cf.\ \cite[(3.37)]{HintzPolyTrap}.) Moreover, the constant $C$ is independent of $\delta$ (though we recall that we are omitting error terms in the estimates whose constants may depend on $\delta$).

  We then plug~\eqref{EqEstTrapPhiu} into the final term and exploit the small factor $\delta^{\frac12}$ to obtain
  \begin{equation}
  \label{EqEstTrapComb}
  \begin{split}
    &\|u\|_{H_{\seop,\eps}^s([0,2\delta]\times[0,3]\times[t_0-\delta_0,t_1])} \\
    &\qquad \leq C\Bigl( \delta^{\frac12}\|u\|_{H_{\seop,\eps}^s([2\delta,4]\times[0,2]\times[t_0-\delta_0,t_1+3\delta_0\eps])} + \|u\|_{H_{\seop,\eps}^{s+1}([0,4]\times[2,4]\times[t_0-\delta_0,t_1+2\delta_0\eps])} \\
    &\qquad \hspace{18em} + \|f\|_{H_{\seop,\eps}^s([0,4]\times[0,4]\times[t_0-1,t_1+3\delta_0\eps])}\Bigr).
  \end{split}
  \end{equation}
  On the left hand side, we replaced $[0,\delta_1]$ by $[0,3]$ since one can obtain control on the piece $[0,2\delta]\times[\delta_1,3]$ near the stable trapped set (but away from the unstable trapped set) by (fixed time $\sim\log\delta_1^{-1}$) real principal type propagation (including cutoffs in $t$) starting from $[0,4]\times[2,4]$ using the equation $\cL u=f$.

  \pfstep{Quantitative propagation for $\cL$.} Fix a number $\beta>0$ with $\max_{\pa\Gamma}\frac{\rho_\infty\ell_1}{\nu_{\rm min}}<\beta<\frac12$. We shall now estimate the first term on the right in~\eqref{EqEstTrapComb} in terms of $\delta^{\frac12-\beta}$ times the left (plus acceptable a priori control terms) via time $\sim\log\delta^{-1}$ propagation along $\rho_\infty H_\cG$. Following \cite[Step~3 in the proof of Theorem~3.9]{HintzPolyTrap}, this involves a commutant which on the characteristic set is given by
  \[
    a = \rho_\infty^{-2 s+1}|\wt\varphi^{\rm s}|^{-2\beta} \psi\Bigl(\log\Bigl|\frac{\wt\varphi^{\rm s}}{\delta}\Bigr|\Bigr)^2\chi(\wt\varphi^{\rm u})^2 \chi_-(t-t_0)^2\chi_+\Bigl(\frac{t-t_1}{\eps}\Bigr)^2
  \]
  where now $\chi_\pm$ now satisfy $\chi_-|_{(-\infty,-3\delta_0]}=0$, $\chi_-|_{[-2\delta_0,\infty)}=1$, $\chi_+|_{(-\infty,-\delta_0]}=1$, $\chi_+|_{[0,\infty)}=0$, and $\mp\chi'_\pm\geq 1$, $\sqrt{\mp\chi_\pm\chi'_\pm}\in\CI$; furthermore $\chi\in\CIc((-3,3))$ equals $1$ on $[-2,2]$ and has $\sqrt{-x\chi\chi'}\in\CI$; and $\psi\in\CIc((\log\frac12,\infty))$ equals $1$ on $[0,4\log\delta^{-1}]$ and $0$ on $[5\log\delta^{-1},\infty)$, with $\sqrt{-\psi\psi'}\in\CI$. The main term in the principal symbol $H_\cG a+2\ell_1 a$ in the usual commutator calculation now arises from $2\ell_1 a$ plus $\rho_\infty^{-2 s+1}$ times the $\rho_\infty H_\cG$-derivative of $|\wt\varphi^{\rm s}|^{-2\beta}$, which gives $\rho_\infty^{-2 s}(2\ell_1-2\beta\wt w^{\rm s})|\wt\varphi^{\rm s}|^{-2\beta}\psi^2\chi^2\chi_-^2\chi_+^2$, and thus is a \emph{negative} square. The derivative of $\psi$ produces a term of the same sign for $|\wt\varphi^{\rm s}|\geq\delta$, whereas the derivative of $\chi$ produces a positive square (as already in the second step of the proof).

  Altogether, using that $|\wt\varphi^{\rm s}|^{-2\beta}$ for $|\wt\varphi^{\rm s}|\in[\delta,1]$ is at least equal to $\sim\delta^{2\beta}$ times its value for $|\wt\varphi^{\rm s}|\in[\frac12\delta,\delta]$, one obtains the estimate
  \begin{align}
    &\|u\|_{H_{\seop,\eps}^s([\delta,4]\times[0,2]\times[t_0-\delta_0,t_1-\delta_0\eps])} \nonumber\\
  \label{EqEstTrapLogProp}
  \begin{split}
    &\qquad \leq C\Bigl( \delta^{-\beta}\|u\|_{H_{\seop,\eps}^s([\frac12\delta,\delta]\times[0,3]\times[t_0-\delta_0,t_1])} + \|u\|_{H_{\seop,\eps}^s([0,5]\times[2,3]\times[t_0-\delta_0,t_1])} \\
    &\qquad \hspace{17em} + \|f\|_{H_{\seop,\eps}^{s-1}([0,5]\times[0,3]\times[t_0-1,t_1])} \Bigr).
  \end{split}
  \end{align}

  \pfstep{Finite fast time propagation near $t=t_1$.} Due to a mismatch of domains (specifically, the end point of the time interval), we cannot yet plug the estimate~\eqref{EqEstTrapLogProp} into~\eqref{EqEstTrapComb}. This is easily remedied by estimating
  \begin{equation}
  \label{EqEstTrapFinalTime}
  \begin{split}
    &\|u\|_{H_{\seop,\eps}^s([2\delta,4]\times[0,2]\times[t_1-\delta_0\eps,t_1+3\delta_0\eps])} \\
    &\qquad \leq C\Bigl( \|u\|_{H_{\seop,\eps}^s([\delta,4]\times[0,2]\times[t_1-2\delta_0\eps,t_1-\delta_0\eps])} + \|u\|_{H_{\seop,\eps}^s([\delta,5]\times[2,3]\times[t_1-2\delta_0\eps,t_1+4\delta_0\eps])} \\
    &\qquad \hspace{17em} + \|f\|_{H_{\seop,\eps}^{s-1}([0,5]\times[0,3]\times[t_1-2\delta_0\eps,t_1+4\delta_0\eps])} \Bigr)
  \end{split}
  \end{equation}
  (into the first term of which one can plug~\eqref{EqEstTrapLogProp}). This estimate can be proved using a positive commutator estimate for $\cL u=f$ which exploits the monotonicity $\rho_\infty H_\cG\tau\geq c_0>0$, $\tau=\frac{t-t_1}{\eps}$. Indeed, propagation along $\rho_\infty H_\cG$ for fast ($\tau$) time $4 c_0^{-1}\delta_0$ starting from $[\delta,4]\times[0,3]\times[t_1-2\delta_0\eps,t_1-\delta_0\eps]$ covers $[2\delta,4]\times[0,2]\times[t_1-\delta_0\eps,t_1+3\delta_0\eps]$ since $\wt\varphi^{\rm s}$ increases from a starting value of $\delta$ only by a factor of the order $\exp(4 c_0^{-1}\delta_0\wt w^{\rm s})<2$ if $\delta_0$ is chosen small enough at the outset of the proof.

  \pfstep{Conclusion.} Plugging~\eqref{EqEstTrapLogProp} into~\eqref{EqEstTrapFinalTime}, the norm $\|u\|_{H_{\seop,\eps}^s([2\delta,4]\times[0,2]\times[t_0-\delta_0,t_1+3\delta_0\eps])}$ is bounded by the right hand side of~\eqref{EqEstTrapLogProp} except $t_1$ there needs to be replaced by $t_1+4\delta_0\eps$ throughout. Using this to estimate the first term in~\eqref{EqEstTrapComb} and fixing $\delta$ so small that $C\delta^{\frac12-\beta}<\frac12$ allows one to absorb the term $C\delta^{\frac12-\beta}\|u\|_{H_{\seop,\eps}^s([0,2\delta]\times[0,3]\times[t_0,t_1])}$ into the left hand side. This gives
  \begin{align*}
    &\|u\|_{H_{\seop,\eps}^s([0,2\delta]\times[0,3]\times[t_0-\delta_0,t_1])} \\
    &\qquad \leq C\Bigl( \|u\|_{H_{\seop,\eps}^{s+1}([0,5]\times[2,4]\times[t_0-\delta_0,t_1+4\delta_0\eps])} + \|f\|_{H_{\seop,\eps}^s([0,5]\times[0,4]\times[t_0-1,t_1+4\delta_0\eps])} \Bigr)
  \end{align*}
  (plus $C(\|u\|_{H_{\seop,\eps}^{s-\frac12}([0,5]\times[0,5]\times[t_0-\delta_0,t_1+4\delta_0\eps])}+\|\chi u\|_{H_{\seop,\eps}^{-N}})$). This completes the proof of the estimate~\eqref{EqEstTrap2} for $-N=s-\frac12$.

  For general $N$, one applies the estimate iteratively to the $H_{\seop,\eps}^{s-\frac12}$ remainder term (with slightly enlarged cutoffs), thus reducing the differentiability order by $\frac12$ at each step until one reaches $-N$ after finitely many steps.
\end{proof}

\subsubsection{Higher s-regularity}
\label{SssEstTraps}

Since $\sigmase^1(\pa_{\hat t})=-i\sigma$ is elliptic on $\Gamma^\pm$ by~\eqref{EqGlDynKerrTrapSigma}, we can again obtain trapping estimates for solutions of $\wt L u=f$ on mixed (se;s)-Sobolev spaces by commuting $\wt L$ with $\pa_t$; this will use~\eqref{EqEstRadsNorm} and Lemma~\ref{LemmaEstRadsComm}.

\begin{prop}[Uniform estimate near the trapped set: s-regularity]
\label{PropEstTraps}
  We use the assumptions of Proposition~\usref{PropEstTrap}. There exists $d_0=d_0(s,N)$ so that the following holds. Let $k\in\N_0$, and assume that $\wt L$ has the regularity~\eqref{EqEstRadsOp}. Then for all neighborhoods $\cU\subset\Sse^*\wt M$ of $\pa\Gamma^+\cap\{t_0\leq t\leq t_1\}$ and all $\chi\in\CIc(\wt M)$ equal to $1$ near the base projection of $\pa\Gamma^+\cap\{t_0\leq t\leq t_1\}$, there exist operators $B_\Gamma,B^{\rm s},G\in\Psi_\seop^0$ with operator wave front sets contained in $\cU$, with $B_\Gamma=\chi B_\Gamma\chi$ etc., with $B_\Gamma$ elliptic at $\pa\Gamma^+\cap\{t_0\leq t\leq t_1\}$ and the wave front set of $B^{\rm s}$ disjoint from $\Gamma^{\rm u,+}$, and a constant $C$ so that
  \begin{equation}
  \label{EqEstTraps}
    \|B_\Gamma u\|_{H_{(\seop;\sop),\eps}^{(s;k),\alpha_\circ,\hat\alpha}} \leq C\Bigl( \|G\wt L u\|_{H_{(\seop;\sop),\eps}^{(s;k),\alpha_\circ,\hat\alpha-2}} + \|B^{\rm s}u\|_{H_{(\seop;\sop),\eps}^{(s+1,k),\alpha_\circ,\hat\alpha}} + \|\chi u\|_{H_{(\seop;\sop),\eps}^{(-N;k),\alpha_\circ,\hat\alpha}}\Bigr)
  \end{equation}
  holds (in the strong sense) for all $u$ with support in $t\geq t_0$.
\end{prop}
\begin{proof}
  Consider first the case $k=0$. A direct modification of the proof of Proposition~\ref{PropEstTrap} is delicate (since we do not have an analogue of~\eqref{EqEstTrapFinalTime} when the upper bounds in the time intervals are $t_1+4\delta_0$ instead of $t_1+4\delta_0\eps$ etc.). Instead, we apply Proposition~\ref{PropEstTrap} (the proof of which handles subprincipal pseudodifferential terms directly, cf.\ the discussion prior to~\eqref{EqEstTrapL1}) with final time $t_1+\delta$ for a small value of $\delta>0$, and thus with $\tau=\frac{t-(t_1+\delta)}{\eps}$; pick then $B^0_\Gamma\in\Psi_\seop^0$ with $B^0_\Gamma=\chi B^0_\Gamma\chi$ so that $B^0_\Gamma$ is elliptic at $\Gamma^+\cap\{t_0\leq t\leq t_1\}$ and at the same time satisfies $\WF_\seop'(B^0_\Gamma)\subset\Ell_\seop(B_\Gamma)$. Since a fortiori $B^0_\Gamma\in\tilde\Psi_\seop^0$ (see also \cite[Remark~3.48]{HintzScaledBddGeo}), microlocal elliptic regularity in $\tilde\Psi_\seop$ implies that
  \[
    \|B_\Gamma^0 u\|_{H_{\seop,\eps}^{s,\alpha_\circ,\hat\alpha}} \leq C\Bigl(\|B_\Gamma u\|_{H_{\seop,\eps}^{s,\alpha_\circ,\hat\alpha}} + \|\chi u\|_{H_{\seop,\eps}^{-N,\alpha_\circ,\hat\alpha}}\Bigr).
  \]
  Similarly, we may pick $G^0\in\Psi_\seop^0$ so that $\WF_\seop'(G)\subset\Ell_\seop(G^0)$, $\WF_\seop'(G^0)\subset\cU$, and $\chi G^0\chi=G^0$. (The latter condition can be arranged if one starts with a cutoff function in Proposition~\ref{PropEstTrap} with support contained in the interior of the support of the present $\chi$.) Then
  \[
    \|G\wt L u\|_{H_{\seop,\eps}^{s,\alpha_\circ,\hat\alpha-2}} \leq C\Bigl(\|G^0\wt L u\|_{H_{\seop,\eps}^{s,\alpha_\circ,\hat\alpha-2}} + \|\chi u\|_{H_{\seop,\eps}^{-N,\alpha_\circ,\hat\alpha}}\Bigr).
  \]
  Proceeding similarly with the term $B^{\rm s}$ from Proposition~\ref{PropEstTrap} thus yields~\eqref{EqEstTraps} for $k=0$ with $B_\Gamma^0$ etc.\ in place of $B_\Gamma$.

  The case of higher $k$ is treated via induction. Writing $\wt L u=f$, we consider equation~\eqref{EqEstRadsOutL1} for $\pa_t u$. Since, near $\hat M^\circ$, we have $\wt L^{(1)}-\wt L\in\eps^{-1}(\CI+\cC_{\seop;\sop}^{(d_0;k-1)})\Psi_\seop^1$, the subprincipal operator of $\wt L^{(1)}$ at $\pa\Gamma^+$ equals that of $\wt L$, and therefore the inductive hypothesis applies to~\eqref{EqEstRadsOutL1}. In view of~\eqref{EqEstRadsNorm}, this completes the inductive step.
\end{proof}

\subsection{Energy estimates}
\label{SsEstEn}

The microlocal estimates proved thus far give estimates of solutions of $\wt L u=f$ at high se-frequencies, but with error terms whose norms are taken on sets which are slightly enlarged by an amount $c\hat\rho$ in the $t$- and $x$-directions near $\hat M$, and by an amount $c$ away from $\hat M$, where $c$ is arbitrary but fixed. To close se-estimates, we use energy estimates to control solutions of wave equations for short times near the subsets $\{\hat r=\bhm\}$ and $\{\frac{t-t_1}{\hat\rho}=0\}$ of $\upbeta^{-1}\wt\Omega$.\footnote{We do not need an analogue for mixed (se;s)-Sobolev spaces. In fact, the sharp localization in $\frac{t-t_1}{\hat\rho}$ below is incompatible with ps.d.o.s of class $\Psi_\seop$ which have s-regular coefficients. Moreover, estimates for fixed size $t$-intervals cannot be proved using simple-minded energy methods; after all, this is the key challenge which we overcome in this work.} This is analogous to \cite[\S{2.1.3}]{HintzVasySemilinear} and \cite[\S{3.3}]{HintzConicWave}, and thus we shall be brief.

Recall that $\wt\Omega\subset\wt M$ is a standard domain associated with $\Omega=\Omega_{t_0,t_1,r_0}$. Working in $\hat r\geq\frac12\bhm$, we take $\hat\rho=r$. For $0<\delta<\min(1,\frac12\bhm)$, we then consider the enlarged domain
\[
  \wt\Omega^\delta := \{ (\eps,t,x) \colon t_0\leq t\leq t_1+\hat\rho\delta,\ |\hat x|\geq\bhm-\delta,\ |x|\leq r_0+2(t_1-t)+\delta \} \subset [\wt M;\hat M_{t_1}].
\]
In terms of
\[
  \tau := \frac{t-t_1}{\hat\rho} = \frac{t-t_1}{r}\,,
\]
the final spacelike hypersurfaces are thus at $\tau=\delta$ (as opposed to $\tau=0$ for $\wt\Omega$), $|\hat x|=\bhm-\delta$, and $|x|=r_0+2(t_1-t)+\delta$. Since $\dd\tau=r^{-1}(\dd t-\tau\,\dd r)$ is a past timelike se-covector for $\tau=0$, this remains true also for sufficiently small $\tau$ on $\wt\Omega^\delta$. The spacelike nature of the other boundary hypersurfaces of $\wt\Omega^\delta$ follows by similar perturbative arguments, much as after Definition~\ref{DefGlDynStd}. Near the boundary hypersurfaces of $\wt\Omega^\delta$, we will work with domains of the form
\begin{equation}
\label{EqEstEnDom}
\begin{split}
  \wt\Omega^{\delta_+}_{\delta_-} :=\ \ &\{ t_1+\hat\rho\delta_- \leq t\leq t_1+\hat\rho\delta_+,\ |\hat x|\geq\bhm-\delta_+,\ |x|\leq r_0+2(t_1-t)+\delta_+ \} \\
    \cup\;& \{ t_0\leq t\leq t_1+\hat\rho\delta_+,\ |x|-(r_0+2(t_1-t))\in[\delta_-,\delta_+] \} \\
    \cup\;& \{ t_0\leq t\leq t_1+\hat\rho\delta_+,\ |\hat x|\in[\bhm-\delta_+,\bhm-\delta_-] \}
\end{split}
\end{equation}
for small $\delta_\pm\in\R$ with $\delta_-<\delta_+$. The initial and final boundary hypersurfaces of such domains are illustrated in Figure~\ref{FigEstEnDom}. As usual, we write $(\Omega^{\delta_+}_{\delta_-})_\eps:=\wt\Omega^{\delta_+}_{\delta_-}\cap M_\eps$.

\begin{figure}[!ht]
\centering
\includegraphics{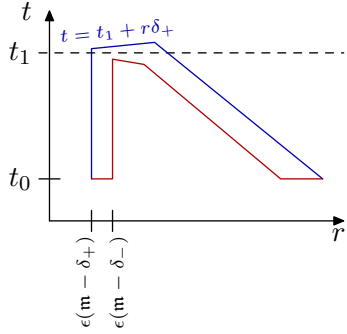}
\caption{A domain $\wt\Omega^{\delta_+}_{\delta_-}$, for a fixed value $\eps>0$ and with $\delta_-<0<\delta_+$. The initial, resp.\ final boundary hypersurfaces are drawn in red, resp.\ blue.}
\label{FigEstEnDom}
\end{figure}

\begin{prop}[Energy estimates near the final boundary hypersurfaces of $\wt\Omega^\delta$]
\label{PropEstEn}
  Let $\alpha_\circ,\hat\alpha\in\R$. Let $\sfs\in\CI(\Sse^*\wt M)$, and assume that in a neighborhood of $\bigcup_{\bullet=-,\,\wedge,\,|}\wt Y^\bullet_{t_0,t_1,r_0}$, the function $\sfs$ is non-increasing along the $\pm\hat\rho^2 H_{\wt G}$-flow in the characteristic set $\wt\Sigma^\pm$. Then for small $\delta_-<\delta_+$, there exist $\eps_0>0$ and $C\in\R$ so that the following holds for all $\eps\leq\eps_0$: given $f\in H_{\seop,\eps}^{\sfs-1,\alpha_\circ,\hat\alpha-2}((\Omega^{\delta_+}_{\delta_-})_\eps)^{\bullet,-}$, there exists a unique solution $u\in H_{\seop,\eps}^{\sfs,\alpha_\circ,\hat\alpha}((\Omega^{\delta_+}_{\delta_-})_\eps)^{\bullet,-}$ of $\wt L u=f$, with the norm of $u$ bounded by $C$ times that of $f$.
\end{prop}
\begin{proof}
  Existence and uniqueness of $u$ follows for each $\eps>0$ from standard hyperbolic theory (exploiting the timelike nature of $\tau$ and, in $\hat r\leq\bhm-\delta_-$, of $\hat r$), at least for constant orders $\sfs$. One can easily accommodate variable orders by first extending $f$ to a larger domain $(\Omega_{\delta_-}^{\tilde\delta_+})_\eps$ with $\tilde\delta_+>\delta_+$, solving the equation on spaces with constant (and sufficiently negative) se-regularity orders there (discussed below), and then using the propagation of singularities for finite affine time (independently of $\eps$) along $\hat\rho^2 H_{\wt G}$ integral curves over $[\wt M;\hat M_{t_1}]$ to get uniform $H_{\seop,\eps}^{\sfs,\alpha_\circ,\hat\alpha}$ control over $(\Omega_{\delta_-}^{\delta_+})_\eps$; the usual error term in microlocal estimates can be taken to be a weak norm on the larger domain $(\Omega_{\delta_-}^{\tilde\delta_+})_\eps$, which however is controlled by virtue of having solved the wave equation on this larger domain.

  For the proof of uniform estimates of $u$ in terms of $f$ on weighted se-Sobolev spaces, we only discuss the case $\sfs=1$; the extension to general orders $\sfs$ follows from the above arguments (together with duality arguments to cover negative $\sfs$, cf.\ the proof of \cite[Proposition~3.7]{HintzConicWave} or \cite[Theorem~6.4]{HintzVasyScrieb}). For $\sfs=1$, one simply runs a standard energy estimate on $(\Omega_{\delta_-}^{\delta_+})_\eps$ using the (future timelike) vector field multiplier
  \[
    V_\digamma=-e^{-\digamma\tau}\rho_\circ^{-2\alpha_\circ}\hat\rho^{-2\hat\alpha+2}(\dd\tau)^\sharp \in \rho_\circ^{-2\alpha_\circ}\hat\rho^{-2\hat\alpha+2}\Vse(\wt M)
  \]
  for large $\digamma>1$. (When $\wt L$ is an operator on bundles, one takes $V_\digamma$ to be a first order se-differential operator with this principal part.) The deformation tensor of $V_\digamma$ with respect to the metric $\wt g$ is, as an se-tensor, a smooth multiple of $\rho_\circ^{-2\alpha_\circ}\hat\rho^{-2\hat\alpha}$, and indeed a \emph{negative} definite multiple when $\digamma$ is sufficiently large (again using the past timelike nature of $\dd\tau$). Applying the divergence theorem to the vector field $T[u](V_\digamma,-)$ where $T[u]=u_{;\mu}u_{;\nu}-\frac12 \wt g_{\mu\nu}|\dd u|_{\wt g^{-1}}^2$ yields the desired estimate.
\end{proof}

\subsection{Uniform regularity estimates on standard domains}
\label{SsEstStd}

We now gather all (micro)local estimates proved in the previous subsections.

\begin{definition}[Admissible order functions]
\fakephantomsection
\label{DefEstAdm}
  \begin{enumerate}
  \item Let $\wt\Omega$ be a standard domain. We call $\sfs\in\CI(\Sse^*\wt M)$, $\alpha_\circ,\hat\alpha\in\R$ \emph{admissible (on $\wt\Omega$)} if
  \begin{enumerate}
    \item\label{ItEstAdmIn} $\sfs+\alpha_\circ-\hat\alpha>\frac12(-1+\vartheta_{\rm in})$ at $\pa\cR_{\rm in}$ (see~\eqref{EqEstRadInTheta});
    \item\label{ItEstAdmOut} $\sfs+\alpha_\circ-\hat\alpha<\frac12(-1-\vartheta_{\rm out})$ at $\pa\cR_{\rm out}$ (see~\eqref{EqEstRadOutTheta});
    \item\label{ItEstAdmHor} $\sfs>\frac12(1+\vartheta_{\cH^+})$ at $\pa\cR_{\cH^+}$ (see~\eqref{EqEstHorTheta});
    \item $\sfs$ is constant near the radial sets and near the trapped set~\eqref{EqGlDynStdTrap}, and monotonically decreasing along the future null-bicharacteristic flow, i.e.\ $\pm\hat\rho^2 H_{\wt G}\sfs\leq 0$ on $\wt\Sigma^\pm$ over a neighborhood of $\wt\Omega$;
    \item the restriction of $\sfs$ to $\Sse^*_{\hat M}\wt M$ is $t$-independent.
  \end{enumerate}
  \item\label{ItEstAdmKerr} We say that $\sfs\in\CI(\Stb^*\hat M_b)$, $\alpha_\cD$ are \emph{Kerr-admissible} if $\sfs$ is the restriction to $\hat M_{t_0}$ (for any $t_0\in I_\cC$, via~\eqref{EqFseBundle3b}) of an order function $\tilde\sfs$ for which $\tilde\sfs$, $\alpha_\circ=\alpha_\cD$, and $\hat\alpha=0$ are admissible. Equivalently, the above threshold, constancy, and monotonicity conditions hold at $\pa\hat\cR_{\rm in}$, $\pa\hat\cR_{\rm out}$, $\pa\hat\cR_{\cH^+}$, $\pa\hat\Gamma_b$, $\pa\hat\Sigma^\pm$.
  \end{enumerate}
\end{definition}

\begin{rmk}[Induced orders]
\label{RmkEstAdmInd}
  An admissible order function $\sfs$ induces a stationary variable 3b-order function, i.e.\ an element of $\CI(\Stb^*_{\hat X_b}\hat M_b)$, via~\eqref{EqFseBundle3b} which by an abuse of notation we denote by $\sfs$ as well. This function in turn induces order functions on the spectral side as recalled after~\eqref{EqFVar3bFT}: namely, a variable b-differential order function on $\Sb^*\hat X_b$ for the zero energy operator $\hat L(0)$; a variable scattering regularity order on $\Ssc^*\hat X_b$ and a variable scattering decay order on $\ol{\Tsc^*_{\pa\hat X_b}}\hat X_b$ for the spectral family $\hat L(\sigma)$ at nonzero real frequencies (and also a variable semiclassical scattering decay order at high real frequencies); a variable semiclassical order (powers of $h$) at high real frequencies; and variable sc-b-transition regularity and (scattering) decay orders for the uniform low energy analysis at positive or negative frequencies, which in turn induce variable scattering regularity and decay orders for the transition face normal operators $L_\tface(\pm 1)$. We will often denote all these induced orders by $\sfs$ simply. See also~\cite[Lemma~4.28]{Hintz3b}.
\end{rmk}

When $\wt\Omega\cap\hat M=\emptyset$, then any constant $\sfs$ is admissible. To construct an (admissible) order function $\sfs$ when the domain $\wt\Omega$ intersects $\hat M$, we recall Lemma~\ref{LemmaGlDynStdMono} and take (in the notation of the Lemma)
\begin{equation}
\label{EqEstAdmIndEx}
  \sfs := s_+ - \chi_1(\hat\xi_\seop)\chi_2(\hat r)(s_+-s_-)
\end{equation}
where $s_+$ satisfies the lower bounds in conditions~\eqref{ItEstAdmIn} and \eqref{ItEstAdmHor} of Definition~\ref{DefEstAdm}, while $s_-$ satisfies the upper bound in condition~\eqref{ItEstAdmOut}; and we require $s_-\leq s_+$. Note also that the function~\eqref{EqEstAdmIndEx} satisfies $\pa_t\sfs|_{\hat M}=0$, and moreover it is admissible also for all perturbations of the wave operator $\wt L$ (within the class considered in~\eqref{ItEstLBundle}--\eqref{ItEstLOp}).

\begin{thm}[Uniform regularity estimates]
\label{ThmEstStd}
  Let $\wt\Omega\subset\wt M\setminus\wt K^\circ$ be a standard domain. Let $\eps_0>0$ be so small that all boundary hypersurfaces of $\Omega_\eps=\wt\Omega\cap M_\eps$ are spacelike for $\eps\in(0,\eps_0]$. Let $\alpha_\circ,\hat\alpha\in\R$. Let $\sfs,\sfs_0\in\CI(\Sse^*\wt M)$ be such that $\sfs,\alpha_\circ,\hat\alpha$ are admissible on $\wt\Omega$, and so that also $\sfs_0$ satisfies the lower bounds in Definition~\usref{DefEstAdm}\eqref{ItEstAdmIn}, \eqref{ItEstAdmHor}. Suppose moreover that there exists a positive definite fiber inner product on $\hat\cE$ (which gives an inner product on $\cE|_{\hat M}$ via pullback) so that $\frac{1}{2 i}(S_{\rm sub}(L)-S_{\rm sub}(L)^*)<\frac12\nu_{\rm min}$, where we recall~\eqref{EqEstTrapNuMin}. Then there exists a constant $C$ so that the estimate
  \begin{equation}
  \label{EqEstStd}
    \|u\|_{H_{\seop,\eps}^{\sfs,\alpha_\circ,\hat\alpha}(\Omega_\eps)^{\bullet,-}} \leq C\Bigl( \|\wt L u\|_{H_{\seop,\eps}^{\sfs,\alpha_\circ,\hat\alpha-2}(\Omega_\eps)^{\bullet,-}} + \|u\|_{H_{\seop,\eps}^{\sfs_0,\alpha_\circ,\hat\alpha}(\Omega_\eps)^{\bullet,-}} \Bigr)
  \end{equation}
  holds (in the strong sense) for all $\eps\leq\eps_0$.
\end{thm}
\begin{proof}
  Given $f\in H_{\seop,\eps}^{\sfs,\alpha_0,\hat\alpha-2}(\Omega_\eps)^{\bullet,-}$, we extend it to $\tilde f\in H_{\seop,\eps}^{\sfs,\alpha_0,\hat\alpha-2}(M_\eps)$ with the same norm, as described in~\S\ref{SsEstFn}. For some fixed small $\delta>0$ and recalling~\eqref{EqEstEnDom}, set
  \[
    \wt\Omega^\delta:=\wt\Omega\cup\wt\Omega_0^\delta\subset[\wt M;\hat M_{t_1}].
  \]
  On the $\eps$-level sets $(\Omega^\delta)_\eps$ of this domain (for all sufficiently small $\eps>0$), we solve $\wt L\tilde u=\tilde f$, with trivial initial data at the initial hypersurface of $\wt\Omega^\delta$. Thus, $\tilde u$ lies in $H^{\sfs+1}$ for any fixed $\eps>0$, and by the uniqueness of solutions of wave equations it equals $u$ on $\Omega_\eps$.

  Starting in $t\in(t_0-1,t_0)$ where $\tilde u$ vanishes, we now propagate se-regularity along future null-bicharacteristics, following the arrows in Figure~\ref{FigGlDynFlow}. To wit, we first propagate uniform se-regularity into the incoming radial set $\pa\cR_{\rm in}$ using Proposition~\ref{PropEstRadIn} and, independently, out of the event horizon $\pa\cR_{\cH^+}$ using Proposition~\ref{PropEstHor}. This in particular gives control of $\tilde u$ at the stable manifold of $\pa\Gamma$ outside of any arbitrarily small neighborhood of $\pa\Gamma$; from there, se-Sobolev bounds on $\tilde u$ can be propagated into the trapped set itself using Proposition~\ref{PropEstTrap}. Since all past null-bicharacteristics starting in a punctured neighborhood of $\pa\cR_{\rm out}$ over $\hat M$ tend to the horizon, the trapped set, or the initial hypersurface, we can now use Proposition~\ref{PropEstRadOut} to propagate uniform se-regularity into the outgoing radial set $\pa\cR_{\rm out}$, and from there onward using real principal type propagation.

  Together with se-microlocal elliptic estimates to estimate $\tilde u$ away from $\wt\Sigma$, this establishes the uniform estimate
  \[
    \|u\|_{H_{\seop,\eps}^{\sfs,\alpha_\circ,\hat\alpha}(\Omega_\eps)^{\bullet,-}} \leq C\Bigl( \| \tilde f \|_{H_{\seop,\eps}^{\sfs,\alpha_\circ,\hat\alpha-2}((\Omega^\delta)_\eps)^{\bullet,-}} + \| \tilde u \|_{H_{\seop,\eps}^{\sfs_0,\alpha_\circ,\hat\alpha}((\Omega^\delta)_\eps)^{\bullet,-}}\Bigr).
  \]
  Let now $\chi\in\CI(\wt\Omega^\delta_{-\delta})$ be a cutoff function which equals $0$ on $\wt\Omega_{-\delta}^{-\delta/2}$ and $1$ on $\wt\Omega^\delta_0$. Then the norm on $\tilde u$ on the right hand side is bounded by the sum of the norms of $\chi\tilde u$ and $(1-\chi)\tilde u=(1-\chi)u$. We only need to estimate the former term; for this purpose, note that
  \begin{equation}
  \label{EqEstStdExt}
    \wt L(\chi\tilde u)=\tilde f':=\chi\tilde f+[\wt L,\chi]\tilde u = \chi\tilde f + [\wt L,\chi]u.
  \end{equation}
  Proposition~\ref{PropEstEn} gives the uniform estimate
  \[
    \|\chi\tilde u\|_{H_{\seop,\eps}^{\sfs_0,\alpha_\circ,\hat\alpha}((\Omega^\delta_{-\delta})_\eps)^{\bullet,-}} \leq C\|\tilde f'\|_{H_{\seop,\eps}^{\sfs_0-1,\alpha_\circ,\hat\alpha-2}((\Omega^\delta_{-\delta})_\eps)^{\bullet,-}}.
  \]
  To estimate $\tilde f'$, note that $\chi\in\CI([\wt M;\hat M_{t_1}])\subset\CI_\seop(\wt M)$ and thus $[\wt L,\chi]\in\CI_\seop\Diffse^{1,0,2}$; and since $\supp\dd\chi\subset\wt\Omega$, we can estimate
  \[
    \|\tilde f'\|_{H_{\seop,\eps}^{\sfs_0-1,\alpha_\circ,\hat\alpha-2}((\Omega^\delta_{-\delta})_\eps)^{\bullet,-}} \leq C\Bigl(\|\tilde f\|_{H_{\seop,\eps}^{\sfs_0-1,\alpha_\circ,\hat\alpha-2}((\Omega^\delta)_\eps)^{\bullet,-}} + \|u\|_{H_{\seop,\eps}^{\sfs_0,\alpha_\circ,\hat\alpha}(\Omega_\eps)^{\bullet,-}}\Bigr).
  \]
  Since $\|\tilde f\|_{H_{\seop,\eps}^{\sfs,\alpha_\circ,\hat\alpha-2}((\Omega^\delta)_\eps)^{\bullet,-}}=\|f\|_{H_{\seop,\eps}^{\sfs,\alpha_\circ,\hat\alpha-2}(\Omega_\eps)^{\bullet,-}}$, this is in turn bounded by the right hand side of~\eqref{EqEstStd}. The proof is complete.
\end{proof}

\begin{rmk}[Uniformity for small perturbations]
\label{RmkEstStdUnif}
  The constants $\eps_0$ and $C$ can be chosen uniformly, for fixed orders $\sfs,\sfs_0,\alpha_\circ,\hat\alpha$, also for sufficiently small perturbations $\wt L'$ of $\wt L$ satisfying the hypotheses made on $\wt L$, where smallness is measured in the supremum norm (on $\wt\Omega$) of up to $d_0$ derivatives along the se-vector fields $\hat\rho\pa_t$, $\hat\rho\pa_x$ of the coefficients of $\eps^{-1}\hat\rho^2(\wt L-\wt L')$, expressed in terms of a fixed spanning set of se-vector fields (such as $\hat\rho\pa_t$, $\hat\rho\pa_x$), for sufficiently large $d_0$. Indeed, the proofs of all estimates used in the proof of Theorem~\ref{ThmEstStd} apply, \emph{without any modifications}, when the coefficients of $\wt L$ have some large but finite amount of se-regularity. (For the microlocal regularity and propagation results, we already made this explicit when discussing estimates on mixed (se;s)-Sobolev spaces in~\S\ref{SssFVarsse}, \ref{SssEstRads}, \ref{SssEstHors}, and \ref{SssEstTraps}. See also~\S\ref{SsNTame} for tame estimates.)
\end{rmk}

\begin{rmk}[Higher s-regularity]
\label{RmkEstStds}
  We \emph{cannot} prove an analogue of Theorem~\ref{ThmEstStd} on mixed (se;s)-spaces with s-regularity order $k\geq 1$: performing the extension/restriction procedure to standard domains with final hypersurface at the later slow time $t=t_1+\delta$ first requires uniform se-bounds for the solution of $\wt L u=f$ for times $t$ between $t_1$ and $t_1+\delta$, but at this point we do not have such a solvability theory yet. (Establishing this solvability requires the invertibility of the Kerr model $L$---which we have not yet addressed---on matching spaces, as discussed in~\S\ref{SsScS} below.) It is conceivable that this is an artifact of our (se-)microlocal approach. In any case, the theory developed in this paper is only useful if solvability on se-spaces holds, as e.g.\ in the setting of Theorem~\ref{ThmScUnif} below.
\end{rmk}

\begin{rmk}[Other settings]
\label{RmkEstOtherUnif}
  If the model $L$ is a wave-type operator on a nontrapping spacetime without horizons as in Remark~\ref{RmkEstOther}, then Theorem~\ref{ThmEstStd} continues to hold with the following simplifications: there is no threshold regularity condition over $\hat M^\circ$ (cf.\ Definition~\ref{DefEstAdm}\eqref{ItEstAdmHor}), and since we do not need to use trapping estimates anymore, one can relax the norm on $\wt L u$ in~\eqref{EqEstStd} to the $H_{\seop,\eps}^{\sfs-1,\alpha_\circ,\hat\alpha-2}$-norm.
\end{rmk}

\subsection{Estimates for the spectral family}
\label{SsEstFT}

The validity of Fredholm estimates for the spectral family $\hat L(\sigma)$ at bounded nonzero energies as well as high energy estimates only use the dynamical structure of Kerr spacetimes (and information on the subprincipal symbol at horizons, radial sets, and the trapped set) described in Lemma~\ref{LemmaGlDynKerrFlow}. The spectral family of the scalar wave operator on Kerr is discussed in a manner directly applicable to the present paper (including using variable order spaces) in \cite[\S{3.5}]{HintzKdSMS}, where the symbolic estimates which we state below are complemented with mode stability assumptions (which we do not make here) to give invertibility and uniform boundedness statements for the spectral family. A more general setting is described in detail in \cite[\S{4.1}]{HintzNonstat}, except the reference discusses a conjugation of the spectral family by $e^{i\sigma r}$ (see also \cite{VasyLAPLag,VasyLowEnergyLag}) and does not discuss trapping (in the semiclassical, i.e.\ high energy regime) and energy estimates near final hypersurface $\hat r=\bhm$ in the black hole interior; we recall that the relevant semiclassical estimates at trapping (which apply also to tensorial equations) are stated in \cite[Theorem~4.7]{HintzVasyQuasilinearKdS}, and the (energy) estimates near the interior final hypersurface are proved in the same manner as Proposition~\ref{PropEstEn}; likewise in the semiclassical setting, in which the basic energy estimate is proved in \cite[Proposition~3.7]{VasyMicroKerrdS}.

The combination of the various (micro)local estimates as in the proof of Theorem~\ref{ThmEstStd} gives the results stated in the remainder of this section; to avoid repetition, we shall thus be very brief.

\begin{prop}[Fredholm estimate at $0$ energy]
\label{PropEstFT0}
  Let $\sfs\in\CI(\Sb^*_{\hat X_b}\hat M_b)$ be monotonically decreasing along the future null-bicharacteristic flow of $H_{\hat G_b}\bmod\pa_{\hat t}$, with $\sfs>\frac12(1+\vartheta_{\cH^+})$ at $\pa\hat\cR_{\cH^+}$ and constant nearby, and suppose that $\alpha_\cD+\frac32\notin\specb(\hat L(0))$. Then the operator
  \begin{equation}
  \label{EqEstFT0Op}
    \hat L(0) \colon \{ u\in\Hb^{\sfs,\alpha_\cD}(\hat X_b;\hat\cE) \colon \hat L(0)u\in\Hb^{\sfs-1,\alpha_\cD+2}(\hat X_b;\hat\cE) \} \to \Hb^{\sfs-1,\alpha_\cD+2}(\hat X_b;\hat\cE)
  \end{equation}
  is Fredholm. Here, we use the density on $\hat X_b$ whose product with $|\dd\hat t|$ is the metric density $|\dd\hat g_b|$.
\end{prop}
\begin{proof}
  Near $\pa\hat X_b$, $\hat L(0)$ is an elliptic b-differential operator, and hence its analysis is analogous to \cite[Lemma~4.1]{HintzNonstat} (where, unlike in the present setting, the Fredholm analysis is complemented with an injectivity assumption on $\hat L(0)$ and its adjoint). The characteristic set lies over the ergoregion, event horizon, and the black hole interior, where the Fredholm analysis at bounded (or even semiclassical) frequencies is completely analogous to \cite[\S{6}]{VasyMicroKerrdS}. See also \cite[Lemma~3.19]{HintzKdSMS}.
\end{proof}

We shall take for $\sfs$ in Proposition~\ref{PropEstFT0} a Kerr-admissible variable 3b-differential order (with $\alpha_\cD\in\R$ arbitrary but fixed, unless stated otherwise), which in turn arises from an admissible se-order function, as discussed in Remark~\ref{RmkEstAdmInd}. We phrase subsequent estimates in this fashion for brevity. Furthermore, the following condition will play an important role in the sequel:
\begin{equation}
\label{EqEstFTSubpr}
  \parbox{0.8\textwidth}{There exists a stationary positive definite inner product on $\hat\pi^*\hat\cE$ so that $\frac{1}{2 i}(S_{\rm sub}(L)-S_{\rm sub}(L)^*)<\frac12\nu_{\rm min}$ at $\hat\Gamma_b$ in the notation~\eqref{EqEstTrapNuMin}.}
\end{equation}

\begin{prop}[Fredholm estimate at bounded nonzero energies]
\label{PropEstFTbdd}
  Let $\sfs$, $\alpha_\cD$ be Kerr-admissible orders. Then for $\sigma\neq 0$, $\Im\sigma\geq 0$, the operator
  \begin{equation}
  \label{EqEstFTbddOp}
    \hat L(\sigma) \colon \{ u\in\Hsc^{\sfs,\sfs+\alpha_\cD}(\hat X_b;\hat\cE) \colon \hat L(\sigma)u\in\Hsc^{\sfs-1,\sfs+\alpha_\cD+1}(\hat X_b;\hat\cE) \} \to \Hsc^{\sfs-1,\sfs+\alpha_\cD+1}(\hat X_b;\hat\cE)
  \end{equation}
  is Fredholm of index $0$. Here the scattering regularity and decay orders are those induced by $\sfs$ for positive, resp.\ negative real frequencies when $\Re\sigma>0$, resp.\ $\Re\sigma<0$; and when $\Im\sigma>0$, then $\sfs$ can be taken to be arbitrary near $\pa\hat X_b$,\footnote{More precisely, only the threshold condition near $\pa\cR_{\cH^+}$ is needed (so one can take $\sfs$ to be a constant order even, though in case $\sfs$ is variable, the monotonicity along the null-bicharacteristic flow is still necessary).} and $\sfs+\alpha_\cD+1$ can be replaced by $\sfs+\alpha_\cD$. For $\sigma$ lying in a compact subset of $\{\sigma\in\C\colon \sigma\neq 0,\ \Im\sigma\geq 0\}$ of spectral parameters for which $\hat L(\sigma)$ is invertible (for a fixed choice of $\sfs,\alpha_\cD$), the inverse is uniformly bounded.
\end{prop}
\begin{proof}
  This is proved similarly to \cite[Proposition~3.18]{HintzKdSMS}. The uniform bounds down to the real axis are a version of the limiting absorption principle, cf.\ \cite[Proposition~5.28]{VasyMinicourse}. The index $0$ property follows from the invertibility energies with large imaginary part (where trapping is irrelevant), proved in Proposition~\ref{PropEstFThi} below.
\end{proof}

\begin{prop}[High energy estimates]
\label{PropEstFThi}
  Let $\sfs$, $\alpha_\cD$ be Kerr-admissible orders, and assume~\eqref{EqEstFTSubpr}. Then for $\sigma\in\C$ with $\Im\sigma\geq 0$ and $|\sigma|$ sufficiently large, the operator~\eqref{EqEstFTbddOp} is invertible. Furthermore, for the semiclassical order, scattering regularity, and scattering decay orders induced by $\sfs$ for positive, resp.\ negative real frequencies, we have for $\pm\Re\sigma>0$ and $0\leq\Im\sigma\leq C_0$ (for any fixed $C_0$) the high energy estimate
  \begin{equation}
  \label{EqEstFThi}
    \|u\|_{H_{\scop,|\sigma|^{-1}}^{\sfs,\sfs+\alpha_\cD,\sfs}(\hat X_b;\hat\cE)} \leq C \| \hat L(\sigma)u \|_{H_{\scop,|\sigma|^{-1}}^{\sfs-1,\sfs+\alpha_\cD+1,\sfs}(\hat X_b;\hat\cE)}
  \end{equation}
  for all sufficiently large $|\sigma|$. For $|\Im\sigma|>c|\Re\sigma|$ for some $c>0$, then this estimate holds, for sufficiently large $|\sigma|$, for orders $\sfs$ which are arbitrary near $\pa\hat X_b$, and the second and third order $\sfs+\alpha_\cD+1$ and $\sfs$ on the right can be replaced by $\sfs+\alpha_\cD$ and $\sfs-1$, respectively.
\end{prop}
\begin{proof}
  This is analogous to \cite[Proposition~3.17]{HintzKdSMS}. Note that $\hat L(\sigma)$, as a semiclassical operator with semiclassical parameter $h=|\sigma|^{-1}$, has order $h^{-2}$; so the estimate~\eqref{EqEstFThi} features a loss of two semiclassical orders (cf.\ the third index $\sfs$ on the right), which is due to trapping \cite{DyatlovSpectralGaps}, \cite[Theorem~4.7]{HintzVasyQuasilinearKdS}. The large $\Im\sigma$ behavior arises from semiclassical scattering ellipticity at spatial infinity and the fact that for large $\Im\sigma$, the timelike nature of $\dd\hat t$ implies that $\hat L(\sigma)$, as a semiclassical operator, features complex absorption (i.e.\ damping) at the semiclassical principal symbol (but classically subprincipal symbol) level, globally on $\hat X_b$. (See also the discussion of \cite[Theorem~7.3]{VasyMicroKerrdS}, and \cite[\S{3.9}]{HintzKdSMS}.)
\end{proof}

To make further progress, we now need to complement the dynamical structure of glued spacetimes with analytic information.

\begin{definition}[Invertibility of transition face normal operators]
\label{DefEstInvtf}
  (Cf.\ \cite[Definition~3.12]{HintzConicWave}.) Recall the quantities $\vartheta_{\rm in}$ and $\vartheta_{\rm out}$ from~\eqref{EqEstRadInTheta} and \eqref{EqEstRadOutTheta}, respectively, and recall the sc,b-operator $L_\tface(\hat\sigma)$ on $[0,\infty]_{\rho'}\times\Sph_\omega^2$ from~\eqref{EqLtf}; write $r'=\rho'{}^{-1}$. Fix a positive definite fiber inner product on $\hat\cE|_{\pa\hat X_b}$. We then say that $L_\tface$ is \emph{invertible at weight $\alpha_\cD\in\R$} if the following conditions hold for all $\hat\sigma=e^{i\theta}$, $\theta\in[0,\pi]$.
  \begin{enumerate}
  \item $\alpha_\cD+\frac32\notin\specb\hat L(0)$, or equivalently $-(\alpha_\cD+\frac32)\notin\specb(L_\tface)$ (cf.\ \eqref{EqFseIdentSpecb}).
  \item{\rm (Injectivity.)} Suppose $L_\tface(\hat\sigma)u=0$ where $|(r'\pa_{r'})^j\pa_\omega^\alpha u|\lesssim r'{}^{-\alpha_\cD-\frac32}$ for $r'\in(0,1]$, and $|(r'\pa_{r'})^j\pa_\omega^\alpha(e^{-i\hat\sigma r'}u)|\lesssim r'^C$ (for $\hat\sigma=\pm 1$), resp.\ $|\pa_{r'}^j\pa_\omega^\alpha u|\lesssim r'{}^{-N}$ (for $\Im\hat\sigma>0$) for $r'\in[1,\infty)$ and for all $j\in\N_0$, $\alpha\in\N_0^2$, $N\in\R$, and for some constant $C$. Then $u=0$.
  \item{\rm (Injectivity of the adjoint.)} Suppose $L_\tface(\hat\sigma)^*u=0$, $|(r'\pa_{r'})^j\pa_\omega^\alpha u|\lesssim r'{}^{\frac12+\alpha_\cD}$ for $r'\in(0,1]$, and $|(r'\pa_{r'})^j\pa_\omega^\alpha(e^{i\hat\sigma r'}u)|\lesssim r'{}^C$ (for $\hat\sigma=\pm 1$), resp.\ $|\pa_{r'}^j\pa_\omega^\alpha u|\lesssim r'{}^{-N}$ (for $\Im\hat\sigma>0$) for $r'\in[1,\infty)$ and for all $j,\alpha,N$ for some constant $C$. Then $u=0$.
  \end{enumerate}
\end{definition}

Via Lemma~\ref{LemmaFseIdent}, this definition is equivalent to the spectral admissibility condition for the reduced edge normal operator family of $L_\circ$ (see~\eqref{EqEstLepsLcirc}) stated in \cite[Definition~3.12]{HintzConicWave}. The following result is then the same as \cite[Lemma~3.13]{HintzConicWave}.

\begin{lemma}[Quantitative invertibility of transition face normal operators]
\label{LemmaEstMcInvft}
  Let $\sfs$, $\alpha_\cD\in\R$ be Kerr-admissible orders. Suppose $L_\tface$ is invertible at weight $\alpha_\cD$. Write $\sfs^\pm_\infty\in\CI({}^{\scop,\bop}S^*\tface)$ and $\sfs^\pm_\scop\in\CI(\ol{{}^{\scop,\bop}T^*_{\rho'{}^{-1}(0)}}\tface)$ for the regularity and scattering decay orders induced by $\sfs$ for the low energy spectral family at positive ($+$), resp.\ negative ($-$) real frequencies (cf.\ Remark~\usref{RmkEstAdmInd}). Fix on $\tface$ the density $r'{}^2|\dd r'\,\dd g_{\Sph^2}|$. Then there exists a constant $C$ so that for all $\hat\sigma=e^{i\theta}$ where $\theta\in[0,\frac{\pi}{4}]$ (for the `$+$' sign), resp.\ $\theta\in[\frac{3\pi}{4},\pi]$ (for the `$-$' sign), the estimate
  \begin{equation}
  \label{EqEstMcInvft}
    \|u\|_{H_{\scop,\bop}^{\sfs_\infty^\pm,\sfs_\scop^\pm+\alpha_\cD,-\alpha_\cD}(\tface;\pi_\tface^*\hat\cE|_{\pa\hat X_b})} \leq C\|L_\tface(\hat\sigma)u\|_{H_{\scop,\bop}^{\sfs_\infty^\pm-2,\sfs_\scop^\pm+\alpha_\cD+1,-\alpha_\cD-2}(\tface;\pi_\tface^*\hat\cE|_{\pa\hat X_b})}
  \end{equation}
  holds for all $u$ for which both sides are finite; moreover, for all $f\in H_{\scop,\bop}^{\sfs_\infty^\pm-2,\sfs_\scop^\pm+\alpha_\cD+1,-\alpha_\cD-2}$ there exists a (unique) $u\in H_{\scop,\bop}^{\sfs_\infty^\pm,\sfs_\scop^\pm+\alpha_\cD,-\alpha_\cD}$ with $L_\tface(\hat\sigma)u=f$. An analogous statement holds for $\theta\in[\frac{\pi}{4},\frac{3\pi}{4}]$, where now the order $\sfs$ is allowed to be arbitrary near $\pa\hat X_b$ and one can replace $\sfs_\scop^\pm+\alpha_\cD+1$ on the right by $\sfs_\scop^\pm+\alpha_\cD$.
\end{lemma}

We also recall that the set of $\alpha_\cD$ so that $L_\tface$ is invertible at weight $\alpha_\cD$ is open: the boundary spectrum is a discrete set, and the Fredholm estimates underlying the proof of Lemma~\ref{LemmaEstMcInvft} hold for an open set of weights; the invertibility for one weight then implies that for nearby weights by standard perturbation theory argument as in \cite[\S{2.7}]{VasyMicroKerrdS}.

\begin{prop}[Conditional estimate near low energies]
\label{PropEstFTlo}
  Let $\sfs$, $\alpha_\cD$ be Kerr-admissible orders so that also $\sfs-1,\alpha_\cD$ are Kerr-admissible. Suppose that the zero energy operator~\eqref{EqEstFT0Op} is invertible and that also $L_\tface$ is invertible at weight $\alpha_\cD$. Then there exists a constant $C$ so that for all $\hat\sigma=e^{i\theta}$ where $\theta\in[0,\frac{\pi}{4}]$, resp.\ $\theta\in[\frac{3\pi}{4},\pi]$ the uniform low energy estimate
  \[
    \|u\|_{H_{\scbtop,|\sigma|}^{\sfs,\sfs+\alpha_\cD,\alpha_\cD,0}(\hat X_b;\hat\cE)} \leq C\|\hat L(\hat\sigma|\sigma|)u\|_{H_{\scbtop,|\sigma|}^{\sfs-1,\sfs+\alpha_\cD+1,\alpha_\cD+2,0}(\hat X_b;\hat\cE)}, \qquad |\sigma|\leq 1,
  \]
  holds where the orders are induced by $\sfs$ for positive, resp.\ negative real frequencies. For $\theta\in[\frac{\pi}{4},\frac{3\pi}{4}]$, this holds for orders $\sfs$ which are arbitrary near $\pa\hat X_b$, and one can moreover replace $\sfs+\alpha_\cD+1$ on the right by $\sfs+\alpha_\cD$.
\end{prop}
\begin{proof}
  The proof is completely analogous to that of \cite[Proposition~3.21]{HintzKdSMS} (although we switch the order in which the transition face and zero energy operator estimates are used). We recall the argument briefly since we will use it in a slightly modified context in \cite{HintzGlueLocIII}. One first establishes (using radial point estimates at the radial sets over the scattering face of $(\hat X_b)_\scbtop$, as well as at $\pa\hat\cR_{\cH^+}$) the estimate
  \begin{equation}
  \label{EqEstFTloEst}
    \|u\|_{H_{\scbtop,|\sigma|}^{\sfs,\sfs+\alpha_\cD,\alpha_\cD,0}} \leq C\Bigl( \|\hat L(\hat\sigma|\sigma|)u\|_{H_{\scbtop,|\sigma|}^{\sfs-1,\sfs+\alpha_\cD+1,\alpha_\cD+2,0}} + \|u\|_{H_{\scbtop,|\sigma|}^{\sfs_0,\sfs_0+\alpha_\cD,\alpha_\cD,0}}\Bigr)
  \end{equation}
  where we fix an order $\sfs_0$ which satisfies $\sfs_0<\sfs-1$ and which has the property that $\sfs_0,\alpha_\cD$ are Kerr-admissible.

  In the second step, one localizes the error term near the transition face using a cutoff $\chi$ (equal to $1$ in a collar neighborhood of $\tface$, and supported in a slightly larger collar neighborhood). Using~\eqref{EqF3scbtNormzf} (with $w=3$) and Lemma~\ref{LemmaEstMcInvft}, we then estimate
  \begin{align}
    \|u\|_{H_{\scbtop,|\sigma|}^{\sfs_0,\sfs_0+\alpha_\cD,\alpha_\cD,0}} &\leq C\Bigl( |\sigma|^{-\alpha_\cD-\frac32}\|\Psi_{|\sigma|}^*(\chi u)\|_{H_{\scop,\bop}^{\sfs_0,\sfs_0+\alpha_\cD,-\alpha_\cD} } + \|(1-\chi)u\|_{H_{\scbtop,|\sigma|}^{\sfs_0,\sfs_0+\alpha_\cD,\alpha_\cD,0}}\Bigr) \nonumber\\
      &\leq C\Bigl( |\sigma|^{-\alpha_\cD-\frac32}\|L_\tface(\hat\sigma)\Psi_{|\sigma|}^*(\chi u)\|_{H_{\scop,\bop}^{\sfs_0-2,\sfs_0+\alpha_\cD+1,-\alpha_\cD-2}} + \|u\|_{H_{\scbtop,|\sigma|}^{\sfs_0,\sfs_0+\alpha_\cD,-N,0}}\Bigr) \nonumber\\
  \label{EqEstFTloEst2}
      &\leq C\Bigl( \|\hat L(\hat\sigma|\sigma|)u\|_{H_{\scbtop,|\sigma|}^{\sfs_0-2,\sfs_0+\alpha_\cD+1,\alpha_\cD+2,0}} + \|u\|_{H_{\scbtop,|\sigma|}^{\sfs_0,\sfs_0+\alpha_\cD+1,\alpha_\cD-1,0}}\Bigr),
  \end{align}
  where in the passage to the third line we replaced $L_\tface(\hat\sigma)$ by $\sigma^{-2}\hat L(\hat\sigma|\sigma|)$ and absorbed the error term (including the one arising from commuting through $\chi$) into the final, error term.

  In the final step, we localize to a collar neighborhood of the zero energy face using a cutoff which we again denote $\chi$. Let $0<\eta\leq 1$ be such that the zero energy operator invertibility in Proposition~\ref{PropEstFT0} holds also for the orders $\sfs,\alpha_\cD-\eta$. We can then estimate the final, error, term of~\eqref{EqEstFTloEst2} further by
  \begin{align*}
    \|u\|_{H_{\scbtop,|\sigma|}^{\sfs_0,\sfs_0+\alpha_\cD+1,\alpha_\cD-1,0}} &\leq C\Bigl(\|\chi u\|_{\Hb^{\sfs_0,\alpha_\cD-\eta}} + \|(1-\chi)u\|_{H_{\scbtop,|\sigma|}^{\sfs_0,\sfs_0+\alpha_\cD+1,\alpha_\cD-1,-N}}\Bigr) \\
      &\leq C\Bigl( \|\hat L(0)(\chi u)\|_{\Hb^{\sfs_0-1,\alpha_\cD-\eta+2}} + \|u\|_{H_{\scbtop,|\sigma|}^{\sfs_0,\sfs_0+\alpha_\cD+1,\alpha_\cD-1,-N}}\Bigr) \\
      &\leq C\Bigl( \|\hat L(\hat\sigma|\sigma|)u \|_{H_{\scbtop,|\sigma|}^{\sfs_0-1,-N,\alpha_\cD-\eta+2,0}} + \|u\|_{H_{\scbtop,|\sigma|}^{\sfs_0,-N,\alpha_\cD-\eta,-1}} \\
      &\quad \hspace{14em} + \|u\|_{H_{\scbtop,|\sigma|}^{\sfs_0,\sfs_0+\alpha_\cD+1,\alpha_\cD-1,-N}}\Bigr),
  \end{align*}
  where in the passage to the third line we replaced $\hat L(0)\chi$ by the operator $\hat L(\sigma)\chi$, $\sigma=\hat\sigma|\sigma|$, which differs from it by an element of $\Diff_\scbtop^{1,-N,-2,-1}$ for all $N$, and commuted this through the cutoff $\chi$.

  Altogether, we have now improved the error term in~\eqref{EqEstFTloEst} to
  \[
    \|u\|_{H_{\scbtop,|\sigma|}^{\sfs_0,\sfs_0+\alpha_\cD+1,\alpha_\cD-\eta,-1}} \leq C\eps^\eta\|u\|_{H_{\scbtop,|\sigma|}^{\sfs,\sfs+\alpha_\cD,\alpha_\cD,0}}
  \]
  where we used that $\sfs>\sfs_0+1$. For small $\eps>0$, this can thus be absorbed into the left hand side of~\eqref{EqEstFTloEst}. This finishes the proof.
\end{proof}

\subsection{Invertibility of the model operator on \texorpdfstring{$M_\circ$}{the lift of the original spacetime} on edge Sobolev spaces}
\label{SsEstMc}

With $\Omega=\Omega_{t_0,t_1,r_0}\subset M$ as before, we now consider on $\Omega_\circ:=\upbeta_\circ^{-1}\Omega$ the operator $L_\circ$ defined in~\eqref{EqEstLepsLcirc}.

\begin{thm}[Solvability and uniqueness on non-refocusing domains]
\label{ThmEstMc}
   Suppose $L_\tface$ is invertible at weight $\alpha_\cD\in\R$ (Definition~\usref{DefEstInvtf}), and set $\ell=-\alpha_\cD$. Let $\sfs\in\CI(\Se^*_\Omega M_\circ)$ be monotonically decreasing along the null-bicharacteristic flow of $|x|^2 H_G$ (see~\S\usref{SssGlDynMc}) and constant near $\pa\cR_{\rm in}$ and $\pa\cR_{\rm out}$. Suppose that
  \begin{equation}
  \label{EqEstMcThres}
    \sfs-\ell>\frac12(-1+\vartheta_{\rm in})\ \ \text{at}\ \ \pa\cR_{\rm in},\qquad
    \sfs-\ell<\frac12(-1-\vartheta_{\rm out})\ \ \text{at}\ \ \pa\cR_{\rm out}.
  \end{equation}
  Then there exists a constant $C$ so that the following holds: for all $f\in\He^{\sfs-1,\ell-2}(\Omega_\circ;E_\circ)^{\bullet,-}$, there exists a unique solution $u\in\He^{\sfs,\ell}(\Omega_\circ;E_\circ)^{\bullet,-}$ of the equation $L_\circ u=f$, and it satisfies the estimate
  \[
    \|u\|_{\He^{\sfs,\ell}(\Omega_\circ;E_\circ)^{\bullet,-}} \leq C\|f\|_{\He^{\sfs-1,\ell-2}(\Omega_\circ;E_\circ)^{\bullet,-}}.
  \]
\end{thm}

In the present work, such order functions $\sfs$ arise as restrictions to $\Sse^*_{M_\circ}\wt M=\Se^*M_\circ$ of order functions as in Theorem~\ref{ThmEstStd}.

\begin{proof}[Proof of Theorem~\usref{ThmEstMc}]
  This is the main result of \cite[Theorem~3.18]{HintzConicWave}, whose hypotheses we proceed to verify. Condition~\eqref{ItGlDynStdNonrf} in Definition~\ref{DefGlDynStdM} is the non-refocusing assumption in \cite[Definition~2.5]{HintzConicWave}. The $P$-admissibility of the orders $\sfs,\ell$ in the sense of \cite[Definition~3.10]{HintzConicWave} is equivalent to our present assumptions~\eqref{EqEstMcThres} (due to an effective sign switch in the definition of $\vartheta_{\rm out}$ in \cite[Definition~3.2]{HintzConicWave}). Finally, the spectral admissibility, \cite[Definition~3.12]{HintzConicWave}, is captured by our Definition~\ref{DefEstInvtf}.
\end{proof}

\section{Linear toy model: scalar waves on glued spacetimes}
\label{SSc}

We fix a glued spacetime $(\wt M,\wt g)$ associated with the data $(M,g)$, $\cC\subset M$, $b=(\bhm,\bha)$, and Fermi normal coordinates $t\in I_\cC\subset\R$, $x\in\R^3$ as in Definition~\ref{DefGl}; write $g_\eps=\wt g|_{M_\eps}$. We study the scalar wave operator
\[
  \wt L = \Box_{\wt g}\qquad \text{(i.e.\ $L_\eps:=\wt L|_{M_\eps}=\Box_{g_\eps}$)}.
\]
Thus, the bundles $\wt E$ and $\hat\cE$ in condition~\eqref{ItEstLBundle} in~\S\ref{SEst} are trivial. The normal operators from~\eqref{EqEstL} and \eqref{EqEstLepsLcirc} are the scalar wave operators
\[
  L=\Box_{\hat g_b},\qquad L_\circ=\Box_g
\]
on the subextremal Kerr spacetime $(\hat M_b^\circ,\hat g_b)$ and on $(M,g)$ (or rather $(M_\circ,\upbeta_\circ^*g)$), respectively. Since $L$ is formally self-adjoint, the threshold shifts in Definition~\ref{DefEstAdm} are
\[
  \vartheta_\bullet=0,\qquad \bullet={\rm in,out},\cH^+.
\]
Unless noted otherwise, we shall only assume the regularity $\hat\rho^{-2}(\CI+\cC_\seop^{\infty,1,1})\Diffse^2$ for $\wt L$. (For any fixed orders in the statements below, a sufficiently large but finite value of the se-regularity order would suffice.)

\begin{rmk}[More general operators]
\label{RmkScGeneral}
  With minor purely notational modifications to the arguments below, we can also consider scalar operators $\wt L$ as in~\eqref{ItEstLOp} whose $\hat M$-normal operator is equal to $\Box_{\hat g_b}$. Such a more general setting is required to accommodate nonlinear settings; see Theorem~\ref{ThmNTame} below for the corresponding (tame) linear statement. More generally still, we can treat any operator $\wt L$ (with principal part $\Box_{\hat g_b}$, but acting on vector bundles, and without any symmetry conditions) as long as its $\hat M$-model $L$ satisfies mode stability in $\Im\sigma\geq 0$, and the condition on the subprincipal symbol at the trapped set from Proposition~\ref{PropEstTrap} is satisfied. No modifications are required for the arguments below, except the interval of weights $\alpha_\cD:=\alpha_\circ-\hat\alpha$ for which the zero energy operator $\hat L(0)$ is invertible needs to be adjusted accordingly as in \cite[Definition~3.14]{HintzNonstat}. (An important example where mode stability fails at $\sigma=0$ is when $L$ is the linearized gauge-fixed Einstein operator on Kerr, discussed in \cite{HaefnerHintzVasyKerr}; uniform estimates in this case are at the heart of \cite{HintzGlueLocIII}.)
\end{rmk}

\subsection{Estimates for scalar waves on Kerr in 3b-spaces}
\label{SsSc3b}

We do not merely have the (automatic) invertibility of the spectral family of the scalar wave operator on Kerr for high frequencies, which follows from Proposition~\ref{PropEstFThi}, but in fact mode stability in all of $\Im\sigma\geq 0$:

\begin{lemma}[Spectral family of the scalar wave operator on Kerr]
\label{LemmaSc3bSpec}
  For $\alpha_\cD\in(-\frac32,-\frac12)$, the operators $\hat L(0)$ in~\eqref{EqEstFT0Op} and $\hat L(\sigma)$, $\sigma\in\C\setminus\{0\}$, $\Im\sigma\geq 0$, in \eqref{EqEstFTbddOp} are invertible. Furthermore, $L_\tface$ is invertible at weight $\alpha_\cD$, and thus the low energy estimates of Proposition~\usref{PropEstFTlo} hold.
\end{lemma}
\begin{proof}
  The invertibility of the zero energy operator is the content of \cite[Lemma~3.19]{HintzKdSMS}. Nonzero real $\sigma$ are discussed in \cite[Proposition~3.18]{HintzKdSMS}; the invertibility of $\hat L(\sigma)$ for $\Im\sigma>0$ is shown in \cite[\S{3.9}]{HintzKdSMS} using \cite{ShlapentokhRothmanModeStability,WhitingKerrModeStability} and a continuity argument. The invertibility of $L_\tface$ is the content of \cite[Lemma~3.20]{HintzKdSMS} for $\hat\sigma=1$ (with $\hat\sigma=-1$ being completely analogous), and the case $\hat\sigma=e^{i\theta}$ for $\theta\in(0,\pi)$ is discussed in \cite[\S{3.9}]{HintzKdSMS} as well.
\end{proof}

\begin{thm}[Scalar waves on Kerr: 3b-estimates]
\label{ThmSc3b}
  Let $\sfs\in\CI(\Stb^*\hat M_b)$, $\alpha_\cD\in(-\frac32,-\frac12)$ be Kerr-admissible orders in the sense of Definition~\usref{DefEstAdm}\eqref{ItEstAdmKerr} with $\vartheta_\bullet=0$. Fix on $\hat M_b$ the metric density $|\dd\hat g_b|$. We work with 3b-Sobolev spaces which have extendible character at the final hypersurface at $\hat r=\bhm$ (though we do not make this explicit in the notation). Then the following statements hold.
  \begin{enumerate}
  \item\label{ItSc3bSol} The operator
  \begin{equation}
  \label{EqSc3b}
    \Box_{\hat g_b} \colon \{ u\in\Htb^{\sfs,\alpha_\cD,0}(\hat M_b) \colon \Box_{\hat g_b}u \in \Htb^{\sfs,\alpha_\cD+2,0}(\hat M_b) \} \to \Htb^{\sfs,\alpha_\cD+2,0}(\hat M_b)
  \end{equation}
  is invertible. In particular, there exists a constant $C=C(b,\sfs,\alpha_\cD)$ so that the estimate
  \begin{equation}
  \label{EqSc3bEst}
    \|u\|_{\Htb^{\sfs,\alpha_\cD,0}(\hat M_b)} \leq C\|\Box_{\hat g_b}u\|_{\Htb^{\sfs,\alpha_\cD+2,0}(\hat M_b)}
  \end{equation}
  holds for all $u$ for which both sides are finite.
  \item\label{ItSc3bFwd} The inverse of~\eqref{EqSc3b} is the forward solution operator, i.e.\ if $\hat t\geq T_0$ on $\supp f$ for some $T_0\in\R$, then also $\hat t\geq T_0$ on $\supp\Box_{\hat g_b}^{-1}f$.
  \item\label{ItSc3bDom} Suppose $\Omega\subset\hat M_b$ is a $4$-dimensional submanifold with corners whose boundary hypersurfaces are all spacelike, and which has exactly one initial boundary hypersurface which is moreover contained in a level set $\{\hat t=\hat t_0\}$. Write $\Htb^{\sfs,\alpha_\cD,0}(\Omega)^{\bullet,-}$ for the space of distributions which have supported, resp.\ extendible character at the initial, resp.\ final boundary hypersurfaces of $\Omega$. Then, for the same constant $C=C(b,\sfs,\alpha_\cD)$ as in~\eqref{EqSc3bEst} (which is thus \emph{independent} of $\Omega$), we have
  \[
    \|u\|_{\Htb^{\sfs,\alpha_\cD,0}(\Omega)^{\bullet,-}} \leq C\|\Box_{\hat g_b}u\|_{\Htb^{\sfs,\alpha_\cD+2,0}(\Omega)^{\bullet,-}}.
  \]
  \end{enumerate}
\end{thm}
\begin{proof}
  In view of the isomorphism~\eqref{EqFVar3bFT}, part~\eqref{ItSc3bSol} is an immediate consequence of those statements in Lemma~\ref{LemmaSc3bSpec} which concern $\hat L(\sigma)$ for $\sigma\in\R$, as well as the high energy estimates of Proposition~\ref{PropEstFThi}. Thus,
  \[
    \cF(\Box_{\hat g_b}^{-1}f)(\sigma) = \wh{\Box_{\hat g_b}}(\sigma)^{-1}(\cF f)(\sigma),\qquad \sigma\in\R.
  \]
  To prove part~\eqref{ItSc3bFwd}, it suffices to consider $f\in\CIc(\hat M_b^\circ)$ since this space is dense in the target space $\Htb^{\sfs,\alpha_\cD+2,0}(\hat M_b)$ of~\eqref{EqSc3b}. Given $\chi\in\CI(\hat X_b^\circ)$ which is equal to $1$ on $\hat\pi(\supp f)$ (in the notation used in~\eqref{EqEstL}), consider then
  \[
    u_\chi(\hat t,-) = I(0),\qquad I(\alpha) := \frac{1}{2\pi}\int_{\R+i\alpha} e^{i\hat t\sigma} \chi\wh{\Box_{\hat g_b}}(\sigma)^{-1}\chi(\cF f)(\sigma)\,\dd\sigma.
  \]
  By the holomorphicity of $\chi\wh{\Box_{\hat g_b}}(\sigma)^{-1}\chi$ in $\Im\sigma>0$ as a map $H^{\sfs-1}\to H^\sfs$, and in view of the high energy estimates that this map satisfies, $I(\alpha)$ is independent of $\alpha>0$. The uniform resolvent estimates in $\Im\sigma\geq 0$ near $\pm(0,\infty)$ allow one to shift the contour to the union of $(-\infty,-1]$, an upper semicircle from $-1$ to $1$, and $[1,\infty)$; and the holomorphicity in $\Im\sigma>0$ together with the uniform low energy estimates of Proposition~\ref{PropEstFTlo} imply that one can shift the semicircle contour down to $[-1,1]$. Therefore, $u_\chi=\chi u$ where $u=\Box_{\hat g_b}^{-1}f$. The claim now follows from the Paley--Wiener theorem (now using the uniformity of the estimates~\eqref{EqEstFThi} as $\Im\sigma\to+\infty$).

  For part~\eqref{ItSc3bDom}, let $f\in\Htb^{\sfs,\alpha_\cD+2,0}(\Omega)^{\bullet,-}$. Denote by $\tilde f\in\dot H_\tbop^{\sfs,\alpha_\cD+2,0}(\{\hat t\geq\hat t_0\})$ the extension of $f$ with minimal norm, so $\|\tilde f\|_{H_\tbop^{\sfs,\alpha_\cD+2,0}(\hat M_b)}=\|f\|_{\Htb^{\sfs,\alpha_\cD+2,0}(\Omega)^{\bullet,-}}$. Let $\tilde u=\Box_{\hat g_b}^{-1}\tilde f\in\dot H_\tbop^{\sfs,\alpha_\cD,0}(\{\hat t\geq\hat t_0\})$. Then $u:=\tilde u|_\Omega$ is the (unique) forward solution of $\Box_{\hat g_b}u=f$. By the definition of norms on spaces of extendible distributions, we have
  \[
    \|u\|_{\Htb^{\sfs,\alpha_\cD,0}(\Omega)^{\bullet,-}} \leq \|\tilde u\|_{\Htb^{\sfs,\alpha_\cD,0}(\hat M_b)} \leq C\|\tilde f\|_{\Htb^{\sfs,\alpha_\cD+2,0}} = C\|f\|_{\Htb^{\sfs,\alpha_\cD+2,0}(\Omega)^{\bullet,-}}. \qedhere
  \]
\end{proof}

Considering the forward mapping properties of $\Box_{\hat g_b}$, the estimate~\eqref{EqSc3bEst} is sharp as far as the weights (i.e.\ decay rates) are concerned; indeed, for $u\in\Htb^{\sfs,\alpha_\cD,0}(\hat M_b)$ one has $\Box_{\hat g_b}u\in\Htb^{\sfs-2,\alpha_\cD+2,0}(\hat M_b)$. For present purposes, we have no need for sharper decay estimates (e.g.\ Price's law \cite{HintzPrice,AngelopoulosAretakisGajicKerr}) when the source term has better decay.

\begin{rmk}[The space $\Htb^{\sfs,\alpha_\cD,0}$]
\label{RmkSc3bSpace}
  In order to give some intuition for the space $\Htb^{\sfs,\alpha_\cD,0}$ in~\eqref{EqSc3b}, we replace it by a closely related space of functions on $\hat M_b=\R_{\hat t}\times\{\hat x\in\R^3\colon|\hat x|\geq\bhm\}$ which can be defined without 3b-ps.d.o.s. To wit, fix $\alpha_\cD\in(-\frac32,-\frac12)$ and $s=0$ (so $s+\alpha_\cD<-\frac12$), and consider for $k\in\N_0$ the space
  \[
    \tilde H_\tbop^{(0;k),\alpha_\cD}(\hat M_b) := \bigl\{ u \colon \pa_{\hat t}^j\pa_{\hat x}^\alpha(\hat r(\pa_{\hat t}+\pa_{\hat r}))^l \pa_\omega^\beta u\in \la\hat x\ra^{-\alpha_\cD}L^2(\hat M_b)\ \forall\,j+|\alpha|+l+|\beta|\leq k \bigr\},
  \]
  where we use the volume density $|\dd\hat t\,\dd\hat x|$ to define $L^2$. The point is that the space of 3b-vector fields which are characteristic at $\tilde\cR_{\rm out}:=\{\hat r=\infty\}\cap\mathspan\frac{\dd(\hat t-\hat r)}{\hat r}\subset\Ttb^*\hat M_b\setminus o$---which is the $\hat t$-invariant extension of $\hat\cR_{\rm out}$ defined in~\eqref{EqGlDynKerrOut}---is spanned over $\CI(\hat M_b)$ by $\pa_{\hat t}$, $\pa_{\hat x}$, $\hat r(\pa_{\hat t}+\pa_{\hat r})$, $\pa_\omega$. (That is, $\tilde H_\tbop^{(0;k),\alpha_\cD}$ captures \emph{module regularity} at $\tilde\cR_{\rm out}$ relative to $\Htb^{0,\alpha_\cD}$, which is a frequently used variant of, or addition to, variable order spaces \cite{BaskinVasyWunschRadMink,GellRedmanHassellShapiroZhangHelmholtz,VasyLAPLag}. More refined versions of this space arise in \cite[\S{3.6}]{HintzGlueLocIII}.) In particular, elements of $\tilde H_\tbop^{(0;1),\alpha_\cD}$ are microlocally in $\Htb^{1,\alpha_\cD}$ away from $\tilde\cR_{\rm out}$ and thus, due to $1+\alpha_\cD>-\frac12$, have above-threshold regularity at the $\hat t$-invariant extension of $\hat\cR_{\rm in}$. Let now $u\in\tilde H_\tbop^{(0;1),\alpha_\cD}$. (We remark that one can show that $\Box_{\hat g_b}^{-1}\colon\tilde H_\tbop^{(0;k),\alpha_\cD+2}\to\tilde H_\tbop^{(0;k),\alpha_\cD}$ for all $k\geq 1$.) Passing to $\hat t_*=\hat t-\hat r$, we thus have $\la\hat x\ra^{\alpha_\cD}u\in L^2(\R_{\hat t_*};\Hb^1(\hat X_b))\cap H^1(\R_{\hat t*};L^2(\hat X_b))$; this is consistent for bounded $\hat t_*$ with $u\sim\la\hat x\ra^{-\eta}$ for any $\eta>\alpha_\cD+\frac32\in(0,1)$, and thus in particular allows for a nontrivial radiation field at $\scri^+$. On the other hand, for $\frac{\hat r}{\hat t}$ near $-1$ (or merely bounded away from $+1$), the fact that $u$ remains in $\la\hat x\ra^{-\alpha_\cD}L^2$ when differentiating along $\hat r(\pa_{\hat t}+\pa_{\hat r})$ means, roughly speaking, that $u$ cannot have a nontrivial radiation field at $\scri^-$. See also \cite[Lemma~5.7]{HintzVasyMink4} for the relationship between such b-regularity on the radial compactification of $\R^{1+3}$ (here away from $\scri^+$ and $\cT^\pm$) and decay near null infinity.
\end{rmk}

\subsection{Uniform estimates for linear waves}
\label{SsScUnif}

We now combine the symbolic estimates from~\S\S\ref{SsEstRad}--\ref{SsEstStd} with the normal operator estimates at $M_\circ$ and $\hat M$ from~\S\ref{SsEstMc} and \S\ref{SsSc3b} to prove uniform bounds for linear scalar waves on glued spacetimes on se-Sobolev spaces.

\begin{thm}[Uniform se-estimates for linear waves on standard domains in glued spacetimes]
\label{ThmScUnif}
  Let $\wt\Omega\subset\wt M\setminus\wt K^\circ$ denote a standard domain (see Definition~\usref{DefGlDynStd}). Fix $\eps_0>0$ so that the conclusions of Lemma~\usref{LemmaGlDynStdFlow}\eqref{ItGlDynStdFlowSpace} hold. Let $\sfs\in\CI(\Sse^*\wt M)$, $\alpha_\circ,\hat\alpha\in\R$ denote admissible orders on $\wt\Omega$ (see Definition~\usref{DefEstAdm}) with $\alpha_\circ-\hat\alpha\in(-\frac32,-\frac12)$, and assume that also the orders $\sfs-3,\alpha_\circ,\hat\alpha$ are admissible. Then there exists $C<\infty$ so that the unique forward solution of $\Box_{g_\eps}u=f$ on $\Omega_\eps=\wt\Omega\cap M_\eps$ satisfies the estimate
  \begin{equation}
  \label{EqScUnif}
    \|u\|_{H_{\seop,\eps}^{\sfs,\alpha_\circ,\hat\alpha}(\Omega_\eps)^{\bullet,-}} \leq C \|f\|_{H_{\seop,\eps}^{\sfs,\alpha_\circ,\hat\alpha-2}(\Omega_\eps)^{\bullet,-}}
  \end{equation}
  for all $\eps\in(0,\eps_0]$.
\end{thm}
\begin{proof}
  For any fixed $\eps_1\in(0,\eps_0]$ and for all $\eps\in[\eps_1,\eps_0]$, the estimate~\eqref{EqScUnif} follows from the standard hyperbolic estimate $\|u\|_{H^\sfs(\Omega_\eps)}\leq C'\|f\|_{H^{\sfs-1}(\Omega_\eps)}$ since weighted se- and standard Sobolev norms are equivalent for $\eps$ bounded away from $0$. When $\wt\Omega$ is a standard domain of the type in Definition~\ref{DefGlDynStd}\eqref{ItGlDynStdNotC}, then the norms on $H_{\seop,\eps}^{\sfs,0,\hat\alpha-\alpha_\circ}$ and $H^\sfs$ are uniformly equivalent for all $\eps\in(0,1)$, and~\eqref{EqScUnif} (multiplied by $\eps^{\alpha_\circ}$) is again equivalent to a standard hyperbolic estimate. It thus suffices to prove~\eqref{EqScUnif} for $\eps\leq\eps_0$ where $\eps_0>0$ may depend on $\sfs,\alpha_\circ,\hat\alpha$, and for domains $\wt\Omega$ of the type in Definition~\ref{DefGlDynStd}\eqref{ItGlDynStdC}, associated with a standard domain $\Omega=\Omega_{t_0,t_1,r_0}$. 

  \pfstep{Control of se-regularity.} Pick $\sfs_0$ with $\sfs_0<\sfs-3-\eta$ for some $\eta>0$ and so that $\sfs_0,\alpha_\circ,\hat\alpha$ are admissible. Then Theorem~\ref{ThmEstStd} gives
  \begin{equation}
  \label{EqScUnifPf}
    \|u\|_{H_{\seop,\eps}^{\sfs,\alpha_\circ,\hat\alpha}(\Omega_\eps)^{\bullet,-}} \leq C\Bigl(\|\Box_{g_\eps}u\|_{H_{\seop,\eps}^{\sfs,\alpha_\circ,\hat\alpha-2}(\Omega_\eps)^{\bullet,-}} + \|u\|_{H_{\seop,\eps}^{\sfs_0,\alpha_\circ,\hat\alpha}(\Omega_\eps)^{\bullet,-}}\Bigr)
  \end{equation}

  \pfstep{Control near $\hat M$.} To estimate the final, error, term, we first use Proposition~\ref{PropEstFnSobRel} to relate, via $\Psi_\eps(t,x)=(\eps,\frac{t-t_0}{\eps},\frac{x}{\eps})=(\eps,\hat t,\hat x)$, its norm with a 3b-norm. Set
  \[
    \hat\Omega_\eps=\bigl\{0\leq\hat t\leq\eps^{-1}(t_1-t_0),\ \bhm\leq|\hat x|\leq\eps^{-1}r_0+2(\eps^{-1}(t_1-t_0)-t)\bigr\}.
  \]
 For $\eta'\in(0,\eta)$, we then have
  \[
    \|u\|_{H_{\seop,\eps}^{\sfs_0,\alpha_\circ,\hat\alpha}(\Omega_\eps)^{\bullet,-}} \leq C\eps^{2-\hat\alpha}\|(\Psi_\eps)_*u\|_{H_\tbop^{\sfs_0+\eta',\alpha_\cD,0}(\hat\Omega_\eps)^{\bullet,-}}
  \]
  for all sufficiently small $\eps>0$ where $\alpha_\cD:=\alpha_\circ-\hat\alpha\in(-\frac32,-\frac12)$. Using Theorem~\ref{ThmSc3b}, we can bound this by $C\eps^{2-\hat\alpha}\|\Box_{\hat g_b}((\Psi_\eps)_*u)\|_{\Htb^{\sfs_0+\eta',\alpha_\cD+2,0}(\hat\Omega_\eps)^{\bullet,-}}$. But since
  \[
    \wt R := \Box_{\wt g} - \eps^{-2}(\Psi_\eps)^*\Box_{\hat g_b} \in (\hat\rho\CI+\cC_\seop^{\infty,1,1})\Diffse^{2,0,2}(\wt M\setminus\wt K^\circ)
  \]
  (cf.\ condition~\eqref{ItEstLOp} in~\S\ref{SEst}) has one order of vanishing more at $\hat M$ than $\Box_{\wt g}$, we can further bound (again using Proposition~\ref{PropEstFnSobRel})
  \begin{align*}
    \|u\|_{H_{\seop,\eps}^{\sfs_0,\alpha_\circ,\hat\alpha}(\Omega_\eps)^{\bullet,-}} &\leq C\eps^{2-\hat\alpha}\eps^2\|(\Psi_\eps)_*(( \Box_{\wt g}-\wt R )u)\|_{H_\tbop^{\sfs_0+\eta',\alpha_\cD+2,0}(\hat\Omega_\eps)^{\bullet,-}} \\
      &\leq C\Bigl(\|\Box_{g_\eps}u\|_{H_{\seop,\eps}^{\sfs_0+\eta',\alpha_\circ,\hat\alpha-2}(\Omega_\eps)^{\bullet,-}} + \|\wt R u\|_{H_{\seop,\eps}^{\sfs_0+\eta',\alpha_\circ,\hat\alpha-2}(\Omega_\eps)^{\bullet,-}}\Bigr) \\
      &\leq C\Bigl(\|\Box_{g_\eps}u\|_{H_{\seop,\eps}^{\sfs_0+\eta',\alpha_\circ,\hat\alpha-2}(\Omega_\eps)^{\bullet,-}} + \|u\|_{H_{\seop,\eps}^{\sfs_0+2+\eta',\alpha_\circ,\hat\alpha-1}(\Omega_\eps)^{\bullet,-}}\Bigr).
  \end{align*}
  Plugging this into~\eqref{EqScUnifPf}, we obtain, for any $\beta\leq 1$ (chosen below),
  \begin{equation}
  \label{EqScUnifPf2}
    \|u\|_{H_{\seop,\eps}^{\sfs,\alpha_\circ,\hat\alpha}(\Omega_\eps)^{\bullet,-}} \leq C\Bigl(\|\Box_{g_\eps}u\|_{H_{\seop,\eps}^{\sfs,\alpha_\circ,\hat\alpha-2}(\Omega_\eps)^{\bullet,-}} + \|u\|_{H_{\seop,\eps}^{\sfs_0+2+\eta',\alpha_\circ,\hat\alpha-\beta}(\Omega_\eps)^{\bullet,-}}\Bigr)
  \end{equation}
  and have thus weakened the $\hat M$-decay order at the (acceptable) expense of a stronger se-regularity order.

  \pfstep{Control near $M_\circ$.} Let now $\chi_\circ=\chi(\frac{\eps}{|x|})$, where $\chi\in\CIc([0,c_0))$ equals $1$ near $0$, and $c_0>0$ is small and chosen momentarily. We then apply the triangle inequality to estimate the norm of $u=\chi_\circ u+(1-\chi_\circ)u$ in $H_{\seop,\eps}^{\sfs_0+2+\eta',\alpha_\circ,\hat\alpha-\beta}(\Omega_\eps)^{\bullet,-}$, with the second term having support disjoint from $M_\circ$, and use Proposition~\ref{PropEstFnSobRel} to estimate
  \begin{subequations}
  \begin{equation}
  \label{EqScUnifPfe1}
    \|u\|_{H_{\seop,\eps}^{\sfs_0+2+\eta',\alpha_\circ,\hat\alpha-\beta}(\Omega_\eps)^{\bullet,-}} \leq C\eps^{-\alpha_\circ}\Bigl(\|\chi_\circ u\|_{H_\eop^{\sfs_\circ+2+\eta'',\ell-\beta}(\Omega)^{\bullet,-}} + \|u\|_{H_{\seop,\eps}^{\sfs_0+2+\eta',-N,\hat\alpha-\beta}(\Omega_\eps)^{\bullet,-}}\Bigr)
  \end{equation}
  where $\ell=\hat\alpha-\alpha_\circ$ and $-N<\alpha_\circ$ is arbitrary but fixed, further $\sfs_\circ=\sfs_0|_{\Sse^*_{M_\circ}\wt M}$; and we fix $\eta''\in(\eta',\eta)$, and take $c_0$ so small that the $\eps$-independent extension of $\sfs_\circ+\eta''$ on $\supp\chi_\circ$ is pointwise larger than $\sfs_0+\eta'$. Choosing $\beta>0$ in~\eqref{EqScUnifPf2} so that $\ell-\beta\in(\frac12,\frac32)$ still, and using that the orders $\sfs_\circ+2+\eta'',\ell-\beta$ satisfy the requirements of Theorem~\ref{ThmEstMc} when $\beta,\eta''$ are sufficiently small, we can estimate
  \begin{equation}
  \label{EqScUnifPfe2}
  \begin{split}
    &\eps^{-\alpha_\circ}\|\chi_\circ u\|_{\He^{\sfs_\circ+2+\eta'',\ell-\beta}(\Omega)^{\bullet,-}} \\
    &\qquad \leq C\eps^{-\alpha_\circ}\|\Box_g(\chi_\circ u)\|_{\He^{\sfs_\circ+1+\eta'',\ell-2-\beta}(\Omega)^{\bullet,-}} \\
    &\qquad \leq C\Bigl(\|\Box_{g_\eps}(\chi_\circ u)\|_{H_{\seop,\eps}^{\sfs_0+1+\eta''',\alpha_\circ,\hat\alpha-2-\beta}(\Omega_\eps)^{\bullet,-}} + \|R_\circ u\|_{H_{\seop,\eps}^{\sfs_0+1+\eta''',\alpha_\circ,\hat\alpha-2-\beta}(\Omega_\eps)^{\bullet,-}}\Bigr),
  \end{split}
  \end{equation}
  where $R_\circ=(\Box_{\wt g}-\Box_g)\circ\chi_\circ\in(\rho_\circ\CI+\cC_\seop^{\infty,1,1})\Diffse^{2,0,2}(\wt M\setminus\wt K^\circ)$ has one order of vanishing more at $M_\circ$ than $\Box_{\wt g}$; moreover, $\eta'''\in(\eta'',\eta)$, and we work with sufficiently small $\eps$ so that $\sfs_\circ+3+\eta''\leq\sfs_0+3+\eta'''$ near $\supp\chi_\circ$. Using now that $[\Box_{g_\eps},\chi_\circ]\in(\rho_\circ^N\CI+\cC_\seop^{\infty,N,1})\Diffse^{1,0,2}$ for all $N$, this is further bounded by a uniform constant times
  \begin{equation}
  \label{EqScUnifPfe3}
    \|\chi_\circ\Box_{g_\eps}u\|_{H_{\seop,\eps}^{\sfs_0+1+\eta''',\alpha_\circ,\hat\alpha-2-\beta}(\Omega_\eps)^{\bullet,-}} + \|u\|_{H_{\seop,\eps}^{\sfs_0+3+\eta''',\alpha_\circ-1,\hat\alpha-\beta}(\Omega_\eps)^{\bullet,-}}.
  \end{equation}
  \end{subequations}

  \pfstep{Conclusion.} Using~\eqref{EqScUnifPfe1}--\eqref{EqScUnifPfe3} to estimate the error term of~\eqref{EqScUnifPf2} gives the uniform estimate
  \begin{equation}
  \label{EqScUnifPfFinal}
    \|u\|_{H_{\seop,\eps}^{\sfs,\alpha_\circ,\hat\alpha}(\Omega_\eps)^{\bullet,-}} \leq C\Bigl(\|\Box_{g_\eps}u\|_{H_{\seop,\eps}^{\sfs,\alpha_\circ,\hat\alpha-2}(\Omega_\eps)^{\bullet,-}} + \|u\|_{H_{\seop,\eps}^{\sfs_0+3+\eta''',\alpha_\circ-1,\hat\alpha-\beta}(\Omega_\eps)^{\bullet,-}}\Bigr).
  \end{equation}
  But since $\sfs_0+3+\eta'''<\sfs$, the error term here is bounded by a constant times
  \[
    \|u\|_{H_{\seop,\eps}^{\sfs,\alpha_\circ-1,\hat\alpha-\beta}(\Omega_\eps)^{\bullet,-}} \leq C\eps^\beta \|u\|_{H_{\seop,\eps}^{\sfs,\alpha_\circ,\hat\alpha}(\Omega_\eps)^{\bullet,-}},
  \]
  and thus can be absorbed into the left hand side for sufficiently small $\eps>0$. This finishes the proof of~\eqref{EqScUnif}.
\end{proof}

\begin{rmk}[Initial value problems]
\label{RmkScUnifIVP}
  One can convert initial value problems for $\Box_{g_\eps}$ (with trivial source terms, for simplicity of exposition) to forcing problems in the usual fashion, say on a standard domain $\wt\Omega$ associated with $\Omega_{0,t_1,r_0}$ where $t_1,r_0>0$. To wit, using standard hyperbolic theory one can solve wave equations for $g_\eps$ on domains $\frac{t}{\hat\rho}\in[0,\delta]$ where $\delta>0$ is small enough so that $\{\frac{t}{\hat\rho}=\delta'\}$ is spacelike on $\wt\Omega$ for $\delta'\in[0,\delta]$. Using a cutoff function $\chi=\chi(\frac{t}{\hat\rho})$ which vanishes on $[\frac{\delta}{2},\infty]$ and equals $1$ on $[0,\frac{\delta}{4}]$, the continuation of this very short time solution $u_{\rm in}$ can be found by solving the forward problem for $\Box_{g_\eps}v=-[\Box_{g_\eps},\chi]u_{\rm in}$: then $u_{\rm in}+v$ solves the initial value problem. We leave the details to the interested reader.
\end{rmk}

\subsubsection{Uniform bounds on small domains without the \texorpdfstring{$M_\circ$-normal operator}{Mo-normal operator}}
\label{SssScUnifSm}

In a spirit similar to \cite[Proposition~3.17]{HintzConicWave} and its proof, one can establish the uniform estimate~\eqref{EqScUnif} on standard domains associated with $\Omega=\Omega_{t_0,t_1,r_0}$ when $t_1-t_0,r_0$ are sufficiently small by only utilizing the se-regularity estimate and the inversion of the $\hat M$-normal operator (i.e.\ $\Box_{\hat g_b}$ above), but \emph{without needing the edge normal operator inversion} (Theorem~\ref{ThmEstMc}). The idea is that in the estimate~\eqref{EqScUnifPf2}, the error term is small compared to the left hand side when the domain is sufficiently localized near $\hat M$, for if $\hat\rho\lesssim\lambda\ll 1$ on $\Omega_\eps$, then the error term is, roughly, bounded by $C\lambda^\beta<\frac12$ times the left hand side and thus can be absorbed. (This observation will become important in \cite{HintzGlueLocIII} where the inversion of the $\hat M$-normal operator estimate loses $\hat t$-decay due to the presence of zero energy bound states and resonances, and the usage of a $M_\circ$-normal operator estimate would become cumbersome.) We thus explain this in some detail. Until (and including) Lemma~\ref{LemmaScUnifSmOp}, we work with general operators $\wt L$ as introduced in~\eqref{ItEstLBundle}--\eqref{ItEstLOp} in~\S\ref{SEst}.

Consider for $\lambda>0$ the map
\begin{equation}
\label{EqScUnifSmResc}
  \wt S_\lambda \colon (0,1)\times\R\times\R^3 \ni (\eps,t',x') \mapsto (\lambda\eps,\ t_0+\lambda t',\ \lambda x').
\end{equation}
In the coordinates $\hat t'=\frac{t'}{\eps}$, $\hat x'=\frac{x'}{\eps}$, and $\hat t=\frac{t-t_0}{\eps}$, $\hat x=\frac{x}{\eps}$, this map is given by $(\eps,\hat t',\hat x')\mapsto(\lambda\eps,\hat t,\hat x)=(\lambda\eps,\hat t',\hat x')$. At $\eps=0$ and in polar coordinates, this map is the restriction to $r'>0$ of the rescaling map
\[
  S_\lambda \colon \R\times[0,\infty)\times\Sph^2 \ni (t',r',\omega) \mapsto (t_0+\lambda t',\lambda r',\omega)
\]
from \cite[Proposition~3.17]{HintzConicWave}. More geometrically, define
\[
  \wt M':=[[0,1)_\eps\times(\R_{t'}\times\R^3_{x'});\{0\}\times\R\times\{0\}],
\]
where we use the coordinates $t',x'$ on $\wt M'$; set $\hat x'=\frac{x'}{\eps}$. Then $\wt S_\lambda$ extends by continuity from~\eqref{EqScUnifSmResc} to a diffeomorphism of $\wt M'$ fixing $\hat M'_0$, the fiber of the front face of $\wt M'$ over $t'=0$. By identifying a neighborhood of $\hat M_0'$ with a neighborhood of $\hat M_{t_0}$ via $t=t_0+t'$, $x=x'$, the map $\wt S_\lambda$ maps a neighborhood of $\hat M'_0$ diffeomorphically to a neighborhood of $\hat M_{t_0}$.

Let now
\begin{subequations}
\begin{equation}
\label{EqScUnifDomain1}
  \Omega':=\Omega_{0,1,1}\subset\R_{t'}\times[0,\infty)_{r'}\times\Sph^2,
\end{equation}
which is a standard domain for the Minkowski metric, and write $\wt\Omega'$ for the associated standard domain in $\wt M'$,
\begin{equation}
\label{EqScUnifDomain2}
  \wt\Omega'\subset \wt M'\setminus(\wt K')^\circ, \qquad
  \wt K' := \{ |\hat x'|\leq\bhm \}.
\end{equation}
\end{subequations}
Then $\Omega_{(\lambda)}:=S_\lambda(\Omega')=\Omega_{t_0,t_0+\lambda,\lambda}$ is a (small) standard domain in $(M,g)$ when $\lambda$ is small, and we can study $\wt g$ and $\wt L=\Box_{\wt g}$ on the associated (small) standard domain $\wt\Omega_{(\lambda)}=\wt S_\lambda(\wt\Omega')\subset\wt M$ uniformly in $\lambda$ by considering
\begin{equation}
\label{EqScUnifSmDefs}
  \wt g_{(\lambda)} := \lambda^{-2}\wt S_\lambda^*\wt g,\qquad
  \wt L_{(\lambda)} := \lambda^2\wt S_\lambda^*\wt L
\end{equation}
on the fixed domain $\wt\Omega'$. Here, $\wt L_{(\lambda)}$ acts on sections of the bundle $\wt S_\lambda^*\wt E$. In a neighborhood of $\hat M_{t_0}$ in $\wt M$, we fix an identification of $\wt E$ with the pullback of $\hat\cE=\wt E|_{\hat M_{t_0}}$ in~\eqref{EqEstLBundle} along $(\eps,t,\hat x)\mapsto(t_0,\hat x)$; we pull this back to an identification of $\wt S_\lambda^*\wt E$ with $[0,1)_\eps\times\R_{t'}\times\hat\cE$.

\begin{lemma}[Limit of rescalings: metrics]
\label{LemmaScUnifSmMetric}
  Write $\wt g_b\in\CI(\wt M'\setminus(\wt K')^\circ;S^2\wt T^*\wt M')$ for the family of Kerr metrics with parameters $\eps b$, i.e.
  \begin{equation}
  \label{EqScUnifSmMetricKerr}
    (\wt g_b)_{\mu\nu}(\eps,t,x)=(\hat g_{\eps b})_{\hat\mu\hat\nu}(x) = (\hat g_b)_{\hat\mu\hat\nu}(x/\eps),
  \end{equation}
  where on the left, resp.\ right we compute the metric coefficients in the coordinates $z=(t,x)$, resp.\ $\hat z=(\hat t,\hat x)$.\footnote{In other words, $\wt g_b|_{M'_\eps}$ is the metric of a Kerr black hole with parameters $\eps\bhm,\eps\bha$, cf.\ \eqref{EqIResc}.} Then, for sufficiently small $\lambda_0>0$, the following statements hold.
  \begin{enumerate}
  \item\label{ItScUnifSmMetricDiff} We have
    \begin{equation}
    \label{EqScUnifSmMetric}
      \wt g_{(\lambda)} - \wt g_b \in \lambda L^\infty\bigl((0,\lambda_0]_\lambda; \hat\rho\CI_\seop(\wt\Omega';S^2\wt T^*\wt M')\bigr);
    \end{equation}
  \item\label{ItScUnifSmMetricFlow} the conclusions of Lemma~\usref{LemmaGlDynStdFlow} hold for $\wt g,\wt\Omega_{(\lambda)}$ for all $\eps\in(0,1)$;
  \item\label{ItScUnifSmMetricMono} an admissible order function $\sfs$ for $\wt g_b$ (relative to some weights $\alpha_\circ,\hat\alpha$) on $\wt\Omega'$ defined as in~\eqref{EqEstAdmIndEx} is admissible also for $\wt g_{(\lambda)}$ on $\wt\Omega'$ for all $\lambda\in(0,\lambda_0]$.
  \end{enumerate}
\end{lemma}

We can equivalently define $\wt g_b$ as
\begin{equation}
\label{EqScUnifSmMetricPsi}
  \wt g_b|_{M'_\eps} = \eps^2\Psi_\eps^*\hat g_b,\qquad
  \Psi_\eps\colon(\eps,t',x')\mapsto\Bigl(\frac{t'}{\eps},\frac{x'}{\eps}\Bigr)\in\cM.
\end{equation}
We remark that the admissibility of $\sfs$ for $\wt g_b$ is equivalent to the following two conditions: $\pa_t\sfs|_{\hat M}=0$; and $\sfs|_{\hat M_t}$, $\alpha_\circ-\hat\alpha$ are Kerr admissible orders (see Definition~\ref{ItEstAdmKerr}) for one (and thus all) $t\in I_\cC$.

\begin{proof}[Proof of Lemma~\usref{LemmaScUnifSmMetric}]
  We verify part~\eqref{ItScUnifSmMetricDiff} near the boundary $\pa M_\circ'$ of the lift $M_\circ'\subset\wt M'$ of $\{0\}\times(\R\times\R^3)$. Using the coordinates $t$, $\rho_\circ=\frac{\eps}{|x|}$, $\hat\rho=|x|$, $\omega=\frac{x}{|x|}$, we have
  \[
    \wt g_{\mu\nu}(t,\rho_\circ,\hat\rho,\omega) = (\hat g_b)_{\hat\mu\hat\nu}(\rho_\circ^{-1}\omega) + \hat\rho\wt h_{\mu\nu}(t,\rho_\circ,\hat\rho,\omega)
  \]
  where $\wt h_{\mu\nu}$ is of class $\CI+\rho_\circ\CI_\seop\subset\CI_\seop$. Using the coordinates $t',\rho_\circ'=\frac{\eps}{|x'|}$, $\hat\rho'=|x'|$, $\omega$ on $\wt M'$, the map $\wt S_\lambda$ is given by
  \begin{equation}
  \label{EqScUnifSmMetCoord}
    (t',\rho_\circ',\hat\rho',\omega)\mapsto(t,\rho_\circ,\hat\rho,\omega)=(t_0+\lambda t',\rho_\circ',\lambda\hat\rho',\omega)
  \end{equation}
  and pulls back $\dd z^\mu\,\dd z^\nu$ to $\lambda^2\,\dd z'{}^\mu\,\dd z'{}^\nu$; therefore,
  \[
    (\wt g_{(\lambda)})_{\mu\nu}(t',\rho_\circ',\hat\rho',\omega) = (\hat g_b)_{\hat\mu\hat\nu}(\rho_\circ'{}^{-1}\omega) + \lambda\hat\rho \wt h_{\mu\nu}(t_0+\lambda t',\rho'_\circ,\lambda\hat\rho',\omega).
  \]
  But $\wt h_{\mu\nu}(t_0+\lambda t',\rho'_\circ,\lambda\hat\rho,\omega)$ is (a fortiori) uniformly bounded (for $\lambda\in(0,\lambda_0]$) in $\CI_\seop$. This proves~\eqref{EqScUnifSmMetric}.

  Parts~\eqref{ItScUnifSmMetricFlow} and~\eqref{ItScUnifSmMetricMono} follow from the convergence $\wt g_{(\lambda)}\to\wt g_b$ as $\lambda\to 0$ and an inspection of the proof of Lemma~\ref{LemmaGlDynStdMono}.
\end{proof}

\begin{lemma}[Limit of rescalings: wave operators]
\label{LemmaScUnifSmOp}
  We use the notation of Lemma~\usref{LemmaScUnifSmMetric} and~\eqref{EqScUnifSmMetricPsi}, and consider an operator $\wt L$ of the general form specified in~\S\usref{SEst}. Write $\wt L_b\in\hat\rho^{-2}\Diffse^2(\wt M'\setminus\wt K';\cE)$ for the operator family with $\wt L_b|_{M'_\eps}=\eps^{-2}\Psi_\eps^*L$ where $L$ is the $\hat M$-normal operator~\eqref{EqEstL}. Then
  \begin{equation}
  \label{EqScUnifSmOp}
    \wt L_{(\lambda)} - \wt L_b \in \lambda L^\infty\bigl((0,\lambda_0]; \hat\rho^{-2}\CI_\seop\Diffse^2(\wt M'\setminus\wt K';\cE)\bigr).
  \end{equation}
\end{lemma}

For the present case $\wt L=\Box_{\wt g}$, we have $\wt L_b:=\Box_{\wt g_b}$, i.e.\ the family of wave operators on Kerr spacetimes with parameters $\eps b$.

\begin{proof}[Proof of Lemma~\usref{LemmaScUnifSmOp}]
  We describe the coefficients of $\wt L$ in the coordinates $t$, $\rho_\circ=\frac{\eps}{|x|}$, $\hat\rho=|x|=r$, $\omega=\frac{x}{|x|}$ near $\pa\hat M_{t_0}$, and use $r\pa_z$, $z=(t,x)$, as a local frame of $\Vse(\wt M)$. Using the (fixed) identifications $\wt E|_{(t,\rho_\circ,\hat\rho,\omega)}\cong\hat\cE|_{(\rho_\circ,\omega)}$, and locally trivializing the latter bundle, it suffices to prove the Lemma for scalar operators. Since $\hat\rho^{-1}\Vse(\wt M)$ is then locally spanned by $\hat\rho^{-1} r\pa_{z^\mu}=\pa_{z^\mu}$, we then note that for $f=f(t,\rho_\circ,\hat\rho,\omega)\in\CI+\eps\CI_\seop=\CI+\rho_\circ\hat\rho\CI_\seop$, and using the same coordinates on $\wt M'$ as in~\eqref{EqScUnifSmMetCoord},
  \[
    \bigl(\lambda\wt S_\lambda^* ( f \pa_z )\bigr)(t',\rho'_\circ,\hat\rho',\omega) = f(t_0+\lambda t',\rho'_\circ,\lambda\hat\rho',\omega) \pa_{z'}\qquad\qquad (z'=(t',x'))
  \]
  is equal to $f(t_0,\rho'_\circ,0,\omega)\pa_{z'}=\eps^{-1} f(t_0,\rho'_\circ,0,\omega)\pa_{\hat z'}$ (where $\hat z'=\frac{z'}{\eps}$) plus an error term of class $\lambda L^\infty((0,\lambda_0];\hat\rho'{}^{-1}\CI_\seop\Vse(\wt M'))$. In a similar vein,
  \[
    \bigl(\lambda\wt S_\lambda^*(\hat\rho^{-1}f)\bigr)(t',\rho'_\circ,\hat\rho',\omega) = \hat\rho'{}^{-1}f(t_0+\lambda t',\rho'_\circ,\lambda\hat\rho',\omega)
  \]
  is equal to $\eps^{-1}\rho'_\circ f(t_0,\rho'_\circ,0,\omega)$ modulo $\lambda L^\infty([0,\lambda_0);\hat\rho'{}^{-1}\CI_\seop(\wt M'))$. Writing $\wt L$ as a sum of products of terms of type $f\pa_z^2$, $\hat\rho^{-1}f\pa_z$, $\hat\rho^{-2}f$, we thus obtain~\eqref{EqScUnifSmOp}.
\end{proof}

\bigskip

Returning to the setting of Theorem~\ref{ThmScUnif}, we apply the above for $\wt L=\Box_{\wt g}$ and the domain $\wt\Omega'$ from~\eqref{EqScUnifDomain1}--\eqref{EqScUnifDomain2}. Thus $\wt L_{(\lambda)}=\Box_{\wt g_{(\lambda)}}$ in the notation~\eqref{EqScUnifSmDefs}; we write $g_{(\lambda),\eps}=\wt g_{(\lambda)}|_{M'_\eps}$ and $g_{b,\eps}=\wt g_b|_{M'_\eps}$. Fix an order function $\sfs$ on $\wt\Omega'$ as in~\eqref{EqEstAdmIndEx} so that\footnote{This is slightly weaker than the requirement in Theorem~\ref{ThmScUnif}.} $\sfs,\alpha_\circ,\hat\alpha$ and $\sfs-2,\alpha_\circ,\hat\alpha$ are admissible for the metric $\wt g_b$ (given by~\eqref{EqScUnifSmMetricKerr}). We then have an analogue of~\eqref{EqScUnifPf},
\begin{equation}
\label{EqScUnifPfLambda}
  \|u\|_{H_{\seop,\eps}^{\sfs,\alpha_\circ,\hat\alpha}(\Omega'_\eps)^{\bullet,-}} \leq C\Bigl(\|\Box_{g_{(\lambda),\eps}}u\|_{H_{\seop,\eps}^{\sfs,\alpha_\circ,\hat\alpha-2}(\Omega'_\eps)^{\bullet,-}} + \|u\|_{H_{\seop,\eps}^{\sfs_0,\alpha_\circ,\hat\alpha}(\Omega'_\eps)^{\bullet,-}}\Bigr),
\end{equation}
for all $\lambda\leq\lambda_0$ and all $\eps\in(0,1)$ where $\lambda_0>0$ is sufficiently small, and $C$ is independent of $\lambda,\eps$. (Cf.\ Remark~\ref{RmkEstStdUnif}). We choose here $\sfs_0<\sfs-2-\eta$ for some $\eta>0$ so that $\sfs_0,\alpha_\circ,\hat\alpha$ are admissible. Theorem~\ref{ThmSc3b}, together with Proposition~\ref{PropEstFnSobRel} to pass between se- and 3b-norms, then gives the estimate
\[
  \|u\|_{H_{\seop,\eps}^{\sfs_0,\alpha_\circ,\hat\alpha}(\Omega'_\eps)^{\bullet,-}} \leq C\|\Box_{g_{b,\eps}}u\|_{H_{\seop,\eps}^{\sfs_0+\eta,\alpha_\circ,\hat\alpha-2}(\Omega'_\eps)^{\bullet,-}}
\]
for all sufficiently small $\eps$. In view of Lemma~\ref{LemmaScUnifSmOp}, this is further bounded by a constant (independent of $\eps,\lambda$) times
\begin{equation}
\label{EqScUnifPfOrig}
  \|\Box_{g_{(\lambda),\eps}}u\|_{H_{\seop,\eps}^{\sfs_0+\eta,\alpha_\circ,\hat\alpha-2}(\Omega'_\eps)^{\bullet,-}} + C\lambda\|u\|_{H_{\seop,\eps}^{\sfs_0+2+\eta,\alpha_\circ,\hat\alpha}(\Omega'_\eps)^{\bullet,-}}.
\end{equation}
For sufficiently small $\lambda$, the second term here can be absorbed into the left hand side of~\eqref{EqScUnifPfLambda}. When $u$ is related to a function on $\wt M$ via pullback along $\wt S_\lambda$, this proves the estimate~\eqref{EqScUnif} for the original operator $\Box_{\wt g}$ on the domain $\wt\Omega_{(\lambda)}=\wt S_\lambda(\wt\Omega')$ for $\eps<\lambda$. We thus obtain the following result.

\begin{thm}[Uniform se-estimates for linear waves on small domains]
\label{ThmScUnifSm}
  Let $\wt\Omega'$ be as in~\eqref{EqScUnifDomain1}--\eqref{EqScUnifDomain2}. Let $\alpha_\circ,\hat\alpha\in\R$ with $\alpha_\circ-\hat\alpha\in(-\frac32,-\frac12)$, and fix an order function $\sfs$ as in~\eqref{EqEstAdmIndEx} so that $\sfs,\alpha_\circ-\hat\alpha$ and $\sfs-2,\alpha_\circ-\hat\alpha$ are Kerr-admissible. Then there exist $\lambda_0,\eps_0>0$ so that for all $\lambda\in(0,\lambda_0]$ and with $\wt\Omega_{(\lambda)}=\wt S_\lambda(\wt\Omega')$, the unique forward solution of $\Box_{g_\eps}u=f$ on $\Omega_{(\lambda),\eps}=\wt\Omega_{(\lambda)}\cap M_\eps$ satisfies the estimate
  \begin{equation}
  \label{EqScUnifSm}
    \|u\|_{H_{\seop,\eps}^{\sfs,\alpha_\circ,\hat\alpha}(\Omega_{(\lambda),\eps})^{\bullet,-}} \leq C \|f\|_{H_{\seop,\eps}^{\sfs,\alpha_\circ,\hat\alpha-2}(\Omega_{(\lambda),\eps})^{\bullet,-}}
  \end{equation}
  for all $\eps<\lambda\eps_0$; here $C=C(\sfs,\alpha_\circ,\hat\alpha,\lambda)$.
\end{thm}

The restriction on $\eps$ arises from the scaling of $\eps$ in~\eqref{EqScUnifSmResc}.

\begin{rmk}[Other domains]
\label{RmkScUnifSmShift}
  The same arguments can be applied for scalings of the domain $\Omega'=\Omega_{t_0,t_1,r_0}$ (instead of~\eqref{EqScUnifDomain1}) for any $t_0<t_1$ and $r_0>0$.
\end{rmk}

\begin{rmk}[Perturbations]
\label{RmkScUnifSmPert}
  If we add to $\wt L$ a perturbation of class $\hat\rho^{-2}\cC_\seop^{\infty,1,1}\Diffse^2$, then the estimate~\eqref{EqScUnifSm} remains valid \emph{for the same value of $\lambda_0$} and sufficiently small $\eps_0>0$ (depending only on the perturbation). Indeed, the perturbation causes an additional term on the right hand side in~\eqref{EqScUnifPfOrig} which is \emph{small} (due to the decay at $\eps=0$ of the perturbation) for sufficiently small $\eps>0$ and can thus be absorbed.
\end{rmk}

\subsection{Higher s-regularity}
\label{SsScS}

We proceed to show that higher s-regularity of the source term $f$ implies matching higher s-regularity of the solution $u$ of $\wt L u=f$. We shall only present the details on small domains as considered in Theorem~\ref{ThmScUnifSm}.

\begin{thm}[Uniform estimates for linear waves on small domains: s-regularity]
\label{ThmScS}
  We use the notation and assumptions of Theorem~\usref{ThmScUnifSm}, except we assume that $\sfs,\alpha_\circ-\hat\alpha$ and $\sfs-4,\alpha_\circ-\hat\alpha$ are Kerr-admissible, and $\sfs$ is as in~\eqref{EqEstAdmIndEx}. Let $k\in\N_0$, and suppose that $\wt L\in\hat\rho^{-2}(\CI+\cC_{\seop;\sop}^{(\infty;k),1,1})\Diffse^2$, i.e.\ its coefficients have s-regularity $k$. (Finite but sufficiently large se-regularity $d_0$, as in~\eqref{EqEstRadsOp}, with $d_0=d_0(\sfs)$, is sufficient.) Then the unique forward solution of $\Box_{g_\eps}u=f$ on $\Omega_{(\lambda),\eps}$ satisfies
  \begin{equation}
  \label{EqScS}
    \|u\|_{H_{(\seop;\sop),\eps}^{(\sfs;k),\alpha_\circ,\hat\alpha}(\Omega_{(\lambda),\eps})^{\bullet,-}} \leq C\|f\|_{H_{(\seop;\sop),\eps}^{(\sfs;k),\alpha_\circ,\hat\alpha-2}(\Omega_{(\lambda),\eps})^{\bullet,-}}
  \end{equation}
  for all $\lambda\in(0,\lambda_0]$ and $\eps<\lambda\eps_0$, with $C=C(\sfs,k,\alpha_\circ,\hat\alpha,\lambda)$. (Here $\lambda_0$ does not depend on $k$, and can moreover be taken to be uniform upon adding to $\wt L$ a bounded perturbation of class $\hat\rho^{-2}\cC_{\seop,\sop}^{(\infty;k),1,1}\Diffse^2$.)
\end{thm}

In~\S\ref{SsNTame}, we will prove a strengthening of Theorem~\ref{ThmScS} which gives \emph{tame estimates} for $u$.

\begin{proof}[Proof of Theorem~\usref{ThmScS}]
  Let $\Omega^\sharp:=\Omega_{0,2,2}$ (the scaling by $\frac12$ of which contains $\Omega'=\Omega_{0,1,1}$ from~\eqref{EqScUnifDomain2}), write $\wt\Omega^\sharp\subset\wt M'\setminus(\wt K')^\circ$ for the associated standard domain in $\wt M'$, and let $\wt\Omega^\sharp_{(\lambda)}=\wt S_\lambda(\wt\Omega^\sharp)$. Suppose~\eqref{EqScS} has been established for $k-1$, the case $k-1=0$ being Theorem~\ref{ThmScUnifSm} (which only requires $\sfs,\alpha_\circ,\hat\alpha$ and $\sfs-2,\alpha_\circ,\hat\alpha$ to be admissible). Taking $\lambda_0$ to be $\frac12$ times the value from Theorem~\ref{ThmScUnifSm}, extend $f\in H_{(\seop;\sop),\eps}^{(\sfs;k),\alpha_\circ,\hat\alpha-2}(\Omega_{(\lambda),\eps})^{\bullet,-}$ to $f^\sharp\in H_{(\seop;\sop),\eps}^{(\sfs;k),\alpha_\circ,\hat\alpha-2}(\Omega^\sharp_{(\lambda),\eps})^{\bullet,-}$ and solve $\Box_{g_\eps}u^\sharp=f^\sharp$ using the inductive hypothesis, which gives
  \begin{equation}
  \label{EqScSusharp}
    u^\sharp\in H_{(\seop;\sop),\eps}^{(\sfs;k-1),\alpha_\circ,\hat\alpha}(\Omega^\sharp_{(\lambda),\eps})^{\bullet,-}.
  \end{equation}

  Differentiating the equation $\Box_{g_\eps}u^\sharp=f^\sharp$ along $\pa_t^k$ gives
  \[
    \Box_{g_\eps}(\pa_t^k u^\sharp) = \pa_t^k f^\sharp + [\Box_{g_\eps},\pa_t^k]u^\sharp \in H_{\seop,\eps}^{\sfs-2,\alpha_\circ,\hat\alpha-2}(\Omega^\sharp_{(\lambda),\eps})^{\bullet,-};
  \]
  here we use\footnote{We do not need to exploit the gain of decay at $\hat M$ of the commutator here.}
  \[
    [\Box_{g_\eps},\pa_t^k]\in\sum_{j=0}^{k-1}\hat\rho^{-2}(\hat\rho\CI+\cC_{\seop;\sop}^{(\infty;j),1,1})\Diffse^2\circ\pa_t^j\subset\hat\rho^{-2}\CI_\seop\Diffse^2\circ\pa_t^{k-1}.
  \]
  Applying Theorem~\ref{ThmScUnifSm} to this equation gives $\pa_t^k u^\sharp\in H_{\seop,\eps}^{\sfs-2,\alpha_\circ,\hat\alpha-2}$. Together with~\eqref{EqScSusharp} (which gives $W u^\sharp\in H_{(\seop;\sop),\eps}^{(\sfs-1;k-1),\alpha_\circ,\hat\alpha}$ for all $W\in\Vse(\wt M)\subset\cV_{[\sop]}(\wt M)$), we have thus shown
  \begin{equation}
  \label{EqScSusharp2}
    u^\sharp \in H_{(\seop;\sop),\eps}^{(\sfs-2;k),\alpha_\circ,\hat\alpha}(\Omega^\sharp_{(\lambda),\eps})^{\bullet,-}.
  \end{equation}
  Using that $\Box_{g_\eps}u^\sharp\in H_{(\seop;\sop),\eps}^{(\sfs;k),\alpha_\circ,\hat\alpha-2}(\Omega^\sharp_{(\lambda),\eps})^{\bullet,-}$, we can now use the microlocal elliptic and propagation estimates on (se;s)-Sobolev spaces (on open subsets of $\Omega^\sharp_{(\lambda),\eps}$) as in the proof of Theorem~\ref{ThmEstStd}---but now using Lemmas~\ref{LemmaFVarsseEll} and \ref{LemmaFVarsseProp} and Propositions~\ref{PropEstRadsIn}, \ref{PropEstRadsOut}, \ref{PropEstHors}, and \ref{PropEstTraps}---to re-gain the $2$ degrees of se-regularity lost in~\eqref{EqScSusharp2}. Upon restriction to $\Omega_{(\lambda),\eps}$, this gives $u=u^\sharp|_{\Omega_{(\lambda),\eps}}\in H_{(\seop;\sop),\eps}^{(\sfs;k),\alpha_\circ,\hat\alpha}(\Omega_{(\lambda),\eps})^{\bullet,-}$.
\end{proof}

\begin{rmk}[s-regularity on standard domains]
\label{RmkScSStd}
  The same method of proof can be used to prove a version of Theorem~\ref{ThmScUnif} with higher s-regularity; the resulting conditions on $\sfs$ are that $\sfs,\alpha_\circ,\hat\alpha$ and $\sfs-5,\alpha_\circ,\hat\alpha$ be admissible. In the context of Theorem~\ref{ThmI} (where $\alpha_\circ=0$ and $\hat\alpha\in(\frac12,\frac32)$), note that the threshold condition $\sfs<-\frac12+\hat\alpha$ is satisfied for $\sfs=0$, whereas $\sfs-5>-\frac12+\hat\alpha$ is satisfied for $\sfs=6$. This is the reason for the shift by $6$ orders in~\eqref{EqIEst}.
\end{rmk}

\subsection{From local to semiglobal solvability}
\label{SsScSG}

We now concatenate the uniform bounds on small domains given by Theorem~\ref{ThmScS} to obtain a quantitative semiglobal solvability result for the linear scalar wave equation. In the following result, we write $\cA=\cA(\wt M)$ for the space of bounded conormal functions on $\wt M$; thus $u\in\cA$ if and only if $u\in\CI_\sop(\wt M)$ and $(\eps\pa_\eps)^j u\in\CI_\sop(\wt M)$ for all $j\in\N_0$.

\begin{thm}[Semiglobal solvability]
\label{ThmScSG}
  Let $K\subset M$ be compact. Let $X\subset M$ be a Cauchy surface with $K\subset J^+(X)$. Denote by $\wt X\subset\wt M$ the lift of $[0,1)\times X$. Suppose that $\wt L\in\hat\rho^{-2}(\CI+\cC_\sop^{\infty,1,1})\Diffse^2(\wt M\setminus\wt K^\circ)$. Then there exists $\eps_0>0$ so that for all $\wt f\in\hat\rho^{-2}H_\sop^\infty(\wt M\setminus\wt K^\circ)$ (by which we mean that $\wt f|_{M_\eps}$ is uniformly bounded in $\hat\rho^{-2}H_{\sop,\eps}^k(M_\eps\setminus\wt K^\circ)$ for all $k$) which vanish to infinite order at $\wt X$, the equation
  \[
    \wt L\wt u=\wt f
  \]
  has a unique solution $\wt u\in H_\sop^\infty(\wt\upbeta^*([0,\eps_0]\times K)\setminus\wt K^\circ)$ vanishing to infinite order at $\wt X$, i.e.\ $\Box_{g_\eps}u_\eps=f_\eps$ on $K\setminus\{|\hat x|\leq\bhm\}$ for all $\eps>0$ where $g_\eps,u_\eps,f_\eps$ are the restrictions of $\wt g,\wt u,\wt f$ to $M_\eps$. If, moreover, we have $\wt L\in\hat\rho^{-2}(\CI+\eps\cA)\Diffse^2(\wt M\setminus\wt K^\circ)$ and $\wt f\in\CIdot(\wt M\setminus\wt K^\circ)$, then also $\wt u\in\CIdot(\wt\upbeta^*([0,\eps_0]\times K)\setminus\wt K^\circ)$.
\end{thm}

On a technical level, Theorem~\ref{ThmScSG} is of no use for solving nonlinear equations on glued spacetimes. However, the method of proof is flexible and used in a nonlinear context in~\S\ref{SssNToyT} below.

\begin{rmk}[Other types of problems]
\label{RmkScSGOther}
  Given $\wt f\in H^\infty_\sop(\wt M\setminus\wt K^\circ)$, we can solve $\Box_{\wt g}\wt u_+=\chi_+\wt f$ where $\chi_+$ is supported in the future of $X$ and equal to $1$ near a later Cauchy surface $X'$, and then $\Box_{\wt g}\wt u_-=(1-\chi_+)\wt f$ (with the lift of $X'$ taken as the Cauchy surface); the sum $\wt u=\wt u_++\wt u_-$ then solves $\Box_{\wt g}\wt u=\wt f$. One can also consider problems with nontrivial Cauchy data at $\wt X$; we leave such modifications to the interested reader.
\end{rmk}

\begin{proof}[Proof of Theorem~\usref{ThmScSG}]
  It suffices to consider the case that $f_\eps$ vanishes in the past of $X$. (Indeed, the extension of $f_\eps|_{J^+(X)}$ by $0$ to the causal past of $X$ satisfies the same regularity assumptions as $\wt f$.) Fix a metric splitting $\R\times X_0$ of $M$ and set $X_\cT:=\{\cT\}\times X_0$; by shifting the time coordinate, we arrange that $J^-(K)\cap X\subset J^+(X_1)$. Modify $X_0$ near $\{\fp\}=X_0\cap\cC$ to a Cauchy hypersurface $X'_0\subset(-1,1)\times X_0$ so that $X'_0$ coincides with a $t$-level set near $\fp$.

  We now apply Proposition~\ref{PropGlDynCover}, with $\eta>0$ there chosen so small that Theorem~\ref{ThmScS} applies on each standard domain produced by the Proposition. We then solve $\Box_{\wt g}\wt u=\wt f$ iteratively in $\Omega_0$, $\Omega_2$, \dots, $\Omega_J$. To wit, suppose we have already solved the equation in $\bigcup_{j=0}^{J'-1}\Omega_j$, $J'\geq 0$. If the initial boundary hypersurface $X_{J'}$ of $\Omega_{J'}$ is contained in $X'_0$, there are two possibilities: either $X_{J'}$ is disjoint from $\cC$, in which case we can solve using standard hyperbolic theory; or $X_{J'}\cap\cC\neq\emptyset$, in which case we can solve $\Box_{\wt g}\wt u=\wt f$ on $\Omega_{J'}$ using Theorem~\ref{ThmScS} for all $k\in\N_0$, obtaining $\wt u\in H^\infty_\sop(\wt\Omega_{J'})^{\bullet,-}$; here $\wt\Omega_{J'}$ denotes the lift to $\wt M\setminus\wt K^\circ$ of $\Omega_{J'}$. 

  If, on the other hand, the initial boundary hypersurface of $\Omega_{J'}$ is contained in $\bigcup_{j=0}^{J'-1}\Omega_j^\circ$, let $\chi\in\CI(\Omega_{J'})$ be equal to $0$ near the initial boundary hypersurface of $\Omega_{J'}$ and equal to $1$ outside of $\bigcup_{j=0}^{J'-1}\Omega_j$. The source term in the equation
  \begin{equation}
  \label{EqScSGwtv}
    \Box_{\wt g}\wt v = \chi\wt f + [\Box_{\wt g},\chi]\wt u
  \end{equation}
  is well-defined and lies in $\hat\rho^{-2} H^\infty_\sop(\wt\Omega_{J'})^{\bullet,-}$. Therefore, we can again appeal to Theorem~\ref{ThmScS} to obtain $\wt v\in H^\infty_\sop(\wt\Omega_{J'})^{\bullet,-}$. Since $\wt v=\chi\wt u$ on $\Omega_{J'}\cap(\bigcup_{j=0}^{J'-1}\Omega_j)$, the solution $\wt v$ of~\eqref{EqScSGwtv} furnishes the extension of $\wt u$ to $\bigcup_{j=0}^{J'}\Omega_j$. Arguing similarly for solving the equation towards the past of $X_0$ finishes the proof of the first part of the Theorem.

  To prove the second part, note first that since $\Box_{\wt g}$ commutes with multiplication by $\eps^N$, we have $\wt u\in\eps^\infty\CI_\sop$ by Sobolev embedding (Proposition~\ref{PropFsSobEmb}). Furthermore, $u_\eps$ is differentiable in $\eps$ for every $\eps>0$, as follows by considering the equation $L_\eps\frac{u_{\eps+h}-u_\eps}{h}=\frac{f_{\eps+h}-f_\eps}{h}-\frac{L_{\eps+h}-L_\eps}{h}u_{\eps+h}$, and taking the limit as $h\searrow 0$ in the estimates for its solution on Sobolev spaces. Differentiating $\wt L\wt u=\wt f$ (now regarded as an equation on $\wt M\setminus(\hat M\cup M_\circ)$, rather than a family of equations parameterized by $\eps$) along $\eps\pa_\eps$ gives
  \[
    \wt L(\eps\pa_\eps\wt u) = \eps\pa_\eps\wt f + [\wt L,\eps\pa_\eps]\wt u.
  \]
  By our assumption on $\wt L$, we have $[\eps\pa_\eps,\wt L]\in\hat\rho^{-2}\eps\cA\Diffse^2$, and thus the right hand side lies in $\eps^\infty\CI_\sop$. Therefore, also $\eps\pa_\eps\wt u\in\eps^\infty\CI_\sop$. Iterating this argument gives $\wt u\in\eps^\infty\cA=\CIdot$, finishing the proof.
\end{proof}

\begin{rmk}[Application in settings with degenerating potentials]
\label{RmkScSGV}
  In the context of Remark~\ref{RmkGlOther}, we recall that in \cite[\S{1.2.1}]{HintzGlueLocI}, a formal solution of $\wt L\wt u=\wt f\in\CIdot(\wt M)$ was constructed where $\wt L=\Box_g+\eps^{-2}V(\frac{x}{\eps})$ with $0\leq V\in\CIc(\R^3)$, and $u=\wt u|_{M_\circ}$ is a given solution of $\Box_g u=0$. Since the $\hat M_t$-model operator $\Box_{\hat{\ubar g}}+V$ (with $\hat{\ubar g}=-\dd\hat t^2+\dd\hat x^2$ the Minkowski metric) satisfies mode stability in $\Im\sigma\geq 0$, the proof of Theorem~\ref{ThmScSG} carries over to produce $\wt v\in\CIdot(\wt M)$ with $\wt L\wt v=-\wt f$; thus $\wt u+\wt v$ is a true solution (i.e.\ $\wt L(\wt u+\wt v)=0$) of the singular perturbation problem.
\end{rmk}

\section{Tame estimates, Nash--Moser iteration, and a nonlinear toy model}
\label{SN}

Consider a semi- or quasilinear wave equation which involves derivatives in the nonlinearity. Due to loss of two derivatives (relative to elliptic estimates) in the trapping estimate~\eqref{EqScUnif}, the existence of solutions, over compact subsets of $M$ lifted to a glued spacetime as in Theorem~\ref{ThmScSG}, cannot be proved using a Picard type iteration, regardless also of the amount of s-regularity of the function spaces one works with. To get around this issue, we shall demonstrate in~\S\ref{SsNTame} that the quantitative estimates on (se;s)-Sobolev spaces are \emph{tame} in the s-regularity order: in~\S\ref{SssNTameMl}, we discuss this on the level of se-microlocal elliptic and propagation estimates, and in~\S\ref{SssNTameSg} we prove tame estimates for forward solutions of $\wt L u=f$. This material is used in~\S\ref{SsNToy} to correct a formal solution of a toy nonlinear equation to a true solution.

We fix a glued spacetime $(\wt M,\wt g)$ corresponding to $(M,g)$, $\cC\subset M$, $b=(\bhm,\bha)$, and Fermi normal coordinates $t,x$ (with $t\in I_\cC\subseteq\R$) as in Definition~\ref{DefGl}.

\subsection{Tame estimates}
\label{SsNTame}

We assume that $M$ is compact; recall from~\S\ref{SsEstFn} that this suffices for local analysis near points in $\cC$.

\subsubsection{se-microlocal s-tame estimates}
\label{SssNTameMl}

We shall use the schematic notation $D_\sop^j$, resp.\ $D_{[\sop]}^j$ for the vector of all $j$-fold compositions of the elements of a fixed finite subset of $\cV_\sop(\wt M)$, resp.\ $\cV_{[\sop]}(\wt M)$ which spans $\cV_\sop(\wt M)$ over $\CI(\wt M)$.

\begin{lemma}[Multiplication on s-Sobolev spaces]
\label{LemmaNTameDaDb}
  Let $a,b\in\N_0$, and fix $s_0>\frac{\dim M}{2}=2$. Then there exists a constant $C_{a,b}$ so that for all $\eps\in(0,1)$, we have
  \[
    \|(D_\sop^a\ell)(D_\sop^b u)\|_{L^2(M)}\leq C_{a,b}\Bigl(\|\ell\|_{L^\infty(M)}\|u\|_{H_{\sop,\eps}^{a+b}(M)} + \|\ell\|_{\cC_{\sop,\eps}^{a+b}(M)}\|u\|_{H_{\sop,\eps}^{s_0}(M)}\Bigr).
  \]
\end{lemma}
\begin{proof}
  Instead of working with an $\eps$-independent volume density on $M$, we shall work with a smooth positive s-density. (The two densities differ by a factor of $\hat\rho^3$; cf.\ the proof of Proposition~\ref{PropFsSobEmb} for $n=3$.)

  Denote by $U_{\eps,\alpha}\subset\{\eps\}\times M$ the unit cells and by $\phi_{\eps,\alpha}\colon U_{\eps,\alpha}\xra{\cong}(-2,2)^n$ the corresponding charts for a parameterized b.g.\ structure (in the sense of \cite[Definition~1.9]{HintzScaledBddGeo}) for which the space of uniformly bounded vector fields is $\CI_\sop\cV_\sop(\wt M)$ (see e.g.\ \eqref{EqFVarsseBG1}--\eqref{EqFVarsseBG2}). For all $k\in\N_0$, we then recall from (the discussion following) \cite[Definition~3.14]{HintzScaledBddGeo} that for all $\chi\in\CIc((-2,2)^n)$ which equal $1$ on $[-1,1]^n$, we have a uniform equivalence of norms
  \[
    \|v\|_{H_{\sop,\eps}^k(M)}^2 \sim \sum_\alpha\|\chi(\phi_\alpha)_*v\|_{H^k(\R^4)}^2.
  \]
  The Lemma then follows from the standard $\R^n$ result
  \[
    \|(D^a\ell)(D^b u)\|_{L^2(\R^n)} \leq C_{a,b}\bigl( \|\ell\|_{L^\infty}\|u\|_{H^{a+b}} + \|\ell\|_{H^{a+b}}\|u\|_{L^\infty}\bigr),
  \]
  see \cite[Chapter~13, Proposition~3.6]{TaylorPDE3}, together with the bound $\|\ell\|_{H^{a+b}}\leq C\|\ell\|_{\cC^{a+b}}$ (valid for all $\ell$ with support in $(-2,2)^n$) and Sobolev embedding $\|u\|_{L^\infty}\leq C\|u\|_{H^{s_0}}$.
\end{proof}

\begin{lemma}[Tame multiplication, se-microlocalized]
\label{LemmaNTameMult}
  Let $\sfs,\sfs_0\in\CI(\Sse^*\wt M)$. Then there exist $d_0,d\in\N$ with the following property. For all $B,B^\sharp\in\Psi_\seop^0(\wt M)$ and $\chi\in\CIc(\wt M)$ with $B=\chi B\chi$, $B^\sharp=\chi B^\sharp\chi$, and $\WF_\seop'(B)\subset\Ell_\seop(B^\sharp)$, and for all $k\in\N_0$, there exists a constant $C_k$ so that for all $\eps\in(0,1)$ and $0\leq j\leq k$,
  \begin{equation}
  \label{EqNTameMult}
  \begin{split}
    \|B(D_\sop^j\ell)(D_\sop^{k-j}u)\|_{H_{\seop,\eps}^\sfs(M)} &\leq C_k\biggl[ \|\ell\|_{\cC_{(\seop;\sop),\eps}^{(d_0;d)}(M)}\Bigl(\|B^\sharp u\|_{H_{(\seop;\sop),\eps}^{(\sfs;k)}(M)} + \|\chi u\|_{H_{(\seop;\sop),\eps}^{(\sfs_0;k)}(M)}\Bigr) \\
      &\quad\hspace{0.8em} + \|\ell\|_{\cC_{(\seop;\sop),\eps}^{(d_0;k)}(M)}\Bigl(\|B^\sharp u\|_{H_{(\seop;\sop),\eps}^{(\sfs;d)}(M)} + \|\chi u\|_{H_{(\seop;\sop),\eps}^{(\sfs_0;d)}(M)}\Bigr) \biggr].
  \end{split}
  \end{equation}
  An analogous estimate holds on weighted spaces.
\end{lemma}

Below, we shall only consider \emph{fixed se-regularity orders}, whereas the s-regularity order $k$ will be arbitrarily large. Thus, one should consider the quantities $d_0,d,\sfs$ in~\eqref{EqNTameMult} as being low regularity orders. (Of course, one could replace $(\sfs;d)$ by $(d_0;d)$ by picking $d_0>\sup\sfs$, but this would make the estimate look more opaque.)

\begin{rmk}[Special cases]
\label{RmkNTameMultSpecial}
  In the special case $B=B^\sharp=\chi=1$ and $\sfs_0=\sfs$, this implies
  \[
    \|\ell u\|_{H_{(\seop;\sop),\eps}^{(\sfs;k)}(M)} \leq C_k\Bigl(\|\ell\|_{\cC_{(\seop;\sop),\eps}^{(d_0;d)}(M)} \|u\|_{H_{(\seop;\sop),\eps}^{(\sfs;k)}(M)} + \|\ell\|_{\cC_{\seop;\sop}^{(d_0;k)}(M)} \|u\|_{H_{(\seop;\sop),\eps}^{(\sfs;d)}(M)} \Bigr)
  \]
  upon using the Leibniz rule. Further restricting to $\sfs=0$, this is a (imprecise) version of the usual tame multiplication estimate stated in \cite[Chapter~13, Proposition~3.7]{TaylorPDE3}, which here gives
  \begin{equation}
  \label{EqNTameMultSpecial}
    \|\ell u\|_{H_{\sop,\eps}^k(M)} \leq C_k\Bigl(\|\ell\|_{\cC_{\sop,\eps}^0(M)}\|u\|_{H_{\sop,\eps}^k(M)} + \|\ell\|_{H_{\sop,\eps}^k(M)}\|u\|_{\cC_{\sop,\eps}^0(M)}\Bigr).
  \end{equation}
\end{rmk}

\begin{proof}[Proof of Lemma~\usref{LemmaNTameMult}]
  By se-microlocal elliptic regularity, it suffices to consider the case that $B^\sharp=I$ microlocally on $\WF_\seop'(B)$, i.e.\ $B(I-B^\sharp)\in\Psi_\seop^{-\infty}$. Furthermore, we may assume without loss of generality that $\sfs_0<\sfs$.

  Write
  \begin{equation}
  \label{EqNTameMultInsert}
    B\ell u = B\ell B^\sharp u + B\ell(I-B^\sharp)u = B\ell B^\sharp u + B\ell(I-B^\sharp)\chi u.
  \end{equation}
  If $\ell\in\cC_{\seop,\eps}^\infty(\wt M)$, then $B\ell B^\sharp\in\tilde\Psi_\seop^0(\wt M)$ is uniformly bounded on $H_{\seop,\eps}^\sfs$, with operator norm depending only on a finite seminorm of $B\ell B^\sharp$ and thus of $\ell$. Similarly, $B\ell(I-B^\sharp)\in\tilde\Psi_\seop^{-\infty}(\wt M)$ is uniformly bounded as a map $H_{\seop,\eps}^{\sfs_0}\to H_{\seop,\eps}^\sfs$, and this persists for $\ell\in\cC_\seop^{d_0}(\wt M)$ when the se-regularity order $d_0$ is sufficiently large (depending on $\sfs,\sfs_0$). This implies~\eqref{EqNTameMult} for $k=j=0$ with $d=0$.

  For $k\leq d$, where $d>2\bar s+2$ with $\bar s:=\max(0,\lceil\sup\sfs\rceil)$ will be fixed below, the estimate~\eqref{EqNTameMult} follows from the case $k=0$ upon plugging in $D_\sop^j\ell$ and $D_\sop^{k-j}u$ in place of $\ell$ and $u$, respectively. Consider now $k>d$. It suffices to estimate $B(D_{[\sop]}^j\ell)(D_{[\sop]}^{k-j} u)$ in $H_{\seop,\eps}^\sfs$ by the right hand side of~\eqref{EqNTameMult}. For $j\leq d$ or $j\geq k-d$, we again apply the case $k=0$ to $D_{[\sop]}^j\ell$ and $D_{[\sop]}^{k-j}u$ in place of $\ell$ and $u$, respectively, to conclude. It remains to consider the terms with $d+1\leq j\leq k-d-1$. We first crudely estimate
  \[
    \|B(D_{[\sop]}^j\ell)(D_{[\sop]}^{k-j}u)\|_{H_{\seop,\eps}^\sfs} \leq C\|B(D_{[\sop]}^j\ell)(D_{[\sop]}^{k-j}u)\|_{H_{\sop,\eps}^{\bar s}} \leq C'\sum_{q=0}^{\bar s} \| D_{[\sop]}^q ( B(D_{[\sop]}^j\ell)(D_{[\sop]}^{k-j}u) ) \|_{L^2}
  \]
  and expand $D_{[\sop]}^q$ using the Leibniz rule, i.e.\ schematically
  \[
    D_{[\sop]}^q (B v) = \sum_{i=0}^q ({\rm ad}_{D_{[\sop]}}^i B)(D_{[\sop]}^{q-i}v),\qquad v:=(D_{[\sop]}^j\ell)(D_{[\sop]}^{k-j}u),
  \]
  where ${\rm ad}_{D_{[\sop]}}^i B$ is an $i$-fold commutator of $B$ with $i$ (possibly different) commutator s-vector fields, and thus satisfies the same assumptions as $B$ relative to the operator $B^\sharp$. We only treat the terms with $q=\bar s$ and $i=0$ (which thus involve the largest number of overall s-derivatives); expanding $D_{[\sop]}^{\bar s}v$, we thus need to estimate
  \[
    \sum_{l=0}^{\bar s} \| B (D_{[\sop]}^{j+l}\ell) (D_{[\sop]}^{k+\bar s-j-l}u) \|_{L^2}.
  \]
  Writing $u=B^\sharp u+(I-B^\sharp)\chi u$ as in~\eqref{EqNTameMultInsert}, we bound the $l$-th summand by
  \begin{equation}
  \label{EqNTameMult2}
    \|B\|_{L^2\to L^2} \Bigl( \| (D_{[\sop]}^{j+l}\ell)(D_{[\sop]}^{k+\bar s-j-l}(B^\sharp u)) \|_{L^2} + \|(D_{[\sop]}^{j+l}\ell) (D_{[\sop]}^{k+\bar s-j-l}((I-B^\sharp)\chi u)) \|_{L^2} \Bigr).
  \end{equation}
  (Here $d+1\leq j\leq k-d-1$ and $0\leq l\leq\bar s$.) For a value $d_1$ with $\bar s<d_1$ and $2 d_1\leq d$ (to be determined below), we now apply Lemma~\ref{LemmaNTameDaDb} to the first term with $D_\sop^{d_1}\ell$, $D_\sop^{d_1}(B^\sharp u)$ and $a=j+l-d_1$, $b=k+\bar s-j-l-d_1$, $s_0=3$; this gives the bound
  \[
    \|(D_{[\sop]}^{j+l}\ell)(D_{[\sop]}^{k+\bar s-j-l}(B^\sharp u))\|_{L^2} \leq C'\Bigl(\|\ell\|_{\cC_{\sop,\eps}^{d_1}}\|B^\sharp u\|_{H_{\sop,\eps}^{k+\bar s-d_1}} + \|\ell\|_{\cC_{\sop,\eps}^{k+\bar s-2 d_1}}\|B^\sharp u\|_{H_{\sop,\eps}^{d_1+3}}\Bigr).
  \]
  But $\|u\|_{H_{\sop,\eps}^{k+\bar s-d_1}}\leq C\|u\|_{H_{(\seop;\sop),\eps}^{(\sfs;k+\bar s+\ubar s-d_1)}}$ where $\ubar s:=\max(0,\lceil\sup(-\sfs)\rceil)$, and likewise $\|u\|_{H_{\sop,\eps}^{d_1+3}}\leq C\|u\|_{H_{(\seop;\sop),\eps}^{(\sfs;d_1+3+\ubar s)}}$. Thus, if we require $d_1\geq\ubar s+\bar s$ and $d\geq d_1+3+\ubar s$, we obtain the bound
  \[
    \|(D_{[\sop]}^{j+l}\ell)(D_{[\sop]}^{k+\bar s-j-l}(B^\sharp u))\|_{L^2} \leq C\Bigl(\|\ell\|_{\cC_{\sop,\eps}^{d_1}}\|B^\sharp u\|_{H_{(\seop;\sop),\eps}^{(\sfs;k)}}+\|\ell\|_{\cC_{\sop,\eps}^k}\|B^\sharp u\|_{H_{(\seop;\sop),\eps}^{(\sfs;d)}}\Bigr).
  \]
  With the stronger requirements $d_1\geq\ubar s_0+\bar s$, $d\geq d_1+3+\ubar s_0$, where $\ubar s_0:=\max(0,\lceil\sup(-\sfs_0)\rceil)$, we can similarly bound the second term in~\eqref{EqNTameMult2} by $\|\ell\|_{\cC_{\sop,\eps}^{d_1}}\|\chi u\|_{H_{(\seop;\sop),\eps}^{(\sfs_0;k)}}+\|\ell\|_{\cC_{\sop,\eps}^k}\|\chi u\|_{H_{(\seop;\sop),\eps}^{(\sfs_0;d)}}$; here we use the uniform boundedness of $I-B^\sharp$ on every mixed (se;s)-Sobolev space.
\end{proof}

We proceed to discuss tame analogues of microlocal regularity results. In the elliptic regularity estimate~\eqref{EqFVarsseEll}, the constant $C$ can be taken to be a function of the $\cC_{(\seop;\sop),\eps}^{(d_0;k)}$-norm of the coefficients of $\wt L$ when expressing $\wt L$ as a finite sum of operators $\ell_j A_j$ where $\ell_j\in\cC_{\seop;\sop}^{(d_0;k)}$ and $A_j\in\Psi_\seop^m(\wt M)$, with the $A_j$ fixed. This estimate can thus be used for s-regularity orders $k$ below some (large but fixed) value $d\in\N$, much as in the proof of the previous result. We shall write $\|\wt L\|_{\cC_{(\seop;\sop),\eps}^{(d_0;k)}}$ for the sum of the $\cC_{(\seop;\sop),\eps}^{(d_0;k)}$-norms of the $\ell_j$.

\begin{lemma}[se-microlocal elliptic regularity: tame estimate]
\label{LemmaNTameEll}
  Given $m,\sfs,N$, there exist $d_0,d\in\N$ so that for all $k\in\N_0$ there exists a constant $C_k$ so that, under the assumptions of Lemma~\usref{LemmaFVarsseEll}, we have a tame estimate
  \begin{equation}
  \label{EqNTameEll}
  \begin{split}
    &\|B u\|_{H_{(\seop;\sop),\eps}^{(\sfs;k)}} \\
    &\quad \leq C_k \biggl[ \|G \wt L u\|_{H_{(\seop;\sop),\eps}^{(\sfs-m;k)}} + \|\chi u\|_{H_{(\seop;\sop),\eps}^{(-N;k)}} + \|\wt L\|_{\cC_{(\seop;\sop),\eps}^{(d_0;k)}} \Bigl( \|G \wt L u\|_{H_{(\seop;\sop),\eps}^{(\sfs-m;d)}} + \|\chi u\|_{H_{(\seop;\sop),\eps}^{(-N;d)}} \Bigr) \biggr].
  \end{split}
  \end{equation}
  The constant $C_k$ can be taken to be uniform for perturbations of $\wt L$ which are sufficiently small in $\cC_{(\seop;\sop),\eps}^{(d_0;d)}\Psi_\seop^m(\wt M)$.
\end{lemma}
\begin{proof}
  We will fix $d=d(m,\sfs,N)$ to be large. Write $f:=\wt L u$. From the above discussion, we have, for $k\leq d$, an estimate
  \[
    \|B u\|_{H_{(\seop;\sop),\eps}^{(\sfs;k)}} \leq C_k\Bigl( \|G f\|_{H_{(\seop;\sop),\eps}^{(\sfs-m;k)}} + \|\chi u\|_{H_{(\seop;\sop),\eps}^{(-N;k)}} \Bigr).
  \]
  For $k>d$, we apply the case $k=0$ to $u':=D_{[\sop]}^k u$; thus, schematically,
  \[
    f':=\wt L u'=D_{[\sop]}^k f+\sum_{j=1}^k ({\rm ad}_{D_{[\sop]}}^j\wt L)D_{[\sop]}^{k-j}u.
  \]
  We shall bound the two terms in
  \begin{equation}
  \label{EqNTameEllGfp}
    \|G f'\|_{H_{\seop,\eps}^{\sfs-m}}\leq\|G D_{[\sop]}^k f\|_{H_{\seop,\eps}^{\sfs-m}}+\sum_{j=1}^k\|G ({\rm ad}_{D_{[\sop]}}^j\wt L)D_{[\sop]}^{k-j}u\|_{H_{\seop,\eps}^{\sfs-m}}
  \end{equation}
  separately. We can estimate the $j$-th summand using Lemma~\ref{LemmaNTameMult} with $k-1$, $j-1$ in place of $k$, $j$. The first term in~\eqref{EqNTameEllGfp} is bounded by $\sum_{l=0}^k\|D_{[\sop]}^l({\rm ad}_{D_{[\sop]}}^{k-l}G)f\|_{H_{\seop,\eps}^{\sfs-m}}\leq C(\|\tilde G f\|_{H_{(\seop;\sop),\eps}^{(\sfs-m;k)}}+\|\chi f\|_{H_{(\seop;\sop),\eps}^{(-N-m;k)}})$ where $\tilde G\in\Psi_\seop^0$ is elliptic on $\WF_\seop'(G)$; but again using the Leibniz rule we can estimate
  \[
    \|\chi f\|_{H_{(\seop;\sop),\eps}^{(-N-m;k)}}=\sum_{j=0}^k\|D_{[\sop]}^j\chi\wt L u\|_{H_{\seop,\eps}^{-N-m}}
  \]
  by the right hand side of~\eqref{EqNTameEll} without the $G\wt L u$ terms. Altogether, we arrive at
  \begin{align*}
    \|G f'\|_{H_{\seop,\eps}^{\sfs-m}} &\leq C_k\biggl[\|\tilde G f\|_{H_{(\seop;\sop),\eps}^{(\sfs-m;k)}} + \|\tilde G u\|_{H_{(\seop;\sop),\eps}^{(\sfs;k-1)}} + \|\chi u\|_{H_{(\seop;\sop),\eps}^{(-N;k-1)}} \\
      &\quad \hspace{2.8em} + \|\wt L\|_{\cC_{(\seop;\sop),\eps}^{(d_0;k)}}\Bigl(\|\tilde G u\|_{H_{(\seop;\sop),\eps}^{(\sfs;d)}} + \|\chi u\|_{H_{(\seop;\sop),\eps}^{(-N;d)}}\Bigr) \biggr].
  \end{align*}

  Using the mapping properties of $\Psi_\seop^0$, we can moreover bound
  \begin{equation}
  \label{EqNTameEllBup}
    \|B u'\|_{H_{\seop,\eps}^\sfs} \geq c\|B u\|_{H_{(\seop;\sop),\eps}^{(\sfs;k)}} - C\|\tilde B u\|_{H_{(\seop;\sop),\eps}^{(\sfs;k-1)}}
  \end{equation}
  where $\tilde B\in\Psi_\seop^0$ is elliptic on $\WF'_\seop(B)$. Altogether, we thus obtain the desired tame estimate~\eqref{EqNTameEll} except with slightly larger cutoffs and with an additional term $C_k(\|\tilde G u\|_{H_{(\seop;\sop),\eps}^{(\sfs;k-1)}}+\|\tilde B u\|_{H_{(\seop;\sop),\eps}^{(\sfs;k-1)}})$ on the right hand side. Choosing, as we may, the operator wave front sets of $G$ and $\tilde B,\tilde G$ to be supported in a small neighborhood of $\WF'_\seop(B)$, we can estimate this term inductively to conclude.
\end{proof}

We next prove tame estimates in the setting of Lemma~\ref{LemmaFVarsseProp}.

\begin{lemma}[se-microlocal real principal type propagation: tame estimate]
\label{LemmaNTameProp}
  Given $m,\sfs,N$, there exist $d_0,d\in\N$ so that for all $k\in\N_0$ there exists a constant $C_k$ so that, under the assumptions of Lemma~\usref{LemmaFVarsseProp}, we have a tame estimate
  \begin{equation}
  \label{EqNTameProp}
  \begin{split}
    \|B u\|_{H_{(\seop;\sop),\eps}^{(\sfs;k)}} &\leq C_k\biggl[ \|G \wt L u\|_{H_{(\seop;\sop),\eps}^{(\sfs-m+1;k)}} + \|E u\|_{H_{(\seop;\sop),\eps}^{(\sfs;k)}} + \|\chi u\|_{H_{(\seop;\sop),\eps}^{(-N;k)}} \\
      &\quad \hspace{2em} + \|\wt L_1\|_{\cC_{(\seop;\sop),\eps}^{(d_0;k)}}\Bigl( \|G \wt L u\|_{H_{(\seop;\sop),\eps}^{(\sfs-m+1;d)}} + \|E u\|_{H_{(\seop;\sop),\eps}^{(\sfs;d)}} + \|\chi u\|_{H_{(\seop;\sop),\eps}^{(-N;d)}} \Bigr) \biggr]
  \end{split}
  \end{equation}
  for all $u$ with support in $t\geq t_0$. The constant $C_k$ can be taken to be uniform for perturbations of $\wt L=\wt L_0+\wt L_1$ which are sufficiently small in that $\wt L_0$ is perturbed in $\Psi_\seop^m(\wt M)$ and $\wt L_1$ in $\cC_{\seop;\sop}^{(d_0;d),\delta,\delta}\Psi_\seop^m(\wt M)$.
\end{lemma}
\begin{proof}
  As before, we only need to consider $k$ above some large but fixed value $d$. As in the second proof of Lemma~\ref{LemmaFVarsseProp}, write $[\wt L,D_{[\sop]}]=A D_{[\sop]}+A'$ where the coefficients of $A,A'\in(\CI+\cC_{\seop;\sop}^{(d_0;k-1)})\Psi_\seop^{m-1}$ are bounded in norm for the parameter value $\eps$ by $\|\wt L_1\|_{\cC_{(\seop;\sop),\eps}^{(d_0;k)}}$. Thus, $(\wt L-A)D_{[\sop]}=D_{[\sop]}\wt L+A'$. Writing $f=\wt L u$, we then have (schematically)
  \begin{equation}
  \label{EqNTamePropfp}
    (\wt L-A)D_{[\sop]}^k u = f' := D_{[\sop]}\wt L D_{[\sop]}^{k-1}u + A' D_{[\sop]}^{k-1}u = D_{[\sop]}^k f + A' D_{[\sop]}^{k-1} u + \sum_{j=2}^k ({\rm ad}_{D_{[\sop]}}^j\wt L) D_{[\sop]}^{k-j}u.
  \end{equation}
  (The operator $\wt L-A$ on the left hand side acts on the vector of $k$-fold derivatives of $u$ along a fixed spanning set of elements of $\cV_{[\sop]}$, with $\wt L$ acting component-wise and $A$ being an appropriate matrix of operators of class $(\CI+\cC_{\seop;\sop}^{(d_0;k-1)})\Psi_\seop^{m-1}$.) We then apply the propagation estimate on $H_{\seop,\eps}^\sfs$ to this equation. The main task is to estimate $\|G f'\|_{H_{\seop,\eps}^{\sfs-m+1}}$, which we do using the arguments following~\eqref{EqNTameEllGfp}, except for two differences.

  First, the sum in~\eqref{EqNTamePropfp} starts at $j=2$; for $2\leq j\leq d$, we use ${\rm ad}_{D_{[\sop]}}^j\wt L\in\cC_{\seop;\sop}^{(d_0;k-j)}\Psi_\seop^m$ and relax the $\Psi_\seop^m$ part to $\Psi_\seop^{m-1}\Diffs^1\supset\Psi_\seop^m$; this yields the estimate
  \begin{equation}
  \label{EqNTamePropEst}
    \|G({\rm ad}_{D_{[\sop]}}^j\wt L)D_{[\sop]}^{k-j}u\|_{H_{\seop,\eps}^{\sfs-m+1}} \leq C\|\wt L\|_{\cC_{(\seop;\sop),\eps}^{(d_0;d)}}\Bigl(\|\tilde G u\|_{H_{(\seop;\sop),\eps}^{(\sfs;k-1)}}+\|\chi u\|_{H_{(\seop;\sop),\eps}^{(-N;k-1)}}\Bigr),
  \end{equation}
  where we use that $k-j+1\leq k-1$. For $j>d$, we can directly quote Lemma~\ref{LemmaNTameMult} to obtain a tame estimate.

  Second, we need to estimate the additional term $\|G A' D_{[\sop]}^{k-1}u\|_{H_{\seop,\eps}^{\sfs-m+1}}$; but since $A' D_{[\sop]}^{k-1}=A' D_{[\sop]} D_{[\sop]}^{k-2}$, with $A' D_{[\sop]}\in\Psi_\seop^{m-1}\Diffs^1$, this can be estimated in exactly the same fashion.

  Finally, using an inductive argument in $k$ to estimate (microlocalized) $H_{(\seop;\sop),\eps}^{(\sfs;k-1)}$-norms of $u$, we conclude the proof of the tame estimate~\eqref{EqNTamePropfp}.
\end{proof}

Lemmas~\ref{LemmaNTameEll}--\ref{LemmaNTameProp} remain valid, \emph{mutatis mutandis}, for weighted operators on weighted Sobolev spaces.

We continue with tame versions of radial point estimates. Since the arguments are very similar to those in the proof of Lemma~\ref{LemmaNTameProp}, we shall only prove a tame version of the outgoing radial point estimate, Proposition~\ref{PropEstRadsOut}, and leave the (notational) modifications required to prove tame versions of the other estimates (Propositions~\ref{PropEstRadsIn} and \ref{PropEstHors}) to the reader.

\begin{prop}[Propagation near $\cR_{\rm out}^+$: tame estimate]
\label{PropNTameRout}
  Let $s,s_0,N\in\R$. Then there exist $d_0,d\in\N_0$ so that the following holds for operators $\wt L$ as in~\eqref{ItEstLBundle}--\eqref{ItEstLOp} in~\S\usref{SEst}, or indeed for operators satisfying the weaker regularity assumptions
  \begin{equation}
  \label{EqNTameRoutOp}
    \wt L=\wt L_0+\wt L_1,\qquad \wt L_0\in\hat\rho^{-2}(\Diffse^2+\Psi_\seop^1),\quad \wt L_1\in\hat\rho^{-2}\cC_{\seop;\sop}^{(d_0;d),1,1}(\Diffse^2+\Psi_\seop^1).
  \end{equation}
  Let $\alpha_\circ,\hat\alpha\in\R$ with $s+\alpha_\circ-\hat\alpha<\frac12(-1-\vartheta_{\rm out})$ in the notation of Proposition~\usref{PropEstRadOut}. Define $\cK=\{\rho_\infty=\hat\rho=\rho_\circ=0,\ \hat\xi_\seop=1,\ t_0\leq t\leq t_1\}$. Then for all neighborhoods $\cU\subset\Sse^*\wt M$ of $\cK$ and all $\chi\in\CIc(\wt M)$ equal to $1$ near $\pa M_\circ\cap\{t_0\leq t\leq t_1\}$, there exist operators $B,E,G\in\Psi_\seop^0(\wt M)$ with $\chi B\chi=B$ and $\WF_\seop'(B)\subset\cU$, similarly for $E,G$, with $B$ elliptic at $\cK$ and $\WF'_\seop(E)\cap\pa W_{\rm out}^+=\emptyset$ so that for all $k\in\N_0$, there exists a constant $C_k$ so that we have the uniform tame estimate
  \begin{equation}
  \label{EqNTameRout}
  \begin{split}
    &\|B u\|_{H_{(\seop;\sop),\eps}^{(s;k),\alpha_\circ,\hat\alpha}} \\
    &\quad \leq C_k\biggl[ \|G\wt L u\|_{H_{(\seop;\sop),\eps}^{(s-1;k),\alpha_\circ,\hat\alpha-2}} + \|E u\|_{H_{(\seop;\sop),\eps}^{(s;k),\alpha_\circ,\hat\alpha}} + \|G u\|_{H_{(\seop;\sop),\eps}^{(s_0;k),\alpha_\circ,\hat\alpha}} + \|\chi u\|_{H_{(\seop;\sop),\eps}^{(-N;k),\alpha_\circ,\hat\alpha}} \\
    &\quad \hspace{4em} + \|\wt L_1\|_{\hat\rho^{-2}\cC_{(\seop;\sop),\eps}^{(d_0;k),1,1}} \Bigl( \|G\wt L u\|_{H_{(\seop;\sop),\eps}^{(s-1;d),\alpha_\circ,\hat\alpha-2}} + \|E u\|_{H_{(\seop;\sop),\eps}^{(s;d),\alpha_\circ,\hat\alpha}} \\
    &\quad \hspace{13em} + \|G u\|_{H_{(\seop;\sop),\eps}^{(s_0;d),\alpha_\circ,\hat\alpha}} + \|\chi u\|_{H_{(\seop;\sop),\eps}^{(-N;d),\alpha_\circ,\hat\alpha}} \Bigr) \biggr].
  \end{split}
  \end{equation}
  The constant $C_k$ can be taken to be uniform for perturbations of $\wt L_0$ in $\hat\rho^{-2}\Diffse^2$ and of $\wt L_1$ in $\hat\rho^{-2}\cC_{\seop;\sop}^{(d_0;d),1,1}\Diffse^2$.
\end{prop}
\begin{proof}
  Starting with $\wt L u=f$, we recall from Lemma~\ref{LemmaEstRadsComm} that $[\wt L,\pa_t]=\sum_{l=1}^N a_l(A_l\hat\rho\pa_t+R_l)$ where
  \begin{equation}
  \label{EqNTameRoutComm}
  \begin{gathered}
    a_l=a_{l,0}+a_{l,1},\qquad a_{l,0}\in\hat\rho^{-2}\CI,\quad a_{l,1}\in\hat\rho^{-2}\cC_{\seop;\sop}^{(d_0;k-1),1,1}, \\
    A_l\in\Psi_\seop^1,\qquad R_l\in\Psi_\seop^2,\ \ \WF'_\seop(R_l)\cap\bar\cU=\emptyset,
  \end{gathered}
  \end{equation}
  provided we shrink $\cU$ so that $\Char_\seop(\hat\rho\pa_t)\cap\bar\cU=\emptyset$.

  We can then rewrite the equation
  \begin{equation}
  \label{EqNTameRoutEq}
    \wt L(\pa_t^k u) = \pa_t^k f + k[\wt L,\pa_t]\pa_t^{k-1} u + \sum_{j=2}^k c_{k,j}({\rm ad}_{\pa_t}^j\wt L)(\pa_t^{k-j}u),
  \end{equation}
  where the $c_{k,j}$ are combinatorial constants, as
  \[
    \Bigl(\wt L-k\sum\nolimits_l a_l A_l\hat\rho\Bigr)\pa_t^k u = f' := \pa_t^k f + k\sum_{l=1}^N a_l R_l\pa_t^{k-1}u + \sum_{j=2}^k c_{k,j}({\rm ad}_{\pa_t}^j\wt L)(\pa_t^{k-j}u).
  \]
  We now apply the radial point estimate in Proposition~\ref{PropEstRadOut} to this equation. Importantly, the subprincipal term $-k\sum_l a_l A_l\hat\rho$ has vanishing principal symbol at $\hat M$ and thus in particular at $\pa\cR_{\rm out}^+$. Note that due to the se-ellipticity of $\hat\rho\pa_t=\hat\rho\cdot\pa_t$, the quantity $\|B\pa_t^k u\|_{H_{\seop,\eps}^{s,\alpha_\circ,\hat\alpha}}$ controls $\|B u\|_{H_{(\seop;\sop),\eps}^{(s;k),\alpha_\circ,\hat\alpha}}$ up to an error $\|\tilde B u\|_{H_{(\seop;\sop),\eps}^{(s;k-1),\alpha_\circ,\hat\alpha}}+\|\chi u\|_{H_{(\seop;\sop),\eps}^{(s_0;k-1),\alpha_\circ,\hat\alpha}}$ where $\tilde B\in\Psi_\seop^0$ is elliptic on $\WF'_\seop(B)$ (cf.\ \eqref{EqNTameEllBup}). For $k\leq d$ with $d$ large but fixed, this gives~\eqref{EqNTameRout}. The main task for proving a tame estimate for $k>d$ is to tamely estimate $\|G f'\|_{H_{\seop,\eps}^{s-1,\alpha_\circ,\hat\alpha-2}}$. This can be done in the same fashion as in the proof of Lemma~\ref{LemmaNTameProp}, except we now also need to estimate the terms
  \begin{equation}
  \label{EqNTameRoutResidual}
    \|G a_l R_l\pa_t^{k-1}u\|_{H_{\seop,\eps}^{s-1,\alpha_\circ,\hat\alpha-2}}.
  \end{equation}
  Commuting $\pa_t^{k-1}$ through $R_l$, one sees that it suffices to bound $\|G a_l\pa_t^j R_{l,j} u\|_{H_{\seop,\eps}^{s-1,\alpha_\circ,\hat\alpha-2}}$ for $j\leq k-1$ where still $R_{l,j}\in\Psi_\seop^2$ with $\WF'_\seop(R_{l,j})\cap\bar\cU=\emptyset$; we focus on the case $j=k-1$ where $R_{l,k-1}=R_l$. We further write $a_l\circ\pa_t^{k-1}$ as a sum of terms $(\pa_t^j a_l)\circ\pa_t^{k-1-j}$, $j=0,\ldots,k-1$; then Lemma~\ref{LemmaNTameMult} (with $k-1$ in place of $k$) gives the bound
  \begin{align*}
    &\|G(\pa_t^j a_l)\pa_t^{k-1-j}R_l u\|_{H_{\seop,\eps}^{s-1,\alpha_\circ,\hat\alpha-2}} \\
    &\quad \leq C_k\biggl[ \|\tilde G R_l u\|_{H_{(\seop;\sop),\eps}^{(s-1;k-1),\alpha_\circ,\hat\alpha}} + \|\chi u\|_{H_{(\seop;\sop),\eps}^{(-N;k-1),\alpha_\circ,\hat\alpha}} \\
    &\quad \qquad + \|a_{l,1}\|_{\cC_{(\seop;\sop),\eps}^{(d_0;k-1),1,1}}\Bigl( \|\tilde G R_l u\|_{H_{(\seop;\sop),\eps}^{(s-1;d),\alpha_\circ,\hat\alpha}} + \|\chi u\|_{H_{(\seop;\sop),\eps}^{(-N;d),\alpha_\circ,\hat\alpha}}\Bigr) \biggr]
  \end{align*}
  where $\tilde G\in\Psi_\seop^0$ (playing the role of $B^\sharp$ in Lemma~\ref{LemmaNTameMult}) is elliptic on $\WF_\seop'(G)$ but still satisfies $\WF_\seop'(\tilde G)\cap\WF_\seop'(R_l)=\emptyset$. Since $\tilde G R_l$ is a smoothing operator, we can estimate $\|\tilde G R_l u\|_{H_{(\seop;\sop),\eps}^{(s-1;k-1),\alpha_\circ,\hat\alpha}}\leq C_{k-1}\|\chi u\|_{H_{(\seop;\sop),\eps}^{(-N;k-1),\alpha_\circ,\hat\alpha}}$ simply.

  Next, observe that $\|E\pa_t^k u\|_{H_{\seop,\eps}^{s,\alpha_\circ,\hat\alpha}}$ is bounded by a constant (which is independent of $\wt L$ of course) times $\|E u\|_{H_{(\seop;\sop),\eps}^{(s;k),\alpha_\circ,\hat\alpha}}$ plus terms arising from the commutator $[E,\pa_t^k]$, which can thus be estimated by $\|\tilde E u\|_{H_{(\seop;\sop),\eps}^{(s;k-1),\alpha_\circ,\hat\alpha}}+\|\chi u\|_{H_{(\seop;\sop),\eps}^{(s_0;k-1),\alpha_\circ,\hat\alpha}}$ analogously to~\eqref{EqNTameEllBup}. The terms involving $G\pa_t^k u$ and $\chi\pa_t^k u$ are treated similarly.

  In summary, we have now proved the estimate~\eqref{EqNTameRout} except for an additional term $C_k\|\tilde G u\|_{H_{(\seop;\sop),\eps}^{(s;k-1),\alpha_\circ,\hat\alpha}}$ on the right, with $\tilde G$ elliptic near $\WF_\seop'(G)\cup\WF_\seop'(E)$; this is estimated by induction on $k$, finishing the proof.
\end{proof}

Finally, we discuss tame estimates at trapping. As in the previous tame estimates, we apply the basic se-estimate (involving microlocalizers with s-regularity, as in Proposition~\ref{PropEstTraps}) to $k$-th order derivatives of $u$. (We remark that we therefore in particular do \emph{not} need to prove (tame) s-regularity of the stable and unstable defining functions constructed in Proposition~\ref{PropTrap}.)

\begin{prop}[Uniform tame estimate near the trapped set]
\label{PropNTameTrap}
  Let $s,N\in\R$. Then there exist $d_0,d\in\N_0$ so that the following holds for operators $\wt L=\wt L_0+\wt L_1$ as in~\eqref{ItEstLBundle}--\eqref{ItEstLOp} in~\S\usref{SEst} except for allowing the more general form~\eqref{EqNTameRoutOp}. Let $\alpha_\circ,\hat\alpha\in\R$. Let $\cU\subset\Sse^*\wt M$ be a neighborhood of $\pa\Gamma^+\cap\{t_0\leq t\leq t_1\}$, and let $\chi\in\CIc(\wt M)$ be equal to $1$ near the base projection of $\pa\Gamma^+\cap\{t_0\leq t\leq t_1\}$. Then there exist $B_\Gamma,B^{\rm s},G\in\Psi_\seop^0$ with $\chi B_\Gamma\chi=B_\Gamma$ and $\WF'_\seop(B_\Gamma)\subset\cU$ etc., with $B_\Gamma$ elliptic at $\pa\Gamma^+\{t_0\leq t\leq t_1\}$ and $\WF_\seop'(B^{\rm s})\cap\Gamma^{\rm u,+}=\emptyset$ so that for all $k\in\N_0$, there exists a constant $C_k$ so that we have the uniform tame estimate
  \begin{equation}
  \label{EqNTameTrap}
  \begin{split}
    &\|B_\Gamma u\|_{H_{(\seop;\sop),\eps}^{(s;k),\alpha_\circ,\hat\alpha}} \\
    &\quad \leq C_k\biggl[ \|G\wt L u\|_{H_{(\seop;\sop),\eps}^{(s;k),\alpha_\circ,\hat\alpha-2}} + \|B^{\rm s}u\|_{H_{(\seop;\sop),\eps}^{(s+1;k),\alpha_\circ,\hat\alpha}} + \|\chi u\|_{H_{(\seop;\sop),\eps}^{(-N;k),\alpha_\circ,\hat\alpha}} \\
    &\quad \hspace{3em} + \|\wt L_1\|_{\hat\rho^{-2}\cC_{(\seop;\sop),\eps}^{(d_0;k),1,1}}\Bigl( \|G\wt L u\|_{H_{(\seop;\sop),\eps}^{(s;d),\alpha_\circ,\hat\alpha-2}} + \|B^{\rm s}u\|_{H_{(\seop;\sop),\eps}^{(s+1;d),\alpha_\circ,\hat\alpha}} + \|\chi u\|_{H_{(\seop;\sop),\eps}^{(-N;d),\alpha_\circ,\hat\alpha}} \Bigr) \biggr]
  \end{split}
  \end{equation}
  for all $u$ with support in $t\geq t_0$. The constant $C_k$ can be taken to be uniform for perturbations of $\wt L_0$ in $\hat\rho^{-2}\Diffse^2$ and of $\wt L_1$ in $\hat\rho^{-2}\cC_{\seop;\sop}^{(d_0;d),1,1}\Diffse^2$.
\end{prop}
\begin{proof}
  We cannot quite follow the same argument as in the proof of Proposition~\ref{PropNTameRout}: this would now involve the estimate~\eqref{EqNTamePropEst} but with $s-m+2$ instead of $\sfs-m+1$ (and now concretely with $m=2$) for $j\geq 2$: one would be forced to use the $H_{(\seop;\sop),\eps}^{(s;k)}$-norm of $\tilde G u$ on the right hand side---which however is what one is trying to estimate. To resolve this issue, we shall rewrite also the $j=2$ term on the right in~\eqref{EqNTameRoutEq} as a contribution to the operator to which we apply the original trapping estimate, and regard only the terms with $j\geq 3$ as error terms which we can estimate inductively.

  Writing $[\wt L,\pa_t]=\sum_{l=1}^N a_l(A_l\hat\rho\pa_t+R_l)$ as in~\eqref{EqNTameRoutComm}, we compute
  \[
    [\pa_t,[\wt L,\pa_t]] = \sum_{l=1}^N (\pa_t a_l)(A_l\hat\rho\pa_t+R_l) + a_l([\pa_t,A_l]\hat\rho\pa_t+[\pa_t,R_l]),
  \]
  which due to $\pa_t\in\cV_{[\sop]}$ we can write as
  \[
    [\pa_t,[\wt L,\pa_t]] = \sum_{m=1}^{N'} b_m(B_m\hat\rho\pa_t+Q_m),\qquad b_m\in\hat\rho^{-2}\CI+\hat\rho^{-2}\cC_{\seop;\sop}^{(d_0;k-2),1,1},
  \]
  where $B_m\in\Psi_\seop^1$ and $Q_m\in\Psi_\seop^2$, $\WF'_\seop(Q_m)\cap\bar\cU=\emptyset$. (Here $\cU$ is chosen so small that $\bar\cU\cap\Char_\seop(\hat\rho\pa_t)=\emptyset$.) We furthermore write $I=B\hat\rho\pa_t+Q$ where $B\in\Psi_\seop^{-1}$ and $Q\in\Psi_\seop^0$ with $\WF'_\seop(Q)\cap\bar\cU=\emptyset$. Writing $f=\wt L u$, we then have
  \begin{align*}
    \wt L(\pa_t^k u) &= \pa_t^k f + c_{k,1}\sum_{l=1}^N a_l(A_l\hat\rho\pa_t+R_l)\pa_t^{k-1}u \\
      &\quad\hspace{2em} + c_{k,2}\sum_{m=1}^{N'} b_m\bigl(B_m\hat\rho(B\hat\rho\pa_t+Q)\pa_t+Q_m\bigr)\pa_t^{k-2}u + \sum_{j=3}^k ({\rm ad}_{\pa_t}^j\wt L) \pa_t^{k-j}u
  \end{align*}
  (using schematic notation for the sum over $j$), which we rewrite as
  \begin{align*}
    &\biggl(\wt L-c_{k,1}\sum_{l=1}^N a_l A_l\hat\rho - c_{k,2}\sum_{m=1}^{N'} b_m B_m\hat\rho B\hat\rho\biggr)\pa_t^k u = \pa_t^k f + f'_1 + f'_2, \\
    &\qquad \quad f'_1 := \sum_{j=3}^k ({\rm ad}_{\pa_t}^j\wt L)\pa_t^{k-j}u, \\
    &\qquad \quad f'_2 := \biggl(c_{k,1}\sum_{l=1}^N a_l R_l+c_{k,2}\sum_{m=1}^{N'}b_m B_m\hat\rho Q\biggr)\pa_t^{k-1}u + c_{k,2}\sum_{m=1}^{N'} b_m Q_m\pa_t^{k-2}u.
  \end{align*}
  We can then apply Proposition~\ref{PropEstTraps} (with $k=0$) to this equation. We need to prove tame estimates for $\|G f'_i\|_{H_{\seop,\eps}^{s,\alpha_\circ,\hat\alpha-2}}$, $i=1,2$. For $i=2$, these can be proved in the same fashion as for the term~\eqref{EqNTameRoutResidual}; this is due to the fact that $\WF_\seop'(G)$ is disjoint from the operator wave front sets of $R_l\in\Psi_\seop^2$, $B_m\hat\rho Q\in\Psi_\seop^{-\infty}$, and $Q_m\in\Psi_\seop^2$. For the estimate on $\|G f'_2\|_{H_{\seop,\eps}^{s,\alpha_\circ,\hat\alpha-2}}$, we note that ${\rm ad}_{\pa_t}^j\wt L\in\hat\rho^{-2}(\CI+\cC_{\seop;\sop}^{(d_0;k-j)})(\Diffse^2+\Psi_\seop^1)$ and use $\Diffse^2+\Psi_\seop^1\subset\Diffs^2\Psi_\seop^0$ to estimate, for $j\geq 3$,
  \[
    \|G({\rm ad}_{\pa_t}^j\wt L)\pa_t^{k-j}u\|_{H_{\seop,\eps}^{s,\alpha_\circ,\hat\alpha-2}} \leq C\Bigl(\|\tilde G u\|_{H_{(\seop;\sop),\eps}^{(s;k-1),\alpha_\circ,\hat\alpha}} + \|\chi u\|_{H_{(\seop;\sop),\eps}^{(-N;k-1),\alpha_\circ,\hat\alpha}}\Bigr),
  \]
  where $\tilde G\in\Psi_\seop^0$ is elliptic on $\WF'_\seop(G)$. Using this estimate for $j\leq d$, and a similar estimate (now involving $k$ s-derivatives on $\wt L$ and $d$ s-derivatives on $u$) for $j\geq k-d$, it remains to estimate those terms with $d+1\leq j\leq k-d-1$ using Lemma~\ref{LemmaNTameMult}. An inductive argument in $k$ finishes the proof.
\end{proof}

\subsubsection{Uniform tame estimates on s-Sobolev spaces}
\label{SssNTameSg}

Assuming the invertibility of the Kerr model operator, we can now combine the s-tame se-microlocal estimates established in~\S\ref{SssNTameMl} to prove a tame analogue of Theorem~\ref{ThmScS}. For concreteness, we only consider the case that the Kerr model operator is equal to the scalar wave operator on $(\hat M_b,\hat g_b)$.

\begin{thm}[Uniform estimates for linear waves on small domains: tame estimates]
\label{ThmNTame}
  There exists a number $d\in\N$ such that the following holds for all operators $\wt L=\wt L_0+\wt L_1$ as in~\eqref{ItEstLBundle}--\eqref{ItEstLOp} in~\S\usref{SEst}, except we now require $\wt L_0\in\hat\rho^{-2}\Diffse^2$ and $\wt L_1\in\hat\rho^{-2}\cC_\sop^{d,1,1}\Diffse^2$;\footnote{Correspondingly, the underlying metric $\wt g$ is only required to be of class $\CI+\cC_\sop^{d,1,1}$ as a section of $S^2\wt T^*\wt M$, cf.\ Definition~\usref{DefGl}\eqref{ItGlMetric}.} and we require the $\hat M$-normal operator of $\wt L$ to be equal to the scalar wave operator $\Box_{\hat g_b}$ on a subextremal Kerr spacetime. Let $t_0\in I_\cC$, and define $\wt\Omega_{(\lambda)}\subset\wt M\setminus\wt K^\circ$ to be the standard domain associated with $\Omega_{(\lambda)}:=\Omega_{t_0,t_0+\lambda,\lambda}$; let $\lambda_1>0$ be such that $[t_0,t_0+\lambda_1]\subset I_\cC$. Let $\alpha_\circ,\hat\alpha\in\R$ with $\alpha_\circ-\hat\alpha\in(-\frac32,-\frac12)$. Then there exist $\lambda_0\in(0,\lambda_1]$ and $\eps_0>0$ depending only on $t_0$ and $\wt L|_{\Omega_{(\lambda_1)}}$ so that for all $\lambda\in(0,\lambda_0]$ and all $\N_0\ni k\geq d$, there exists a constant $C=C(\lambda,k)$ so that for the solution $u$ of $\wt L u=f$, the tame estimate
  \begin{equation}
  \label{EqNTame}
  \begin{split}
    &\|u\|_{H_{\sop,\eps}^{k,\alpha_\circ,\hat\alpha}(\Omega_{(\lambda),\eps})^{\bullet,-}} \\
    &\qquad \leq C\Bigl( \|f\|_{H_{\sop,\eps}^{k+d,\alpha_\circ,\hat\alpha-2}(\Omega_{(\lambda),\eps})^{\bullet,-}} + \|\wt L_1\|_{\hat\rho^{-2}\cC_{\sop,\eps}^{k+d,1,1}(\Omega_{(\lambda),\eps})}\|f\|_{H_{\sop,\eps}^{d,\alpha_\circ,\hat\alpha-2}(\Omega_{(\lambda),\eps})^{\bullet,-}}\Bigr)
  \end{split}
  \end{equation}
  holds uniformly for $\eps\in(0,\lambda\eps_0)$.\footnote{Here, as before, the norm on $\wt L_1$ is the sum of norms of the coefficients of $\wt L_1$ when expressed in terms of the standard se-vector fields $\hat\rho\pa_t$, $\hat\rho\pa_x$.} The constants $\lambda_0$ and $C$ can be chosen to be uniform for perturbations of $\wt L_1|_{\wt\Omega_{(\lambda)}}$ whose $\hat\rho^{-2}\cC_{\sop,\eps}^{d,1,1}\Diffse^2$-norm is bounded by $1$ (say) for $\eps\in(0,\eps_0]$.
\end{thm}

One can prove an analogous result on general standard domains by building on Theorem~\ref{ThmScUnif}. However, we restrict to small domains here since in \cite{HintzGlueLocIII} we will only prove an analogue of Theorem~\ref{ThmScUnifSm}. One can also obtain uniform estimates for perturbations of the smooth part $\wt L_0$, but since this flexibility is not used in our applications, we do not discuss this further.

\begin{rmk}[Other domains]
\label{RmkNTameOther}
  Analogously to Remark~\ref{RmkScUnifSmShift}, Theorem~\ref{ThmNTame} remains valid for the standard domains associated with $\Omega_{t_0+c_1\lambda,t_0+c_2\lambda,r_0\lambda}$ for any $c_1<c_2$ and $r_0>0$, with $\lambda_0$ depending on $c_1,c_2,r_0$.
\end{rmk}

The proof will use an extension/restriction procedure. For the minimal norm extensions used e.g.\ in the proof of Theorem~\ref{ThmEstStd} for a fixed choice of function space, one cannot deduce boundedness properties on other (less regular) function spaces. For tame applications however, we need to quantitatively control low and high regularity norms of extensions. This is accomplished by the following result.

\begin{lemma}[Extension operators]
\label{LemmaNTameExt}
  Consider $M'=\R_t\times\R^3_x$, $\wt M'=[[0,1)\times M';\{0\}\times\R\times\{0\}]$. Let $t_0<t_1<t_2$ and $r_0<r_1$. Denote by $\wt\Omega,\wt\Omega^\sharp\subset\wt M'$ the lifts of $[0,1)_\eps$ times $\Omega_{t_0,t_1,r_0},\Omega_{t_0,t_2,r_1}\subset M'$, and write $\Omega_\eps=\wt\Omega\cap M_\eps$, $\Omega^\sharp_\eps=\wt\Omega^\sharp\cap M_\eps$. Then there exists a family $E_\eps$, $\eps\in(0,1)$, of linear operators $E_\eps\colon\CI(\Omega_\eps)\to\CI(\Omega^\sharp_\eps)$ with the following properties:
  \begin{enumerate}
  \item $E_\eps$ is an extension operator, i.e.\ $(E_\eps u)|_{\Omega_\eps}=u$ for all $u\in\CI(\Omega_\eps)^{\bullet,-}$;
  \item there exist constants $C_j$, $j\in\N_0$ so that we have uniform (in $\eps$) estimates
    \[
      \|E_\eps u\|_{H_{\sop,\eps}^j(\Omega^\sharp_\eps)^{\bullet,-}} \leq C_j\|u\|_{H_{\sop,\eps}^j(\Omega_\eps)^{\bullet,-}},\qquad
      \|E_\eps u\|_{\cC_{\sop,\eps}^j(\Omega^\sharp_\eps)} \leq C_j\|u\|_{\cC_{\sop,\eps}^j(\Omega_\eps)}.
    \]
  \end{enumerate}
  An analogous statement holds when replacing $\Omega_\eps$ and $\Omega_\eps^\sharp$ by their intersections with $\{\frac{|x|}{\eps}\geq\hat r_0\}$ and $\{\frac{|x|}{\eps}\geq\hat r_1\}$, respectively, where $\hat r_1<\hat r_0$.
\end{lemma}
\begin{proof}
  This is a variant of Seeley extension, see \cite{SeeleyExtension}. For notational simplicity, we shall only consider the local extension from $K_\eps:=\{-1\leq t\leq 0\}\cap\{\eps<r=|x|<1<r_0\}$ to $K^\sharp_\eps:=\{-1\leq t\leq 1\}\cap\{\eps<r=|x|<r_0\}$; and we ignore spherical variables. (Extensions near the corners of $\Omega_\eps$ can then be defined by first extending across one and then the other hypersurface, as in \cite[\S{1.4}]{MelroseDiffOnMwc}.) Given a function $u=u(t,r)$ on $K_\eps$, we set
  \[
    \tilde u(t,r) := \begin{cases} u(t,r), & t<0, \\ \sum_{l=0}^\infty c_l\cdot\chi\bigl(-\frac{t}{\delta_l}\bigr) u\bigl(-\frac{t}{\delta_l},r\bigr), & t>0, \end{cases}
  \]
  where $\chi\in\CIc((-1,0])$ equals $1$ near $0$, and with $c_l\in\R$ and $\delta_l>0$ (with $\delta_l\searrow 0$) fixed so that, for all $j\in\N_0$, one has $\sum_{l=0}^\infty c_l\delta_l^{-j}=(-1)^j$ and $\sum_{l=0}^\infty|c_l|\delta_l^{-j}<\infty$. (We may define $c_l$ via $\sin(\frac{\pi}{2}z)=\sum_{l=0}^\infty c_l z^l$ and take $\delta_l=3^{-l}$; thus $\sum_{l=0}^\infty c_l\delta_l^{-j}=\sin(\frac{\pi}{2}3^j)=(-1)^j$ since $3^j\equiv(-1)^j\bmod 4$, and $\sum_{l=0}^\infty|c_l||z|^l<\infty$ for all $z\in\C$ (by absolute convergence) gives also $\sum_{l=0}^\infty|c_l|\delta_l^{-j}<\infty$ for all $j$.) This implies that $\tilde u$ is smooth across $t=0$. Furthermore, we have
  \[
    \|\tilde u\|_{L^2(\{0\leq t\leq 1\})} \leq \sum_{l=0}^\infty |c_l|\delta_l^{\frac12}\|u\|_{L^2(\{-1\leq t\leq 0\})} \leq C\|u\|_{L^2},
  \]
  similarly for derivatives in $t$ (each of which produces a power of $\delta_l^{-1}$) and along $(\eps^2+r^2)^{\frac12}\pa_r$ (which do not produce additional powers). Since near $K_\eps^\sharp$, the space of s-vector fields is spanned by $\pa_t$ and $(\eps^2+r^2)^{\frac12}\pa_r$, $L^2$-bounds on s-derivatives of $u$ on $K_\eps$ thus imply the same bounds (up to $\eps$-independent factors) on s-derivatives $\tilde u$ on $K^\sharp_\eps\setminus K_\eps$. Similar arguments apply also to $L^\infty$-bounds.
\end{proof}

\begin{proof}[Proof of Theorem~\usref{ThmNTame}]
  \pfstep{Preliminary constructions and simplifications.} Fix an order function $\sfs$ as in~\eqref{EqEstAdmIndEx} on $\wt M$ near $\wt\Omega_{(\lambda_1)}$ so that $\sfs,\alpha_\cD$ and $\sfs-4,\alpha_\cD$ are Kerr-admissible both for $\alpha_\cD=-\frac32$ and for $\alpha_\cD=-\frac12$. Then $\sfs,\alpha_\circ-\hat\alpha$ are Kerr-admissible for all $\alpha_\circ,\hat\alpha$ as in the statement of the Theorem. Similarly to the proof of Theorem~\ref{ThmScUnifSm}, we work on standard domains on $\wt M':=[[0,1)\times\R^{1+3};\{0\}\times\R\times\{0\}]$ of fixed size. Concretely, in terms of the functions $\eps\geq 0$, $t'\in\R$, $x'\in\R^3$ on $\wt M'$, and for $\delta\in[0,1]$, we set
  \[
    \wt\Omega^{(\delta)} := \upbeta'{}^*\bigl([0,1)_\eps\times\Omega_{0,1+\delta,1+\delta}\bigr) \cap \{ |\hat r'|\geq (1-c\delta)\bhm \},\qquad \hat r':=\frac{|x'|}{\eps};
  \]
  here $\upbeta'\colon\wt M'\to[0,1)\times\R^{1+3}$ is the blow-down map, and we fix any $c<\bhm-\hat r_b^-=\sqrt{\bhm^2-a^2}$ (cf.\ Lemma~\ref{LemmaGlCoord}). We will prove tame estimates for the forward solution operator of $\wt L_{(\lambda)}=\lambda^2\wt S_\lambda^*\wt L$ (see~\eqref{EqScUnifSmResc} and \eqref{EqScUnifSmDefs}) on (the $\eps$-level sets $\Omega^{(0)}_\eps$ of) $\wt\Omega^{(0)}$. For sufficiently small $\lambda_0>0$, we have, for every fixed $\lambda\in(0,\lambda_0]$, uniform se-estimates for $\wt L_{(\lambda)}$ on $\Omega^{(\delta)}_\eps$ for $\eps<1$ and $\delta\in[0,1]$, and moreover $\sfs,\alpha_\circ,\hat\alpha$ and $\sfs-4,\alpha_\circ,\hat\alpha$ are admissible for $\wt L_{(\lambda)}$ (see Definition~\ref{DefEstAdm}). We now fix such a value of $\lambda$ and relabel $\wt M',t',x',\wt L_{(\lambda)}$ as $\wt M,t,x,\wt L$. The remainder of our argument will take place on $\wt\Omega^{(1)}$; and the goal is to prove the tame estimate~\eqref{EqNTame} with $\Omega^{(0)}_\eps$ in place of $\Omega_{(\lambda),\eps}$, with uniformity in $\eps\in(0,1)$.

  Let $\sfs_\pm\in\N_0$ be such that $-\sfs_-<\sfs<\sfs_+$. Then the inclusion maps
  \begin{equation}
  \label{EqNTamePassages}
  \begin{alignedat}{2}
    H_{(\seop;\sop),\eps}^{(\sfs;k)}&\hra H_{\sop,\eps}^{k-\sfs_-},&\qquad k&\geq\sfs_-, \\
    H_{\sop,\eps}^{k+\sfs_+}&\hra H_{(\seop;\sop),\eps}^{(\sfs;k)}, &\qquad k&\in\N_0, \\
    \cC_{\sop,\eps}^{d_0+k}&\hra\cC_{(\seop;\sop),\eps}^{(d_0;k)}, &\qquad k&\in\N_0,
  \end{alignedat}
  \end{equation}
  are uniformly bounded; here, in the notation of~\S\ref{SsEstFn}, the domains are $M_\eps$ or $\Omega^{(\delta)}_\eps$, with the spaces having supported/extendible character at initial/final boundary hypersurfaces in the second case. The strategy is to pass from s-spaces to (se;s)-spaces, prove tame estimates using the microlocal elliptic and propagation estimates proved above, and at the end pass back from (se;s)-spaces to s-spaces. In particular, the parameter $d$ will be chosen to exceed $\sfs_++\sfs_-$ as well as the values of $d$ from the above tame microlocal estimates. (We shall not track the value of $d$, and allow $d$ to change throughout the argument, although it will only increase by a finite amount overall which is independent of $\wt L$, $k$, $\alpha_\circ$, $\hat\alpha$, $t_0$, $\lambda$.)

  \pfstep{Extensions from $\wt\Omega^{(0)}$ to $\wt\Omega^{(1)}$.} The requirement that the tame estimate~\eqref{EqNTame} must only feature the norm of $\wt L_1$ on $\Omega^{(0)}_\eps$ (and be stable under perturbations of the coefficients of $\wt L_1$ as measured by $\cC_\sop$-norms \emph{on this set}) means that we cannot work with $\wt L_1$ on $\wt\Omega^{(1)}$, as the $\cC_\sop$-norms of $\wt L_1$ on $\wt\Omega^{(1)}$ cannot be quantitatively bounded by those on $\wt\Omega^{(0)}$. Instead, we apply the extension operator $E_\eps$ of Lemma~\ref{LemmaNTameExt} to the coefficients of $\wt L_1|_{\Omega^{(0)}_\eps}$ and obtain a new operator $\wt L_1'\in\hat\rho^{-2}\cC_{\sop,\eps}^{\infty,1,1}(\Omega^{(1)}_\eps)$ for which we do have quantitative bounds
  \begin{equation}
  \label{EqNTameExtL}
    \|\wt L_1'\|_{\hat\rho^{-2}\cC_{\sop,\eps}^{j,1,1}(\Omega^{(1)}_\eps)} \leq C_k\|\wt L_1\|_{\hat\rho^{-2}\cC_{\sop,\eps}^{j,1,1}(\Omega^{(0)}_\eps)},\qquad j\in\N_0.
  \end{equation}
  We then consider $\wt L':=\wt L_0+\wt L_1'$. For all sufficiently small $\eps>0$, we still have the uniform bounds $\|u\|_{H_{\seop,\eps}^{\sfs,\alpha_\circ,\hat\alpha}(\Omega^{(\delta)}_\eps)^{\bullet,-}}\leq C\|\wt L'u\|_{H_{\seop,\eps}^{\sfs,\alpha_\circ,\hat\alpha-2}(\Omega^{(\delta)}_\eps)^{\bullet,-}}$ (for all $\delta\in[0,1]$) given by Theorem~\ref{ThmScUnifSm}; cf.\ Remark~\ref{RmkScUnifSmPert}.

  \pfstep{Non-tame estimates.} Consider now
  \[
    f\in H_{\sop,\eps}^{k+d,\alpha_\circ,\hat\alpha-2}(\Omega^{(0)}_\eps)^{\bullet,-}
  \]
  We then set
  \[
    \tilde f := E_\eps f \in H_{\sop,\eps}^{k+d,\alpha_\circ,\hat\alpha-2}(\Omega^{(1)}_\eps)^{\bullet,-} \subset H_{(\seop;\sop),\eps}^{(\sfs;k),\alpha_\circ,\hat\alpha-2}(\Omega^{(1)}_\eps)^{\bullet,-},
  \]
  which by Lemma~\ref{LemmaNTameExt} and using~\eqref{EqNTamePassages} satisfies
  \begin{equation}
  \label{EqNTameExtf}
    \|\tilde f\|_{H_{(\seop;\sop),\eps}^{(\sfs;j),\alpha_\circ,\hat\alpha-2}(\Omega^{(1)}_\eps)^{\bullet,-}} \leq C_k\|f\|_{H_{\sop,\eps}^{j+d,\alpha_\circ,\hat\alpha-2}(\Omega^{(0)}_\eps)^{\bullet,-}}.
  \end{equation}
  We can now solve $\wt L'\tilde u=\tilde f$, with the solution satisfying uniform bounds
  \[
    \|\tilde u\|_{H_{(\seop;\sop),\eps}^{(\sfs;k),\alpha_\circ,\hat\alpha}(\Omega_\eps^{(1)})^{\bullet,-}} \lesssim \|\tilde f\|_{H_{(\seop;\sop),\eps}^{(\sfs;k),\alpha_\circ,\hat\alpha-2}(\Omega_\eps^{(1)})^{\bullet,-}},
  \]
  with the implicit constant depending on $\|\wt L_1'\|_{\hat\rho^{-2}\cC_{(\seop,\sop),\eps}^{(d_0;k),1,1}(\Omega_\eps^{(1)})}$ (thus, this estimate is not yet tame) for some large but fixed $d_0\in\N_0$; this follows by repeating the arguments in the proof of Theorem~\ref{ThmScUnifSm}. Upon restriction to $\Omega_\eps^{(0)}$, and using the bounds~\eqref{EqNTameExtL} and \eqref{EqNTameExtf}, this gives
  \begin{equation}
  \label{EqNTameEstNontame}
    \|u\|_{H_{\sop,\eps}^{k,\alpha_\circ,\hat\alpha}(\Omega^{(0)}_\eps)^{\bullet,-}} \leq F_k\bigl(\|\wt L_1\|_{\hat\rho^{-2}\cC_{\sop,\eps}^{k+d,1,1}(\Omega^{(0)}_\eps)}\bigr) \|f\|_{H_{\sop,\eps}^{k+d,\alpha_\circ,\hat\alpha}(\Omega^{(0)}_\eps)^{\bullet,-}}
  \end{equation}
  where $F_k\colon[0,\infty)\to[0,\infty)$ is some non-decreasing function. For $k$ below any fixed finite value, this is of the form~\eqref{EqNTame}.

  \pfstep{Tame estimates.} In order to obtain a tame estimate for all $k$, we shall use a mild adaptation of the arguments for Theorem~\ref{ThmScUnifSm}. We will use an inductive argument in which we control $j$ degrees of s-regularity on the domain $\wt\Omega^{(1-\frac{j}{k})}$ (which shrinks as $j$ increases). To wit, consider $k'\in\{2,\ldots,k\}$, and suppose we have already obtained a tame estimate
  \begin{equation}
  \label{EqNTameInd}
  \begin{split}
    \|\tilde u\|_{H_{(\seop;\sop),\eps}^{(\sfs;k'-1),\alpha_\circ,\hat\alpha}(\Omega_\eps^{(1-\frac{k'-1-\eta}{k})})^{\bullet,-}} &\leq C_{k'-1}\biggl(\|\tilde f\|_{H_{(\seop;\sop),\eps}^{(\sfs;k'-1),\alpha_\circ,\hat\alpha-2}(\Omega^{(1-\frac{k'-1-\eta}{k})}_\eps)^{\bullet,-}} \\
      &\quad \hspace{4em} + \|\wt L'_1\|_{\hat\rho^{-2}\cC_{(\seop;\sop),\eps}^{(d_0;k'-1),1,1}(\Omega^{(1)}_\eps)} \|\tilde f\|_{H_{(\seop;\sop),\eps}^{(\sfs;d),\alpha_\circ,\hat\alpha-2}(\Omega^{(1)}_\eps)^{\bullet,-}}\biggr)
  \end{split}
  \end{equation}
  for $\eta=0,1$. This is true for all $k'$ below any fixed, i.e.\ $k$-independent, finite value $d$. Note that the domain on which $\tilde f$ is estimated can be replaced by $\Omega^{(1)}_\eps$ or $\Omega^{(0)}_\eps$ upon adjusting the constant $C_{k'-1}$, due to the fact that $\tilde f$ is controlled by $f$ via~\eqref{EqNTameExtf}.

  Consider now $k'>d$ with $k'\leq k$. We have
  \[
    \wt L'(\pa_t^{k'}\tilde u)=\pa_t^{k'}\tilde f+\sum_{j=1}^{k'} c_{k',j}({\rm ad}_{\pa_t}^j\wt L')\pa_t^{k'-j}\tilde u
  \]
  for some combinatorial constants $c_{k,j}$; the basic se-estimate for $\wt L'$ with se-regularity order $\sfs-2$ and s-regularity order $0$ thus gives
  \begin{align*}
    &\|\pa_t^{k'}\tilde u\|_{H_{\seop,\eps}^{\sfs-2,\alpha_\circ,\hat\alpha}(\Omega^{(1-\frac{k'-1}{k})}_\eps)^{\bullet,-}} \\
    &\qquad \leq C\Bigl( \|\pa_t^{k'}\tilde f\|_{H_{\seop,\eps}^{\sfs,\alpha_\circ,\hat\alpha-2}(\Omega^{(1-\frac{k'-1}{k})}_\eps)^{\bullet,-}} + \|\tilde u\|_{H_{(\seop;\sop),\eps}^{(\sfs;k'-1),\alpha_\circ,\hat\alpha}(\Omega^{(1-\frac{k'-2}{k})})^{\bullet,-}} \\
    &\qquad\quad\hspace{10em} + \|\wt L'_1\|_{\hat\rho^{-2}\cC_{(\seop;\sop),\eps}^{(d_0;k),1,1}(\Omega_\eps^{(1)})}\|\tilde u\|_{H_{(\seop;\sop),\eps}^{(\sfs;d),\alpha_\circ,\hat\alpha}(\Omega^{(1-\frac{k'-2}{k})})^{\bullet,-}}\Bigr).
  \end{align*}
  Here we used Lemma~\ref{LemmaNTameMult} (with $k,j$ in the Lemma being equal to $k'-1$, $j-1$ in present notation) to produce a tame estimate for $\sum_{j=1}^{k'}|c_{k',j}| \|({\rm ad}_{\pa_t}^j\wt L')\pa_t^{k'-j}u\|_{H_{\seop,\eps}^{\sfs,\alpha_\circ,\hat\alpha-2}(\Omega^{(1-\frac{k'-1}{k})}_\eps)^{\bullet,-}}$ using that ${\rm ad}_{\pa_t}^j\wt L'\in\hat\rho^{-2}(\CI+\cC_{\seop;\sop}^{(d_0;k-j)})\Diffse^2$ when $\wt L'_1$ has coefficients of regularity $\hat\rho^{-2}\cC_{\seop;\sop}^{(d_0;k)}$. The extendible nature of the variable order se-norm here is the reason for enlarging the domain in the final two norms on the right.

  Observe then that
  \[
    \|\tilde u\|_{H_{(\seop;\sop),\eps}^{(\sfs-2;k'),\alpha_\circ,\hat\alpha}} \sim \|\pa_t^{k'}\tilde u\|_{H_{\seop,\eps}^{\sfs-2,\alpha_\circ,\hat\alpha}} + \|\tilde u\|_{H_{(\seop;\sop),\eps}^{(\sfs-1;k'-1),\alpha_\circ,\hat\alpha}}.
  \]
  Therefore, we can use the inductive hypothesis~\eqref{EqNTameInd} with $\eta=1$ to deduce a tame estimate for $\|\tilde u\|_{H_{(\seop;\sop),\eps}^{(\sfs-2;k'),\alpha_\circ,\hat\alpha}(\Omega^{(1-\frac{k'-1}{k})}_\eps)^{\bullet,-}}$---namely, this is bounded by the right hand side of~\eqref{EqNTameInd} for $k'$ in place of $k'-1$, and with $k'-2$ in place of $k'-1-\eta$.

  Finally, we upgrade the se-regularity of $\tilde u$ from $\sfs-2$ to $\sfs$ by using the tame se-microlocal results---Lemmas~\ref{LemmaNTameEll} and \ref{LemmaNTameProp} as well as Propositions~\ref{PropNTameRout} (and its variants for the incoming and horizon radial sets) and \ref{PropNTameTrap}---in the same fashion as in the proof of Theorem~\ref{ThmEstStd}. This gives tame control on $\tilde u$ in $H_{(\seop;\sop),\eps}^{(\sfs;k'),\alpha_\circ,\hat\alpha}$ on any slightly smaller region than where we have already established tame $H_{(\seop;\sop),\eps}^{(\sfs-2;k'),\alpha_\circ,\hat\alpha}$-control. Choosing as this smaller region the set $\Omega^{(1-\frac{k'}{k})}_\eps$ finishes the proof of~\eqref{EqNTameInd} for $k'$ in place of $k'-1$, and with $\eta=0$; keep in mind here that the domain on which $\tilde f$ is estimated on the right of~\eqref{EqNTameInd} can be any domain between $\Omega^{(0)}_\eps$ and $\Omega^{(1)}_\eps$.\footnote{We stress that this step, however innocent-looking, is the backbone of our tame theory: without it, we would lose 2 orders of se-regularity for every gain of 1 order of s-regularity.} The case $\eta=1$ can be proved in exactly the same manner by working throughout with slightly larger domains.

  The estimate~\eqref{EqNTameInd}, with $k$ in place of $k'-1$, finishes the proof of the desired tame estimate using the inclusions~\eqref{EqNTamePassages} and the bounds~\eqref{EqNTameExtL} and \eqref{EqNTameExtf} similarly to the proof of~\eqref{EqNTameEstNontame}.
\end{proof}

\subsubsection{Nash--Moser theorem; smoothing operators}
\label{SssNTameNM}

To solve nonlinear equations, the tame estimates from Theorem~\ref{ThmNTame} provide the key estimates allowing for an application of a Nash--Moser iteration scheme. The following is a variant of the main result of \cite{SaintRaymondNashMoser} which essentially already featured in \cite{HintzVasyQuasilinearKdS}.

\begin{thm}[Nash--Moser]
\label{ThmNTameNM}
  Let $(B^s,|\cdot|_s)$ and $(\bfB^s,\|\cdot\|_s)$ be Banach spaces for $s\in\N_0$. Suppose that $B^s\subset B^t$ with $|\cdot|_t\leq|\cdot|_s$ for $s\geq t$, and set $B^\infty=\bigcap_{s=0}^\infty B_s$. For $\eta\in[0,1]$, let $B^s_\eta\subset B^s$ be a linear subspace (with the induced norm), with $B^s_\eta\subseteq B^s_{\eta'}$ whenever $\eta\leq\eta'$. We make the analogous definitions for, and assumptions on, $\bfB^s$. Suppose $\Phi\colon B^\infty\to\bfB^\infty$ is a $\cC^2$ map, defined for all $u\in B^\infty$ with $|u|_{3 d}<\delta$, which satisfies $\Phi(B^\infty_\eta)\subset\bfB^\infty_\eta$ for all $\eta\in[0,1]$. Suppose there exist constants $d\in\N$, $\delta>0$ and $C_i$, $i\in\N$, so that for all $u,v,w\in B^\infty$ with $|u|_{3 d}<\delta$, we have
  \begin{subequations}
  \begin{gather}
  \label{EqNTameNMPhi}
    \|\Phi(u)\|_s \leq C_s(1+|u|_{s+d})\quad\forall\,s\geq d, \\
  \label{EqNTameNMPhip}
    \|\Phi'(u)v\|_{2 d}\leq C_1|v|_{3 d},\qquad
    \|\Phi''(u)(v,w)\|_{2 d}\leq C_2|v|_{3 d}|w|_{3 d}.
  \end{gather}
  \end{subequations}
  Suppose moreover that for every $u\in B^\infty$ with $|u|_{3 d}<\delta$, there exist a linear operator $\Psi(u)\colon\bfB^\infty\to B^\infty$ mapping $\bfB^\infty_\eta\to B^\infty_\eta$ for all $\eta\in[0,1]$ and satisfying $\Phi'(u)\Psi(u)f=f$ for all $f\in\bfB^\infty$ together with a tame estimate
  \begin{equation}
  \label{EqNTameNMPsi}
    |\Psi(u)f|_s \leq C_s(\|f\|_{s+d}+|u|_{s+d}\|f\|_{2 d})\quad\forall\,s\geq d,\ f\in\bfB^\infty.
  \end{equation}
  Suppose there exist operators $S_\theta$, $\theta>1$, mapping $B_{\theta^{-1/2}}^\infty\to B_0^\infty$ and indeed $B_\eta^\infty\to B_{\eta-\theta^{-1/2}}^\infty$ which for all $s,t\geq 0$ and $v\in B_{\theta^{-1/2}}^\infty$ satisfy
  \begin{equation}
  \label{EqNTameNMSmoothing}
    |S_\theta v|_s \leq C_{s,t}\theta^{s-t}|v|_t\quad\text{for}\ s\geq t, \qquad
    |v-S_\theta v|_s \leq C_{s,t}\theta^{s-t}|v|_t\quad\text{for}\ s\leq t.
  \end{equation}
  Then if $\|\Phi(0)\|_{2 d}$ is sufficiently small depending on $\delta$ and the constants $C_s$, $C_{s,t}$ for $s,t\leq 16 d^2+43 d+24$, there exists $u\in B^\infty$ with $|u|_{3 d}<\delta$ and $\Phi(u)=0$.
\end{thm}
\begin{proof}
  This follows by repeating the arguments of \cite{SaintRaymondNashMoser}; we only need to address the filtrations by $\eta\in[0,1]$ (and the fact that $S_\theta$ is only defined on $B^\infty_\eta$ for $\eta\geq\theta^{-1/2}$), the estimates themselves being otherwise unaffected. To start, since $u_0:=0\in B^\infty_1$, we have $\Phi(u_0)\in\bfB^\infty_1$ and therefore (using the notation of \cite[Lemma~1]{SaintRaymondNashMoser}) $v_0:=-\Psi(u)\Phi(0)\in B^\infty_1$; hence $u_1:=u_0+S_{\theta_0}v_0\in B^\infty_{1-\theta_0^{-1/2}}$. Now, the parameters $\theta_j$ in the proof of \cite[Lemma~1]{SaintRaymondNashMoser} are fixed by $\theta_j=\theta_0^{(5/4)^j}$, with $\theta_0$ large. It then follows that $u_J\in B^\infty_{\eta_J}$ where $\eta_J=1-\sum_{j=0}^{J-1}\theta_j^{-1/2}$. For sufficiently large $\theta_0$, we have $\eta_J>0$ for all $J$, and therefore all iterates $u_J$ are well-defined. Their limit $u$ in $B^{3 d+3}$ then solves $u\in B^\infty$ and $\Phi(u)=0$ as in \cite{SaintRaymondNashMoser}.
\end{proof}

We motivate the presence of the filtration $B^s_\eta$, $\eta\in[0,1]$, as follows: if $B^s$ is a space of $H^s$-functions with support in $t\geq 0$, then one could take $B^s_\eta$ to be the subspace of elements with support in $t\geq\eta$. If $\Phi$ is a nonlinear wave operator and $\Psi(u)$ is its forward solution operator, then $\Phi$ and $\Psi(u)$ preserve supports in $t\geq\eta$. On the other hand, standard smoothing operators do not respect supports. We now show, following \cite[Lemma~5.9]{HintzVasyQuasilinearKdS}, how to construct smoothing operators which enlarge supports only by an amount $\theta^{-\delta}$, $\delta\in(0,1)$ (so $\delta=\frac12$ is a possible choice).

\begin{lemma}[Smoothing operators]
\label{LemmaNTameSmooth}
  Let $\delta\in(0,1)$. Then there exist linear operators $S_\theta\colon\sD'(\R^n)\to\CI(\R^n)$, $\theta>1$, with the following properties:
  \begin{enumerate}
  \item\label{ItNTameSmoothSupp0} for all $u\in\sD'(\R^n)$, $p\in\supp u$, and $q\in\supp(S_\theta u)$, we have $|p-q|<\theta^{-\delta}$;
  \item\label{ItNTameSmoothEst} for all $s,t\geq 0$, we have estimates $\|S_\theta v\|_{H^s}\leq C_{s,t}\theta^{s-t}\|v\|_{H^t}$, $s\geq t$, and $\|v-S_\theta v\|_{H^s}\leq C_{s,t}\theta^{s-t}\|v\|_{H^t}$, $s\leq t$.
  \end{enumerate}
\end{lemma}
\begin{proof}
  Fix $\phi\in\CIc(\R^n)$ so that $\phi(\xi)=1$ for $|\xi|<1$. Let $\chi=\cF^{-1}\phi\in\sS(\R^n)$ and set
  \[
    (S_\theta v)(x) = \int_{\R^n} \phi(\theta^\delta y) \theta^n\chi(\theta y) v(x-y)\,\dd y,\qquad
    (S_\theta^0 v)(x) = \int_{\R^n} \theta^n\chi(\theta y)v(x-y)\,\dd y.
  \]
  The estimates of part~\eqref{ItNTameSmoothEst} hold for $S_\theta^0$ since $\wh{S_\theta^0 v}(\xi)=\phi(\frac{\xi}{\theta})\hat v(\xi)$, so $\la\xi\ra^{s-t}\phi(\frac{\xi}{\theta})\leq C_{s,t}\theta^{s-t}$ for $s\geq t$ and $\la\xi\ra^{s-t}(1-\phi(\frac{\xi}{\theta}))\leq C_{s,t}\theta^{s-t}$ for $s\leq t$. Consider then
  \[
    ((S_\theta-S_\theta^0)v)(x) = \int_{\R^n} (1-\phi(\theta^\delta y))\theta^n\chi(\theta y)v(x-y)\,\dd y.
  \]
  Since on $\supp(1-\phi(\theta^\delta y))$ we have $|y|\geq\theta^{-\delta}$, we can estimate
  \begin{align*}
    \int |1-\phi(\theta^\delta y)| \theta^n|\chi(\theta y)|\,\dd y &\leq \theta^n\int_{|\theta y|\geq\theta^{1-\delta}} C_N(1+|\theta y|)^{-N}\,\dd y \\
      &= C'_N \int_{r>\theta^{1-\delta}} r^{-N}\,r^{n-1}\dd r \leq C''_N\theta^{-(N-n)(1-\delta)}
  \end{align*}
  for any $N$. Since $\delta<1$, this implies $\|(S_\theta-S_\theta^0)v\|_{L^2}\leq C'''_N\theta^{-N}\|v\|_{L^2}$ for all $N$. To estimate derivatives of $(S_\theta-S_\theta^0)v$, note that any finite number of $y$-derivatives of $(1-\phi(\theta^\delta y))\theta^n\chi(\theta y)$ has $L^1$-norm bounded by $\theta^{-N}$ for any $N$ by a similar estimate. Therefore, we in fact have $\|(S_\theta-S_\theta^0)v\|_{H^s}\leq C_N\theta^{-N}\|v\|_{L^2}$ for all $N$; and thus $S_\theta$ satisfies the desired estimates.
\end{proof}

\begin{lemma}[Smoothing operators on extendible spaces]
\label{LemmaNTameSmoothExt}
  There exist linear operators $S_\theta\colon L^2(\R^n_+)\to\bar H^\infty(\R^n_+)$, $\theta>1$, satisfying part~\eqref{ItNTameSmoothSupp0} of Lemma~\usref{LemmaNTameSmooth} as well as the estimates~\eqref{ItNTameSmoothEst}, where now $\|\cdot\|_{H^s}$ denotes the quotient norm on $\bar H^s(\R^n_+)$.
\end{lemma}
\begin{proof}
  Fix $\chi\in\CIc([0,\infty))$ to be equal to $1$ near $0$. For $u\in L^2(\R^n_+)$, set $(E u)(x,y)=u(x,y)$ when $x>0$ and $(E u)(x,y)=\sum_{l=0}^\infty c_l\cdot\chi(-\frac{x}{\delta_l})u(-\frac{x}{\delta_l},y)$ where $c_l\in\R$ and $\delta_l>0$ with $\delta_l\searrow 0$ are chosen as in the proof of Lemma~\ref{LemmaNTameExt}. Then $E\colon\bar H^s(\R^n_+)\to H^s(\R^n)$ is well-defined and continuous for $s\in\N_0$ and thus for all $s\geq 0$ by interpolation. Denote by $R\colon H^s(\R^n)\to\bar H^s(\R^n_+)$ the restriction operator. We can then set
  \[
    S_\theta v := R \tilde S_\theta E v,\qquad v\in L^2(\R^n_+),
  \]
  where $\tilde S_\theta\colon H^0(\R^n)\to H^\infty(\R^n)$ is any family of linear operators satisfying part~\eqref{ItNTameSmoothSupp0} of Lemma~\usref{LemmaNTameSmooth} and the estimates~\eqref{ItNTameSmoothEst} (e.g.\ the operators constructed in Lemma~\ref{LemmaNTameSmooth}). The bound $\|S_\theta v\|_{\bar H^s}\leq C_{s,t}\theta^{s-t}\|v\|_{\bar H^t}$ for $s\geq t$ follows from the corresponding bounds for $\tilde S_\theta$ and the continuity of $R,E$. Since $R E=I$ on $\bar H^s(\R^n_+)$ for all $s\geq 0$, the bound on $v-S_\theta v=R(E v-\tilde S_\theta E v)$ follows from the corresponding bound on $v'-\tilde S_\theta v'$ with $v'=E v$.
\end{proof}

Combining the methods of proof of Lemmas~\ref{LemmaNTameSmooth}--\ref{LemmaNTameSmoothExt}, we now obtain:

\begin{cor}[Smoothing operators on manifolds with corners]
\label{CorNTameSmooth}
  Let $\Omega$ be a manifold with corners, with boundary defining functions $\rho_1,\ldots,\rho_N\in\CI(\Omega)$. Fix $j\in\{0,\ldots,N\}$, and set $\Omega_\eta:=\Omega\cap\{\rho_i\geq\eta,\ i=1,\ldots,j\}$ for $\eta\in[0,1]$. Assume that the differentials of the $\dd\rho_i$, $i\leq j$, are linearly independent on their $\eta$-level sets for $\eta\in[0,1]$ so that $\Omega_\eta$ is a manifold with corners.\footnote{This can be arranged by multiplying $\rho_1,\ldots,\rho_j$ by sufficiently large constants.} Write $H^s(\Omega_\eta)^{\bullet,-}$ for the space of $H^s$-functions with supported character at $\rho_i=\eta$, $i\leq j$, and extendible character at $\rho_i=0$, $i>j$. Then there exist linear operators $S_\theta$, $\theta>1$, mapping $L^2(\Omega_\eta)^{\bullet,-}\to H^\infty(\Omega_{\eta-\theta^{-1/2}})^{\bullet,-}$ satisfying for all $v\in H^t(\Omega_{\theta^{-1/2}})^{\bullet,-}$ the estimates
  \begin{alignat*}{2}
    \|S_\theta v\|_{H^s(\Omega)^{\bullet,-}} &\leq C_{s,t}\theta^{s-t}\|v\|_{H^t(\Omega)^{\bullet,-}}\quad\text{for}\ s\;&\geq t, \\
    \|v-S_\theta v\|_{H^s(\Omega)^{\bullet,-}} &\leq C_{s,t}\theta^{s-t}\|v\|_{H^t(\Omega)^{\bullet,-}}\quad\text{for}\ s\;&\leq t.
  \end{alignat*}
\end{cor}
\begin{proof}
  Using a partition of unity, it suffices to construct smoothing operators on the local models $[0,\infty)_x^k\times\R_y^{n-k}$, with supported character at $x^i=\eta$ for $i\leq j$ and extendible character at $x^i=0$ for $i>j$. Using extension operators across each of the latter hypersurfaces, the task is reduced to the construction of a smoothing operator on $[0,\infty)_x^j\times\R_z^{n-j}$ acting on functions with support in $x^i>\eta$, $i=1,\ldots,j$, which was done in Lemma~\ref{LemmaNTameSmooth}.
\end{proof}

In the setting of interest in this paper, we similarly have:

\begin{cor}[Smoothing operators on s-Sobolev spaces]
\label{CorNTameSmooths}
  Consider $M'=\R_t\times\R^3_x$, $\wt M'=[[0,1)\times M';\{0\}\times\R\times\{0\}]$. Let $t_0<t_1$ and $r_0>0$. Denote by $\wt\Omega_\eta\subset\wt M'$, $\eta\in[0,1]$, the lift of $[0,1)_\eps\times\Omega_{t_0+c\eta,t_1,r_0}$ where $c<t_1-t_0$, and write $(\Omega_\eta)_\eps=\wt\Omega_\eps\cap M_\eps$. Then there exist linear operators $S_\theta$, $\theta>1$, mapping $H_{\sop,\eps}^{0,\alpha_\circ,\hat\alpha}((\Omega_\eta)_\eps)^{\bullet,-}\to H_{\sop,\eps}^{\infty,\alpha_\circ,\hat\alpha}((\Omega_{\eta-\theta^{-1/2}})_\eps)^{\bullet,-}$ for $\eta\in[\theta^{-1/2},1]$, so that for all integer $s,t\geq 0$, there exist constants $C_{s,t}$ so that the estimates
  \begin{alignat*}{2}
    \|S_\theta v\|_{H_{\sop,\eps}^s(\Omega_\eps)^{\bullet,-}} &\leq C_{s,t}\theta^{s-t}\|v\|_{H_{\sop,\eps}^t(\Omega_\eps)^{\bullet,-}}\quad\text{for}\ s\;&\geq t, \\
    \|v-S_\theta v\|_{H_{\sop,\eps}^s(\Omega_\eps)^{\bullet,-}} &\leq C_{s,t}\theta^{s-t}\|v\|_{H_{\sop,\eps}^t(\Omega_\eps)^{\bullet,-}}\quad\text{for}\ s\;&\leq t,
  \end{alignat*}
  with $v\in H_{\sop,\eps}^t((\Omega_{\theta^{-1/2}})_\eps)^{\bullet,-}$, hold uniformly for $\eps\in(0,1)$.
\end{cor}

We only restrict to integer orders here since we have not previously defined $H_{\sop,\eps}^s$ for non-integer $s$ (although this can easily be done, as they are equal to uniform Sobolev spaces for the b.g.\ structure~\eqref{EqFVarse1}--\eqref{EqFVarse2}).

\subsection{Solution of a nonlinear toy problem}
\label{SsNToy}

In this section, we shall study the following semilinear wave equation.

\begin{thm}[Nonlinear toy problem]
\label{ThmNToy}
  Let $(M,g)$ be a $4$-dimensional globally hyperbolic spacetime, and suppose $u\in\CI(M)$ solves
  \[
    P(u) := \Box_g u - u^2 = 0.
  \]
  Let $(\wt M,\wt g)$ be a glued spacetime associated with $M$, an inextendible timelike geodesic $\cC\subset M$, and subextremal Kerr parameters $\bhm,\bha$, and write $t\in I_\cC$, $x\in\R^3$ for Fermi normal coordinates around $\cC$; cf.\ Definition~\usref{DefGl}. Let $X$ be a Cauchy hypersurface in $M$ which near its intersection point $\fp$ with $\cC$ is equal to $\{t=0\}$, and let $\wt X\subset\wt M\setminus\wt K^\circ$ denote the lift of $[0,1)\times X$. Write $P_\eps(u)=\Box_{g_\eps}u-u^2$ and $\wt P(\wt u)=(P_\eps(\wt u|_{M_\eps}))_{\eps\in(0,1)}$.
  \begin{enumerate}
  \item\label{ItNToyF}{\rm (Formal solution.)} There exists $\wt u\in\cA_\phg^{\cE_\circ,\hat\cE}(\wt M\setminus\wt K^\circ)$, with index sets\footnote{Recall that $(1,*)$ is an index set contained in $(1+\N_0)\times\N_0\subset\C\times\N_0$. Thus, $\wt u$ is log-smooth on $\wt M$.} $\cE_\circ=(0,0)\cup(1,*)$ and $\hat\cE=(0,0)\cup(1,*)$, so that $\wt u|_{M_\circ}=\upbeta_\circ^*u$, further $\wt u|_{\hat M_t}=u(t,0)$, and $\wt P(\wt u)\in\CIdot(\wt M\setminus\wt K^\circ)$ vanishes to infinite order at $M_\circ\cup\hat M$ and also at $\wt X$.
  \item\label{ItNToyT}{\rm (Correction to a true solution.)} Let $\wt u$ be a formal solution as in part~\eqref{ItNToyF}. Let $K\subset M$ be compact. Then there exist $\eps_0>0$ and $\wt v\in\CIdot(\wt M\setminus\wt K^\circ)$ so that $\wt P(\wt u+\wt v)=0$ on $K$ for all $\eps\in(0,\eps_0]$.
  \end{enumerate}
\end{thm}

We will prove part~\eqref{ItNToyF} in~\S\ref{SssNToyF} and part~\eqref{ItNToyT} in~\S\ref{SssNToyT}.

Since any simple variant of the problem studied in Theorem~\ref{ThmNToy} is just as artificial, we do not aim for any sort of general statement here, and instead focus on a simple concrete example to illustrate how to operate the machinery developed in this paper. We thus leave it to the interested reader to study generalizations to equations such as $\Box_{g_\eps}u_\eps=|\dd u_\eps|_{g_\eps}^2$, or to quasilinear wave equations (with lower order terms). One could also add to $\Box_{g_\eps}$ a potential $\eps^{-2}V(\frac{x}{\eps})$ under the assumption of mode stability for $\Box_{\hat g_b}+V$ (which is in general a rather delicate problem \cite{MoschidisSuperradiant}); and one can also consider geometrically simpler settings as discussed in Remark~\ref{RmkGlOther}.

Note that Theorem~\ref{ThmNToy} is a gluing, or singular perturbation, problem: what is glued in at the fiber $\hat M_t=\ol{\R^3_{\hat x}}$ of the front face $\hat M$ of $\wt M$ is merely the constant function $u(t,0)$, which is a stationary solution of the relevant zero energy problem, $\wh{\Box_{\hat g_b}}(0)u(t,0)=0$ (with $t$ a parameter). See \cite[\S{1.2.1}]{HintzGlueLocI} for a less trivial, albeit somewhat more contrived and in any case linear, setup.

\subsubsection{Formal solution}
\label{SssNToyF}

Using the perspective introduced in \cite{HintzGlueLocI}, we sketch the proof of Theorem~\ref{ThmNToy}\eqref{ItNToyF} via the iterative solution of linear model problems at $M_\circ$ and $\hat M$. To wit, let $\wt u_0\in\CI(\wt M\setminus\wt K^\circ)$ be any function with $\wt u_0|_{M_\circ}=\upbeta_\circ^*u$ and $\wt u_0|_{\hat M_t}=u(t,0)$. Then
\begin{equation}
\label{EqNToyF0}
  \Err_0 := \wt P(\wt u_0) \in \rho_\circ\hat\rho^{-1}\CI(\wt M\setminus\wt K^\circ) = \cA_\phg^{(1,0),(-1,0)}.
\end{equation}
Indeed, since $\Box_{\wt g}\in\hat\rho^{-2}\Diffse^2$, we have $\Box_{\wt g}\wt u_0\in\hat\rho^{-2}\CI$ with leading order behavior at $\hat M_t$ given by $\eps^{-2}\Box_{\hat g_b}(\wt u_0|_{\hat M_t})=0$, so in fact $\Box_{\wt g}\wt u_0\in\hat\rho^{-1}\CI$; and $\wt u_0^2\in\CI$, so $\wt P(\wt u_0)\in\hat\rho^{-1}\CI$, with restriction to $M_\circ$ given by $\Box_g u-u^2=0$, giving~\eqref{EqNToyF0}.

Write $L_\circ v:=D_u P(v)=\Box_g v-2 u v$ and $L:=\Box_{\hat g_b}$ for the normal operators of the linearization of $\wt P$ around $\wt u_0$ (or any other element of $\cA_\phg^{\cE_\circ,\hat\cE}$ with $\cE_\circ,\hat\cE=(0,0)\cup(1,*)$ having the same restrictions to $M_\circ$ and $\hat M$ as $\wt u_0$).

Consider now $f_0:=(\eps^{-1}\Err_0)|_{M_\circ}\in r^{-2}\CI(M_\circ)$. As in \cite[\S{1.2.1}, Step 2]{HintzGlueLocI}, one can construct $h\in\cA_\phg^{(0,*)}(M_\circ)$ with $L_\circ h=f_0$ by first constructing a formal solution at $r=0$ and then solving away the remaining error (which lies in $\CIdot(M_\circ)\subset\upbeta_\circ^*\CI(M)$) using an initial value problem for the wave operator $L_\circ$. We then set $\wt u^1:=\wt u_0+\eps h$ (where we denote the $\eps$-independent extension of $h$ by the same symbol). Thus, for $\wt u^1\in\cA_\phg^{(0,0),(0,0)\cup(1,*)}(\wt M\setminus\wt K^\circ)$, we have $\Err^1:=\wt P(\wt u^1)\in\cA_\phg^{(2,0),(-1,*)}$. (This improves the order of vanishing at $M_\circ$ and does not make that at $\hat M$ worse except for the possibility of logarithmic factors.)

Suppose for simplicity that $\Err^1\in\cA_\phg^{(2,0),(-1,0)\cup(0,*)}$. (The general case is then dealt with similarly.) Consider then $f^1:=(\eps\Err^1)|_{\hat M}$, which is a smooth function of $t\in I_\cC$ with values in $\cA^{(3,0)}(\hat X_b)$. Using the invertibility of $\hat L(0)$ (see Lemma~\ref{LemmaSc3bSpec}) and standard normal operator arguments to extract asymptotic expansions at $\hat r=\infty$ (namely, using that $\hat L(0)-\Delta_{\hat x}\in\la\hat x\ra^{-3}\Diffb^2(\hat X_b)$, with the boundary spectrum of $\Delta_{\hat x}$ being the integers), one can thus find $h\in\CI(\R_t;\cA^{(1,*)}(\hat X_b))$ so that $\hat L(0)h(t,\cdot)=f^1(t,\cdot)$ for all $t\in I_\cC$. With $\hat\chi\in\CI(\wt M)$ denoting a cutoff, equal to $1$ and supported in a collar neighborhood of $\hat M$ where the Fermi normal coordinates are valid, we then regard $\hat\chi h$ as a function $(\eps,t,\hat x)\mapsto \hat\chi(\eps,t,\hat x)h(t,\hat x)$ which thus lies in $\cA^{(1,*),(0,0)}(\wt M\setminus\wt K^\circ)$. Setting $\wt u_1:=\wt u^1+\eps\chi h\in\cA^{(0,0)\cup(2,*),(0,0)\cup(1,*)}$, we then have $\Err_1:=\wt P(\wt u_1)\in\cA^{(2,*),(0,*)}$. (This improves the order of vanishing at $\hat M$ and makes that at $M_\circ$ worse by a power of $\log\rho_\circ$ only.)

One then proceeds in this fashion (with simple modifications to deal with logarithmic terms), thus constructing $\wt u^{k+1}=\wt u_k+\eps^{k+1}h$ with $h\in\cA^{(0,*),(0,*)}(\wt M\setminus\wt K^\circ)$ and $\wt u_{k+1}=\wt u^{k+1}+\eps^{k+1}h$ with $h\in\cA^{(1,*),(0,*)}(\wt M\setminus\wt K^\circ)$ so that, for all $k$, we have $\wt P(\wt u^k)\in\cA^{(k+1,*),(k-2,*)}$ and $\wt P(\wt u_k)\in\cA^{(k+1,*),(k-1,*)}$. Defining $\wt u_\infty$ to be an asymptotic sum $\wt u_\infty\sim\wt u_0+\sum_{k=0}^\infty((\wt u^{k+1}-\wt u_k) + (\wt u_{k+1}-\wt u^{k+1}))$ of $\wt u_0$ and all correction terms, the remaining error $\wt P(\wt u_\infty)$ is then polyhomogeneous and vanishes to all orders at $M_\circ$ and $\hat M$, so $\wt P(\wt u_\infty)\in\CIdot(\wt M\setminus\wt K^\circ)$ indeed.

Finally, we correct $\wt u_\infty$ in Taylor series at $\wt X$ as follows: fix $\ft\in\CI(\wt M)$ with $\wt X=\ft^{-1}(0)$; thus $\dd\ft$ is past timelike for the metric $g_\eps$ near $\wt X$ for all sufficiently small $\eps>0$. Then $\wt P(\wt u_\infty+\ft^2\wt h)\equiv\wt P(\wt u_\infty)+[\Box_{\wt g},\ft^2]\wt h$ modulo terms vanishing at $\wt X$. But the restriction to $\ft=0$ of $[\Box_{\wt g},\ft^2]\wt h$ is a smooth nonvanishing multiple of $\wt h$; therefore, we can choose $\wt h\in\CIdot(\wt M\setminus\wt K^\circ)$ so that $\wt P(\wt u_\infty+\ft^2\wt h)$ vanishes at $\ft=0$. One then constructs further corrections $\ft^j\wt h$, $j\geq 3$, in an iterative fashion. The sought-after formal solution $\wt u$ is then the asymptotic sum (in the sense of Taylor series at $M_\circ\cup\hat M\cup\{\ft=0\}$) of $\wt u_\infty$ and all these corrections.

\subsubsection{True solution}
\label{SssNToyT}

We now turn to the proof of part~\eqref{ItNToyT} of Theorem~\ref{ThmNToy}. We shall prove the existence of the correction term $\wt v=(v_\eps)$ for sufficiently small $\eps>0$ in weighted s-Sobolev spaces, as a solution of a modification of the equation
\[
  \breve\Phi_\eps(v_\eps) = 0,\qquad \breve\Phi_\eps(v_\eps) := P_\eps(u_\eps + v_\eps) = (\Box_{g_\eps}u_\eps-u_\eps^2)+\bigl((\Box_{g_\eps}-2 u_\eps)v_\eps-v_\eps^2\bigr),\quad u_\eps:=\wt u|_{M_\eps}.
\]
Let $f_\eps:=P_\eps(u_\eps)$, then $\wt f=(f_\eps)_{\eps\in(0,1)}$ lies in $\CIdot(\wt M\setminus\wt K^\circ)$ and vanishes to infinite order at $\wt X$.

\medskip

\textbf{Tame estimates on small domains.} Fix some large $d\in\N_0$. Consider a standard domain $\wt\Omega\subset\wt M$ associated with $\Omega_{-\lambda,\lambda,\lambda}$, with $\lambda>0$ small and fixed below. Fix quantities $\alpha_\cD\in(-\frac32,-\frac12)$, $N\geq 3$, and let
\begin{equation}
\label{EqNToyTWeights}
  \alpha_\circ=\alpha_\cD+N,\qquad\hat\alpha=N.
\end{equation}
Let $v_\eps\in H_{\sop,\eps}^{\infty,\alpha_\circ,\hat\alpha}(\Omega_\eps)^{\bullet,-}=\eps^N H_{\sop,\eps}^{\infty,\alpha_\cD,0}(\Omega_\eps)^{\bullet,-}$ (the space $B^\infty$ in the notation of Theorem~\ref{ThmNTameNM}), and suppose that $\|v_\eps\|_{H_{\sop,\eps}^{d,\alpha_\circ,\hat\alpha}(\Omega_\eps)^{\bullet,-}}\leq 1$. Denote by $H=H(t)$ the Heaviside function. Then we can estimate
\begin{equation}
\label{EqNToyTPhiDef}
  \Phi_\eps(v_\eps) := H(t)f_\eps + (\Box_{g_\eps}-2 u_\eps)v_\eps - v_\eps^2
\end{equation}
in $H_{\sop,\eps}^{s-2,\alpha_\circ,\hat\alpha-2}(\Omega_\eps)^{\bullet,-}$ (the space $\bfB^s$ in the notation of Theorem~\ref{ThmNTameNM}) via
\begin{equation}
\label{EqNToyTPhi}
  \|\Phi_\eps(v_\eps)\|_{H_{\sop,\eps}^{s-2,\alpha_\circ,\hat\alpha-2}(\Omega_\eps)^{\bullet,-}} \leq C_{s,N}\eps^N + C_s\|v_\eps\|_{H_{\sop,\eps}^{s,\alpha_\circ,\hat\alpha}(\Omega_\eps)^{\bullet,-}}.
\end{equation}
Indeed, due to the infinite order vanishing of $f_\eps$ at $t=0$, the family $(H(t)f_\eps)$ lies in $\CIdot(\wt M\setminus\wt K^\circ)$ and now vanishes in $t\leq 0$; moreover, we used that~\eqref{EqNTameMultSpecial} (with weights) and Sobolev embedding (Proposition~\ref{PropFsSobEmb}) imply
\[
  \eps^{-\hat\alpha+\frac32}\|v_\eps^2\|_{H_{\sop,\eps}^{s-2,\alpha_\circ,\hat\alpha}} \leq C_s\|v_\eps^2\|_{H_{\sop,\eps}^{s-2,2\alpha_\circ,2\hat\alpha-\frac32}} \leq C_s\|v_\eps\|_{H_{\sop,\eps}^{s-2,\alpha_\circ,\hat\alpha}}\|v_\eps\|_{H_{\sop,\eps}^{d,\alpha_\circ,\hat\alpha}},\qquad d>2,
\]
in view of $\alpha_\circ+\hat\alpha-\frac32\leq 2\alpha_\circ$; thus $\|v_\eps^2\|_{H_{\sop,\eps}^{s-2,\alpha_\circ,\hat\alpha}}\leq C_s\eps^{N-\frac32}\|v_\eps\|_{H_{\sop,\eps}^{s-2,\alpha_\circ,\hat\alpha}}$. The estimate~\eqref{EqNToyTPhi} implies~\eqref{EqNTameNMPhi}. We also note that for every $\delta>0$ and every $N\geq 2$ in~\eqref{EqNToyTWeights}, there exists $\eps_{\delta,N}>0$ so that $\|\Phi_\eps(0)\|_{H_{\sop,\eps}^{s-2,\alpha_\circ,\hat\alpha-2}(\Omega_\eps)^{\bullet,-}}<\delta$ for $\eps<\eps_{\delta,N}$.

The estimates~\eqref{EqNTameNMPhip} follow in a similar fashion from the expressions
\[
  \Phi_\eps'(v_\eps)w_\eps=(\Box_{g_\eps}-2 u_\eps)w_\eps-2 v_\eps w_\eps
\]
and $\Phi_\eps''(v_\eps,w_\eps)z_\eps=-2 w_\eps z_\eps$. The tame estimate~\eqref{EqNTameNMPsi} is provided by Theorem~\ref{ThmNTame} (see also Remark~\ref{RmkNTameOther}): this provides values for $d$, $\lambda$, $\eps_0>0$ so that for all $\eps\leq\eps_0$,
\begin{align*}
  &\|w_\eps\|_{H_{\sop,\eps}^{s,\alpha_\circ,\hat\alpha}(\Omega_\eps)^{\bullet,-}} \\
  &\qquad \leq C_s\Bigl( \|\Phi_\eps'(v_\eps)w_\eps\|_{H_{\sop,\eps}^{s,\alpha_\circ,\hat\alpha-2}(\Omega_\eps)^{\bullet,-}} + \|v_\eps\|_{\cC_{\sop,\eps}^{s+d,1,-1}(\Omega_\eps)}\|\Phi_\eps'(v_\eps)w_\eps\|_{H_{\sop,\eps}^{d,\alpha_\circ,\hat\alpha-2}(\Omega_\eps)^{\bullet,-}}\Bigr),
\end{align*}
where we can further estimate $\|v_\eps\|_{\cC_{\sop,\eps}^{s+d,1,-1}}\leq C_s\|v_\eps\|_{H_{\sop,\eps}^{s+2 d,1,\frac12}}\leq C'_s\|v_\eps\|_{H_{\sop,\eps}^{s+2 d,\alpha_\circ,\hat\alpha}}$ via Sobolev embedding.

\medskip

\textbf{Nonlinear solution on small domains.} The forward solution operator $\Phi'_\eps(v_\eps)^{-1}$ preserves the property of being supported in $t\geq-\lambda+\eta\lambda$, $\eta\in[0,1]$. If we define $B_\eta^s$, resp.\ $\bfB_\eta^s$ to be the subspace of $H_{\sop,\eps}^{s,\alpha_\circ,\hat\alpha}(\Omega_\eps)^{\bullet,-}$, resp.\ $H_{\sop,\eps}^{s-2,\alpha_\circ,\hat\alpha-2}(\Omega_\eps)^{\bullet,-}$ consisting of all elements with support in $t\geq-\lambda+\eta\lambda$, all hypotheses of Theorem~\ref{ThmNTameNM} are thus verified for the map~\eqref{EqNToyTPhiDef}, with uniform constants for all $\eps\leq\eps_0$. Therefore, applying Theorem~\ref{ThmNTameNM} for all $\eps\in(0,\eps_0]$ produces
\begin{equation}
\label{EqNToyTN}
  v^+_\eps \in H_{\sop,\eps}^{\infty,\alpha_\circ,\hat\alpha}(\Omega_\eps)^{\bullet,-}
\end{equation}
so that $\Phi_\eps(v^+_\eps)=0$ for all $\eps\in(0,\eps_0]$. By finite speed of propagation, we have $t\geq 0$ on $\supp v^+_\eps$. Moreover, by Sobolev embedding, we have $\wt v^+=(v^+_\eps)_{\eps\in(0,\eps_0]}\in\cC_\sop^{\infty,\alpha_\circ,\hat\alpha-\frac32}(\wt\Omega)^{\bullet,-}\subset\eps^{N-\frac32}\cC_\sop^\infty(\wt\Omega)^{\bullet,-}$. Since the same arguments apply for larger values of $N$, uniqueness of solutions of nonlinear wave equations implies that $\wt v^+$ does not depend upon the choice of $N$. Therefore,
\begin{equation}
\label{EqNToyTinfty}
  (v^+_\eps)_{\eps\in(0,\eps_0]} \in \eps^\infty\cC_\sop^\infty(\wt\Omega)^{\bullet,-}.
\end{equation}
(Note that the smallness of $\Phi_\eps(0)$ required by Theorem~\ref{ThmNTameNM} imposes an $N$-dependent upper bound on $\eps$; but for $\eps$ bounded away from $0$, the membership~\eqref{EqNToyTinfty} is equivalent to~\eqref{EqNToyTN} for any fixed value of $N$.)

To prove regularity of $v^+_\eps$ in $\eps$, we differentiate the equation $\Phi_\eps(v^+_\eps)=0$ in $\eps$ to find
\[
  \Phi'_\eps(v^+_\eps)(\pa_\eps v^+_\eps) + (\pa_\eps\Phi_\eps)(v^+_\eps) = 0.
\]
Now $(\pa_\eps\Phi_\eps)(v^+_\eps)=H(t)\pa_\eps f_\eps + [\pa_\eps,\Box_{g_\eps}-2 u_\eps]v^+_\eps\in\eps^\infty\cC_\sop^\infty(\wt\Omega)^{\bullet,-}$, and thus Theorem~\ref{ThmScS} shows that $(\pa_\eps v^+_\eps)\in\eps^\infty\cC_\sop^\infty(\wt\Omega)^{\bullet,-}$ as well. Proceeding iteratively implies
\[
  \wt v^+=(v^+_\eps)_{\eps\in(0,\eps_0]}\in\CIdot(\wt\Omega),\qquad t\geq 0\ \text{on}\ \supp\wt v,
\]
for the solution of $\breve\Phi_\eps(v^+_\eps)=0$ on $\Omega_\eps\cap\{t\geq 0\}$ which have constructed.

\medskip

\textbf{Semiglobal existence.} We now argue as in the proof of Theorem~\ref{ThmScSG}: one iteratively solves the nonlinear equation $\wt P(\wt u+\wt v)=0$ on the domains produced by Proposition~\ref{PropGlDynCover}, with $\eta>0$ chosen so small that the tame estimates of Theorem~\ref{ThmNTame} apply to the linearization of $P_\eps$ around the formal solution $u_\eps$ on each standard domain in Proposition~\ref{PropGlDynCover}. Carefully note that Theorem~\ref{ThmNTame} applies uniformly (as far as the value of $\lambda_0$ is concerned) also for bounded perturbations of the linearizations of $P_\eps$ (measured in a norm with fixed finite regularity and decay orders). Crucially, this means that the same estimates apply, \emph{without having to shrink the domains further} (but possibly reducing the upper bound on $\eps$), also for the linearization of $P_\eps$ around $u_\eps+\chi v_\eps$ where $v_\eps$ is the correction term in previous steps of the iterative construction (which thus lies in $\CIdot$, and contributes $\CIdot$ terms to the coefficients of the linearized operator, which are therefore small in the required sense when $\eps$ is small enough), and where $\chi$ is a cutoff function as in the proof of Theorem~\ref{ThmScSG} which equals $0$ in the later part of the domain $\wt\Omega_{J'}$ into which one is currently extending the nonlinear solution, and $1$ near the initial hypersurface of $\wt\Omega_{J'}$. This completes the proof of Theorem~\ref{ThmNToy}.

\appendix
\section{Basic notions of geometric singular analysis}
\label{SB}

\textbf{Manifolds with corners.} An $n$-dimensional \emph{manifold with corners} $X$ is diffeomorphic to $[0,\infty)^k\times\R^{n-k}$ in a neighborhood of each of its points $p\in X$; if the smallest possible number $k\in\{0,\ldots,n\}$, for a given point $p$, is equal to $0$, then $p$ lies in the interior $X^\circ$. We require all boundary hypersurfaces of $X$ to be embedded submanifolds; a boundary hypersurface is the closure of a connected component of the set of $p$ for which $k=1$. This ensures that each boundary hypersurface $H\subset X$ admits a \emph{defining function}, i.e.\ a nonnegative function $\rho\in\CI(X)$ so that $H=\rho^{-1}(0)$ and $\dd\rho\neq 0$ everywhere on $X$. If $U\subset X$ is open, then a function $\rho\in\CI(U)$ is a \emph{local defining function} for $H$ if for all compact $K\subset U$ there exists a defining function $\rho'\in\CI(X)$ for $H$ so that $\rho'=\rho$ on $K$. The quotient of any two (local) defining functions of $H$ is a smooth positive function on $X$ (resp.\ on $U$). Furthermore, we write $\cM_1(X)$ for the set of all boundary hypersurfaces of $X$, and $\CIdot(X)$ for the space of smooth functions on $X$ which vanish to infinite order at $\pa X$. See \cite{MelroseDiffOnMwc} for a comprehensive treatment.

\medskip
\textbf{b- and scattering structures; radial compactification; densities.} On a manifold with corners $X$, we define the space (in fact, $\CI(X)$-module) of \emph{b-vector fields} $\Vb(X)\subset\cV(X)=\CI(X;T X)$ to consist of all smooth vector fields on $X$ which are tangent to all boundary hypersurfaces of $X$ \cite{MelroseMendozaB,MelroseAPS,GrieserBasics}. In local coordinates $x_1,\ldots,x_k\geq 0$, $y_1,\ldots,y_{n-k}\in\R$, such vector fields are linear combinations, with smooth coefficients, of the vector fields
\[
  x_1\pa_{x_1},\ \ldots,\ x_k\pa_{x_k}, \ \pa_{y_1},\ \ldots,\ \pa_{y_{n-k}}.
\]
These vector fields are a frame of the b-tangent bundle $\Tb X\to X$, which over $X^\circ$ is thus equal to the ordinary tangent bundle. Smooth positive sections $\mu_0$ of the associated density bundle $\Omegab X\to X$ are called \emph{b-densities}. Given a \emph{weight family} $w\colon\cM_1(X)\to\R$ and a collection $\rho=(\rho_H)_{H\in\cM_1(X)}$ of defining functions, we write $\rho^w:=\prod_{H\in\cM_1(X)} \rho_H^{w(H)}$, and then densities on $X^\circ$ of the form $\rho^w\mu_0$ for a b-density $\mu_0$ are called \emph{weighted b-densities} (with weight $w$).

When $X$ is a manifold with boundary and $\rho\in\CI(X)$ denotes a boundary defining function, then
\[
  \Vsc(X) := \rho\Vb(X) = \{ \rho V \colon V \in \Vb(X) \}
\]
is the space of \emph{scattering vector fields} \cite{MelroseEuclideanSpectralTheory,VasyMinicourse}. (It is independent of the choice of $\rho$.) In local coordinates $\rho\geq 0$, $y=(y_1,\ldots,y_{n-1})\in\R^{n-1}$ near a boundary point, they are linear combinations, with smooth coefficients, of the frame
\[
  \rho^2\pa_\rho,\ \rho\pa_{y_1},\ \ldots,\ \rho\pa_{y_{n-1}}
\]
of the \emph{scattering tangent bundle} $\Tsc X\to X$. The most important example is $X=\ol{\R^n}$, where we denote by $\ol{\R^n}$ the \emph{radial compactification}
\[
  \ol{\R^n} := \R^n \sqcup \Bigl( [0,\infty)_\rho \times \Sph^{n-1} \Bigr) / \sim,\qquad 0\neq x=r\omega \sim (\rho,\omega)=(r^{-1},\omega);
\]
here $r=|x|$ (Euclidean norm), $\omega=\frac{x}{|x|}$. In the region where the first component $x_1$ is relatively large, i.e.\ $x_1>c\max_{j\neq 1}|x_j|$ for some $c>0$, the functions
\[
  \rho_1 := \frac{1}{x_1},\qquad
  \hat x_j := \frac{x_j}{x_1}\ \ (j\neq 1),
\]
extend by continuity from $\R^n$ to smooth coordinates on an open subset of $\ol{\R^n}$; they are called \emph{projective coordinates}. Since $\pa_{x_1}=-\rho_1^2\pa_{\rho_1}-\sum_{j\neq 1}\hat x_j\pa_{\hat x_j}$ and $\pa_{x_j}=\rho\pa_{\hat x_j}$, we see that scattering vector fields are linear combinations, with $\CI(\ol{\R^n})$ coefficients, of the coordinate vector fields $\pa_{x_1},\ldots,\pa_{x_n}$. The Euclidean density $|\dd x_1\cdots\dd x_n|=r^{n-1}|\dd r\,\dd g_{\Sph^{n-1}}|=\rho^{-n}|\frac{\dd\rho}{\rho}\,\dd g_{\Sph^{n-1}}|$ is a weighted b-density with weight $-n$.

Since invertible linear maps on $\R^n$ lift to diffeomorphisms of $\ol{\R^n}$, the radial compactification of a finite-dimensional real vector space is well-defined. We can therefore define the fiber-radial compactification $\bar E\to X$ of a real rank $k$ vector bundle $E\to X$ to be the fiber bundle with fibers $\bar E_p=\ol{E_p}$, $p\in X$.

\medskip
\textbf{Blow-ups \cite{MelroseDiffOnMwc}.} A \emph{p-submanifold} $Y\subset X$ of a manifold with corners $X$ is an embedded submanifold so that for all $p\in Y$ there exist coordinates on $X$ for which $Y$ is given near $p$ as the zero set of a subset of these coordinates; we call such coordinates \emph{adapted}. (The `p' stands for `product'.) If $Y$ is contained in a boundary hypersurface, we call $Y$ a \emph{boundary p-submanifold}. When $Y$ is a boundary p-submanifold, the blow-up of $X$ along $Y$ is defined as
\[
  [X;Y] := (X\setminus Y) \sqcup S\,{}^+N Y,
\]
where ${}^+N Y={}^+T_Y X/T Y$ is the (non-strictly) inward pointing normal bundle (with ${}^+T_p X$, $p\in Y$, consisting of all $V\in T_p X$, which are not non-strictly inward pointing, i.e.\ $V\rho\geq 0$ for the defining functions $\rho$ of all boundary hypersurfaces of $X$ containing $p$), and the inward pointing spherical normal bundle $S\,{}^+N Y={}^+N Y/\R_+$ is its quotient by fiber-dilations. One can give $[X;Y]$ the structure of a smooth manifold with corners: in adapted coordinates $x_1,\ldots,x_k\geq 0$ and $y_1,\ldots,y_{n-k}\in\R$, in which $Y=\{x_1=\ldots=x_q=0,\ y_1=\ldots,y_p=0\}$, set $R=(\sum_{i=1}^q x_i^2+\sum_{j=1}^p y_j^2)^{1/2}$; then $R^{-1}(x_1,\ldots,x_q,y_1,\ldots,y_p)\in\Sph^{q+p-1}$ and $x_{q+1},\ldots,x_k$, $y_{p+1},\ldots,y_{n-k}$ are local coordinates on $[X;Y]$, with $S\,{}^+N Y$ being identified with $R=0$. One calls $S\,{}^+N Y$ the \emph{front face} of $[X;Y]$, while the closure of $X\setminus Y$ is called the \emph{lift} of $Y$. The \emph{blow-down map} $\upbeta\colon[X;Y]\to X$ is defined to be the identity on $X\setminus Y$ and the base projection on $S\,{}^+N Y$; in local coordinates, $\upbeta$ is thus the polar coordinate map. Given any subset $T\subset X$, the lift $\upbeta^*T\subset[X;Y]$ of $T$ is defined to be $\upbeta^{-1}(T)$ when $T\subset Y$, and the closure of $\upbeta^{-1}(T\setminus Y)$ otherwise.

It is often computationally advantageous to use projective coordinates: in the region where $x_1>c\max_{1\neq i\leq q} x_i$ and $x_1>c\max_{j\leq p}|y_j|$, the functions
\[
  x_1,\quad \hat x_i:=\frac{x_i}{x_1}\ (1\neq i\leq q),\ x_{q+1},\ \ldots,\ x_k,\quad \hat y_j:=\frac{y_j}{x_1}\ (j\leq p),\ y_{p+1},\ \ldots,\ y_{n-k}
\]
extend by continuity from $X\setminus Y$ to a smooth coordinate chart on $[X;Y]$. Taking $c<1$, this chart and analogous charts where $x_i$ ($i\leq q$) or $|y_j|$ ($j\leq p$) plays the role of $x_1$ cover a neighborhood of the front face.

\medskip
\textbf{3-body scattering structures.} When $X$ is a manifold with boundary, and $Y\subset\pa X$ is a boundary p-submanifold, then following \cite{VasyThreeBody} we define the 3-body scattering tangent bundle
\[
  \Ttsc[X;Y] := \upbeta^*(\Tsc X),\qquad \upbeta\colon [X;Y]\to X.
\]
The space of smooth sections of this bundle is $\Vtsc([X;Y])$: these can be thought of as scattering vector fields on $X$ but with coefficients that are singular in $Y$ in the precise fashion that they lift to be smooth on $[X;Y]$.

A non-degenerate section of $S^2\,\Ttsc^*[X;Y]$ is a 3-body scattering metric. In this paper, we encounter these in the case $X=\ol{\R_t\times\R^{n-1}_x}$, $Y=\pa(\ol{\R_t\times\{0\}})\subset\pa X$. Working in the subset $t>0$ of $\R^{1+n}_{t,x}$, local coordinates near the interior of the front face of $[X;Y]$ are $\frac{1}{t}$ and $\frac{x/t}{1/t}=x$. In the region $|x|>1$ on the other hand, local coordinates on $[X;Y]$ are $\rho_\cD:=\frac{1}{|x|}$, $\rho_\cT:=\frac{|x|}{t}$, and $\frac{x}{|x|}\in\Sph^{n-2}$. (In particular, each of the two components of the front face is diffeomorphic to $\ol{\R^n_x}$. Moreover, the projection map $(t,x)\mapsto x$ lifts to a diffeomorphism $[X;Y]\to\ol{\R^n}$. See also \cite[Lemma~3.4]{HintzGlueLocI}.) A particular example of a 3-body scattering metric is thus any \emph{stationary} metric which, when expressed in terms of $\dd t$, $\dd x_1$, $\ldots$, $\dd x_{n-1}$, has coefficients of class $\CI(\ol{\R^{n-1}_x})$.

\medskip
\textbf{Differential operators.} Given the Lie algebra of vector fields $\Vb(X)$, one can define the space $\Diffb^m(X)$ of \emph{$m$-th order b-differential operators} to consist of locally finite sums of up to $m$-fold compositions of elements of $\Vb(X)$ (for $m=0$: multiplication operators by elements of $\CI(X)$); the space $\Diffsc^m(X)$ is defined analogously relative to $\Vsc(X)$. Given a weight family $\alpha\colon\cM_1(X)\to\R$, we furthermore define the space $\Diffb^{m,\alpha}(X)=\rho^{-\alpha}\Diffb^m(X)=\{\rho^{-\alpha}P\colon P\in\Diffb^m(X)\}$. Since $\rho^\alpha[V,\rho^{-\alpha}]\in\CI(X)$ for $V\in\Vb(X)$, one has the composition rule $\Diffb^{m_1,\alpha_1}(X)\circ\Diffb^{m_2,\alpha_2}(X)\subset\Diffb^{m_1+m_2,\alpha_1+\alpha_2}(X)$. Analogous results hold for $\Diffsc^{m,\alpha}(X)=\rho^{-\alpha}\Diffsc^m(X)$.

\medskip
\textbf{Boundary spectrum.} Let $X$ be a manifold with boundary with boundary defining function $\rho$, and let $P\in\rho^{-\alpha}\Diffb^m(X)$. In a collar neighborhood $[0,\eps)_\rho\times\pa X$ of $\pa X$, we can write
\[
  \rho^\alpha P = \sum_{j=0}^m (\rho\pa_\rho)^j P_j(\rho),\qquad P_j\in\CI\bigl([0,\eps)_\rho;\Diff^{m-j}(\pa X)\bigr).
\]
The \emph{b-normal operator} of $P$ is then defined as
\[
  N(P) = \rho^{-\alpha}\sum_{j=0}^m (\rho\pa_\rho)^j P_j(0) \in \rho^{-\alpha}\Diff_{\bop,\rm I}^m([0,\infty)_\rho\times\pa X);
\]
this is homogeneous of degree $-\alpha$ under dilations in $\rho$, hence the subscript `$\rm I$' (for `invariant'). We further define its \emph{Mellin transformed normal operator family} to be the holomorphic operator family
\[
  \hat N(P,\xi) := \sum_{j=0}^m \xi^j P_j(0) \in \Diff^m(\pa X)
\]
We then set
\[
  \specb(P) := \{ \xi\in\C \colon \hat N(P,\xi)\colon \CI(\pa X)\to\CI(\pa X)\ \text{is not invertible} \}.
\]
When $X$ is compact and $P$ is an elliptic b-operator, then $\specb(P)\subset\C$ is discrete, and its intersection with every strip on which $\Re\xi$ is bounded is finite.

\medskip
\textbf{Conormality.} Let $X$ be a manifold with boundary, and let $\rho\in\CI(X)$ denote a boundary defining function. For $\alpha\in\R$, we define the space
\[
  \cA^\alpha(X)\subset\CI(X^\circ)
\]
of \emph{conormal functions (with weight $\alpha$)} to consist of all smooth functions $u$ on $X^\circ$ so that $\rho^{-\alpha}P u\in L^\infty_\loc(X)$ for all $P\in\Diffb(X)$.

Next, recall that an \emph{index set} is a subset $\cE\subset\C\times\N_0$ so that for all $C\in\R$ only finitely many $(z,k)\in\cE$ have $\Re z\leq C$, and $(z,k)\in\cE$ implies $(z+1,k)\in\cE$ and also $(z,k-1)\in\cE$ in case $k\geq 1$. The space $\cA_\phg^\cE(X)$ of \emph{$\cE$-smooth functions} (or \emph{polyhomogeneous functions (with index set $\cE$)}) consists of all $u\in\cA^\alpha(X)$ (where $\alpha<\min_{(z,k)\in\cE}\Re z$) so that there exist $u_{(z,k)}\in\CI(\pa X)$, $(z,k)\in\cE$, so that in a collar neighborhood $[0,1)_\rho\times\pa X$ of $\pa X$ and for a cutoff function $\chi\in\CIc([0,1)_\rho)$ which is identically $1$ near $0$ one has
\[
  u - \chi(\rho)\sum_{\genfrac{}{}{0pt}{}{(z,k)\in\cE}{\Re z\leq C}} \rho^z(\log\rho)^k u_{(z,k)} \in \cA^C(X)
\]
for all $C$. We set
\[
  (z,k) := \{ (z+j,l) \colon j\in\N_0,\ 0\leq l\leq k \},
\]
and shall moreover write $(z,*)$ for an (unspecified) index set contained in $(z+\N_0)\times\N_0$. In particular, $\cA_\phg^{(\alpha,0)}(X)=\rho^\alpha\CI(X)$ and $\cA_\phg^{(\alpha,k)}(X)=\sum_{j=0}^k \rho^\alpha(\log\rho)^j\CI(X)$.

Suppose next that $X$ is a manifold with corners. Then for weight families $\alpha\colon\cM_1(X)\to\R$, the space $\cA^\alpha(X)$ is defined analogously to before in the case that $X$ is a manifold with boundary. Given a collection $\cH\subset\cM_1(X)$ of boundary hypersurfaces, write $\tilde\rho=\prod_{H\in\cH}\rho_H$ for the product of their defining functions; then for $\delta\geq 0$ we set
\[
  \cA_{\cH,\delta}^\alpha(X) = \{ u\in\CI(X^\circ) \colon P u \in \rho^{-\alpha}\tilde\rho^{-\delta m}L^\infty_\loc(X)\ \forall\,P\in\Diffb^m(X),\ m\in\N_0 \}.
\]
Let next $\cE=(\cE_H)_{H\in\cM_1(X)}$ where $\cE_H\subset\C\times\N_0$ is an index set for all $H\in\cM_1(X)$, and let $\alpha_H<\min_{(z,k)\in\cE_H}\Re z$; then $\cA_\phg^\cE(X)$ consists of all $u\in\cA^\alpha(X)$ (where $\alpha=(\alpha_H)_{H\in\cM_1(X)}$) so that, in a collar neighborhood $[0,1)_{\rho_H}\times H$ of $H\in\cM_1(X)$,
\begin{equation}
\label{EqBPhgExp}
  u - \sum_{\genfrac{}{}{0pt}{}{(z,k)\in\cE_H}{\Re z\leq C}} \rho_H^z(\log\rho_H)^k u_{H,(z,k)} \in \cA^{\alpha'(C)}(X)
\end{equation}
for all $C$ where $\alpha'(C)_H=C$ and $\alpha'(C)_{H'}=\alpha_{H'}$ for $H'\neq H$; here $u_{H,(z,k)}\in\cA^{\alpha^H}(H)$ where $(\alpha^H)_{H'}=\alpha_{H'}$ for all $H'\in\cM_1(X)$ with $H'\neq H$, $H'\cap H\neq\emptyset$. See again \cite{MelroseDiffOnMwc}, and also \cite[\S{2A}]{MazzeoEdge}, \cite{MelrosePushfwd}.

\medskip
\textbf{Sobolev spaces.} Associated with $\Vb(X)$ are \emph{weighted b-Sobolev spaces}. Suppose first $X$ is compact. Fixing a weighted b-density to define $L^2(X)$, one defines
\[
  \Hb^m(X) = \{ u\in L^2(X) \colon P u\in L^2(X)\ \forall\,P\in\Diffb^m(X) \},\qquad m\in\N_0,
\]
and then $\Hb^{m,\alpha}(X)=\rho^\alpha\Hb^m(X)=\{\rho^\alpha u\colon u\in\Hb^m(X)\}$. This can be given the structure of a Hilbert space, with squared norm $\|u\|_{\Hb^{m,\alpha}(X)}^2=\sum_j\|P_j u\|_{L^2(X)}^2$ where $\{P_j\}\subset\Diffb^{m,\alpha}(X)$ is a finite subset spanning $\Diffb^{m,\alpha}(X)$ over $\CI(X)$. For arbitrary $m\in\R$, these spaces can be defined using duality and interpolation. The spaces $\Hsc^{m,\alpha}(X)$ are defined analogously. When $E$ is a vector bundle, one can define $\Hb^{m,\alpha}(X;E)$ to consist of $E$-valued distributions which, upon multiplication with a smooth cutoff on whose support $E$ is trivialized, are $(\rank E)$-tuples of elements of $\Hb^{m,\alpha}(X)$.

For general $X$ and a choice of weighted b-density, we can define the spaces $H_{\bop,\loc}^m(X)$ (distributions $u$ so that $\chi u\in\Hb^m$ when $\chi\in\CIc(X)$ is a cutoff to a coordinate chart) and $H_{\bop,\cp}^m(X)\subset H_{\bop,\loc}^m(X)$ (compactly supported elements). Suppose $U\subset X$ is a precompact open subset; then the space
\begin{subequations}
\begin{equation}
\label{EqBHsupp}
  \dot H_\bop^m(\ol U) := \{ u \in H_{\bop,\loc}^m(X) \colon \supp u\subset\ol U \}
\end{equation}
can be given the structure of a Hilbert space; indeed, as the squared norm of $u$ one can take the sum of squared $L^2(U)$-norms of up to $m$ derivatives of $u$ along a fixed finite spanning set of smooth b-vector fields near $\ol{U}$. The space
\begin{equation}
\label{EqBHext}
  \bar H_\bop^m(U) := \{ \tilde u|_U \colon \tilde u\in H_{\bop,\loc}^m(X) \}
\end{equation}
is given the quotient topology via $\bar H_\bop^m(U)=H_{\bop,\loc}^m(X)/\dot H_\bop^m(X\setminus U)$. When $X$ is compact, this endows $\bar H_\bop^m(U)$ with the structure of a Hilbert space. (Only the compactness of $\ol{U}$ is required for this, as one can then replace $X$ in this discussion by a \emph{compact} manifold with corners containing $\ol{U}$ in its topological interior.) Elements of $\dot H_\bop^m(\ol U)$, resp.\ $\bar H_\bop^m(U)$ are called \emph{supported}, resp.\ \emph{extendible distributions}. We also encounter mixed versions of such spaces: suppose $f,g\in\CI(X)$ are such that $\dd f$ and $\dd g$ are linearly independent on $\{f=g=0\}$. Set $\Omega_{f_0,g_0}:=\{f\geq f_0,\ g\geq g_0\}$, and suppose $\Omega_{-1,-1}$ is compact. Then
\begin{equation}
\label{EqBHsuppext}
  H_\bop^m(\Omega_{0,0})^{\bullet,-} := \{ \tilde u|_{\Omega_{-1,0}} \colon \tilde u\in\dot H_\bop^m(\Omega_{-1,-1}),\ \supp\tilde u\subset\Omega_{0,-1} \}
\end{equation}
\end{subequations}
is a space of distributions with supported character at $\Omega_{0,0}\cap\{f=0\}$ and extendible character at $\Omega_{0,0}\cap\{g=0\}$. The dual spaces with respect to $L^2(X)$ are
\[
  (\dot H_\bop^m(\ol U))^* = \bar H_\bop^{-m}(U),\qquad
  (\Hb^m(\Omega_{0,0})^{\bullet,-})^* = \Hb^{-m}(\Omega_{0,0})^{-,\bullet}.
\]
Weighted and vector bundle valued versions of these spaces are defined analogously. 

\medskip
\textbf{Symbols.} Given a vector bundle $E\to X$, we shall write
\begin{equation}
\label{EqBPm}
  P^m(E),\ \text{resp.}\ P^{[m]}(E) \subset \CI(E)
\end{equation}
for the subspace of smooth functions which are, on each fiber $E|_x$, polynomials, resp.\ homogeneous polynomials, of degree $m$. The space of symbols of order $m$ is the space
\[
  S^m(E) := \cA^{-m}(\bar E)
\]
where $\bar E$ is the fiber-radial compactification, and the weight refers to fiber infinity $S E\subset\bar E$. In a local trivialization $U\times\R^k$ of $E$ over $\R^n\cong U\subset X$, membership of $a=a(x,e)$ in $S^m(E)$ is equivalent to the usual condition $|\pa_x^\alpha\pa_e^\beta a(x,e)|\leq C_{K\alpha\beta}\la e\ra^{m-|\beta|}$ for $x$ in a compact subset $K\subset\R^n$ and all $e\in\R^k$.

For $\delta\geq 0$, the space
\[
  S^m_{1-\delta,\delta}(E):=\cA_{S E,\delta}^{-m}(E)
\]
consists precisely of H\"ormander type $(\rho,\delta)$ (with $\rho=1-\delta$) symbols on $E$. They arise in the present paper for variable order symbols. For a bounded function $\sfm\in\CI(S E)$, called a \emph{variable order function}, denote by $\tilde\sfm\in\CI(\bar E)$ any bounded extension of $\sfm$. Let $\rho\in\CI(\bar E)$ be a defining function of $S E$. Then we set
\[
  S^\sfm(E) := \rho^{-\tilde\sfm}\bigcap_{\delta>0} S^0_{1-\delta,\delta}(E).
\]
Allowing for $\delta>0$ here implies that this definition is independent of the choice of extension $\tilde\sfm$ and boundary defining function $\rho$ (in essence since logarithmic factors arising from differentiation of the weight are bounded by $\rho^{-\delta}$ for all $\delta>0$).

\medskip
\textbf{Bounded geometry.} We recall the basic notions from Shubin \cite{ShubinBounded}, and refer the reader to \cite{AmmannLauterNistorLie,AmmannLauterNistorLieGeometry} for general results relating certain Lie algebras of vector fields to manifolds of bounded geometry and pseudodifferential calculi. Let $X$ be a smooth manifold without boundary. A smooth Riemannian metric $g$ on $X$ is called a \emph{metric of bounded geometry} if the Riemann curvature tensor is uniformly bounded together with all of its covariant derivatives, and the injectivity radius of $(X,g)$ is positive. Equivalently, one can cover $X$ by coordinate charts $\phi_i\colon U_i\to B_2\subset\R^n$ (where $B_R$ is the ball of radius $R$) with the following properties: all transition functions $\phi_i\circ\phi_j^{-1}$, as maps between open subsets of $\R^n$, are uniformly bounded together with all derivatives; there exists $J<\infty$ so that the intersection of more than $J$ pairwise distinct elements of $\{U_i\}$ is empty; and the unions of all $\phi_i^{-1}(B_1)$ still cover $X$. A smooth manifold equipped with such a covering is called a \emph{manifold of bounded geometry}.

A \emph{vector bundle of bounded geometry} is a smooth vector bundle $E\to X$ which admits trivializations over the $U_i$ whose transition functions are uniformly bounded; similarly for fiber bundles of bounded geometry. Important examples are $E=T X,T^*X$. One can then define uniform symbol spaces $S_{\rm uni}^s(E)$ (and also variable order versions $S_{\rm uni}^\sfs(E)$ for $\sfs\in\CI_{\rm uni}(S E)$, which we shall not discuss explicitly here).

Given $a\in S_{\rm uni}^s(T^*X)$, we can define its quantization $\Op_{\rm uni}(a)$ as follows: pick $\chi,\tilde\chi\in\CIc(B_2)$ so that $\chi=1$ on $\ol{B_1}$ and $\tilde\chi=1$ on $\supp\chi$, and set $\Xi=\sum_i\phi_i^*\chi\in\CI_{\rm uni}(X)$ and $\chi_i:=(\phi_i^*\chi)/\Xi$, $\tilde\chi_i=\phi_i^*\tilde\chi$. Then
\begin{equation}
\label{EqBOpUni}
  \Op_{\rm uni}(a):=\sum_i \phi_i^* \Op\bigl( \tilde\chi_i (\phi_i^{-1})^*(\chi_i a) \bigr) (\phi_i)_*,
\end{equation}
where $(\Op(b)u)(x)=(2\pi)^{-n}\int e^{i(x-y)\cdot\xi}b(x,\xi)u(y)\,\dd\xi\,\dd y$ is the standard (left) quantization on $\R^n$. One then defines the space of \emph{uniform pseudodifferential operators} as
\[
  \Psi_{\rm uni}^s(X):=\Op_{\rm uni}(S^s_{\rm uni}(T^*X))+\Psi_{\rm uni}^{-\infty}(X)
\]
where $\Psi_{\rm uni}^{-\infty}(X)$ consists of all operators whose Schwartz kernels lie in $\CI_{\rm uni}(X\times X)$ and have support in a bounded neighborhood of the diagonal. The principal symbol of $A=\Op_{\rm uni}(a)+R$, where $a\in S_{\rm uni}^s(T^*X)$ and $R\in\Psi_{\rm uni}^{-\infty}(X)$, is defined as the equivalence class
\[
  \upsigma_{\rm uni}^s(A) := [a] \in S_{\rm uni}^s(T^*X)/S_{\rm uni}^{s-1}(T^*X).
\]
One can also consider quantizations of H\"ormander type symbols $a\in S_{1-\delta,\delta,\rm uni}^s(T^*X)$ when $\delta\in(0,\frac12)$; the principal symbol is then well-defined modulo $S_{\rm uni}^{s-1+2\delta}(T^*X)$. Another variant concerns operators $A=\Op_{\rm uni}(a)$ with $a\in S_{\rm uni}^\sfs(T^*X)$ having variable order $\sfs\in\CI_{\rm uni}(S^*X)$; in this case the principal symbol lies in $S^\sfs_{\rm uni}(T^*X)/\bigcap_{\delta\in(0,\frac12)}S^{\sfs-1+2\delta}_{\rm uni}(T^*X)$.

\emph{Uniform Sobolev spaces} are defined to consist of all distributions with finite squared norm
\[
  \|u\|_{H_{\rm uni}^s(X)}^2 := \sum_i \|(\phi_i^{-1})^*(\chi_i u)\|_{H^s(\R^n)}^2.
\]
One can also consider weighted uniform Sobolev spaces, where the admissible weights are positive functions $w\in\CI(X)$ so that with $w_i:=(\phi_i^{-1})^*w\in\CI(B_2)$ the ratio $\sup w_i/\inf w_i$ is uniformly bounded, and $w_i/\inf w_i$ is uniformly bounded in $\CI(B_2)$. Then
\[
  \|u\|_{H_{\rm uni}^{s,w}(X)}^2 := \sum_i \|w_i^{-1}(\phi_i^{-1})^*(\chi_i u)\|_{H^s(\R^n)}^2.
\]
If one replaces $w_i$ here by $\inf w_i$ or $\sup w_i$, one obtains an equivalent norm. Furthermore, one can define uniform Sobolev spaces with variable differential order $\sfs\in\CI_{\rm uni}(S^*X)$ in the same fashion, with the local Sobolev norms now being the variable order norms $H^{\sfs_i}(\R^n)$ where $\sfs_i$ interpolates between $(\phi_i^{-1})^*\sfs$ over $B_2$ and a fixed constant far away, in such a way that the $\sfs_i$ are uniformly bounded in $\CI(S^*\R^n)$; see \cite{VasyMinicourse} for a detailed treatment of variable order spaces.

Uniform ps.d.o.s act boundedly between (weighted) uniform Sobolev spaces. One can also define classes $\Psi_{\rm uni}^{s,w}(X)$ of weighted uniform ps.d.o.s via simple modifications of the above definitions; in particular, $\Psi_{\rm uni}^{-\infty,w}(X)$ consists of operators with Schwartz kernels supported in a bounded neighborhood of the diagonal and of class $(\pi^*w)^{-1}\CI_{\rm uni}(X\times X)$ where $\pi\colon X\times X\to X$ is the projection to the left (or right) factor.

\section{Perturbations of the trapped set; proof of Proposition~\usref{PropTrap}}
\label{STrap}

Let $\hat Y$ be a smooth manifold, and let $\eps_0>0$. On $Y:=[0,\eps_0]_\eps\times I_t\times\hat Y$, we write
\begin{align*}
  \cV_\seop(Y) &:= \{ V=(V_\eps)_{\eps\in(0,\eps_0]} \colon V_\eps=a_\eps\eps\pa_t+W_\eps, \\
    &\qquad \text{with}\ a_\eps\in\CI(I\times Y),\ W_\eps\in\CI(I\times Y;T(I\times Y))\ \text{uniformly bounded} \}.
\end{align*}
We furthermore define $\cC_\seop^k(Y)$ to consist of all functions on $(0,\eps_0]\times I\times\hat Y$ which remain uniformly bounded (in $\eps$) over every fixed compact subset of $I\times\hat Y$ upon application of any finite number of elements of $\cV_\seop$. We shall write $\eps^\alpha\cC_\seop^k\Vse(Y)$ for the space of locally finite linear combinations of products $\eps^\alpha f V$ where $f\in\cC_\seop^k(Y)$ and $V\in\Vse(Y)$.

We shall deduce Proposition~\ref{PropTrap} from the following general result.

\begin{thm}[Extensions of stable and unstable manifolds: general result]
\label{ThmTrap}
  Let $\hat\cX$ be a smooth manifold, and suppose we are given the following data on $\hat\cX$:
  \begin{enumerate}
  \myitem{ItTraphatVVF}{$\hat V$.1} a smooth vector field $\hat V\in\cV(\hat\cX)=\CI(\hat\cX;T\hat\cX)$;
  \myitem{ItTraphatVNhyp}{$\hat V$.2} a compact $2$-codimensional subset $\hat\Gamma\subset\hat\cX$ which is invariant under the $\hat V$-flow, and which is $r$-normally hyperbolic for all $r\in\N$ (see below) with $1$-codimensional orientable local $\CI$ unstable and stable manifolds $\hat\Gamma^{\rm u}$ and $\hat\Gamma^{\rm s}\subset\hat\cX$, respectively.
  \end{enumerate}
  Shrink $\hat\cX$ so that it is given by a neighborhood
  \[
    \hat\cX = \hat\Gamma \times (-2\delta_0,2\delta_0)_u \times (-2\delta_0,2\delta_0)_s
  \]
  of $\hat\Gamma$ in such a way that $\hat\Gamma^{\rm u}=\{s=0\}$ and $\hat\Gamma^{\rm s}=\{u=0\}$.\footnote{Therefore, $u$ is a coordinate along $\hat\Gamma^{\rm u}$, while $s$ is a defining function of $\hat\Gamma^{\rm u}$.} Let $-\infty<t_0<t_1<\infty$, $I=(t_0-\delta,t_1+\delta)$ for some $\delta>0$, and define the space
  \[
    \cX := [0,\eps_0]_\eps \times I_t \times \hat\cX.
  \]
  Let $k\in\N\cup\{\infty\}$. Suppose we are given $V\in\cC_\seop^0\Vse(\cX)$ so that
  \begin{enumerate}
  \myitem{ItTrapV}{$V$} for some $\alpha>0$, we have $V-(\hat V+\eps\pa_t)\in\eps^\alpha\cC_\seop^k\Vse(\cX)$.
  \end{enumerate}
  Then, upon shrinking $\eps_0>0$ if necessary, there exist functions $\varphi^{\rm u},\varphi^{\rm s}$ on $\cX$ with
  \begin{equation}
  \label{EqTrapExist}
    \varphi^{\rm u}-s,\ \varphi^{\rm s}-u\in \eps^\alpha\cC^k_\seop(\cX),\qquad \pa_s(\varphi^{\rm u}-s)=0=\pa_u(\varphi^{\rm s}-u),
  \end{equation}
  so that in a neighborhood of $[0,\eps_0]_\eps\times[t_0,t_1]\times\hat\Gamma$ in $\cX$, the vector field $V$ is tangent to $\Gamma^{\rm u}:=(\varphi^{\rm u})^{-1}(0)$ and $\Gamma^{\rm s}:=(\varphi^{\rm s})^{-1}(0)$.
\end{thm}

The second condition in~\eqref{EqTrapExist} means that $\Gamma^\bullet$ is a graph over $[0,\eps_0]_\eps\times I_t\times\hat\Gamma^\bullet$ of a function of class $\eps^\alpha\cC_\seop^k$, for $\bullet={\rm u,s}$.

\begin{rmk}[Generalizations]
\label{RmkTrapGeneralize}
  The case $\codim\hat\Gamma=2$ is the case of interest in this paper. The case of higher codimensions can be treated with the same methods; we leave this to the interested reader. Furthermore, the existence of the (un)stable manifolds $\hat\Gamma^{\rm s/u}$ follows from the $r$-normal hyperbolicity for every $r$, as shown in \cite[Theorem~(4.1)]{HirschPughShubInvariantManifolds}. We make the orientability assumption for convenience here; the result remains true (with minor modifications to the proof) without it.
\end{rmk}

\begin{rmk}[Fast time]
\label{RmkTrapFast}
  If one introduces the fast time $\hat t=\frac{t-t_0}{\eps}$, then the $\eps$-level set of $\cX$ is $(-\eps^{-1},(t_1-t_0+2)\eps^{-1})\times\hat\cX\subset\R_{\hat t}\times\hat\cX=:\hat\cM$, and condition~\eqref{ItTrapV} means that $V$ equals $\pa_{\hat t}+\hat V$ up to $\cO(\eps^\alpha)$ correction terms (with regularity in $\eps^{-1}\pa_{\hat t}$ and $\hat\cX$). The unperturbed vector field $\pa_{\hat t}+\hat V$ is tangent to $\R_{\hat t}\times\hat\Gamma^{\rm u}=\{s=0\}$. The conclusion of Theorem~\ref{ThmTrap} is thus that one can construct $\Gamma^{\rm u}$ close to $\R_{\hat t}\times\hat\Gamma^{\rm u}$ also for perturbations $V$ of this vector field \emph{on long time scales $\sim\eps^{-1}$}.
\end{rmk}

\begin{rmk}[Higher regularity]
\label{RmkTrapReg}
  The regularity~\eqref{EqTrapExist} means the uniform boundedness (in $\eps$) of the supremum norm up to $k$ derivatives of $\eps^{-\alpha}(\varphi^{\rm u}-s)$ along a fixed spanning set of `se-vector fields' $\eps\pa_t$, $\cV(\hat\cX)$. If instead we require `s-regularity' of $V-(\hat V+\eps\pa_t)$, i.e.\ regularity of the coefficients under application of `s-vector fields' $\pa_t$, $\cV(\hat\cX)$, then also $\varphi^{\rm u}$ and $\varphi^{\rm s}$ have the matching amount of s-regularity; in the proof, this requires modifying the Lipschitz constant~\eqref{EqTrapLipschitz} to use the temporal distance $|t-t'|$ instead of $\frac{|t-t'|}{\eps}$.
\end{rmk}

We recall the notion of  $r$-normal hyperbolicity \cite{HirschPughShubInvariantManifolds}: there exists a constant $T_0>0$ so that for the time $T_0$ flow
\begin{equation}
\label{EqTraphatf}
  \hat f=e^{T_0\hat V}
\end{equation}
of $\hat V$, the following holds. Set $\hat N^\bullet:=T_{\hat\Gamma}\hat\Gamma^\bullet/T\hat\Gamma$ for $\bullet={\rm u,s}$, and write
\[
  \hat\Gamma_p\hat f:=D_p\hat f|_{T_p\hat\Gamma},\qquad
  \hat N_p^\bullet\hat f := D_p\hat f|_{\hat N_p^\bullet}\quad(\bullet={\rm u,s}).
\]
For a linear map $A$ between normed vector spaces, we set $m(A)=\inf_{\|x\|=1}\|A x\|$ (so $m(A)=\|A^{-1}\|^{-1}$ when $A$ is invertible). Then there exist positive definite fiber inner products on $T\hat\Gamma$, $\hat N^{\rm u}$, and $\hat N^{\rm s}$ so that for all $p\in\hat\Gamma$ and $0\leq k\leq r$,
\begin{equation}
\label{EqTrapNhyp}
  m(\hat N_p^{\rm u}\hat f) > \|\hat\Gamma_p\hat f\|^k, \qquad
  \|\hat N_p^{\rm s}\hat f\| < m(\hat\Gamma_p\hat f)^k;
\end{equation}
(The conditions~\eqref{EqTrapNhyp} amount to the \emph{immediate relative $r$-normal hyperbolicity} of $e^{T\hat V}$; they follow from the \emph{eventual relative $r$-normal hyperbolicity} of the time $1$ flow $e^{\hat V}$ by taking $T_0$ sufficiently large. The latter is the notion used in the subextremal Kerr context in \cite[Assumption~(8)]{DyatlovWaveAsymptotics} and \cite[(1.6)]{WunschZworskiNormHypResolvent}.)

\begin{proof}[Proof of Proposition~\usref{PropTrap}, given Theorem~\usref{ThmTrap}]
  Let $\hat\cU\subset\hat X_b^\circ$ be a precompact open set containing the projection to the base of $\hat\Gamma$ (which is compact). Set $\hat\cX=S^*_{\hat\cU}\hat M_b^\circ\cap\pa\hat\Sigma_b^+$. We fix a fiber bundle isomorphism $S^*_{\hat\cU}\hat M_b^\circ=\hat\cU\times\Sph^3$ via polar coordinates in the fibers corresponding to the fiber-linear coordinates $-\sigma,\xi$ induced by the coordinates $\hat t,\hat x$; and then $\hat\cX\cong\hat\cU\times\Sph^2\subset\hat\cU\times\Sph^3$. Working locally near $\pa\hat\Gamma_b^+$, we set $\rho_\infty=\sigma^{-1}$ and set $\hat V:=(\rho_\infty H_{\hat G_b}\hat t)^{-1}\rho_\infty H_{\hat G_b}-\pa_{\hat t}$. While a priori this is a vector field on $S^*\hat M_b^\circ$, it annihilates $\hat t$ and thus restricts to a vector field on $S^*_{\hat X_b^\circ}\hat M_b^\circ$ which is moreover tangent to $\pa\hat\Sigma_b^+$; therefore, $\hat V\in\cV(\hat\cX)$.

  Set $\wt\cU=[0,\eps_0]_\eps\times I_t\times\cU$. The coordinates $-\sigma,\xi$ also induce an identification $\Sse^*_{\wt\cU}\wt M\cong [0,\eps_0]_\eps\times I\times(\hat\cU\times\Sph^3)$. Since on this set $\eps^2\wt G_\eps$ is a fiber-wise nondegenerate quadratic function with $\CI+\eps^\alpha\cC_\seop^k$ regularity in the base, we can write $\Sse^*_{\wt\cU}\wt M\cap\pa\wt\Sigma^+$ as the normal graph over $\wt\cX=[0,\eps_0]_\eps\times I_t\times\hat\cX$ of a function $\wt\Phi$ of class $\eps^\alpha\cC^k_\seop([0,\eps_0]\times I_t\times\hat\cU)$. To prove this, pick local coordinates $z\in\R^5$, $w\in\R$ on $S^*_{\hat\cU}\hat M_b^\circ$ near $\pa\hat\Sigma_b^+$ so that $\hat G_b=w$. Then $\wt\cX\subset[0,\eps_0]_\eps\times I_t\times S^*_{\hat\cU}\hat M_b^\circ$ is given by $w=0$ in the coordinates $\eps,t,z,w$. Write $\rho_\infty^2\eps^2\wt G=\rho_\infty^2\hat G_b+\eps^\alpha\wt H_0$ where $\wt H_0\in\cC_\seop^k$ (i.e.\ uniform $L^\infty$ bounds for up to $k$ differentiations along $\eps\pa_t$, $\pa_z$, $\pa_w$). The function $\wt\Phi=\eps^\alpha\wt\Phi_0(\eps,t,z)$ then must satisfy $\wt\Phi_0+\wt H_0(\eps,t,z,\eps^\alpha\wt\Phi_0)=0$ which one can solve (parametrically in $\eps,t,z$) using the contraction mapping principle (using that $k\geq 1$). The $\cC_\seop^k$-regularity of $\wt\Phi_0$ for $k\geq 2$ is then inherited from that of $\wt H_0$, as follows by direct differentiation of this equation. We now define $V:=\wt\Phi^*(\rho_\infty H_{\eps^2\wt G})$, which is of regularity $\CI+\eps^\alpha\cC_\seop^{k-1}$. (See \cite[Lemma~4.6]{HintzPolyTrap} for a similar construction.)

  The $r$-normal hyperbolicity of $\hat V$ for every $r$ was proved in \cite[Propositions~3.6 and 3.7]{DyatlovWaveAsymptotics}. Applying Theorem~\ref{ThmTrap} produces functions which we now denote $\wt\varphi^{\rm u},\wt\varphi^{\rm s}$ and which satisfy $\wt\varphi^{\rm u}-s,\wt\varphi^{\rm s}-u\in\eps^\alpha\cC^k_\seop([0,\eps_0]_\eps\times I_t\times\hat\cX)$. We push these forward along $\wt\Phi$ to get the desired functions on $\Sse^*\wt M$ near $\pa\Gamma^+\cap t^{-1}([t_0,t_1])$.

  Since $\rho_\infty H_{\eps^2\wt G}\wt\varphi^{\rm u}=0$ on the set $(\wt\varphi^{\rm u})^{-1}(0)$ it suffices for obtaining~\eqref{EqTrapHamPhisu} to observe that the function $\wt\varphi^{\rm u}$ vanishes simply for small $\eps\geq 0$. The argument for $\wt\varphi^{\rm s}$ is analogous. Finally, the positivity of $\wt w^{\rm s/u}$ on $\cU$ follows from that of $w^{\rm s/u}$ when $\cU$ is sufficiently small.
\end{proof}

\bigskip

\emph{In the remainder of this section, we prove Theorem~\usref{ThmTrap}.} We only explain the construction of $\wt\Gamma^{\rm u}$; the construction of $\wt\Gamma^{\rm s}$ then follows by working with $-\hat V,-\wt V,-t$ in place of $\hat V,\wt V,t$.

\medskip

\textbf{Step 1. Modification of $V$ in $t<t_0$; new setup.} For notational simplicity, we work with $I=(t_0-1,t_1+1)$ (i.e.\ $\delta=1$ in the statement of Theorem~\ref{ThmTrap}). Fix a smooth function $\chi\in\CI(\R)$ with $\chi(t)=1$ for $t<t_0-\frac12$ and $\chi(t)=0$ for $t>t_0-\frac14$. Consider then on $[0,\eps_0]\times(-\infty,t_1+1)\times\hat\cX$ the vector field
\[
  V'' := \chi(\hat V + \eps\pa_t) + (1-\chi)V.
\]
This vector field still satisfies~\eqref{ItTrapV}; but in $t<t_0-\frac12$ where it equals $\hat V+\eps\pa_t$, it is tangent to
\[
  \wt\Gamma_0^{\rm u}=s^{-1}(0) = [0,\eps_0]_\eps\times\R_t\times\hat\Gamma^{\rm u}.
\]
By shrinking $\delta_0,\eps_0>0$, we may assume that $\eps^{-1}V'' t$ has a positive lower bound on $[0,\eps_0]\times(-\infty,t_++1)\times\hat\cX$ (since this equals $1$ at $\eps=0$); and then we may set
\[
  V' := \frac{1}{\eps^{-1}V't}V'',
\]
which also satisfies~\eqref{ItTrapV}.

Let $T_0$ be as in~\eqref{EqTraphatf}--\eqref{EqTrapNhyp}, and shrink $\eps_0>0$ further so that $\eps_0 T_0<1$. The idea is to construct the set $\wt\Gamma^{\rm u}$ as (a subset of) the flow-out under $V'$ of $\wt\Gamma^{\rm u}_{\rm germ}:=(\wt\Gamma_0^{\rm u}\cap\{t_0-2<t<t_0-1\}) \cup (\{0\}\times(t_0-2,t_1+1)\times\hat\Gamma^{\rm u})$. In terms of the map $f':=e^{T_0 V'}$, we can equivalently define $\wt\Gamma^{\rm u}$ as the intersection of $\{t_0-\frac12<t<t_1+\frac12\}$ with $\bigcup_{n\geq 0}f'{}^n(\wt\Gamma^{\rm u}_{\rm germ})$ where $f'{}^n=f'\circ f'\circ\cdots\circ f'$ (with its maximal domain of definition). Carefully note however that in order to cover a unit time interval at the parameter value $\eps$, one needs to iterate $f$ on the order of $\eps^{-1}$ many times.\footnote{Without any structural assumptions on the $\hat V$-flow, already much fewer iterations would typically already produce unit size deviations from $\wt\Gamma^{\rm u}_0$: for example, if $V$ is $\cC^1$ in $\eps$, then this typically happens after $\sim\log\eps^{-1}$ many iterations.} We shall thus instead use ideas from \cite{HirschPughShubInvariantManifolds} and construct $\wt\Gamma^{\rm u}$ using a graph transform involving $f$.

To set this up, we now rename $V'$ as $V$. Shifting and scaling $\eps$, $t$, $\hat V$, and $V$, and giving ourselves some room for notational convenience, we assume that
\begin{equation}
\label{EqTrapV}
  V=\hat V+\eps\pa_t\ \ \text{for}\ \ t<0;\qquad
  \eps^{-1}V t=1\ \ \text{on}\ \ \cX:=[0,\eps_0]_\eps\times[-2,2]_t\times\hat\cX.
\end{equation}
Moreover, we assume that $\hat f=e^{\hat V}$ satisfies the assumptions~\eqref{EqTrapNhyp} of immediate relative $r$-normal hyperbolicity for any fixed $r\geq 1$; and we set $f=e^V$.

For $\delta<\min(\delta_0,\eps_0)$, write
\[
  \hat\Gamma^{\rm u}(\delta):=\hat\Gamma\times[-\delta,\delta]_u,\qquad
  \Gamma^{\rm u}(\delta)=[0,\delta]_\eps\times[-1,1]_t\times\hat\Gamma^{\rm u}(\delta).
\]
We shall write points in $\hat\Gamma^{\rm u}(\delta)$ as $\hat x^{\rm u}=(\hat x,u)\in\hat\Gamma^{\rm u}(\delta)$. We shall construct $\wt\Gamma^{\rm u}$ as the image of a section
\begin{equation}
\label{EqTrapSigma}
\begin{split}
  \sigma\colon\Gamma^{\rm u}(\delta)\to\cS(\delta) := \Gamma^{\rm u}(\delta)\times[-\delta,\delta],\qquad
  &\sigma(\eps,t,\hat x^{\rm u})=(\eps,t,\hat x^{\rm u},\breve\sigma(\eps,t,\hat x^{\rm u})), \\
  &\qquad \breve\sigma|_{\eps=0}=0,\quad \breve\sigma|_{t\leq 0}=0.
\end{split}
\end{equation}
We will then set $\wt\varphi^{\rm u}(\eps,t,\hat x,u,s):=s-\breve\sigma(\eps,t,\hat x,u)$, which inherits the regularity from $\breve\sigma$.

We introduce the notation
\[
  \hat\pi^{\rm u}\colon\hat\cX\ni(\hat x,u,s)\mapsto(\hat x,u)\in\hat\Gamma^{\rm u}(\delta),\qquad
  \pi^{\rm u}\colon(\eps,t,\hat x,u,s)\mapsto(\eps,t,\hat\pi^{\rm u}(\hat x,u,s))\in\Gamma^{\rm u}(\delta).
\]

\medskip

\textbf{Step 2. Constructing a continuous section $\sigma$.} We shall find $\sigma$ in~\eqref{EqTrapSigma} in the space
\begin{equation}
\label{EqTrapSigmadelta}
  \Sigma(\delta) = \bigl\{ {\rm sections}\ \sigma\colon\Gamma^{\rm u}(\delta) \to \cS(\delta) \colon |\breve\sigma(\eps,t,\hat x^{\rm u})|\leq C_\Sigma\eps^\alpha H(t),\ L_{\eps,(t,\hat x^{\rm u})}(\breve\sigma)\leq C_\Sigma\eps^\alpha H(t) \bigr\};
\end{equation}
here $H$ is the Heaviside function. We extend $\breve\sigma$ by $0$ to $t<0$, and thus $\sigma(\eps,t,\hat x^{\rm u})=(\eps,t,\hat x^{\rm u},0)$ for $t<0$. In~\eqref{EqTrapSigmadelta}, the (large) constant $C_\Sigma$ will be chosen in the course of the proof (and then our argument works for all sufficiently small $\delta>0$), and we use the local Lipschitz constant
\begin{equation}
\label{EqTrapLipschitz}
\begin{split}
  &L_{\eps,(t,\hat x^{\rm u})}(\breve\sigma) := \limsup_{(t',\hat y^{\rm u})\to(t,\hat x^{\rm u})} \frac{|\breve\sigma(\eps,t,\hat x^{\rm u})-\breve\sigma(\eps,t',\hat y^{\rm u})|}{d_\eps((t,\hat x^{\rm u}),(t',\hat y^{\rm u}))}\,, \\
  &\hspace{9em} d_\eps((t,\hat x^{\rm u}),(t',\hat y^{\rm u})):=\frac{|t-t'|}{\eps}+d(\hat x^{\rm u},\hat y^{\rm u}),
\end{split}
\end{equation}
with the Riemannian distance function $d$ on $\hat\Gamma^{\rm u}$ (given by the definition of $r$-normal hyperbolicity) being used to define $d(\hat x^{\rm u},\hat y^{\rm u})$. Equipped with the supremum norm of $\breve\sigma$, the space $\Sigma(\delta)$ is complete. (The choice of space~\eqref{EqTrapSigmadelta} is the main difference between the present proof and \cite{HintzPolyTrap}.) The strategy is to iterate the \emph{graph transform}
\[
  \Sigma(\delta) \ni \sigma \mapsto f_\sharp\sigma := f\sigma g,
\]
where $g$ is a right inverse of $\pi^{\rm u}f\sigma$. (Thus, the image of the section $f_\sharp\sigma$ is equal to $f$ applied to the image of $\sigma$.) We shall show that this map is well-defined and defines a contraction on $\Sigma(\delta)$.

\medskip

\textit{\underline{Step 2.1.} Control of $g$.} This closely follows analogous arguments in \cite{HirschPughShubInvariantManifolds,HintzPolyTrap}, and hence we shall omit some details. We first work near a point $\hat x\in\hat\Gamma$. We flatten out $\hat\Gamma$ inside of $\hat\Gamma^{\rm u}$ by introducing a map $h_{\hat x}\colon T_{\hat x}\hat\Gamma^{\rm u}\to T_{\hat x}\hat\Gamma^{\rm u}$, vanishing quadratically at $0$, for which the equality of Riemannian exponential maps $\exp^{\hat\Gamma}_{\hat x}(\hat v)=\exp^{\hat\Gamma^{\rm u}}_{\hat x}(\hat v+h_{\hat x}(\hat v))$ holds for all sufficiently small $\hat v\in T_{\hat x}\hat\Gamma$. For small $\hat v^{\rm u}\in T_{\hat x}\hat\Gamma^{\rm u}=T_{\hat x}\hat\Gamma\oplus\R$, we then define charts for $\hat\Gamma^{\rm u}$ and $\hat\cX$ by
\[
  \hat\chi_{\hat x}(\hat v^{\rm u}) := \exp_{\hat x}^{\hat\Gamma^{\rm u}}\bigl(\hat v^{\rm u}+h_{\hat x}(\hat v^{\rm u})\bigr),\qquad
  \hat e_{\hat x}(\hat v^{\rm u},s) := (\hat\chi_{\hat x}(\hat v^{\rm u}),s)
\]
We express $\hat f$ in such local charts as
\begin{equation}
\label{EqTrapfExp}
\begin{split}
  &\hat f_{\hat x} := \hat e_{\hat f(\hat x)}^{-1}\hat f\hat e_{\hat x}\colon T_{\hat x}\hat\Gamma^{\rm u}\oplus\R\to T_{\hat f(\hat x)}\hat\Gamma^{\rm u}\oplus\R, \\
  &\hat f_{\hat x} = D_{\hat x}\hat f+\hat r_{\hat x},\qquad
  D_{\hat x}\hat f = \begin{pmatrix} \hat\Gamma_{\hat x}\hat f \oplus \hat N_{\hat x}^{\rm u}\hat f & 0 \\ 0 & \hat N_{\hat x}^{\rm s}\hat f \end{pmatrix},\quad
  \hat r_{\hat x}=\cO(|\hat v^{\rm u}|^2+s^2).
\end{split}
\end{equation}
For later use, we note that since $\hat f$ preserves $\hat\Gamma^{\rm u}$, the $s$-component of $\hat r_{\hat x}(\hat v^{\rm u},0)$ equals $0$.

We extend the above charts to charts of $\Gamma^{\rm u}(\delta)$ and $\cX$ via
\[
  \chi_{\hat x}(\eps,t,\hat v^{\rm u}) = \bigl(\eps,t,\hat\chi_{\hat x}(\hat v^{\rm u})\bigr),\qquad
  e_{\hat x}(\eps,t,\hat v^{\rm u},s) = \bigl(\eps,t,\hat e_{\hat x}(\hat v^{\rm u},s)\bigr),
\]
and then define the local coordinate representation of $f$ as $f_{\hat x}:=e_{\hat f(\hat x)}^{-1} f e_{\hat x}$. By~\eqref{EqTrapV} and assumption~\eqref{ItTrapV}, we then have (writing $\hat f_{\hat x}$ also for the product of the identity map on $[0,\delta]_\eps\times[-2,2]_t$ with $\hat f_{\hat x}$)
\begin{equation}
\label{EqTrapfDiff}
  f_{\hat x} - \hat f_{\hat x} \colon (\eps,t,\hat v^{\rm u},s) \mapsto \bigl(0,\eps, \tilde f_{\hat x}(\eps,t,\hat v^{\rm u},s)\bigr), \quad \tilde f_{\hat x}\in\eps^\alpha\cC^k_\seop([0,\delta]_\eps\times[-2,\tfrac32]_t\times T_{\hat x}\hat\Gamma^{\rm u}\times\R_s);
\end{equation}
here $\tilde f_{\hat x}$ takes values in $T_{\hat f(\hat x)}\hat\Gamma^{\rm u}\oplus\R$ and vanishes for $t<-\eps$. We shrink the $t$-interval from $[-2,2]$ to $[-2,\frac32]$ here since $f$ increases $t$ by the amount $\eps\leq\delta$, and we require $\delta<\frac12$.

Given $\sigma\in\Sigma(\delta)$, define similarly $\sigma_{\hat x}:=e_{\hat x}^{-1}\sigma\chi_{\hat x}\colon [0,\delta]_\eps\times[-1,1]_t\times T_{\hat x}\hat\Gamma^{\rm u}\to[0,\delta]_\eps\times[-1,1]_t\times T_{\hat f(\hat x)}\hat\Gamma^{\rm u}$, which is thus given by
\[
  \sigma_{\hat x}(\eps,t,\hat v^{\rm u})=\bigl(\eps,t,\breve\sigma_{\hat x}(\eps,t,\hat v^{\rm u})\bigr),\qquad
  \breve\sigma_{\hat x}(\eps,t,\hat v^{\rm u})=\breve\sigma\bigl(\eps,t,\hat\chi_{\hat x}(\hat v^{\rm u})\bigr);
\]
here $\breve\sigma_{\hat x}$ satisfies the bounds~\eqref{EqTrapSigmadelta} albeit with an additional factor of $1+\cO(\delta)$ on the right to account for the varying nature of the metric on $\hat\Gamma^{\rm u}$. Then the local chart representation of $\pi^{\rm u}f\sigma$ is
\begin{align}
  &\chi_{\hat f(\hat x)}^{-1}\pi^{\rm u}f\sigma\chi_{\hat x} = \pi^{\rm u} \circ e_{\hat f(\hat x)}^{-1}f e_{\hat x}\circ e_{\hat x}^{-1}\sigma\chi_{\hat x}= \pi^{\rm u}f_{\hat x}\sigma_{\hat x} \colon \nonumber\\
  &\qquad (\eps,t,\hat v^{\rm u}) \mapsto \pi^{\rm u}\bigl(f_{\hat x}(\eps,t,\hat v^{\rm u},\breve\sigma_{\hat x}(\eps,t,\hat v^{\rm u})\bigr) \nonumber\\
\label{EqTrapfLoc}
  &\qquad\hspace{6em} = \Bigl(\eps,t,\hat\pi^{\rm u}\hat f_{\hat x}\bigl(\hat v^{\rm u},\breve\sigma_{\hat x}(\eps,t,\hat v^{\rm u})\bigr)\Bigr) + \Bigl(0,\eps, \hat\pi^{\rm u}\tilde f_{\hat x}\bigl(\eps,t,\hat v^{\rm u},\breve\sigma_{\hat x}(\eps,t,\hat v^{\rm u})\bigr)\Bigr).
\end{align}

\medskip

\textit{\underline{Step 2.1.1.} Local injectivity.} For $\delta_0>0$, write $T_{\hat x}\hat\Gamma(\delta_0)$ for the open $\delta_0$-ball around $0\in T_{\hat x}\hat\Gamma$; define $\hat N^{\rm u}_{\hat x}(\delta_0)$ analogously. We first claim that if $\delta_0>0$ is sufficiently small, then
\begin{equation}
\label{EqTrapfsigmaInj}
  \chi_{\hat f(\hat x)}^{-1}\pi^{\rm u}f\sigma\chi_{\hat x}\ \ \text{is injective on}\ \ [0,\delta_0]_\eps\times[-1,1]_t\times\bigl(T_{\hat x}\hat\Gamma(\delta_0)\oplus\hat N_{\hat x}^{\rm u}(\delta_0)\bigr).
\end{equation}
Note first that this map takes $\eps\mapsto\eps$ and $t\mapsto t+\eps$, so two points can only have the same image if their $\eps$- and $t$-values agree. Set
\[
  \ubar\mu := m(\hat\Gamma_{\hat x}\hat f) > 0;
\]
for $\eps\in[0,\delta_0]$, $t\in[-1,1]$ then, we estimate using~\eqref{EqTrapfExp} and~\eqref{EqTrapfLoc} and the $\cO(\eps^\alpha)$ bound on the Lipschitz constant of $\breve\sigma$:
\begin{align}
  &d\bigl(\pi^{\rm u}f_{\hat x}(\eps,t,\hat v^{\rm u},\breve\sigma_{\hat x}(\eps,t,\hat v^{\rm u})), \pi^{\rm u}f_{\hat x}(\eps,t,\hat w^{\rm u},\breve\sigma_{\hat x}(\eps,t,\hat w^{\rm u}))\bigr) \nonumber\\
  &\qquad \geq |D_{\hat x}\hat f(\hat v^{\rm u}-\hat w^{\rm u})| - \bigl| \hat r_{\hat x}\bigl(\hat v^{\rm u},\breve\sigma_{\hat x}(\eps,t,\hat v^{\rm u})\bigr) - \hat r_{\hat x}\bigl(\hat w^{\rm u},\breve\sigma_{\hat x}(\eps,t,\hat w^{\rm u})\bigr) \bigr| \nonumber\\
  &\qquad \qquad - \bigl|\hat\pi^{\rm u}\tilde f_{\hat x}\bigl(\eps,t,\hat v^{\rm u},\breve\sigma_{\hat x}(\eps,t,\hat v^{\rm u})\bigr) - \hat\pi^{\rm u}\tilde f_{\hat x}\bigl(\eps,t,\hat w^{\rm u},\breve\sigma_{\hat x}(\eps,t,\hat w^{\rm u})\bigr) \bigr| \nonumber\\
  &\qquad \geq \bigl(\ubar\mu - C_{\hat r}\delta_0(1 + \cO(\eps^\alpha)) - C_{\tilde f}\cO(\eps^\alpha) \bigr)(1-C\delta_0) |\hat v^{\rm u}-\hat w^{\rm u}| \nonumber\\
\label{EqTrapfsigmaEst}
  &\qquad \geq (\ubar\mu-o(1))|\hat v^{\rm u}-\hat w^{\rm u}|,\qquad \delta_0\to 0.
\end{align}
Here $C_{\hat r}$ controls the $\cC^2$ norm of $\hat r_{\hat x}$, and $C_{\tilde f}$ the $\cC^1$ norm of $\eps^{-\alpha}\tilde f_{\hat x}$ in $\hat v^{\rm u},s$; moreover, the factor $1-C\delta_0$ accounts for the difference between the norm in $T_{\hat x}\hat\Gamma^{\rm u}$ and the Riemannian distance function on $\hat\Gamma^{\rm u}$. This gives~\eqref{EqTrapfsigmaInj}.

\medskip

\textit{\underline{Step 2.1.2.} Coverage of the range.} Let now
\begin{equation}
\label{EqTrapfsigmaMapParams}
  0 < \mu < \ubar\mu=m(\hat\Gamma_{\hat x}\hat f),\qquad
  1 < \lambda < m(\hat\Gamma^{\rm u}_{\hat x}\hat f),\qquad
  0<\omega<1,\qquad \lambda\omega>1.
\end{equation}
We then claim that for sufficiently small $\delta\leq\delta_0$
\begin{equation}
\label{EqTrapfsigmaMap}
\begin{split}
  &[0,\delta]_\eps\times[-1,1]_t\times\bigl(T_{\hat f(\hat x)}\hat\Gamma(\mu\omega\delta) \oplus \hat N_{\hat f(\hat x)}^{\rm u}(\lambda\omega\delta)\bigr) \\
  &\qquad \subset \chi_{\hat f(\hat x)}^{-1}\pi^{\rm u}f\sigma\chi_{\hat x} \bigl( [0,\delta]_\eps\times[-2,1]_t\times\bigl(T_{\hat x}\hat\Gamma(\omega\delta) \oplus \hat N_{\hat x}^{\rm u}(\omega\delta)\bigr) \bigr).
\end{split}
\end{equation}
To prove this, we need to show (recalling~\eqref{EqTrapfLoc}) that given a point $(\eps,t,\hat w^{\rm u})$, $\hat w^{\rm u}=(\hat w,u')$, in the space on the left one can solve
\begin{align*}
  &D_{\hat x}\hat f(\hat v^{\rm u}) + \hat q_{\hat x}(\eps,t-\eps,\hat v^{\rm u}) = \hat w^{\rm u}, \\
  &\qquad \hat q_{\hat x}(\eps,t-\eps,\hat v^{\rm u}) := \hat r_{\hat x}\bigl(\hat v^{\rm u},\breve\sigma_{\hat x}(\eps,t-\eps,\hat v^{\rm u})\bigr) + \hat\pi^{\rm u}\tilde f_{\hat x}\bigl(\eps,t-\eps,\hat v^{\rm u},\breve\sigma_{\hat x}(\eps,t-\eps,\hat v^{\rm u})\bigr),
\end{align*}
for $\hat v^{\rm u}=(\hat v,u)$. This is accomplished via a fixed point argument using the map
\begin{equation}
\label{EqTrapfFixedPt}
  \hat v^{\rm u} \mapsto (D_{\hat x}\hat f)^{-1}\bigl( \hat w^{\rm u} - \hat q_{\hat x}(\eps,t-\eps,\hat v^{\rm u})\bigr),
\end{equation}
which, due to the choices~\eqref{EqTrapfsigmaMapParams}, is easily seen to be a contraction on $T_{\hat x}\hat\Gamma(\omega\delta)\oplus\hat N_{\hat x}^{\rm u}(\omega\delta)$ when $\delta$ in~\eqref{EqTrapfsigmaMap} is sufficiently small.

In combination with~\eqref{EqTrapfsigmaInj}, the statement~\eqref{EqTrapfsigmaMap} implies, upon shrinking $\delta$ even further if necessary, that one can define a right inverse $g$ of $\pi^{\rm u}f\sigma$ as a map (recalling~\eqref{EqTrapSigma})
\[
  g \colon S(\delta) \to [0,\delta]_\eps \times [-1,1]_t \times \bigcup_{\hat x\in\hat\Gamma} \hat\chi_{\hat x}\bigl(T_{\hat x}\Gamma(\omega\delta) \oplus \hat N_{\hat x}^{\rm u}(\omega\delta)\bigr).
\]
The estimate~\eqref{EqTrapfsigmaEst} and the characterization of $g(\eps,t,\hat\chi_{\hat x}(\hat v,u))$ as the fixed point of~\eqref{EqTrapfFixedPt} furthermore implies, for $\hat x^{\rm u},\hat y^{\rm u}$ lying in a single chart, the bound
\begin{equation}
\label{EqTrapgEst}
  d_\eps\bigl( g(\eps,t,\hat y^{\rm u}), g(\eps,t',\hat x^{\rm u}) \bigr) \leq (\ubar\mu^{-1} + o(1)) |\hat y^{\rm u}-\hat x^{\rm u}| + (1+o(1))\frac{|t-t'|}{\eps},\qquad \eps\leq\delta\to 0;
\end{equation}
here, the $\cO(\eps^\alpha)=o(1)$-bound on the Lipschitz constant of $\breve\sigma_{\hat x}$ in $t$ enters.

\medskip
\textit{\underline{Step 2.2.} $f_\sharp$ as a map on $\Sigma(\delta)$.} Let $\sigma\in\Sigma(\delta)$, then $f_\sharp\sigma=f\sigma g\colon\Gamma^{\rm u}(\delta)\to\cS(1)$ is well-defined. Write $\pi^{\rm s}$ for the projection onto the $s$-coordinate. We first claim that $f_\sharp\sigma$ maps into $\cS(\delta)$. This follows in the above local coordinates near $\hat x,\hat f(\hat x)$, and writing $g(\eps,t,\hat w^{\rm u})=:(\eps,t-\eps,\hat v^{\rm u})$, from
\begin{align}
  |\pi^{\rm s}f_\sharp\sigma(\eps,t,\hat w^{\rm u})| &= \bigl|\pi^{\rm s} f\bigl(\eps,t-\eps,\hat v^{\rm u},\breve\sigma(\eps,t-\eps,\hat v^{\rm u})\bigr) \bigr| \nonumber\\
    &\leq |\hat N^{\rm s}_{\hat x}\hat f(\breve\sigma(\eps,t-\eps,\hat v^{\rm u}))| + \bigl|\pi^{\rm s}\hat r_{\hat x}\bigl(\hat v^{\rm u},\breve\sigma_{\hat x}(\eps,t-\eps,\hat v^{\rm u})\bigr)\bigr| \nonumber\\
    &\quad\hspace{9.25em} + \bigl|\pi^{\rm s}\tilde f_{\hat x}\bigl(\eps,t-\eps,\hat v^{\rm u},\breve\sigma_{\hat x}(\eps,t-\eps,\hat v^{\rm u})\bigr)\bigr| \nonumber\\
\label{EqTrapfsharpsigma}
    &\leq \|\hat N^{\rm s}_{\hat x}\hat f\| C_\Sigma\eps^\alpha H(t-\eps) + C_{\hat r}\delta\cdot C_\Sigma\eps^\alpha H(t-\eps) + C_{\tilde f}\eps^\alpha H(t).
\end{align}
Here, we used $\pi^{\rm s}\hat r_{\hat x}(\hat v^{\rm u},0)=0$ as well as the vanishing of $\tilde f_{\hat x}(\eps,t-\eps,\cdots)$ for $t-\eps\leq-\eps$. Thus, if we choose $C_\Sigma$ so large that $C_\Sigma\sup_{\hat\Gamma}\|\hat N^{\rm s}\hat f\|+C_{\tilde f}<C_\Sigma$, then for sufficiently small $\delta$, the expression~\eqref{EqTrapfsharpsigma} is bounded by $C_\Sigma\eps^\alpha H(t)$.

As for the Lipschitz constant, write $\hat v^{\rm u}=(\hat v,u)$; similar considerations as for~\eqref{EqTrapfsharpsigma}, together with~\eqref{EqTrapgEst}, then allow us to estimate
\begin{align*}
  L_{\eps,(t,\hat w^{\rm u})}(\pi^{\rm s}f_\sharp\sigma) &\leq L_{\eps,(t-\eps,\hat v^{\rm u})}(\pi^{\rm s}f\sigma)\times L_{\eps,(t,\hat w^{\rm u})}(g) \\
    &\leq \bigl(\|\hat N_{\hat x}^{\rm s}\hat f\| C_\Sigma\eps^\alpha H(t-\eps) + C_{\tilde r}\delta\cdot C_\Sigma\eps^\alpha H(t-\eps) + C_{\tilde f}\eps^\alpha H(t)\bigr) \\
    &\qquad \times \max(\ubar\mu^{-1},1)(1+o(1)) \\
    &\leq C_\Sigma\eps^\alpha H(t),
\end{align*}
where we use~\eqref{EqTrapNhyp} for $k=0,1$.

\medskip
\textit{\underline{Step 2.3.} $f_\sharp$ is a contraction.} Recall that $\Sigma(\delta)$ is equipped with the (complete) sup norm metric. Given $\sigma$, $\sigma'\in\Sigma(\delta)$ and the right inverses $g$, $g'$ of $\pi^{\rm u}f\sigma$, $\pi^{\rm u}f\sigma'$ constructed above, one easily obtains the pointwise bound $|g-g'|\leq o(1)|\sigma-\sigma'|$, $\delta\to 0$ from the description of $g(\eps,t,\hat w^{\rm u})$ as the fixed point of~\eqref{EqTrapfFixedPt}. Therefore, we have the pointwise bound
\[
  |\pi^{\rm s}f_\sharp\sigma-\pi^{\rm s}f_\sharp\sigma'|\leq|\pi^{\rm s}f\sigma g-\pi^{\rm s}f\sigma' g| + |\pi^{\rm s}f\sigma'g-\pi^{\rm s}f\sigma'g'| \leq (\|\hat N^{\rm s}\hat f\|+o(1))|\sigma-\sigma'|,\quad \delta\to 0,
\]
which by~\eqref{EqTrapNhyp} for $k=0$ is bounded by $\theta|\sigma-\sigma'|$ for some $\theta<1$ when $\delta$ is small enough.

The contraction mapping principle thus produces the desired section $\sigma\in\Sigma(\delta)$ with $f_\sharp\sigma=\sigma$, with the bounds on $\breve\sigma$ and its local Lipschitz constants given in~\eqref{EqTrapSigmadelta}.

\medskip

\textbf{Step 3. Higher regularity.} Since the modification of the arguments in \cite[\S{2.3}]{HintzPolyTrap} (which follow \cite[\S{4}]{HirschPughShubInvariantManifolds}) to the present setting proceeds in a similar fashion as the previous two steps of the proof, we shall be very brief.

Pointwise differentiability is proved by defining a graph transform analogous to $f_\sharp$ on sections over $\Gamma^{\rm u}(\delta)$ of the bundle of Lipschitz jets of local sections of $\cS(\delta)$ which at their respective base point of $\Gamma^{\rm u}(\delta)$ agree with the already constructed invariant section $\sigma$. Using the regularity of $f$ under differentiations along $\eps\pa_t$ and vector fields on $\hat\cX$, this graph transform preserves the space of sections of the closed subbundle of \emph{differentiable} Lipschitz jets (with uniform bounds for $\eps\in(0,\delta]$); and therefore the Lipschitz jets of the invariant section $\sigma$ are in fact differentiable.

Continuous differentiability is proved by showing, by analogous means, the existence of a fixed point for a graph transform acting on continuous sections of a bundle of linear maps, which end up being the local linear approximations of $\sigma$. Up to this point, only the $r$-normal hyperbolicity for $r=1$ and the $\cC_\seop^1$-regularity of the vector field perturbation in condition~\eqref{ItTrapV} of Theorem~\ref{ThmTrap} are used.

Higher regularity is proved via induction in $r$ similarly to \cite[Step~6 in the proof of Theorem~2.3]{HintzPolyTrap}. We remark here that higher regularity of $\sigma$ is clear for any fixed value of $\eps>0$, since $\sigma$ can be obtained by applying $f$ a finite number of times (as explained in Step~1). Thus, for proving $r$-fold differentiability, we may replace $\delta(1):=\delta$ by a value $\delta(r)$; but by applying $f_\sharp$ a finite number of times (depending on the values of $\delta(r)$ and $\delta(1)$, and utilizing the expanding nature of $\hat f$ in the unstable directions), one can then recover the $r$-fold differentiability of $\sigma$ as a section of the \emph{fixed} base $\Gamma^{\rm u}(\delta)$. This completes the proof of Theorem~\ref{ThmTrap}.

\bibliographystyle{alphaurl}


\end{document}